\newcommand{\blind}{1}
\newcommand{\bolding}{1}
\documentclass[authoryear,12pt]{article}
\usepackage[a4paper, total={6.5in, 10in}]{geometry}
\usepackage{amsthm, amsmath, amssymb, nccmath, xcolor, mathrsfs, mathtools}
\usepackage{bm}
\usepackage{pxfonts}  %esint (no support for \oiintclockwise)
\usepackage{microtype} % for overfull hbox warning reduction
\usepackage{setspace}
\usepackage{natbib}
\usepackage[export]{adjustbox}
\usepackage{import}
\usepackage{xifthen}
\usepackage{pdfpages}
\usepackage{transparent}
\usepackage{xr}
\usepackage[colorlinks=true,linkcolor=black, citecolor=blue, urlcolor=blue]{hyperref}
\usepackage{comment}
\usepackage{pdflscape}
\usepackage{rotating}
\usepackage{textcomp}
\usepackage{siunitx}
\usepackage[font=small]{caption}
\usepackage{fontenc,inputenc,subcaption,graphicx, titlesec, fancyhdr}
\usepackage{paralist,multirow,arydshln}%algorithm, 
\usepackage[ruled, vlined]{algorithm2e}
\usepackage{float}
\usepackage{titling} % used for multiple titles in same document
\usepackage{cleveref}

\newcommand{\myquad}[1][1]{\hspace*{#1em}\ignorespaces}

\makeatletter
\newcommand{\pushright}[1]{\ifmeasuring@#1\else\omit\hfill$\displaystyle#1$\fi\ignorespaces}
\newcommand{\pushleft}[1]{\ifmeasuring@#1\else\omit$\displaystyle#1$\hfill\fi\ignorespaces}
\makeatother

\begin{document}
	
	% \DeclareMathOperator*{\argmax}{arg\,max}
% \DeclareMathOperator*{\argmin}{arg\,min}
%\startlocaldefs
% my-commands
%\newcolumntype{C}[1]{>{\centering\arraybackslash}m{#1}}
\renewcommand{\T}{\mathrm{\scriptscriptstyle T} }
\newcommand{\prm}{^{\prime}} % prime
\newcommand{\dprm}{^{\prime\prime}} % double prime
\newcommand{\fim}[1]{\trans{\bX}#1\bX} % Fisher's Information Matrix
\newcommand{\dd}[1]{\mfrac {\mathrm{d}}{\mathrm{d}#1}} % differentiation
\newcommand{\del}[2]{\frac{\partial#1}{\partial #2}} % partial-differentiation-1
\newcommand{\deltwo}[1]{\mfrac{\partial^2}{\partial #1^2}} % partial-differentiation-2
\newcommand\scalemath[2]{\scalebox{#1}{\mbox{\ensuremath{\displaystyle #2}}}}
\newcommand{\beginsupplement}{%
        \setcounter{equation}{0}
        \renewcommand{\theequation}{S\arabic{equation}}
        \setcounter{section}{0}
        \renewcommand{\thesection}{S\arabic{section}}
        \setcounter{table}{0}
        \renewcommand{\thetable}{S\arabic{table}}%
        \setcounter{figure}{0}
        \renewcommand{\thefigure}{S\arabic{figure}}%
     }
% stat-local defs
\newcommand{\E}{\mathbb{E}}
\newcommand{\Cov}{{\rm Cov}}
\newcommand{\Var}{{\rm Var}}
\newcommand{\vecc}{{\rm vec}} % vectorization operator for matrices
\newcommand{\vech}{{\rm vech}} % half-vectorization operator for symmetric matrices
\newcommand{\diag}{{\rm diag}} % diagonal matrix
\newcommand{\tildeK}{\widetilde{K}} %covariance function
\newcommand{\nablaa}{\widetilde{\nabla}} %unique derivatives -- vech
\newcommand{\Partials}{\if1\bolding{\widetilde{\bm{\partial}_\bs}}\fi\if0\bolding{\widetilde{\partial}_\bs}\fi}
\newcommand{\Sigmaa}{\widetilde{\Sigma}} %covariance after marginalizing Z
\newcommand{\given}{\, | \;} % conditional probability
\newcommand{\mstar}{m^*}
\def\A{\mathscr{A}}
\def\C{\mathscr{C}}
\def\G{\if1\bolding{\bm{\mathscr{G}}}\fi\if0\bolding{\mathscr{G}}\fi}
\def\calK{\if1\bolding{\bm{\mathcal{K}}}\fi\if0\bolding{\mathcal{K}}\fi}
\def\L{\if1\bolding{\bm{\mathcal{L}}}\fi\if0\bolding{\mathcal{L}}\fi}
\def\M{\mathcal{M}}
\def\N{\mathcal{N}}
\def\P{\mathcal{P}}
\def\D{\mathcal{D}}
\def\Y{\if1\bolding{\bm{\mathcal{Y}}}\fi\if0\bolding{\mathcal{Y}}\fi}
\def\Z{\if1\bolding{\bm{\mathcal{Z}}}\fi\if0\bolding{\mathcal{Z}}\fi}
\def\R1{\mbox{$\mathfrak{R}$}}
\def\R2{\mbox{$\mathfrak{R}^2$}}
\def\R3{\mbox{$\mathfrak{R}^3$}}
\def\Rd{\mbox{$\mathfrak{R}^{d}$}}
\def\Rq{\mbox{$\mathfrak{R}^{q}$}}
% my theorems
\newtheorem{definition}{Definition}
\newtheorem{theorem}{Theorem}
\newtheorem{result}{Result}
\newtheorem{condition}{Condition}
\newtheorem{lemma}{Lemma}
\newtheorem{corollary}{Corollary}
\newtheorem{proposition}{Proposition}
\newtheorem{assumption}{Assumption}
\newtheorem{remark}{Remark}
% boldface: \usepackage{bm}
\def\bnabla{\if1\bolding{\bm{\nabla}}\fi\if0\bolding{\nabla}\fi}
\def\bpartial{\if1\bolding{\bm{\partial}}\fi\if0\bolding{\partial}\fi}
\def\0{\if1\bolding{\bm{0}}\fi\if0\bolding{0}\fi}
\def\1{\if1\bolding{\bm{1}}\fi\if0\bolding{1}\fi}
\def\ba{\if1\bolding{\bm{a}}\fi\if0\bolding{a}\fi}
\def\be{\if1\bolding{\bm{e}}\fi\if0\bolding{e}\fi}
\def\bu{\if1\bolding{\bm{u}}\fi\if0\bolding{u}\fi}
\def\bv{\if1\bolding{\bm{v}}\fi\if0\bolding{v}\fi}
\def\bw{\if1\bolding{\bm{w}}\fi\if0\bolding{w}\fi}
\def\bW{\if1\bolding{\bm{W}}\fi\if0\bolding{W}\fi}
\def\bx{\if1\bolding{\bm{x}}\fi\if0\bolding{x}\fi}
\def\bs{\if1\bolding{\bm{s}}\fi\if0\bolding{s}\fi}
\def\bDelta{\if1\bolding{\bm{\Delta}}\fi\if0\bolding{\Delta}\fi}
\def\bP{\if1\bolding{\bm{\mathrm{P}}}\fi\if0\bolding{P}\fi} % matrix
\def\bmu{\if1\bolding{\bm{\mu}}\fi\if0\bolding{\mu}\fi}
\def\btheta{\if1\bolding{\bm{\theta}}\fi\if0\bolding{\theta}\fi}
\def\bbeta{\if1\bolding{\bm{\beta}}\fi\if0\bolding{\beta}\fi}
\def\bgamma{\if1\bolding{\bm{\gamma}}\fi\if0\bolding{\gamma}\fi}
\def\bL{\if1\bolding{\bm{L}}\fi\if0\bolding{L}\fi}
\def\bV{\if1\bolding{\bm{\mathrm{V}}}\fi\if0\bolding{V}\fi} % matrix
\def\bn{\if1\bolding{\bm{n}}\fi\if0\bolding{n}\fi}
\def\m{\if1\bolding{\bm{m}}\fi\if0\bolding{m}\fi}
\def\bH{\if1\bolding{\bm{\mathrm{H}}}\fi\if0\bolding{H}\fi} % matrix
\def\bN{\if1\bolding{\bm{\mathrm{N}}}\fi\if0\bolding{N}\fi} % matrix
\def\bGamma{\if1\bolding{\bm{\Gamma}}\fi\if0\bolding{\Gamma}\fi}
\def\bF{\if1\bolding{\bm{F}}\fi\if0\bolding{F}\fi}
\def\bA{\if1\bolding{\bm{\mathrm{A}}}\fi\if0\bolding{A}\fi} % matrix
\def\bB{\if1\bolding{\bm{\mathrm{B}}}\fi\if0\bolding{B}\fi} % matrix
\def\bK{\if1\bolding{\bm{\mathrm{K}}}\fi\if0\bolding{K}\fi} % matrix
\def\bG{\if1\bolding{\bm{\mathrm{G}}}\fi\if0\bolding{G}\fi} % matrix
\def\bI{\if1\bolding{\bm{\mathrm{I}}}\fi\if0\bolding{I}\fi} % matrix
\def\bX{\if1\bolding{\bm{\mathrm{X}}}\fi\if0\bolding{X}\fi} % matrix
\def\bR{\if1\bolding{\bm{\mathrm{R}}}\fi\if0\bolding{R}\fi} % matrix
\def\bE{\if1\bolding{\bm{\mathrm{E}}}\fi\if0\bolding{E}\fi} % matrix
\def\bO{\if1\bolding{\bm{\mathrm{O}}}\fi\if0\bolding{O}\fi} % matrix
\def\bU{\if1\bolding{\bm{\mathrm{U}}}\fi\if0\bolding{U}\fi}
\def\bM{\if1\bolding{\bm{\mathrm{M}}}\fi\if0\bolding{M}\fi}
	
	\def\spacingset#1{\renewcommand{\baselinestretch}%
		{#1}\small\normalsize} \spacingset{1}
	
	\if1\blind
	{
		\title{\bf Bayesian Spatiotemporal Wombling}
		\author{Aritra Halder$^{a}$, Didong Li$^{b}$ and Sudipto Banerjee$^{c}$ 
			\\
			$^{a}$Department of Epidemiology \& Biostatistics,\\
			Drexel University,
			Philadelphia, PA, USA.\\
			$^{b}$Department of Biostatistics,\\
			University of North Carolina,
			Chapel Hill, North Carolina, USA.\\
			$^{c}$Department of Biostatistics,\\
			University of California,
			Los Angeles, California, USA.\\
		}
		\date{}
		\maketitle
	} \fi

	\if0\blind
	{
		\bigskip
		\bigskip
		\bigskip
		\begin{center}
			{\LARGE\bf Bayesian Spatiotemporal Wombling}
		\end{center}  
		\medskip
	} \fi
	
	\begin{abstract}
		\noindent Stochastic process models for spatiotemporal data underlying random fields find substantial utility in a range of scientific disciplines. Subsequent to predictive inference on the values of the random field (or spatial surface indexed continuously over time) at arbitrary space-time coordinates, scientific interest often turns to gleaning information regarding zones of rapid spatial-temporal change. We develop Bayesian modeling and inference for directional rates of change along a given surface. These surfaces, which demarcate regions of rapid change, are referred to as \emph{``wombling'' surface boundaries}. Existing methods for studying such changes have often been associated with curves and are not easily extendable to surfaces resulting from curves evolving over time. Our current contribution devises a fully model-based inferential framework for analyzing differential behavior in spatiotemporal responses by formalizing the notion of a ``wombling'' surface boundary using conventional multi-linear vector analytic frameworks and geometry followed by posterior predictive computations using triangulated surface approximations. We illustrate our methodology with comprehensive simulation experiments followed by multiple applications in environmental and climate science; pollutant analysis in environmental health; and brain imaging. 
	\end{abstract}
	% 168 words abstract (limit 200 words)
	
	
	\spacingset{1.9} 
	
	\section{Introduction}
	Statistical inference for rates of change in spatial random fields is instrumental in data-driven scientific discoveries. Such inference is pursued in diverse scientific investigations and can be traced back to the origins of evolutionary biology that have resulted in the discovery of lurking inexplicable differences \citep[see, e.g.,][on the ``The Wallace, Weber and Lydekker Lines"]{mayr1944wallace}. %,van1997worlds,ali2021wallace
	More generally, the study of spatial rates of change in the context of scientific experiments is often referred to as ``wombling'' \citep{womble1951differential, gleyze2001wombling}. ``Wombling'' broadly refers to delineating zones of rapid change on spatial surfaces usually defined in terms of locations with statistically significant directional gradients \citep[]{banerjee2003directional}, but also by testing paths or boundaries for rapid change on the spatial field \citep{banerjee2006bayesian, fitzpatrick2010ecological, qu2021boundary}. % can save one line here YES.

	The methods employed for ``wombling'' depend on whether the data features a continuous or discrete spatial indexing. For recent approaches to discrete (areal) data, see \cite{wu2026model}, and references therein.
	%``Wombling'' relies upon methods that adhere to the different base scales of spatial mapping for boundary analysis.  %While different base scales dictate the methods for spatial analysis, 
	{H}ere we will treat space and time as continuous when modeling the random field. We consider point-referenced spatiotemporal data, where variables of interest are mapped at locations within an Euclidean coordinate frame (e.g., using a planar map projection) and indexed over time as a realization of a smoothly varying stochastic process. %Smoothness considerations for the process are crucial in devising a legitimate probabilistic inferential framework for rates of change 
	Probabilistic inference on higher order rates of change relies on the smoothness of the parent process \citep[see, e.g.,][]{adler1981geometry, greenwood1984unified, kent1989continuity, morris1993bayesian, mardia1996kriging, stein1999interpolation, banerjee2003smoothness, guindani2006smoothness}. Formal inference on rates of change over spatial random fields enjoys a significant literature in diverse %modeling 
	scenarios % morris1993bayesian, mardia1996kriging, AH: Have removed these refs for page limit
	\citep[see, e.g.,][]{banerjee2003directional, majumdar2006gradients, liang2009bayesian, gabriel2011, heaton2014wombling, terres2015using, wang2016estimating, terres2016spatial, wang2018process, qu2021boundary}, including recent investigations into optimal plug-in models for gradients \citep[see, e.g.,][]{liu_optimal_2026} and inference on spatial curvature processes \citep[see, e.g.,][]{halder2024bayesian}. % AH: pasted the abstracts for your reference SB: Thanks That was the paper I was looking for.
	% References stand at 4.5 pages SB: yeah, probably need to reduce to at most 3 pages
	% I have removed some refs from the applications..sticking to one per application.
	
	This manuscript develops statistical inference for ``wombling'' boundaries within a hierarchical spatiotemporal modeling framework \citep{cressie2015statistics}. The literature on rates of change on spatiotemporal random fields is considerably sparser than for spatial random fields. \cite{quick2013modeling} explored temporal gradients in the context of dynamic Markov random fields, while \cite{quick2015bayesian} developed inference restricted to first-order and mixed spatial-temporal gradients on spatiotemporal random fields. More recently \cite{yu2023bayesian} discussed the utility of temporal rates of change to conduct inference on stationary points in Gaussian process regression models. 
	
	% Unlike in response surface analysis for computer experiments \citep{morris1993bayesian} or in ``kriging'' with derivative information \citep{mardia1996kriging}, ``wombling'' does not learn from derivative information; rather, rates of change are to be inferred over arbitrary space-time coordinates as well as on geographic features that may be delineating zones of rapid change (e.g., mountains, rivers, administrative boundaries).
	Wombling can be contrasted and compared with popular machine learning approaches to detecting clusters \citep[see, e.g.,][]{neal_markov_2000,birant_stdbscan_2007}, image segmentation \citep[see, e.g.,][]{ronneberger_unet_2015}, %he_deep_2016
	hot-spot detection \citep[see, e.g,][]{bivand_comparing_2018, di2018spatiotemporal} { and level set methods \citep[see, e.g.,][]{osher_level_2001,lie2006binary,boledi_level-set_2022}}. Wombling seeks predictive inference over curves (and surfaces) offering fully model-based uncertainty quantification leveraging rates of change. Clustering aims to find homogeneous regions in the reference domain, unable to offer inference over curves separating such regions. Image segmentation aims to locate boundaries in the interpolated surface images. It often requires careful and expensive training \citep[see, e.g,][]{akkus_deep_2017} and seldom provides uncertainty quantification. % Wombling can provide uncertainty quantification over such boundaries.
	Hot-spot detection techniques for point-referenced data are only able to track regions of rapid change. { Level set approaches leverage finite differences to detect boundaries but are unable to offer any uncertainty quantification on them.} Wombling yields statistical inference over arbitrary space-time coordinates as well as, on geographic features that may or may not be delineating zones of rapid change (e.g., mountains, rivers, administrative boundaries). 
	
	Within a spatiotemporal setting, curves tracking rapid spatial change evolve over time to produce surfaces posing challenges for wombling which has been typically associated to curves. Furthermore, the extent of change may not be captured by first-order derivatives. %We expand over these existing approaches in the following ways. 
	For example, the rate of change in neural activity may show substantial increments over time---{  persistently increasing patterns in event related potentials (ERPs) over time}, requiring inference on higher-order gradient processes like the spatiotemporal curvature. Here, we construct these processes by jointly modeling the parent Gaussian process and its derivatives to perform inference on higher-order rates of change for spatiotemporal fields. More pertinently for ``wombling'', % exercises,
	we seek to evaluate average higher-order rates of change not only at arbitrary space-time coordinates in a posterior predictive fashion, but we also extend inference to integrals over surfaces representing %static geographic entities, for e.g., a mountain range, as well as 
	temporally evolving zones that are posited to be harboring rapid change, for e.g., large pockets of smog traversing a geographic space. In this context, spatiotemporal random fields offer significant challenges over purely spatial settings as we develop inference for integrals over surfaces resulting from evolving boundaries in contrast to boundaries that remain static.
	
	The article evolves as follows. Section~\ref{sec:STDP} sets the stage for our methodological framework with an outline of the calculus of spatial-temporal derivative processes and the related distribution theory for legitimate probabilistic inference. Section~\ref{sec:spt-curv-womb} % aims for a comprehensive 
	develops measures of higher-order rates of change and their averages over { evolving boundaries} (Section~\ref{sec:spt-womb}). Section~\ref{sec:bi&c} develops a Bayesian inferential framework underscoring computation using surface triangulation to avoid high-dimensional quadrature. Sections~\ref{sec:sim}~and~\ref{sec:app} provide % extensive 
	simulation experiments and spatiotemporal data analysis to illustrate inference on rates of change at points and over surfaces arising from neuroimaging as ERPs in electroencephalography (EEG) sessions to detect predisposition to alcoholism. Applications featuring analysis of precipitation in northern California and PM\textsubscript{2.5} levels during the Canada wildfires are in the Supplement. Section~\ref{sec:diss} concludes with a discussion.
	
	% AH -> SB see end of pg. 19 & 24 for extra two/three line gap that is generated due to formatting
	
	\section{Spatiotemporal Derivative Processes}\label{sec:STDP} %	\subsection{Preliminaries}
	\subsection{Calculus for spatiotemporal derivatives}\label{subsec:calc-stdp}
	Let $\{Z(\bs,t):(\bs,t)\in \mathfrak{R}^d\times \mathfrak{R}^{+}\}$ be a spatial-temporal process, where $\bs \in \mathfrak{R}^d$ (usually $d=2$) denotes a spatial location and $t$ is the temporal reference on $\mathfrak{R}^+=[0,\infty)$. We require formal spatiotemporal derivative processes induced by $Z(\bs,t)$. Such derivatives can be purely spatial, temporal or jointly spatial and temporal underscoring the need for higher-order derivative processes to achieve probabilistic inference. The dimensions for such processes increase with the order of derivatives---the $k$th-order derivative is a $k$-fold tensor array---so we use conventional multi-linear algebra in developing higher-order derivatives% \citep[see, e.g.,][]{holmquist1996d}
	.
	
	Our subsequent developments rely heavily upon the gradient operator applied to $Z(\bs,t)$ to produce the $(d+1)\times 1$ vector of partial derivatives with respect to each coordinate, i.e., $\bnabla Z(\bs,t) := \left(\partial_{1} Z(\bs,t), \partial_{2} Z(\bs,t),\ldots,\partial_d Z(\bs,t), \partial_t Z(\bs,t)\right)^{\T} = \left(\bpartial_\bs Z(\bs,t)^{\T}, \partial_t Z(\bs,t)\right)^{\T}$, where $\partial_j$ denotes the partial derivative with respect to the $j$-th coordinate of $\bs$, $\bpartial_\bs Z(\bs,t)$ is the $d\times 1$ spatial gradient and $\partial_t$ is the partial derivative with respect to time. Following \cite{banerjee2003directional}, one can obtain spatial-temporal derivatives as limits of finite difference processes. We define the shift operator $D_{h,\bu,v}Z(\bs,t) = Z(\bs+h \bu,t + h v)$, where $\bu\in \mathfrak{R}^d$ is a unit vector representing a displacement in spatial location, $v\in [0,1]$ is a scaled temporal shift and $h$ is a scalar denoting the absolute value of the displacements. Then, $\bnabla := \lim_{h\to 0} h^{-1}\left((D_{h,\be_1,0} - 1),\ldots,(D_{h,\be_d,0} - 1),(D_{h,0,1}-1)\right)^{\T}$, where $\{\be_1,\ldots,\be_d\}$ are the $d$ orthonormal vectors in $\mathfrak{R}^d$. We will specifically use the tensor product $\bnabla^{\otimes 2} := \bnabla \otimes \bnabla$, which yields the $(d+1)^2\times 1$ vector $\bnabla \otimes (\bnabla Z(\bs,t))$ by applying $\bnabla$ to each element of $\bnabla Z(\bs,t)$. Higher-orders are defined recursively as $\bnabla^{\otimes r} := \bnabla \otimes (\bnabla^{\otimes (r-1)})$. The $(d+1)^r$ vector $\bnabla^{\otimes r}Z(\bs,t)$ contains all partial derivatives of order $r$ in both space and time. Section~S1 provides explicit constructions for these tensor products.      
	
	The process $Z(\bs,t)$ is {  $k$-th} order mean-squared differentiable if the elements of $\bnabla^{\otimes r}Z(\bs,t)$ are continuously differentiable for $r<k$, $\bnabla^{\otimes k}Z(\bs,t)$ are continuous and, for any $h\in \mathfrak{R}^+$, $Z(\bs+h \bu,t+h v)=\sum\limits_{r=0}^{k}\frac{h^r}{r!}\;(\bu^{\T}, v)^{\otimes r}\cdot\bnabla^{\otimes r} Z(\bs,t)+o(h^{k+1}||(\bu^{\T},v)^{\T}||^{k+1})$, % NOTE: giving a space before equation below reduces space overleaf puts before equations 
	%    \begin{equation}\label{eq:2}
		%        % --> In the Taylor expansion \oplus is not required since all gradients w.r.t space and time (first and higher order) appear in an additive fashion
		% Z(\bs+h \bu,t+h v)=\sum\limits_{r=0}^{k}\frac{h^r}{r!}\;(\bu^{\T}, v)^{\otimes r}\cdot\bnabla^{\otimes r} Z(\bs,t)+o(h^{k+1}||(\bu^{\T},v)^{\T}||^{k+1}),
		%    \end{equation}
	where $\bu$ is any vector of unit length in $\mathfrak{R}^{d}$, $v\in [0,1]$ and $(\bu^{\T}, v)^{\otimes r} = (\bu^{\T}, v) \otimes \cdots \otimes (\bu^{\T}, v)$ (Kronecker product with itself $r$ times) is $1\times (d+1)^r$. We refer to $Z(\bs,t)$ as the parent process and $\bnabla^{\otimes r}Z(\bs,t)$ as the { $r$-th} order derivative process induced by the parent. {\em Directional differential processes} are of interest when the change in $Z(\bs,t)$ is measured along a particular direction. For $r=1$ and $\bu\in \mathfrak{R}^{d}$, $\lim_{h\to 0} h^{-1}\left((D_{h,\bu,0} - 1) Z(\bs,t),(D_{h,0,v}-1) Z(\bs,t)\right)^{\T} = (\bu^{\T}\oplus v)\bnabla Z(\bs,t) = (\bu^{\T}\bpartial_\bs Z(\bs,t), v\partial_tZ(\bs,t))^{\T}$, where $\oplus$ is the direct sum operator and $\bu^{\T}\oplus v = \left(\begin{smallmatrix} \bu^{\T} & 0\\ \0_d^{\T} & v\end{smallmatrix}\right)$. More generally, the { $r$-th} order directional derivative process is $(\bu^{\T}\oplus v)^{\otimes r}\bnabla^{\otimes r} Z(\bs,t)$. 
	
	\subsection{Distribution Theory for Spatiotemporal Derivative processes}
	
	We assume that $Z(\bs,t)$ is a weakly stationary spatiotemporal process that is zero centered with finite second moment \citep{gneiting2010handbook}. We write $Z(\bs,t) \sim GP(0, K(\cdot,\cdot;\btheta))$ to denote a zero-centered Gaussian process (GP), where $K(\bDelta,\delta;\btheta) = \Cov(Z(\bs,t),Z(\bs',t'))$ is a stationary spatiotemporal covariance function indexed by process parameters $\btheta$ with $\bDelta=\bs-\bs'$ and $\delta=t-t'$. To ease notation, we suppress the dependence on $\btheta$ and simply write $ K(\cdot,\cdot;\btheta)= K(\cdot,\cdot)$. The covariance function satisfies $\Var\left\{\sum_{i=1}^{N} a_i Z( \bs_i,t_i)\right\}=\sum_{i=1}^{N}\sum_{j=1}^{N} a_i a_j K(\bDelta_{ij},\delta_{ij}) >0$ for any finite collection of spatiotemporal coordinates $\{(\bs_i,t_i) : i=1,\ldots,N\}$, where $\bDelta_{ij}=\bs_i-\bs_j$, $\delta_{ij}=t_i-t_j$ and $a_i,  a_j\in \mathfrak{R}$, $i,j=1,2,\ldots,N$ are not all zero. The process $Z(\bs,t)$ is mean-squared continuous at $(\bs,t)$ if $E\left[Z(\bs+\bDelta,t+\delta)-Z(\bs,t)\right]^2\to 0$ as $(\bDelta,\delta)\to \0_{d+1}$. Under isotropy, we have $K = K(\bDelta,\delta)=\tildeK(||\bDelta||,|\delta|)=\tildeK$. 
	
	%Our measurements are recorded over the set $\D=\{(s_i,t_i) : i=1,\ldots,N\} \subset \mathfrak{R}^d\times \mathfrak{R}^{+}$. We seek a predictive distribution for the vectorized process comprising the parent and the derivative processes of different orders. 
	Our subsequent developments rely upon joint processes consisting of the parent and derivatives of different orders. To that end, let $\bL_r:{\mathfrak R} \to {\mathfrak R}^{m_r}$ be the {\em differential operator} of order $r$, where $m_r = \sum\limits_{i=1}^{r}(d+1)^i$, which maps $Z(\bs,t)$ to the derivative processes {\em up to} order $r$. %(excluding possible repetitions for mixed derivatives).
	Thus, $\bL_r Z(\bs,t)=\left(\bnabla^{\otimes 1}Z(\bs,t)^{\T},\ldots,\bnabla^{\otimes r}Z(\bs,t)^{\T}\right)^{\T}$ is $m_r\times 1$. For $\bL_r Z(\bs,t)$ to be a valid process we require unique pure and partial derivative processes. %Furthermore, practical applications seldom require all partial and pure derivative processes of higher orders. 
	{  In general, $\bL_r Z(\bs,t)$ includes redundancies. For example, $\bL_2Z(\bs,t)$ includes $\frac{\partial^2}{\partial \bs\partial t}$ and $\frac{\partial^2}{\partial t\partial \bs}$, while $\bL_4Z(\bs,t)$ includes $\frac{\partial^4}{\partial \bs^3\partial t}$ which is not required.} We use elimination and permutation matrices to transform $\bL_rZ(\bs,t)$ to $\L_r Z(\bs,t):{\mathfrak R} \to {\mathfrak R}^{m_r^*}$, $m_r^*=\sum\limits_{i=0}^{r}\binom{i+d}{i}$, comprised of unique partial and pure derivatives. % needed for inference. %(Section~S2 offers details). 
	% only unique derivatives:: $m_r = \sum\limits_{i=0}^{r}\binom{i+d}{i}$
	
	Let $\bpartial_\bs^{r} = \partial^{r}/\partial_1^{i_1}\ldots\partial_d^{i_d}$ with $\sum_{k=1}^{d}i_k=r$ and $i_k\geq0$. If $\partial_t^{j_1+j_2}\bpartial_\bs^{r_1 + r_2-j_1-j_2} K(\bDelta, \delta)$ exists,
	then $\Cov\left(\partial_t^{j_1}\bpartial_\bs^{r_1-j_1}Z(\bs,t), \partial_t^{j_2}\bpartial_\bs^{r_2-j_2}Z(\bs + \bDelta,t + \delta)\right)=(-1)^{r_1}\partial_t^{j_1+j_2}\bpartial_\bs^{r_1 + r_2-j_1-j_2}K(\bDelta, \delta)$ is the covariance between the processes, $\partial_t^{j_1}\bpartial_\bs^{r_1-j_1}Z(\bs,t)$ and $\partial_t^{j_2}\bpartial_\bs^{r_2-j_2}Z(s+\bDelta,t+\delta)$ (see Section~S2). The partial derivatives operate element-wise on $K(\cdot,\cdot)$. % Section~S3 provides more detail.
	Also, $\L_rZ(\bs,t)$ is a valid GP if $\partial_t^{2j}\bpartial_\bs^{2(r-j)}K(\0_d,0)$ exists for every $j = 0,1,\ldots, r$.   %The cross-covariance matrix of the finite difference process approaches the cross covariance matrix for $\L_r$, $V_{\L_r}$ in limit. The derivatives $\nabla_t^{2j}\nabla_s^{2(r-j)}K(0_d,0)$ must exist for $j = 0,1,\ldots,r$ to ensure that all entries in $V_{\L_r}$ to be well-defined. 
	In particular, the case $j=r-j=2$ requires $\partial^4_t\bpartial_\bs^4K(\0_d,0)$ to exist. A GP assumption on $Z(\bs,t)$ has immediate consequences. If $Z_1(\bs,t)\sim GP(0, K_1(\cdot,\cdot))$ and $Z_2(\bs,t) \sim GP(0, K_2(\cdot,\cdot))$ independently, then $\L_{r_1} Z_1(\bs,t)$ and $\L_{r_2} Z_2(\bs,t)$ are independent. For $a_1, a_2\in \mathfrak{R}$, $\L_r (a_1Z_1(\bs,t)+a_2Z_2(\bs,t))=a_1\L_r Z_1(\bs,t)+a_2\L_rZ_2(\bs,t)$ is a GP. Any sub-vector of $\L_r Z(\bs,t)$ is a GP. %The first two properties specify algebraic rules for the geometry of two independent spatiotemporal stationary processes. 
	This is useful if selected derivative processes are required for inference as is seen in the next paragraph. In addition, inference extends to differential geometric operators like the divergence and the Laplacian (see Section~S3).
	% Interpreting higher order spatial-temporal derivative processes can be challenging in practical applications. \cite{quick2015bayesian} considers $\L_1Z(\bs,t) = \left(Z(\bs,t), \nabla_s Z(\bs,t)^{\T}\right)^{\T}$, developing distribution theory for $\widetilde{\L_1}Z(\bs,t) = \left(\L_1Z(\bs,t),\nabla_t\L_1Z(\bs,t)^{\T}\right)^{\T}$ using only gradients to characterize spatiotemporal change. \cite{halder2024bayesian} studies surface geometry within a spatial context through, $\L_2Z(s) = \left(Z(s), \nabla_s Z(s)^{\T},\nablaa_s^2Z(s)^{\T}\right)^{\T}$, leveraging both spatial gradients and curvature. Here, we combine these approaches to study and characterize spatiotemporal change jointly using spatial and temporal gradient and curvature processes. Consequently, selected derivatives of orders $r\leq4$ are of interest to us. %We denote the spatiotemporal $r$-order derivative processes as $\nabla^{r}Z(\bs,t) = \nabla^{r_t}_t\nabla^{r_s}_sZ(\bs,t)$, where $r_t+r_s=r$. 
	% \begin{equation}\label{eq:spt-diff}
		%     \L^* Z(\bs,t)=\left(\L_sZ(\bs,t)^{\T}, \partial_t\L_sZ(\bs,t)^{\T}, \partial_t^2\L_sZ(\bs,t)^{\T}\right)^{\T},
		% \end{equation}
	
	We estimate $\L^* Z(\bs,t)=\left(\L_\bs Z(\bs,t)^{\T}, \partial_tZ(\bs,t),\partial_t\L_\bs Z(\bs,t)^{\T}, \partial_t^2Z(\bs,t),\partial_t^2\L_\bs Z(\bs,t)^{\T}\right)^{\T}$, where $\L_\bs Z(\bs,t) = \left(\bpartial_\bs Z(\bs,t)^{\T},\Partials^2Z(\bs,t)^{\T}\right)^{\T}$ and $\Partials^2$ is the operator producing the $\binom{d}{2}+d$ vector of unique second order spatial derivatives. The elements of $\L^* Z(\bs,t)$ capture spatial and temporal rates of change in $Z(\bs,t)$. The process $\L_\bs Z(\bs,t)$ captures spatial gradients and curvature in $Z(\bs,t)$ while, $\partial_t \L_\bs Z(\bs,t)$ and $\partial_t^2 \L_\bs Z(\bs,t)$ capture temporal gradients and curvature in $\L_\bs Z(\bs,t)$ respectively. Casting $\L^*$ as a spatiotemporal differential operator, $\L^*:\mathfrak{R}\to \mathfrak{R}^{\mstar}$, $\mstar=3\left(1+d+\frac{d(d+1)}{2}\right)-1$. $\left(\begin{smallmatrix}
		Z(\bs,t)\\\L^*Z(\bs,t)
	\end{smallmatrix}\right)$ is a zero-centered stationary process. Its cross-covariance function $\bV_{Z, \L^* Z}(\bDelta,\delta) = \Cov\left\{\left(\begin{smallmatrix}
		Z(\bs,t)\\\L^*Z(\bs,t)
	\end{smallmatrix}\right), \left(\begin{smallmatrix}
		Z(\bs',t')\\\L^*Z(\bs',t')
	\end{smallmatrix}\right)\right\} %= [\Cov\{(\L^*Z(\bs,t))_l, (\L^*Z(\bs',t'))_{l'}\}]
	$ is a matrix-valued function from $\mathfrak{R}^{d}\times \mathfrak{R}$ to $\mathfrak{R}^{(\mstar + 1)\times (\mstar + 1)}$ with $(l,l')$-th element being the covariance between the $l$-th element of $\left(\begin{smallmatrix}
		Z(\bs,t)\\\L^*Z(\bs,t)
	\end{smallmatrix}\right)$ and $l'$-th element of $\left(\begin{smallmatrix}
		Z(\bs',t')\\\L^*Z(\bs',t')
	\end{smallmatrix}\right)$. The cross-covariance $\bV_{Z,\L^*Z}(\bDelta,\delta)$ is expressible as a block matrix in terms of $K(\bDelta,\delta)$ and its derivatives with blocks % NOTE: giving a space before equation below reduces space overleaf puts before equations  
	$   \left(\begin{smallmatrix}
		\partial_t^{i-1}\partial_t^{j-1}K(\bDelta,\delta) & \partial_t^{i-1}\partial_t^{j-1}\bpartial_\bs K(\bDelta,\delta)^{\T} & \partial_t^{i-1}\partial_t^{j-1}\Partials^2K(\bDelta,\delta)^{\T}\\
		-\partial_t^{i-1}\partial_t^{j-1}\bpartial_\bs K(\bDelta,\delta) & -\partial_t^{i-1}\partial_t^{j-1}\bH_K^{11}(\bDelta,\delta) & -\partial_t^{i-1}\partial_t^{j-1}\bH_K^{12}(\bDelta,\delta)\\ 
		\partial_t^{i-1}\partial_t^{j-1}\Partials^2K(\bDelta,\delta) & \partial_t^{i-1}\partial_t^{j-1}\bH_K^{21}(\bDelta,\delta) & \partial_t^{i-1}\partial_t^{j-1}\bH_K^{22}(\bDelta,\delta)
	\end{smallmatrix}\right)$, $i,j=1,2,3$.
	% NOTE:: can be done inline
	% \begin{equation*}%\label{eq:spt-cov-block}
		%   \left(\begin{array}{ccc}
			%          \partial_t^{i-1}\partial_t^{j-1}K(\bDelta,\delta) & \partial_t^{i-1}\partial_t^{j-1}\bpartial_\bs K(\bDelta,\delta)^{\T} & \partial_t^{i-1}\partial_t^{j-1}\Partials^2K(\bDelta,\delta)^{\T}\\
			%          -\partial_t^{i-1}\partial_t^{j-1}\bpartial_\bs K(\bDelta,\delta) & -\partial_t^{i-1}\partial_t^{j-1}\bH_K^{11}(\bDelta,\delta) & -\partial_t^{i-1}\partial_t^{j-1}\bH_K^{12}(\bDelta,\delta)\\ 
			%          \partial_t^{i-1}\partial_t^{j-1}\Partials^2K(\bDelta,\delta) & \partial_t^{i-1}\partial_t^{j-1}\bH_K^{21}(\bDelta,\delta) & \partial_t^{i-1}\partial_t^{j-1}\bH_K^{22}(\bDelta,\delta)
			%     \end{array}\right),\quad i,j=1,2,3 
		% \end{equation*}
	where $\partial_t$ acts element-wise on matrices. The entries $\bH_K^{11}(\bDelta,\delta)= \Var(\bpartial_\bs Z(\bs,t))$, $\bH_K^{12}(\bDelta,\delta)=\Cov(\bpartial_\bs Z(\bs,t),\Partials^2 Z(\bs',t')) = {\bH_K^{21}(\bDelta,\delta)}^{\T}$ and $\bH_K^{22}(\bDelta,\delta) = \Var(\Partials^2 Z(\bs,t))$ (see Section~S4).
	% \citep[see, ][eq. (1)]{halder2024bayesian}. 
	
	% {\color{red}AH: need to make distinction between $\bpartial_\bs$ and $\partial_s$ here. $d_{s}=\frac{\partial}{\partial||\Delta||}$. Same with $\partial_t$!!}    
	
	% {  
		Under isotropy, we have some simplifications when obtaining the entries of $\bV_{Z,\L^*Z}(\bDelta,\delta)$. We write the selected entries as follows: $\bpartial_\bs  K = \left(\frac{d \tildeK}{d||\bDelta||}\right)  \frac{\bDelta}{||\bDelta||}$, $\partial_t K =\left(\frac{d \tildeK}{d|\delta|}\right) \frac{\delta}{|\delta|}$ and $\partial_t^2\Partials^2 K= \left(\frac{d^2}{d|\delta|^2}\frac{d}{d||\bDelta||} \tildeK\right)\frac{\bP_{2,1}}{||\bDelta||}+\left(\frac{d^2}{d|\delta|^2}\frac{d^2}{d||\bDelta||^2}\tildeK-\frac{1}{||\bDelta||}\;\frac{d^2}{d|\delta|^2}\frac{d}{d||\bDelta||}\tildeK\right)\frac{\bP_{2,2}}{||\bDelta||^2}$ where % and $\partial_t^4\bpartial_\bs^4K=\left(\frac{d^4}{d|\delta|^4}\frac{d^2}{d||\bDelta||^2}\tildeK-\frac{1}{||\bDelta||}\;\frac{d^4}{d|\delta|^4}\frac{d}{d||\bDelta||}\tildeK\right)\left(-3\frac{\bP_{4,1}}{||\bDelta||^4}+\frac{\bP_{4,2}}{||\bDelta||^2}+15\frac{\bP_{4,3}}{||\bDelta||^6}-3\frac{\bP_{4,4}}{||\bDelta||^4}\right)+\left(\frac{d^4}{d|\delta|^4}\frac{d^3}{d||\bDelta||^3}\tildeK\right)\left(\frac{\bP_{4,1}}{||\bDelta||^3}-6\frac{\bP_{4,3}}{||\bDelta||^5}+\frac{\bP_{4,4}}{||\bDelta||^3}\right)+ \left(\frac{d^4}{d|\delta|^4}\frac{d^4}{d||\bDelta||^4}\tildeK\right)\frac{\bP_{4,3}}{||\bDelta||^4}$, where 
		%$d_t=\frac{\partial}{\partial|\delta|}$, $d_s=\frac{\partial}{\partial||\bDelta||}= \frac{\partial}{\partial\bDelta}\left(\frac{\partial ||\bDelta||}{\partial\bDelta}\right)^{-1}$ and 
		$\bP_{\cdot,\cdot}$ are permutation matrices (see Section~S5.1 for details). Simplifications owing to isotropy also arise in the expressions for the covariance of directional spatiotemporal derivative processes. For instance, the variance for the directional spatiotemporal curvature process has the following expression: $\frac{3}{||\bDelta||^2}\left(1-5\frac{(\bu^{\T}\bDelta)^2}{||\bDelta||^2}\right)\left(1-\frac{(\bu^{\T}\bDelta)^2}{||\bDelta||^2}\right)\left(\frac{d^4}{d|\delta|^4}\frac{d^2}{d||\bDelta||^2}\tildeK-\frac{1}{||\bDelta||}\frac{d^4}{d|\delta|^4}\frac{d}{d||\bDelta||}\tildeK\right)+\frac{6}{||\bDelta||}\frac{(\bu^{\T}\bDelta)^2}{||\bDelta||^2}\left(1-\frac{(\bu^{\T}\bDelta)^2}{||\bDelta||^2}\right)\left(\frac{d^4}{d|\delta|^4}\frac{d^3}{d||\bDelta||^3}\tildeK\right)+\left(\frac{\bu^{\T}\bDelta}{||\bDelta||}\right)^4\left(\frac{d^4}{d|\delta|^4}\frac{d^4}{d||\bDelta||^4}\tildeK\right)$.
		% \begin{equation}\label{eq:isotropic-cov}
			%    \begin{split}
				%    &\frac{3}{||\Delta||^2}\left(1-5\frac{(u^{\T}\Delta)^2}{||\Delta||^2}\right)\left(1-\frac{(u^{\T}\Delta)^2}{||\Delta||^2}\right)\left(\partial_t^4\partial_s^2\tildeK(||\Delta||,|\delta|)-\frac{\partial_t^4\partial_s\tildeK(||\Delta||,|\delta|)}{||\Delta||}\right)+\\
				%    &\hspace{.5cm}\frac{6}{||\Delta||}\frac{(u^{\T}\Delta)^2}{||\Delta||^2}\left(1-\frac{(u^{\T}\Delta)^2}{||\Delta||^2}\right)\partial_t^4\partial_s^3\tildeK(||\Delta||,|\delta|)+\left(\frac{u^{\T}\Delta}{||\Delta||}\right)^4\partial_t^4\partial_s^4\tildeK(||\Delta||,|\delta|).
				%    \end{split}
			% \end{equation}
		Similar results can be obtained for all spatiotemporal directional derivatives (see Section~S5.2). % These extend the results obtained in \cite{halder2024bayesian} to a space-time setting. %The foregoing developments show 
		We observe that $\L^* Z(\bs,t)$ is no longer isotropic. %} 
	
	We use spectral theory to characterize covariance functions that admit such processes. The existence of fourth-order spectral moments for space and time ensures admittance of the differential processes within $\L^*Z(\bs,t)$. We adopt non-separable spatiotemporal covariance functions of the form, $\displaystyle K(\bDelta,\delta)=\tildeK(||\bDelta||, |\delta|)=\frac{\sigma^2}{\psi(|\delta|^2)^{d/2}}\varphi\left(\frac{||\bDelta||^2}{\psi(|\delta|^2)}\right)$, where $\varphi(x)$ and $\psi(x)$ are complete monotone, and positive with complete monotone derivative on $[0,\infty)$, respectively \citep[see, e.g.][]{gneiting2002nonseparable}. Specifically, we use the Mat\'ern kernel, $\varphi(x)=\left(2^{\nu-1}\Gamma(\nu)\right)^{-1}(\phi_sx^{1/2})^{\nu}K_\nu(\phi_s x^{1/2})$, $\phi_s>0, \nu>0$, where $K_{\nu}(\cdot)$ is the modified Bessel function of the second kind, and $\psi(x)=(\phi_t^2x^\alpha+1)^\beta$, $a>0$, $0<\alpha\leq 1$, $0\leq\beta\leq1$. We fix $\alpha=\beta=1$, producing the covariance kernel, 
	
	\begin{equation}\label{eq:cov-matern}
		\tildeK(||\bDelta||,|\delta|;\btheta)=\frac{\sigma^2}{2^{\nu-1}\Gamma(\nu)A_t^{d/2}}\left(\frac{\phi_s||\bDelta||}{A_t^{1/2}}\right)^\nu K_{\nu}\left(\frac{\phi_s||\bDelta||}{A_t^{1/2}}\right), \quad A_t = \phi_t^2|\delta|^2+1.
	\end{equation}
	Existence of the fourth spectral moment requires the fractal parameter, $\nu>2$. Although irregularities exist in smoothness for separable covariance functions \citep[see, e.g.,][]{stein2005space}, characterizing smoothness for such functions is relatively simpler. Section~S6 provides further details. Establishing the validity of $\L^* Z(\bs,t)$, requires no further assumptions.
	
	Inference on $\L^*Z(\bs_0,t_0)$ at an arbitrary space-time location $(\bs_0,t_0)$ is sought through the joint distribution, $P[\Z, \L^*Z(\bs_0,t_0)]$, where $\Z = (Z(\bs_1,t_1), \ldots, Z(\bs_N, t_N))^{\T}$ is an $N\times 1$ vector of space-time measurements. Let $\Sigma_{\Z}= [\Cov\{Z(\bs_i,t_i), Z(\bs_{i'},t_{i'})\}] = [K(\bDelta_{ii'},\delta_{ii'})]$ be the $N\times N$ covariance matrix associated with our measurements where, $\bDelta_{ii'}=\bs_i-\bs_{i'}$, $\delta_{ii'}=t_i-t_{i'}$, $i,i' = 1,\ldots, N$. Let $\calK_0 = \Cov(\Z,\L^* Z(\bs_0,t_0)) = \left(\begin{smallmatrix}
		\calK_{10}(\bDelta_{10},\delta_{10})^{\T}\\\vdots\\\calK_{N0}(\bDelta_{N0},\delta_{N0})^{\T}
	\end{smallmatrix}\right)$ be a $N\times \mstar$ matrix with columns $\calK_{i0}(\bDelta_{i0},\delta_{i0}) = \Cov\{Z(\bs_i,t_i),\L^*Z(\bs_0,t_0)\}$ for $i = 1,\ldots, N$. The % required 
	joint distribution is % NOTE: giving a space before equation below reduces space overleaf puts before equations
	
	\begin{equation}\label{eq:diffp-full-post}
		\left(\begin{array}{c}\Z\\\L^*Z(\bs_0,t_0)\end{array}\right) \sim \N_{N+\mstar}\left(\0_{N+\mstar}, \left(\begin{array}{cc}
			\Sigma_{\Z} &  {\calK}_0\\
			{{\calK}_0}^{\T} & \bV_{\L^* Z}(\0_d,0)
		\end{array}\right)\right),
	\end{equation}
	where $\bV_{\L^* Z}(\0_d,0)$ is the cross-covariance matrix of $\L^* Z(\bs,t)$ evaluated at $(\0_d,0)$. \Cref{sec:bi&c} develops posterior predictive inference for $\L^* Z(\bs_0,t_0)$.
	
	\section{Wombling on Spatiotemporal Surfaces}\label{sec:spt-curv-womb}
	% \ahal{Some connection to what is mentioned in \cite{banerjee2018comments}!}
	
	Bayesian spatiotemporal wombling requires inference for integrals over surfaces that enclose regions of rapid change in space-time coordinates. Such surfaces {  are defined} from planar curves tracking rapid change in space and across time. In traditional wombling, which hitherto has been purely spatial, the curves remain fixed over time (e.g., geographic features like rivers, mountain ranges or administrative boundaries) and we seek to measure change along these curves. % We refer to this as \emph{static wombling}. %m Alternatively, investigators may be interested, for example, in detecting rapid change in an exposure surface over a month. 
	Alternatively, the regions of interest can vary over time, e.g., {  scalp regions with ERPs over a certain threshold} or, tracking smog arising from wildfires (see the Supplement, S15.2) in the interest of public health, regions evolve with changing wind patterns. This will be the focus of our discussion. % We call this \emph{dynamic wombling}. %We begin by defining measures with a surface that discern whether it forms a wombling boundary in the spatial-temporal domain. 
	% In \Cref{sec:param-surf} we formally develop the average spatial-temporal derivative and curvature for a surface. This is followed by spatiotemporal wombling for curves on surfaces in \Cref{sec:curves}. Static and dynamic wombling on surfaces are discussed in \Cref{,sec:stat-womb,sec:dyn-womb}, respectively.
	
	We focus on real valued processes, $Z(\bs,t)$, over $\mathfrak{H}^3=\mathfrak{R}^2\times \mathfrak{R}^+\subset \mathfrak{R}^3$ using familiar concepts in differential geometry \citep[see, e.g.,][]{spivak1979comprehensive} for the ensuing developments. We explore spatiotemporal derivatives and curvatures over smooth and regular surfaces with boundaries. This extends the development of derivatives and curvatures over planar curves required for curvilinear wombling within a purely spatial context \citep[see, e.g.,][]{banerjee2006bayesian, halder2024bayesian} to surfaces in space-time.
	
	\subsection{Wombling Measures for Surfaces}\label{sec:param-surf}
	In general, we assume that $\C$ is a surface in $\mathfrak{H}^3$. Our aim is to statistically estimate %if $\C$ forms a wombling boundary with respect to $Z(\bs,t)$. This requires associating average directional differential processes with $\C$. 
	the average directional gradient of $Z(\bs,t)$ along $\C$ from the realized data. We define such measures of interest on parametric surfaces. Let $\C=\left\{(s_x(\omega,\upsilon),s_y(\omega,\upsilon),t(\upsilon)):(\omega,\upsilon) \in \mathcal{D}_{\omega}\times \mathcal{D}_\upsilon\right\}$ be a smooth and regular surface % \citep[see, e.g.,][p. 418 and p. 80 respectively]{baxandall1986vector,do2016differential}
	with domain, $\mathcal{D}_{\omega}\times \mathcal{D}_\upsilon\subset \mathfrak{H}^2 = \mathfrak{R}\times \mathfrak{R}^+$ and let $t(\upsilon) = \upsilon$ for convenience. The \emph{outward} pointing normal and unit normal to $\C$ at $(\omega,\upsilon)$ are
	
	\begin{equation}\label{eq:normal}
		\overline{\bn}(\bs(\omega,\upsilon), \upsilon)=\left(\frac{\partial s_y(\omega,\upsilon)}{\partial \omega},-\frac{\partial s_x(\omega,\upsilon)}{\partial \omega}, \left|\frac{\partial(s_x, s_y)}{\partial(\omega,\upsilon)}\right|\right)^{\T}, \quad \bn(\bs(\omega,\upsilon),\upsilon) = \frac{\overline{\bn}(\bs(\omega,\upsilon), \upsilon)}{||\overline{\bn}(\bs(\omega,\upsilon), \upsilon)||}\;,
	\end{equation}
	respectively, where $\left|\frac{\partial(s_x, s_y)}{\partial(\omega,\upsilon)}\right| = \frac{\partial s_x(\omega,\upsilon)}{\partial\omega}\frac{\partial s_y(\omega,\upsilon)}{\partial\upsilon}-\frac{\partial s_y(\omega,\upsilon)}{\partial\omega}\frac{\partial s_x(\omega,\upsilon)}{\partial\upsilon}$. Let $\overline{\bn}_\bs(\omega,\upsilon) = \left(\frac{\partial s_y(\omega,\upsilon)}{\partial \omega},-\frac{\partial s_x(\omega,\upsilon)}{\partial \omega}\right)^{\T}$ and $\overline{n}_t(\omega,\upsilon) = \left|\frac{\partial(s_x, s_y)}{\partial(\omega,\upsilon)}\right|$. If $s_x$, $s_y$ are injective functions and all partial derivatives in \cref{eq:normal} exist, then the surface $\C$ is smooth. Also, $\C$ is regular if $||\overline{\bn}_\bs(\omega,\upsilon)||\ne0$ and the unit normal $\bn(\bs(\omega,\upsilon),\upsilon)$ in \cref{eq:normal} is well-defined. % The equation of the tangent plane at $(\omega,\upsilon)$ is $\left\{(x,y,z):n(s(\omega,\upsilon), \upsilon)^{\T}\left(\begin{smallmatrix}
		%      x-s_x(\omega,\upsilon)\\
		%     y-s_y(\omega,\upsilon)\\
		%     z-\upsilon
		% \end{smallmatrix}\right)=0\right\}$.
	We denote the surface area of $\C$, as $\A(\C)$ given by
	
	\begin{equation}\label{eq:fff} % first fundamental form for \C
		\iint\limits_{\mathcal{D}_\upsilon\times\mathcal{D}_\omega}||\overline{\bn}(\bs(\omega,\upsilon),\upsilon)||\;d\omega\;d\upsilon =  \iint\limits_{\mathcal{D}_\upsilon\times\mathcal{D}_\omega} \left\{\left(\frac{\partial s_y(\omega,\upsilon)}{\partial \omega}\right)^2+\left(\frac{\partial s_x(\omega,\upsilon)}{\partial \omega}\right)^2 + \left|\frac{\partial(s_x, s_y)}{\partial(\omega,\upsilon)}\right|^2\right\}^{\frac{1}{2}}\;d\omega\;d\upsilon.
	\end{equation}
	If $\mathcal{D}_\upsilon=[\upsilon_0,\upsilon_1]$ and $\D_{\omega}=[\omega_0,\omega_1]$, then $\A(\C)=\A_{\omega_0,\upsilon_0}(\omega_1,\upsilon_1)=\int\limits_{\upsilon_0}^{\upsilon_1}\int\limits_{\omega_0}^{\omega_1}||\overline{\bn}(\bs(\omega,\upsilon),\upsilon)||\;d\omega\;d\upsilon$ and, hence, ${\rm d}\A_{\upsilon_0,\omega_0}(\omega,\upsilon)=||\overline{\bn}(\bs(\omega,\upsilon),\upsilon)||\;d\omega\;d\upsilon$. We denote, $n_t(\omega,\upsilon)=n_t$ and $\bn_\bs(\omega,\upsilon)=\bn_\bs$ for brevity of ensuing expressions. For any point $(\bs(\omega,\upsilon),\upsilon)$ on $\C$, let $\L^*_{n_t,\bn_\bs}Z(\bs(\omega,\upsilon),\upsilon) = \bN_{\bs t}\L^*Z(\bs(\omega,\upsilon),\upsilon)$, be the $8\times 1$ directional differential process along the \emph{normal direction matrix}, $\bN_{\bs t}= \left(\widetilde{\bn}_{\bs,-1} \oplus n_t \widetilde{\bn}_\bs \oplus n_t^2 \widetilde{\bn}_\bs\right)$, where $\widetilde{\bn}_\bs = 1 \oplus \widetilde{\bn}_{\bs,-1}$, $\widetilde{\bn}_{\bs,-1}=\bn_\bs^{\T}\oplus \{\widetilde{\bn_\bs}^{\otimes 2}\}^{\T}$ and $\widetilde{\bn_\bs}^{\otimes 2} = \left(\begin{smallmatrix}
		1 & 0 & 0 & 0\\0 & 1& 1& 0\\0 & 0 & 0 & 1
	\end{smallmatrix}\right)\bn_\bs^{\otimes 2}$. We define \emph{wombling measures} as the total and average differential over $\C$,  
	
	\begin{equation}\label{eq:st-womb-measure}
		\bGamma(\C) %= \iint_{\C}\L^*_{n_t,n_s}^{-Z}Z(\bs,t)\;{\rm d}\A 
		= \iint_{\C}\L^*_{n_t,\bn_\bs}Z(\bs(\omega,\upsilon),\upsilon)\;||\overline{\bn}(\bs(\omega,\upsilon),\upsilon)||\;d\omega\;d\upsilon 
		\quad \mbox{and}\quad \overline{\bGamma}(\C) = \frac{\bGamma(\C)}{\A(\C)}\;, 
	\end{equation} 
	%    $\Gamma(\C) = \iint_{\C}\L^*_{n_t,n_s}^{-Z}Z(\bs,t)\;{\rm d}\A = \iint_{\C}\L^*_{n_t(\omega,\upsilon),n_s(\omega,\upsilon)}^{-Z}Z(s(\omega,\upsilon),\upsilon)\;||\overline{n}(s(\omega,\upsilon),\upsilon)||\;d\omega\;d\upsilon$ and $ \overline{\Gamma}(\C) = \frac{\Gamma(\C)}{\A(\C)}$ 
	respectively, where the double integral acts element-wise% and $\A(\cdot)$ is an appropriate measure. If $\A$ is the three dimensional Lebesgue measure, then $\A(\C)=0$ and $\overline{\Gamma}(\C)$ is undefined. % any R^{n-1} dimensional subspace has 0 n-dimensional Lebesgue meaure.For subsequent developments, it suffices to let $\A$ be the area measure
	. %Parametric surfaces offer further insights.   
	Therefore, $\bGamma$ and $\overline{\bGamma}$ are $8\times 1$ vectors. A large value for any entry of $\bGamma(\C)$ (or $\overline{\bGamma}(\C)$) indicates that $\C$ tracks a zone of rapid spatiotemporal change. More generally, $\L^*_{n_t,\bn_\bs} Z(\bs,t)$ can be replaced with, $g\left\{\L^* Z(\bs,t)\right\}$, where $g(\cdot)$ is any linear functional. Section~S7 offers more details on parameterized surfaces. 
	
	Specifically, we are interested in surfaces with a boundary (in the topological sense). In the subsequent examples, the boundary for $\C$ are open or closed parametric planar curves, $C_{\upsilon_i}=\{(s_x(\omega,\upsilon_i),s_y(\omega,\upsilon_i)):\omega\in[\omega_0,\omega_1]\}$, $i=0,1$. If, in particular, $C_{\upsilon_0}=C_{\upsilon_1}=C_\upsilon=C$, for every $\upsilon\in[\upsilon_0,\upsilon_1]$, the curve stays fixed over time so $s_x$ and $s_y$ are functions of $\omega$ only. 
	%then the resulting surface, $\C=\{ (s_x(\omega),s_y(\omega),\upsilon):(\omega,\upsilon) \in [\omega_0,\omega_1]\times [\upsilon_0,\upsilon_1]\}$. 
	From \cref{eq:fff}, $\A(\C) =  (\upsilon_1-\upsilon_0)\int\limits_{\omega_0}^{\omega_1} \left\{\left(\frac{\partial s_y(\omega)}{\partial \omega}\right)^2+\left(\frac{\partial s_x(\omega)}{\partial \omega}\right)^2\right\}^{\frac{1}{2}}\;d\omega=(\upsilon_1-\upsilon_0)\;\ell(C)$, where $\ell(C)$ is the arc-length of the parametric curve $C=\{(s_x(\omega),s_y(\omega)):\omega\in[\omega_0,\omega_1]\}$. % which ranges over the time interval, $[\upsilon_0,\upsilon_1]$ to eventually produce $\C$. % We discuss spatiotemporal wombling for planar curves on $\C$ in the next discussion before describing static wombling. 
	% NOTE: If $C$ is an arc, the above is a special case of the Pappus' Theorem for surfaces of revolution, see p. 467, exercise 4(a), Baxandall & Liebeck  
	
	\subsection{Curves on surfaces}\label{sec:curves}
	Curves of interest on $\C$ are \emph{planar curves}. They mark spatial regions of rapid change at time, $\upsilon=\upsilon'$. Let $C_{\upsilon'}=\{(s_x(\omega,\upsilon'),s_y(\omega,\upsilon')):\omega\in[\omega_0,\omega_1], \omega_0\ne\omega_1\}$ denote an open parametric planar curve. The wombling measures in \cref{eq:st-womb-measure} are defined with respect to the normal, $\bn_\bs(\omega,\upsilon')$, to $C_{\upsilon'}$. One could, in principle, consider defining wombling measures with respect to the tangential direction to $C_{\upsilon'}$. The tangent is $\overline{\bu}(\omega,\upsilon')=\left(\frac{\partial s_x(\omega,\upsilon')}{\partial\omega},\frac{\partial s_y(\omega,\upsilon')}{\partial\omega}\right)^{\T}$, with corresponding unit tangent, $\bu(\omega,\upsilon')=\frac{\overline{\bu}(\omega,\upsilon')}{||\overline{\bu}(\omega,\upsilon')||}$. However, wombling boundaries are, in general, representative of geographic features that delineate zones of rapid change. This suggests that such boundaries should be identified based upon rates of change orthogonal, and not tangential, to the curve. Furthermore, replacing normals with tangents in \cref{eq:st-womb-measure} for closed curves yields the value of $0$ in each element of $\bGamma(\C)$. For example, if $C_{\upsilon'}$ is an administrative or natural boundary of a closed region (e.g., zip-code, county, boundary of water body), then the corresponding wombling measures along the tangential direction will be zero and, hence, lack interpretation. On the other hand, the wombling measures defined with normals, as in \cref{eq:st-womb-measure}, for closed curves correspond to the notion of \emph{spatiotemporal flux} associated with the enclosed region \citep[see, e.g.,][for spatial flux]{banerjee2006bayesian}.    
	% The average spatiotemporal curvature over $C_{\upsilon_*}$ along $u(\omega,\upsilon_*)$ is $\frac{1}{\ell(C_{\upsilon_*})}\int_{\omega_0}^{\omega_1} \{\widetilde{u^{\otimes 2}}(\omega,\upsilon_*)\}^{\T}\; \partial_t^2\Partials^2Z(s(\omega,\upsilon_*),\upsilon_*)\;||\overline{u}(\omega,\upsilon_*)||\;d\omega=\frac{1}{\ell(C_{\upsilon_*})}\;(u^{\T}(\omega_1,\upsilon_*)\;\partial_t^2\partial_sZ(s(\omega_2,\upsilon_*),\upsilon_*)-u^{\T}(\omega_0,\upsilon_*)\;\partial_t^2\partial_sZ(s(\omega_1,\upsilon_*),\upsilon_*))$. This is a consequence of path independence. If $C_{\upsilon_*}$ is closed, then the spatiotemporal curvature is in fact, zero. %results in the difference in mixed temporal curvature and spatial gradient at the end-points of the open parametric curve. Analogous results are derived similarly for the other elements in $\L^{-Z}_{4,u_t,u_s}Z(\bs,t)$. %If $C_{\upsilon_*}$ is closed, then, $\omega_0=\omega_1$ and $\frac{1}{\ell(C_{\upsilon_*})}\oint_{\omega_0}^{\omega_1} \{\widetilde{u^{\otimes 2}}(\omega,\upsilon_*)\}^{\T}\; \partial_t^2\Partials^2Z(s(\omega,\upsilon_*),\upsilon_*)\;||\overline{u}(\omega,\upsilon_*)||\;d\omega = 0$. 
	
	In particular, if for every $\upsilon\in [\upsilon_0,\upsilon_1]$, $C_{\upsilon}$ is a closed curve, then the surface $\C$ encloses a topologically compact region $W$ with a \emph{piece-wise smooth boundary} $\partial W$. We denote the integral over the boundary by $\oiintctrclockwise_{\partial W}$. The total wombling measure is $\bGamma(\C)=\oiintctrclockwise_{\partial W}\L^*_{n_t,\bn_\bs}Z(\bs,t)\;{\rm d}\A$. For such volumes the spatiotemporal flux and rate of change in spatiotemporal flux represent rates of change in the spatiotemporal process enclosed by $W$. If $Z(\bs,t)$ is three times differentiable, then the closed surface integral can be expressed as a volume integral over $W$ without requiring explicit parametric representation. Let $\bF_1(\bs,t)=\left(\begin{smallmatrix}\bpartial_\bs Z(\bs,t)\\\partial_tZ(\bs,t)\end{smallmatrix}\right)$ and $\bF_2(\bs,t)=\left(\begin{smallmatrix}
		F_{21}\\\vdots\\F_{26}
	\end{smallmatrix}\right)=\bnabla^{\otimes 2}Z(\bs,t)$. The spatiotemporal flux and the rate of change in spatiotemporal flux satisfy the identities, $\oiintctrclockwise_{\C=\partial W} \bn(\bs,t)^{\T} \bF_1(\bs,t)\;\mathrm{d}\A=\iiint\limits_{W}{\rm div}\; \bF_1\;\mathrm{d}W$ and $\oiintctrclockwise_{\C=\partial W} {\bn(\bs,t)^{\otimes 2}}^{\T} \bF_2(\bs,t)\;\mathrm{d}\A=\iiint\limits_{W}\left\{{\rm div}\; n_{s_x} \left(\begin{smallmatrix}
		F_{21}\\F_{22}\\F_{23}
	\end{smallmatrix}\right) + {\rm div}\; n_{s_y} \left(\begin{smallmatrix}
		F_{22}\\F_{24}\\F_{25}
	\end{smallmatrix}\right) + {\rm div}\; n_t \left(\begin{smallmatrix}
		F_{23}\\F_{25}\\F_{26}
	\end{smallmatrix}\right)\right\}\;\mathrm{d}W$, where $\bn(\bs,t)=\left(\begin{smallmatrix}
		n_{s_x}\\n_{s_y}\\n_t
	\end{smallmatrix}\right)$ and $\mathrm{d}W=ds_x\;ds_y\;dt$ is the volume element. The integrand on the right consists of elements from $\bnabla^{\otimes 3}Z(\bs,t)$. The identities are obtained using the Divergence Theorem due to Gauss \citep[see, e.g.,][Vol. 4, Ch. 7, p. 192]{spivak1979comprehensive}% ~\citep[see, e.g.,][Theorem 10.33 (p. 272) and Theorem 10.5.1 (p. 512) respectively]{rudin1976principles, baxandall1986vector}
	. The first identity concerning spatiotemporal flux requires $Z(\bs,t)$ to be twice differentiable (see Section~S8). 
	
	Turning to practical computation of the spatiotemporal flux (or its rate of change) $\iiint_W$ can be evaluated using a customary Riemann approximation (with cubes inside the region $W$), which, while simple in concept, can be demanding on resources and requires smoother process specifications. On the other hand, evaluating $\oiintctrclockwise_{\C=\partial W}$ is conceptually harder and benefits from the triangulation procedure discussed in Section~\ref{sec:tri-surf}. Finally, we remark that one could possibly use Stokes' theorem on surfaces $\C=\left\{(s_x(\omega(\upsilon)),s_y(\omega(\upsilon)),\upsilon):\upsilon \in \mathcal{D}_\upsilon\right\}$. This is a parameterization from $\mathfrak{R}^+\to\mathfrak{R}^3$. Wombling measures would then be line integrals defined on curves bounding a surface $\oint_{C} \bn(\bs,t)^{\T} \bF_1(\bs,t)\;\mathrm{d}\ell=\iint_{\C}{\rm curl}\; \bF_1\;\mathrm{d}\A$. These would then be defined on curves and be incongruous with our aim of associating them with surfaces. We do not pursue such developments here.

	{
		\subsection{ Spatiotemporal Wombling}\label{sec:spt-womb} } Our scientific applications include, for example, { temporally evolving wombling measures for neural activity over regions of the human scalp}. %tracking regions of deteriorating air quality due to wildfires. %We extend the previous setting by allowing 
	%Here, we allow the curve of interest to evolve over time. 
	{ We consider temporally indexed curves, say $C_t$, which generate surfaces. For example, a line segment is described as $C_t = \{\bs_t + \omega\bu_t:\omega\in[0,1]\}$ on $t\in[t_0,t_1]$. We write $C_0=\{\bs_0+\omega \bu :\omega\in[0,1]\}$ and $C_1=\{\bs_0 + \bv + \omega (\bw-\bv):\omega\in[0,1]\}$ for $t=t_0$ and $t=t_1$, respectively, where $\bs_{t_0} = \bs_0$, $\bu_{t_0} = \bu$, $\bv = \bs_{t_1} - \bs_0$ is the displacement vector for $\bs_0$ and $\bw = \bu_{t_1} + \bv$ is the corresponding displacement of $\bu_{t_1}$. Hence, $C_{t_0}$ evolves continuously %over $\{t_0+\upsilon:\upsilon\in[0,1]\}$ 
		to $C_{t_1}$ generating planes from two sets of non-collinear points, $\left\{\bigl(\begin{smallmatrix}
			\bs_0\\ t_0
		\end{smallmatrix}\bigr),\bigl(\begin{smallmatrix}
			\bs_0 + \bu\\ t_0
		\end{smallmatrix}\bigr),\bigl(\begin{smallmatrix}
			\bs_0 + \bv\\ t_1
		\end{smallmatrix}\bigr)\right\}$ and $\left\{\bigl(\begin{smallmatrix}
			\bs_0 + \bu\\ t_0
		\end{smallmatrix}\bigr),\bigl(\begin{smallmatrix}
			\bs_0 + \bv\\ t_1
		\end{smallmatrix}\bigr), \bigl(\begin{smallmatrix}
			\bs_0 + \bw\\ t_1
		\end{smallmatrix}\bigr)\right\}$.} The resulting surface, $\C$, is a union of two triangular planes given by $\C_0=\left\{\bigl(\begin{smallmatrix}
		\bs_0 \\ t_0
	\end{smallmatrix}\bigr)+\omega\;\bigl(\begin{smallmatrix}
		\bu \\ 0
	\end{smallmatrix}\bigr)+\upsilon\;\bigl(\begin{smallmatrix}
		\bv \\ 1
	\end{smallmatrix}\bigr):(\omega,\upsilon)\in T\right\}$ and $\C_1=\left\{\bigl(\begin{smallmatrix}
		\bs_0 + \bw \\ t_1
	\end{smallmatrix}\bigr)+\omega\;\bigl(\begin{smallmatrix}
		\bv - \bw \\ 0
	\end{smallmatrix}\bigr)+\upsilon\;\bigl(\begin{smallmatrix}
		\bu - \bw \\ -1
	\end{smallmatrix}\bigr):(\omega,\upsilon)\in T\right\}$, where $T=\{(\omega,\upsilon):\omega,\upsilon\in [0,1],\omega+\upsilon\leq 1\}$ is a triangular region in $\mathfrak{R}^2$ { (left panel of \Cref{fig:examples})}. { The unit normals (independent of $(\omega,\upsilon)$) are $\bn_i = \frac{\overline{\bn}_i}{||\overline{\bn}_i||}$ for $i=0,1$, where $\overline{\bn}_0=\left(\begin{smallmatrix}\bv\\1\end{smallmatrix}\right)\times\left(\begin{smallmatrix}\bu\\0\end{smallmatrix}\right)=\left(\begin{smallmatrix}\overline{\bn}_{\bs,0}\\\overline{n}_{t,0}\end{smallmatrix}\right)$ and $\overline{\bn}_1=\left(\begin{smallmatrix}\bv-\bw\\0\end{smallmatrix}\right)\times\left(\begin{smallmatrix}\bw-\bu\\1\end{smallmatrix}\right)=\left(\begin{smallmatrix}\overline{\bn}_{\bs,1}\\\overline{n}_{t,1}\end{smallmatrix}\right)$, respectively, and $\times$ is the vector product. The }%The normals are $\overline{\bn}_0=\left(\begin{smallmatrix}\bv\\1\end{smallmatrix}\right)\times\left(\begin{smallmatrix}\bu\\0\end{smallmatrix}\right)=\left(\begin{smallmatrix}\overline{\bn}_{\bs,0}\\\overline{n}_{t,0}\end{smallmatrix}\right)$ and $\overline{\bn}_1=\left(\begin{smallmatrix}\bv-\bw\\0\end{smallmatrix}\right)\times\left(\begin{smallmatrix}\bw-\bu\\1\end{smallmatrix}\right)=\left(\begin{smallmatrix}\overline{\bn}_{\bs,1}\\\overline{n}_{t,1}\end{smallmatrix}\right)$ respectively, where $\times$ is the vector product. % We have $||n_1||=\left(u_2^2+u_1^2+(u_1v_2-u_2v_1)^2\right)^{1/2}$ and $||n_2||=((v_2-w_2)^2+(v_1-w_1)^2+\left((v_1-w_1)(w_2-u_2)-(v_2-w_2)(w_1-u_1)\right)^2)^{1/2}$. 
	%They %, $n_1=(u_2,-u_1,u_1v_2-u_2v_1)$ and $n_2=(v_2-w_2,-(v_1-w_1),(u_1-w_1)(v_2-w_2)-(u_2-w_2)(v_1-w_1))$ 
	%are independent of $(\omega,\upsilon)$. The unit normals are $\bn_i = \frac{\overline{\bn}_i}{||\overline{\bn}_i||}$, the 
	areas are $\A(\C_i)=\iint_{T}||\overline{\bn}_i||\;d\omega\;d\upsilon=\frac{||\overline{\bn}_i||}{2}$, and wombling measures in \cref{eq:st-womb-measure} are $\bGamma(\C_i)=||\overline{\bn}_i||\iint_T\L^*_{n_{t,i},\bn_{\bs,i}}Z(\bs(\omega,\upsilon),\upsilon)\;d\omega\;d\upsilon$, for $i=0,1$. The total and average wombling measures over the plane $\C = \C_0\cup\C_1$ with area $\A(\C)=\A(\C_0)+\A(\C_1)$ are $\bGamma(\C)= \bGamma(\C_0)+\bGamma(\C_1)$ and $\overline{\bGamma}(\C)= \frac{\bGamma(\C)}{\A(\C)}$, respectively. 
	
	The right panel of \Cref{fig:examples} displays another example: a surface generated from arcs of parametric semi-circles, $\C=\left\{\left(r\upsilon\cos \omega, r \upsilon\sin \omega,\upsilon\right):\omega\in[0,\pi],\upsilon\in[t_0,t_0+1],r>0\right\}$. The radii of the arcs increase linearly with time. The normal to $\C$ is $\overline{\bn}(\bs(\omega,\upsilon),\upsilon)= \left(r\upsilon\cos\omega,r\upsilon\sin\omega,-r^2\upsilon\right)^{\T}$. The unit normal is $\bn(\bs(\omega,\upsilon),\upsilon)=\frac{1}{\sqrt{1+r^2}}\left(\cos\omega,\sin\omega,-r\right)^{\T}$ and the area is $\A(\C)=\int\limits_{t_0}^{t_0+1}\int\limits_{0}^{\pi}||\overline{\bn}(\bs(\omega,\upsilon),\upsilon)||\;d\omega\; d\upsilon = \left(t_0+\frac{1}{2}\right)\pi r\sqrt{1+r^2}$. The average wombling measure is $\overline{\bGamma}(\C) = \frac{1}{\left(t_0+\frac{1}{2}\right)\pi}\int\limits_{t_0}^{t_0+1}\int\limits_{0}^{\pi}\L^*_{n_t(\upsilon),\bn_\bs(\omega,\upsilon)}Z(\bs(\omega,\upsilon),\upsilon)\;\upsilon\;d\omega\;d\upsilon$. { Sec. S14 discusses static wombling, where the curve stays fixed over time.}
	
	\begin{figure}[t]
		\centering
		\includegraphics[scale = 0.24]{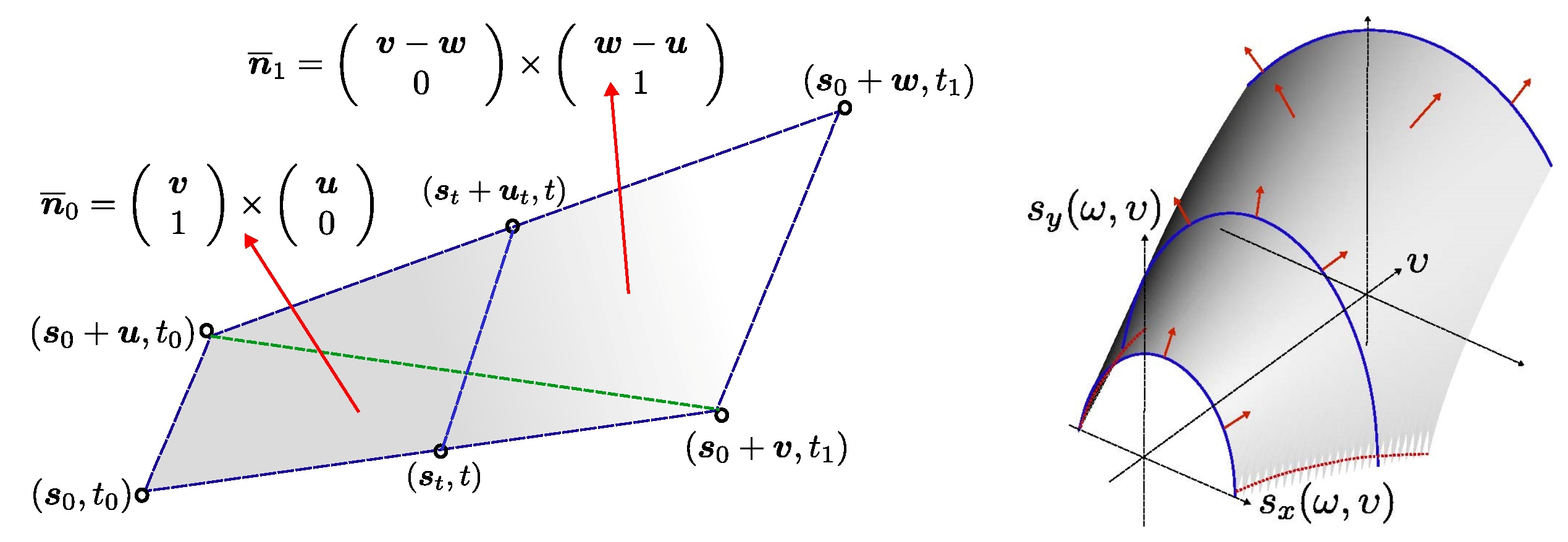}
		\caption{Illustrative examples for wombling: (left) parametric planes (right) parametric circles. Normals to the surface are indicated using red arrows and $\times$ denotes the vector product.}
		\label{fig:examples}
	\end{figure}
	
	\subsection{Distribution Theory for Wombling Measures}\label{sec:dist-womb}
	We extend the distribution theory in \Cref{sec:STDP} to derive the joint distribution for the $N\times 1$ process $\Z$ and the wombling measures $\bGamma(\C)$ in \cref{eq:st-womb-measure} for conducting predictive inference. Suppose $\C$ is generated over intervals, $\mathcal{D}_{\upsilon_0}=[0,\upsilon_0]$ and $\mathcal{D}_{\omega_0}=[0,\omega_0]$. For any $\upsilon^*\in \mathcal{D}_{\upsilon_0}$ and $\omega^*\in \mathcal{D}_{\omega_0}$, let $\C_{*}$ denote the surface restricted to $\D_*=\D_{\upsilon^*}\times \D_{\omega^*}=[0,\upsilon^*]\times [0,\omega^*]$. The wombling measures in \cref{eq:st-womb-measure} are $\bGamma(\C_{*})=\iint_{\D_*}\L^*_{n_t,\bn_\bs}Z(\bs(\omega,\upsilon),\upsilon)\;||\overline{\bn}(\bs(\omega,\upsilon),\upsilon)||\;d\omega\; d\upsilon$. The directional differential processes, $\L^*_{n_t,\bn_\bs} Z(\bs,t)$, are valid zero-mean Gaussian processes (GPs). Hence, $\bGamma(\C_{*})\sim \N_8\left(\0_8,\bK_{\bGamma}(\C_{*},\C_{*})\right)$, where % the $8\times 8$ cross-covariance matrix,
	
	\begin{equation}\label{eq:cov-womb-measures}
		\begin{split}
			\bK_\bGamma(\C_{*},\C_{*})&=\idotsint\limits_{\D_*\times\D_*} \bN_{\bs t}\;\bV_{\L^*Z}(\bDelta(\omega_1,\omega_2,\upsilon_1,\upsilon_2),\delta(\upsilon_1,\upsilon_2))\;\bN_{\bs t}^{\T}\;\cdot \\
			&\myquad[6]||\overline{\bn}(\bs(\omega_1,\upsilon_1),\upsilon_1)||\cdot||\overline{\bn}(\bs(\omega_2,\upsilon_2),\upsilon_2)||\;d\omega_1\;d\upsilon_1\;d\omega_2\;d\upsilon_2,
		\end{split}
	\end{equation}
	$\bDelta(\omega_1,\omega_2,\upsilon_1,\upsilon_2) = s(\omega_2,\upsilon_2)-s(\omega_1,\upsilon_1)$ and $\delta(\upsilon_1,\upsilon_2) = \upsilon_2-\upsilon_1$. The repeated integral acts element-wise on matrices. $\bGamma(\C_{*})$ is a \emph{dependent GP} on $\D_{\upsilon_0}\times\D_{\omega_0}$---for any two points, $(\omega_1^*,\upsilon_1^*)$ and $(\omega_2^*,\upsilon_2^*)$, the dependence is specified through, $\left(\begin{smallmatrix}
		\bGamma(\C_{1*})\\\bGamma(\C_{2*})
	\end{smallmatrix}\right)\sim \N_{16}\left[\0_{16},\left(\begin{smallmatrix}%{*{2}{c}}
		\bK_\bGamma(\C_{1*},\C_{1*}) & \bK_\bGamma(\C_{1*},\C_{2*})\\
		\bK_\bGamma(\C_{2*},\C_{1*}) & \bK_\bGamma(\C_{2*},\C_{2*})
	\end{smallmatrix}\right)\right]$, 
	where $\bGamma(\C_{i*})$ is the wombling measure computed over $[0,\omega_i^*]\times[0,\upsilon_i^*]$ and each $\bK_\bGamma(\C_{i*},\C_{j*})$ is $8\times 8$ obtained from \cref{eq:cov-womb-measures} for $i,j=1,2$. $\bGamma(\C_*)$ is mean-square continuous. Although $Z(\bs,t)$ is stationary, $\bGamma(\C_*)$ need not be. Let $\bG(\C_{*})= \Cov\{\Z,\bGamma(\C_*)\} = \left(\begin{smallmatrix}\bgamma_1(\C_{*})^{\T}\\ \vdots \\ \bgamma_N(\C_{*})^{\T}\end{smallmatrix}\right)$ as the $N\times 8$ matrix, where $\bgamma_i(\C_{*})=\Cov(Z(\bs_i,t_i), \bGamma(\C_*))$, $\bDelta_i(\omega,\upsilon)=\bs_i-\bs(\omega,\upsilon)$, $\delta_i(\upsilon)=t_i-\upsilon$,
	
	\begin{equation}\label{eq:cross-cov}
		\bgamma_i(\C_{*})=\iint_{\D_*}\bN_{\bs t}\calK_{i}(\bDelta_{i}(\omega,\upsilon),\delta_{i}(\upsilon))\;||\overline{\bn}(\bs(\omega,\upsilon),\upsilon)||\;d\omega\;d\upsilon, \quad i=1,\ldots,N\;,
	\end{equation}
	and the integrals act element-wise on the vector. %requiring $\Delta_i(\omega,\upsilon) = s(\omega,\upsilon)-s_i$, $\delta_i(\upsilon) = \upsilon-t_i$ where,
	% \begin{equation}\label{eq:cross-covariance-wm}
		%     \begin{split}
			%         \gamma_{ij}(\omega^*,\upsilon^*) &=\left(\begin{smallmatrix}
				%                 \int\limits_{0}^{\upsilon^*}\int\limits_{0}^{\omega^*}\tilde{c}_{1,1}^{\T}\;\partial_sK(\Delta_i(\omega,\upsilon),\delta_j(\upsilon))\;||n(s(\omega,\upsilon),t(\upsilon))||\;{\rm d}\omega\;{\rm d}\upsilon\\
				%                \int\limits_{0}^{\upsilon^*}\int\limits_{0}^{\omega^*}\tilde{c}_{1,2}^{\T}\;\partial_s^2K(\Delta_i(\omega,\upsilon),\delta_j(\upsilon))\;||n(s(\omega,\upsilon),t(\upsilon))||\;{\rm d}\omega\;{\rm d}\upsilon\\
				%                \vdots\\
				%                \int\limits_{0}^{\upsilon^*}\int\limits_{0}^{\omega^*}c_0^2\;\tilde{c}_{1,2}^{\T}\;\partial_t^2\partial_s^2K(\Delta_i(\omega,\upsilon),\delta_j(\upsilon))\;||n(s(\omega,\upsilon),t(\upsilon))||\;{\rm d}\omega\;{\rm d}\upsilon
				%                \end{smallmatrix}\right)
			%     \end{split}
		% \end{equation}
	The desired distribution $P[\Z, \bGamma(\C_*)]$ is 
	
	\begin{equation}\label{eq:joint-dist-womb-measure}
		\left(\begin{array}{c}
			\Z\\
			\bGamma(\C_*)
		\end{array}\right)\sim \N_{N+8}\left[\0_{N+8}, \left(\begin{array}{cc}
			\Sigma_\Z  & \bG(\C_{*})\\
			\bG^{\T}(\C_{*}) & \bK_\bGamma(\C_{*},\C_{*})
		\end{array}\right)\right].
	\end{equation}
	
	\section{Bayesian Inference and Computation}\label{sec:bi&c}
	\subsection{Spatiotemporal hierarchical model}\label{sec:spt_hier_model}
	
	Let $\Y = (y(\bs_1,t_1),\ldots, y(\bs_N,t_N))^{\T}$ be the $N\times 1$ vector of observations % observed outcomes 
	over spatiotemporal coordinates $\{(\bs_1,t_1),\ldots,(\bs_N,t_N)\}$. We construct the spatiotemporal hierarchical model,
	
	\begin{equation}\label{eq:hier-spt-model}
		y(\bs_i,t_i) = \mu(\bs_i,t_i) + Z(\bs_i,t_i) + \epsilon(\bs_i,t_i),
	\end{equation}
	where $\mu(\bs_i,t_i)=\bx(\bs_i,t_i)^{\T}\bbeta$, each $\bx(\bs_i,t_i)$ is a $p\times 1$ vector of spatiotemporal covariates, {  $Z(\bs_i,t_i)$s} are realizations of the process $Z(\bs,t)\sim GP(0,K(\cdot,\cdot,\btheta))$ and $\epsilon(\bs_i,t_i)\overset{iid}{\sim} N(0,\tau^2)$ accounts for white noise. The process parameters are $\Theta = \{\btheta, \tau^2,\bbeta\}$, where $\btheta = \{\sigma^2, \phi_s,\phi_t\}$. %The joint posterior, $P(\L^*^{-Z}Z(s_0,t_0)\given\Y) = \int P(\L^*^{-Z}Z(s_0,t_0)\given\Z,\Theta)\;P(\Z\given\Y,\Theta)\;P(\Theta\given\Y)\;d\Theta\;d\Z$. For a surface $\C$ of interest, the wombling measures, $\Gamma(\C)$ are sampled from the posterior, $P(\Gamma(\C)\given\Y) = \int P(\Gamma(\C)\given\Z,\Theta)\;P(\Z\given\Y,\Theta)P(\Theta\given\Y)\;d\Theta\;d\Z$. 
	With customary prior specifications we obtain the posterior,
	
	\begin{equation}\label{eq:post-full}
		P(\Theta\given \Y) \propto \pi(\btheta,\tau^2)\times {  \N}(\bbeta\given \bmu_{\bbeta},\Sigma_{\bbeta})\times 
		\mathcal{N}_N\left(\Y\given \bX(\bs,t)\bbeta, \Sigma_\Z + \tau^2\bI_N\right),
	\end{equation}
	where $\pi(\btheta,\tau^2) = U(\phi_s\given a_{\phi_s},b_{\phi_s})\times U(\phi_t \given a_{\phi_t},b_{\phi_t}) \times IG(\sigma^2 \given a_{\sigma},b_{\sigma})\times IG(\tau^2\given a_{\tau},b_{\tau})$, $\Sigma_{\Z} = \sigma^2\bR(\phi_s,\phi_t)$ is the $N\times N$ spatiotemporal covariance matrix constructed from the Mat\'ern covariance kernel in \cref{eq:cov-matern}, $IG$ denotes the inverse-gamma distribution, {  $\N$ denotes univariate Gaussian and $\N_N$ denotes $N$-variate Gaussian distribution}. 
	
	We compute $P(\Theta\given \Y)$ by marginalizing out $\Z$, thereby producing a ``collapsed", more efficient alternative% \citep[see, e.g.,][]{finley2019efficient,alaimo2023bayesian}
	. In a more general setting, if learning smoothness of the process from available data is of interest, then $\nu\sim U(a_\nu,b_\nu)$. Depending on the posterior estimate for $\nu$, the spatiotemporal differential processes are sampled using closed-form Mat\'ern kernels with half-integer values for $\nu$ as follows, $\nu = 1+\frac{1}{2}$ or, $2+\frac{1}{2}$. For instance, if $2\leq\nu\leq 3$, then a Mat\'ern with $\nu=\frac{5}{2}$ is used to generate posterior samples for the differential processes. This features in our implementation. {  Alternatively, models can be fit simultaneously with $\nu=\frac{1}{2},\frac{3}{2}$ and $\frac{5}{2}$ and a model selection criteria can be used to select $\nu$ for the best fit.} The posterior distribution in \cref{eq:post-full} is sampled using a combination of Gibbs and % random walk 
	Metropolis-Hastings (MH) steps. We first sample $\btheta$ and $\tau^2$ using MH steps, % featuring adaptive learning rates, 
	followed by % To ensure computational stability log-normal proposals are used.
	Gibbs samples for $\bbeta$ and $\Z$, one-for-one corresponding to samples of $\btheta$ and $\tau^2$. 
	
	Predictive inference on spatiotemporal differential processes and wombling measures for the latent surface is achieved by sampling from $P(\L^* Z(\bs_0,t_0)\given\Y)$ and $P(\bGamma(\C)\given\Y)$, respectively, % where $\Y$ is stacked similarly to $\Z$
	where $\C$ is a surface of interest. The posterior predictive distribution of the derivative processes %$\L^{-Z}_4Z(s_0,t_0)$ at $(s_0,t_0)$ 
	is $ P\left(\L^* Z(\bs_0,t_0)\given\Y\right) = \int P\left(\L^*Z(\bs_0,t_0)\given\Z,\Theta\right)\;P(\Z\given\Y,\Theta)\;P(\Theta\given\Y)\;d\Theta$.
	Posterior inference for any sub-vector, say, $\partial_t^2\Partials^2Z(\bs_0,t_0)$ is obtained as
	$P\left(\partial_t^2\Partials^2Z(\bs_0,t_0)\given \Y\right) = \idotsint P\left(\partial_t^2\Partials^2Z(\bs_0,t_0)\given \widetilde{\L}^{*}Z(\bs_0,t_0),\Z,\Theta\right)P\left(\widetilde{\L}^{*}Z(\bs_0,t_0)\given{\cal Z},\Theta\right)P(\Z\given\Y,\Theta)P(\Theta\given\Y)d\Theta\; d\widetilde{\L}^{*}Z$, where $\widetilde{\L}^{*}Z$ is the vector of spatiotemporal derivative processes excluding $\partial_t^2\Partials^2Z(\bs_0,t_0)$.
	Posterior sampling from the above conditionals is done one-for-one using \cref{eq:diffp-full-post}. One instance of $\L^*Z(\bs_0,t_0)$ is sampled for each posterior sample of $\theta$ obtained from the posterior, $P\left(\btheta\given \Z\right)$. The conditional predictive distribution of the spatiotemporal differential process, $\L^*Z(\bs_0,t_0)\given{\cal Z},\btheta\sim\N_{m-1}\left(\bmu_\L,\Sigma_\L\right)$, where $\bmu_\L=-{\cal K}_0\Sigma_{\Z}^{-1}\Z$  %$\mu_\L=\L^{-Z}_4{\cal M}(s_0,t_0)-{\cal K}_0\Sigma_{\cal Z}^{-1}\left({\cal Z}-{\cal M}\right)$ 
	and $\Sigma_\L=\bV_{\L^*  Z}(\0_d,0)-{\calK}_0\Sigma_{\cal Z}^{-1}{{\calK}_0}^{\T}$. Similar arguments follow for any element or sub-vector of $\L^*Z(\bs_0,t_0)$. 
	
	Posterior samples for the wombling measures are obtained similarly.  The posterior, $P(\bGamma(\C)\given\Y) = \int P(\bGamma(\C)\given\Z,\Theta)\;P(\Z\given\Y,\Theta)\;P(\Theta\given\Y)\;d\Theta\;d\Z$ is used to sample wombling measures. For each posterior sample of $\Z$ and $\btheta$, we draw $\bGamma(\C)\given\Z,\btheta\sim \N_8\big(-\bG(\C)^{\T}\Sigma_\Z^{-1}\Z,\bK_\bGamma(\C,\C)-\bG(\C)^{\T}\Sigma_\Z^{-1}\bG(\C)\big)$, where, $\bK_\bGamma(\C,\C)$, $\bG(\C)$ are computed using \cref{eq:cov-womb-measures,eq:cross-cov} using an appropriate parameterization (see \Cref{sec:tri-surf}). % and $\mu_\Gamma=\int\limits_{0}^{\upsilon^*}\int\limits_{0}^{\omega^*}\L^{-Z}_{4,n_t,n_s}\mu(s(\omega,\upsilon),\upsilon)\;||n(s(\omega,\upsilon),\upsilon)||\;d\omega\; d\upsilon$.
	Posterior samples for $\Z$, $\L^* Z$ and $\bGamma(\C)$ are generated together to avoid repeated computation of $\Sigma_{\Z}^{-1}$. We leverage parallel computation to improve algorithmic complexity. 
	
	% \section{Computational simplifications}\label{sec:simplify}
	% For a surface to be identified as a wombling boundary, posterior sampling of wombling measures requires a parametric form. Here, we choose parametric triangular planes and are concerned only with surfaces that can be triangulated (compact and connected surfaces in the topological sense%\citep[see, e.g.,][]{thomassen1992jordan}
	% ). % We use parametric triangular planes (see \Cref{sec:stat-womb,sec:dyn-womb} and \Cref{fig:planes}) to triangulate surfaces. It suffices to perform the wombling exercise on a single triangular plane. \Cref{fig:st-wombling} illustrates the process. 
	% \Cref{sec:tri-surf} details triangulation of surfaces and associated computation of wombling measures. Posterior sampling benefits from further simplifications resulting from triangulation.
	%of wombling measures also requires analytic evaluation of the repeated integral $K_\Gamma(\C,\C)$ in \cref{eq:cov-womb-measures}, which is computationally expensive. Triangulation produces further simplifications in this regard. %---the multiple integral is reduced to a double integral. % In particular, in the context of curvilinear wombling, within a purely spatial context, no analytic evaluation is required. 
	% \Cref{sec:var-womb} elaborates further.
	
	\begin{figure}[t]
		\centering
		\includegraphics[scale=0.5, page=1]{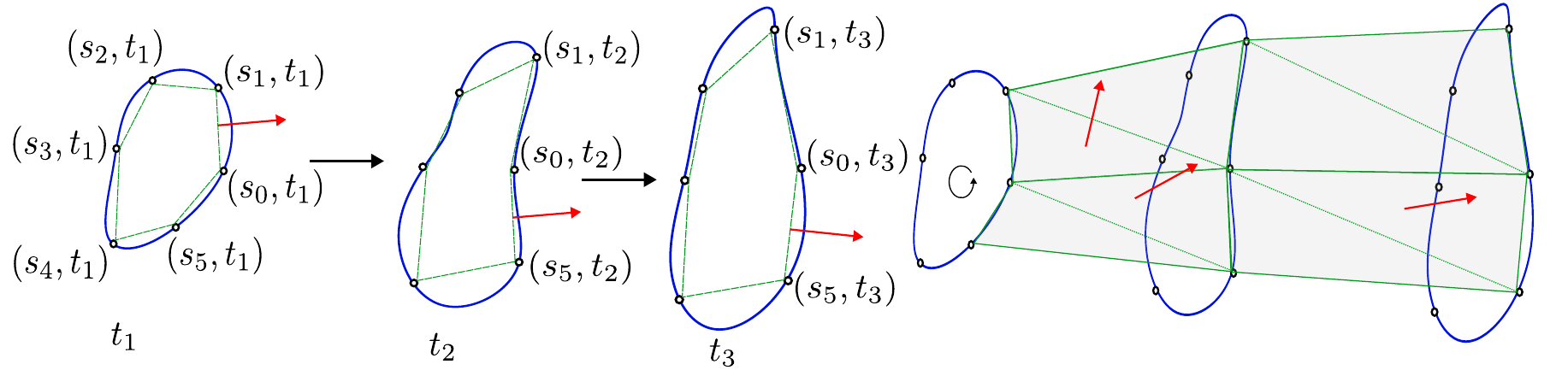}
		\caption{Illustration of spatiotemporal wombling using three time points. At each time point, for a closed planar curve, a partition consisting of six points is chosen to triangulate the surface. An anti-clockwise orientation is used for outward pointing normals which are marked with red arrows.}\label{fig:st-wombling}
	\end{figure}
	
	\subsection{Triangulation of Surfaces}\label{sec:tri-surf} 
	{  For a surface (compact and connected in the topological sense) to be identified as a wombling boundary, posterior sampling of wombling measures requires a parametric form. We choose to perform spatiotemporal wombling using} simple parametric regions, specifically triangular planes as shown in \Cref{fig:examples} (left). The spatiotemporal wombling measures in \cref{eq:st-womb-measure} are computed on each region followed by an aggregation to produce the wombling measure for the surface. The surface $\C$ is triangulated using a partition. We partition $\D_\omega=[\omega_0,\omega_1]$ using $\omega_0=\omega'_1<\cdots<\omega'_{n_\omega}=\omega_1$ and $\D_\upsilon=[t_j,t_{j+1}]$ using $t_j=\upsilon_0=\upsilon'_1<\cdots<\upsilon'_{n_\upsilon}=\upsilon_1=t_{j+1}$. Denoting the partition as $\mathscr{P}$, its norm is $|\mathscr{P}|_A=\max\limits_{\substack{1\leq i \leq n_\omega-1\\1\leq j \leq n_\upsilon-1}}(\omega'_{i+1}-\omega'_{i})\;(\upsilon'_{j+1}-\upsilon'_j)$, which is a function of the area of the largest parallelogram. Let $C_{t_j}$ be {  a planar curve} on $\C$. The curves have a polygonal approximation $\widetilde{C}_{t_j}= \bigcup\limits_{i=1}^{n_\omega-1}\widetilde{C}_{ij}$ with $\widetilde{C}_{ij} = \{\bs(\omega'_{i},t_j)+\omega\;\bu_{ij}: \bu_{ij} = \bs(\omega'_{i+1},t_j)-\bs(\omega'_{i},t_j),\omega\in[0,1]\}$. Considering $\widetilde{C}_{ij}$ and $\widetilde{C}_{ij+1}$, which are rectilinear segments at adjacent time points $t_j$ and $t_{j+1}$, respectively. The triangulation of $\C$ restricted to the rectangle $[\omega'_{i},\omega'_{i+1}]\times [t_j,t_{j+1}]$ follows the first example in \Cref{sec:spt-womb}. The triangular planes are generated using the set, $\left\{\bigl(\begin{smallmatrix}
		\bs(\omega'_i,t_j)\\ t_j \end{smallmatrix}\bigr),\bigl(\begin{smallmatrix}
		\bs(\omega'_{i+1},t_j)\\ t_j \end{smallmatrix}\bigr),\bigl(\begin{smallmatrix}
		\bs(\omega'_i,t_{j+1})\\ t_{j+1} \end{smallmatrix}\bigr),\bigl(\begin{smallmatrix}
		\bs(\omega'_{i+1},t_{j+1})\\ t_{j+1} \end{smallmatrix}\bigr)\right\}$, of four points which are not co-planar. They are % where, $\left\{s(\omega'_{j},t_i),\;s(\omega'_{j+1},t_i),\;s(\omega'_{j},t_{i+1})\right\}$ and $\left\{\;s(\omega'_{j+1},t_i),\;s(\omega'_{j},t_{i+1}),\;s(\omega'_{j+1},t_{i+1}))\right\}$ are non-collinear. Hence,
	$\C_{ij1}= \left\{\bigl(\begin{smallmatrix}
		\bs(\omega'_i,t_j)\\ t_j \end{smallmatrix}\bigr)+\omega\;\bigl(\begin{smallmatrix}
		\bu_{ij}\\0 \end{smallmatrix}\bigr)+\upsilon\;\bigl(\begin{smallmatrix}
		\bv_{ij+1}\\t_{j+1}-t_{j} \end{smallmatrix}\bigr):(\omega,\upsilon)\in T\right\}$ and $\C_{ij2} = \left\{\bigl(\begin{smallmatrix}
		\bs(\omega'_{i+1},t_{j+1})\\ t_{j+1} \end{smallmatrix}\bigr)+\omega\;\bigl(\begin{smallmatrix}
		\bu_{ij+1}\\ 0 \end{smallmatrix}\bigr)+\upsilon\;\bigl(\begin{smallmatrix}
		\bv_{i+1j}\\ t_j-t_{j+1} \end{smallmatrix}\bigr):(\omega,\upsilon)\in T\right\}$, where, $\bv_{ij+1}=\bs(\omega'_{i},t_{j+1})-\bs(\omega'_{i},t_{j})$, $\bv_{i+1j}=\bs(\omega'_{i+1},t_j)-\bs(\omega'_{i+1},t_{j+1})$ and $T$ is a parametric triangle in $\mathfrak{R}^2$ (see \Cref{sec:spt-womb}). 
	
	The normals to the planes are $\overline{\bn}_{ij1} = \left(\begin{smallmatrix} \bv_{ij+1}\\t_{j+1}-t_j\end{smallmatrix}\right)\times\left(\begin{smallmatrix} \bu_{ij}\\0\end{smallmatrix}\right)$ and $\overline{\bn}_{ij2} = \left(\begin{smallmatrix} -\bu_{ij+1}\\0\end{smallmatrix}\right) \times \left(\begin{smallmatrix} -\bv_{i+1j}\\t_{j+1}-t_j\end{smallmatrix}\right)$, respectively, where $\times$ is the vector product. These quantities are free of $(\omega,\upsilon)$. The corresponding unit normals are $\bn_{ijk}=\frac{\overline{\bn}_{ijk}}{||\overline{\bn}_{ijk}||}=\left(\begin{smallmatrix}\bn_{ijk,\bs}\\\bn_{ijk,t}\end{smallmatrix}\right)$, $k=1,2$. Iteration over successive intervals triangulates $\C$ as $\widetilde{\C}=\bigcup\limits_{i=1}^{n_\omega-1}\bigcup\limits_{j=1}^{n_\upsilon-1}\bigcup\limits_{k=1}^{2}\C_{ijk}$. \Cref{fig:st-wombling} illustrates this procedure. The wombling measure in \cref{eq:st-womb-measure} is computed on $\C_{ijk}$ as, $\bGamma(\C_{ijk})=\iint_T \bN_{ijk,\bs t}\L^* Z(\bs(\omega,\upsilon), \upsilon)\;||\overline{\bn}_{ijk}||\;d\omega\; d\upsilon=\iint_T\L^*_{n_{ijk,t},\bn_{ijk,\bs}} Z(\bs(\omega,\upsilon), \upsilon)\;||\overline{\bn}_{ijk}||\;d\omega\; d\upsilon$, for $k=1,2$. Aggregating, we obtain, $\bGamma(\widetilde{\C})=\sum\limits_{i=1}^{n_\omega-1}\sum\limits_{j=1}^{n_\upsilon-1}\sum\limits_{k=1}^{2}\bGamma(\C_{ijk})$. As $|\mathscr{P}|_A\to0$, $\bGamma(\widetilde{\C})\stackrel{a.s.}{\to}\bGamma(\C)$. Hence, $\bGamma(\widetilde{\C})$ is a strongly consistent estimator for $\bGamma(\C$). In general, if we consider $g(\L^*_{n_t,\bn_\bs} Z(\bs,t))$, where $g(\cdot)$ is uniformly continuous  over $[\omega_0,\omega_1]\times [t_1,t_N]$, then the corresponding estimator, $\bGamma_g(\C_{ijk})$, converges almost surely to $\bGamma_g(\C)$ (see Theorem~2 in Section~S9).
	
	Posterior samples of $\bGamma(\C_{ijk})$ are obtained from \cref{eq:joint-dist-womb-measure}. From \cref{eq:cov-womb-measures}, $\bK_\bGamma(\C_{ijk},\C_{ijk})=\idotsint\limits_{T\times T}\bN_{ijk,\bs t}\;\bV_{\L^* Z}(\bDelta_k(\omega_1,\omega_2,\upsilon_1,\upsilon_2),\delta_k(\upsilon_1,\upsilon_2))\;\bN_{ijk,\bs t}^{\T}\; ||\overline{\bn}_{ijk}||^2\;d\omega_1\;d\upsilon_1\;d\omega_2\;d\upsilon_2$, with $\delta_1(\upsilon_2,\upsilon_1)=(\upsilon_2-\upsilon_1)(\upsilon'_{j+1}-\upsilon'_j)$, $\delta_2(\upsilon_2,\upsilon_1)=-\delta_1(\upsilon_2,\upsilon_1)$, $\bDelta_1(\omega_1,\omega_2,\upsilon_1,\upsilon_2)=(\omega_2-\omega_1)\bu_{ij}+(\upsilon_2-\upsilon_1)\bv_{ij+1}$ and $\bDelta_2(\omega_1,\omega_2,\upsilon_1,\upsilon_2)=(\omega_2-\omega_1)\bu_{ij+1}+(\upsilon_2-\upsilon_1)\bv_{ij}$. From \cref{eq:cross-cov}, $\bG(\C_{ijk}) = \left(\begin{smallmatrix}
		\bgamma_1(\C_{ijk})^{\T}\\\vdots\\\bgamma_N(\C_{ijk})^{\T}
	\end{smallmatrix}\right)$, where for any $(\bs_{i'},t_{i'})$, $\bgamma_{i'}(\C_{ijk})=\iint\limits_{T}\bN_{ijk,\bs t}\calK_{i'}(\bDelta_{i'k}(\omega,\upsilon),\delta_{i'k}(\upsilon))\;||\overline{\bn}_{ijk}||\;d\omega\;d\upsilon$, where $\bDelta_{i'1}(\omega,\upsilon)=\bDelta_{i'ij}-\omega \bu_{ij}-\upsilon \bv_{ij+1}$, $\bDelta_{i'2}(\omega,\upsilon)=\bDelta_{i'i+1j+1}-\omega \bu_{ij+1}-\upsilon \bv_{ij}$, $\delta_{i'1}(\upsilon)=\delta_{i'j}-\upsilon\;\delta_{j+1j}$ and $\delta_{i'2}(\upsilon)=\delta_{i'j+1}+\upsilon\;\delta_{j+1j}$ for $i'=1,\ldots, N$, where $\bDelta_{i'ij}=\bs_{i'}-\bs(\omega'_i,t_j)$ and $\delta_{i'j}=t_{i'}-t_j$. Analytically evaluating $\bK_\bGamma(\C_{ijk},\C_{ijk})$ using {  high-dimensional} quadrature is tedious. % This segues into our ensuing discussion on simplifications that mitigate the computational burden.
	
	\subsection{Cross-covariance of Wombling Measures}\label{sec:var-womb}
	Our choice of simple parametric planes allows us to avoid high dimensional quadrature in \cref{eq:cov-womb-measures} required for posterior sampling of wombling measures. In the integrand for $\bK_\bGamma(\C_{ijk},\C_{ijk})$, $\bDelta_k(\omega_1,\omega_2,\upsilon_1,\upsilon_2)=\bDelta_k(\omega_2-\omega_1,\upsilon_2-\upsilon_1)$, $\delta_k(\upsilon_1,\upsilon_2)=\delta_k(\upsilon_2-\upsilon_1)$, $k=1,2$ are functions of $\omega_2-\omega_1$ and $\upsilon_2-\upsilon_1$. A simple change of variable results in $\bK_\bGamma(\C_{ijk},\C_{ijk})=\frac{||\overline{\bn}_{ijk}||^2}{2}\int\limits_{-1}^{1}\int\limits_{-1-x_\upsilon}^{1-x_\upsilon}\bN_{ijk,\bs t}\;\bV_{\L^* Z}(\bDelta_k(x_\omega,x_\upsilon),\delta_k(x_\upsilon))\;\bN_{ijk,\bs t}^{\T}\;dx_\omega\;dx_\upsilon$, where $x_\omega = \omega_2-\omega_1$ and $x_\upsilon = \upsilon_2-\upsilon_1$. This reduces a four-dimensional integral to a double integral. % Analytic evaluation using 
	Bivariate quadrature is performed to obtain the cross-covariance. Similar changes can be effected in the purely spatial case%\citep[see, e.g.,][]{banerjee2006bayesian,halder2024bayesian}% eq. (12), p. 1493, and eq. (10) respectively
	---for the Gaussian kernel we have closed forms available (see Section~S10).
	
	\section{Simulation Experiments}\label{sec:sim}
	We devise simulation experiments to calibrate inference for spatiotemporal differential processes and wombling measures. Algorithms for posterior sampling of the differential processes and the wombling measures 
	are available in Section~S11. The required sub-routines are developed for implementation in the \texttt{R}-statistical environment { \citep[][]{r_core_team_2021}}. The experiments are performed on a Linux OS powered by a 12\textsuperscript{th} Generation Intel\textsuperscript{\textregistered} Core\textsuperscript{\texttrademark}~i9-12.9K CPU with 24 cores and 64GB of RAM. Vignettes are available in the GitHub repository \if1\blind\href{https://github.com/arh926/sptwombling}{https://github.com/arh926/sptwombling}\fi \if0\blind {\url{(REDACTED)}}\fi.
	
	\paragraph{Data generation} Gradients and curvatures in realizations are not observed, but are induced by a parent process. We set up simulation experiments with known true values of differential processes and wombling measures. We consider locations $(s_x,s_y,t)\in [0,1]\times[0,1]\times{\cal T}\subseteq\mathfrak{R}^2\times\mathfrak{R}^+$ over the unit square across nine time points, ${\cal T} = \{1,2,\ldots, 9\}$.  We generate synthetic data from four distributions, $y(\bs,t)\sim {  \N}(\mu_i(\bs,t), \tau^2)$, $i = 1,\ldots,4$ and $\tau^2=1$ using the following specifications: (a) Pattern 1: $\mu_1=10\left\{\sin(3\pi s_x)+\cos(3\pi s_y)\cos\left(\frac{\pi t}{7}\right)\right\}$; (b) Pattern 2: $\mu_2 = 10\sin(3\pi s_x)\cos(3\pi s_y)\cos\left(\frac{\pi t}{7}\right)$. % (c) Pattern 3: $\mu_3 = \frac{\mu_1}{2}$ and (d) Pattern 4: $\mu_4 = \frac{\mu_2}{2}$. % They are adapted from simulation experiments carried out in explorations addressing spatiotemporal gradients and spatial curvature wombling \citep[see, e.g.,][Section 4 and Section 3 respectively]{quick2015bayesian, halder2024bayesian}. 
	\Cref{fig:true-plot-a,fig:true-plot-b} present interpolated plots for the spatiotemporal data generated from these patterns, respectively. These data generating patterns are motivated by smooth spatiotemporal data that do not arise from \cref{eq:post-full}, where the true spatiotemporal differential processes and wombling measures are computed in closed form. This allows us to calibrate the proposed inferential framework. The patterns also represent scenarios that manifest in real-time % spatiotemporally indexed 
	data. The pattern in \Cref{fig:true-plot-a} displays regions enclosing peaks and troughs that relocate spatially over time. The pattern in \Cref{fig:true-plot-b} displays regions enclosing peaks (troughs) that change to troughs (peaks) over time. % but do not display significant spatial relocation over time. 
	Pattern 1 has a zero mixed spatial derivative---changes in space along the $x$-axis does not affect changes along $y$-axis, while Pattern 2 has all non-zero derivatives. 
	
	% all plotting has been changed to ggplot2
	\begin{figure}[t]
		\centering
		\begin{subfigure}{.5\textwidth}
			\centering
			\hspace*{-0.6cm}    
			\includegraphics[scale = 0.3]{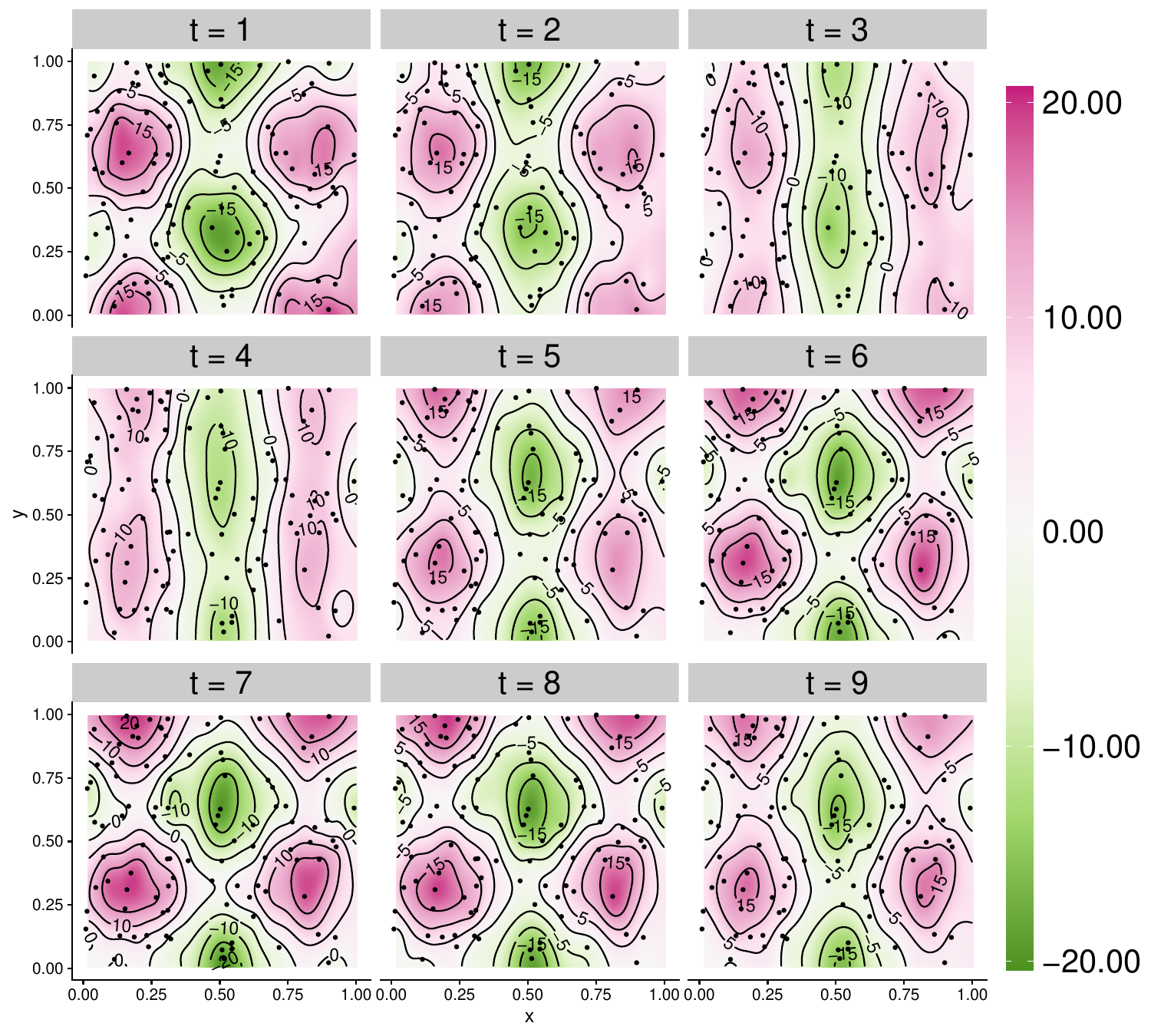}
			\caption{Pattern 1}\label{fig:true-plot-a}
		\end{subfigure}%
		\begin{subfigure}{.5\textwidth}
			\centering
			\includegraphics[scale = 0.3]{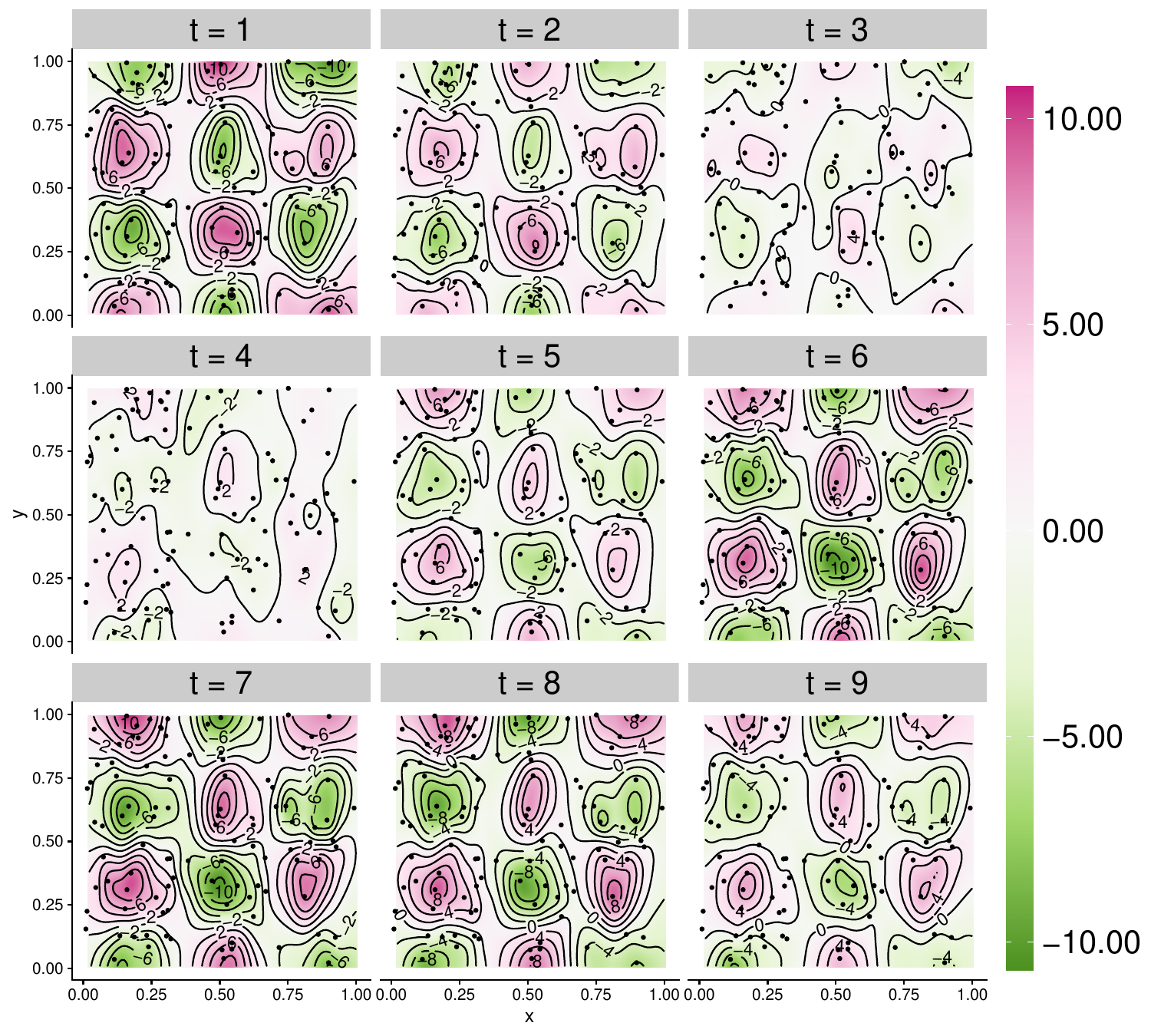}
			\caption{Pattern 2}\label{fig:true-plot-b}
		\end{subfigure}
		\caption{Spatiotemporal surfaces of patterned data used for experiments.}
	\end{figure}
	
	The true values for the spatiotemporal differential processes induced by the patterns are as follows. We list the spatial gradients and curvature followed by their temporal gradients and curvature. $\bpartial_\bs\mu_1(\bs,t) = 30\pi\left(\cos(3\pi s_x),-\sin(3\pi s_y)\cos\left(\frac{\pi t}{7}\right)\right)^{\T}$ and $\bpartial_\bs^2\mu_1(\bs,t) = -90\pi^2{\rm diag}\left\{\sin(3\pi s_x),\;\cos(3\pi s_y)\cos\left(\frac{\pi t}{7}\right)\right\}$. The true temporal derivatives are $\partial_t\mu_1(\bs,t)= -\frac{10\pi}{7}\cos(3\pi s_y)\sin\left(\frac{\pi t}{7}\right)$, $\partial_t\partial_{s_x}\mu_1(\bs,t)= 0$, $\partial_t\partial_{s_y}\mu_1(\bs,t) = \frac{30\pi^2}{7}\sin(3\pi s_y)\sin\left(\frac{\pi t}{7}\right)$, $\partial_t\partial^2_{s_x}\mu_1(\bs,t)= \partial_t\partial^2_{s_xs_y}\mu_1(\bs,t) = 0$ and  $\partial_t\partial^2_{s_y}\mu_1(\bs,t)=\frac{90\pi^3}{7}\cos(3\pi s_y)\sin\left(\frac{\pi t}{7}\right)$. The true temporal curvatures are $\partial_t^2\mu_1(\bs,t) = -\frac{10\pi^2}{49}\cos(3\pi s_y)\cos\left(\frac{\pi t}{7}\right)$, $\partial_t^2\partial_{s_x}\mu_1(\bs,t)=0$, $\partial_t^2\partial_{s_y}\mu_1(\bs,t)=\frac{30\pi^3}{49}\sin(3\pi s_y)\cos\left(\frac{\pi t}{7}\right)$, $\partial_t^2\partial^2_{s_x}\mu_1(\bs,t) = \partial_t^2\partial^2_{s_xs_y}\mu_1(\bs,t) = 0$, $\partial_t^2\partial^2_{s_y}\mu_1(\bs,t)=\frac{90\pi^4}{49}\cos(3\pi s_y)\cos\left(\frac{\pi t}{7}\right)$. Similarly, for pattern 2, $\bpartial_\bs\mu_2(\bs,t) = 30\pi\cos\left(\frac{\pi t}{7}\right)\bA_1(\bs)$, $\bpartial_\bs^2\mu_2(\bs,t) = -90\pi^2\cos\left(\frac{\pi t}{7}\right)\bA_2(\bs)$. The true temporal derivatives are $\partial_t\mu_2(\bs,t) = -\frac{10\pi}{7}\sin(3\pi s_x)\cos(3\pi s_y)\sin\left(\frac{\pi t}{7}\right)$, $\partial_t\bpartial_\bs\mu_2(\bs,t)= -\frac{30\pi^2}{7}\sin\left(\frac{\pi t}{7}\right)\bA_1(\bs)$, $\partial_t\bpartial_\bs^2\mu_2(\bs,t) = \frac{90\pi^3}{7}\sin\left(\frac{\pi t}{7}\right)\bA_2(\bs)$. The true temporal curvatures are as follows, $\partial_t^2\mu_2(\bs,t) = -\frac{10\pi^2}{49}\sin(3\pi s_x)\cos(3\pi s_y)\cos\left(\frac{\pi t}{7}\right)$, $\partial_t^2\bpartial_\bs\mu_2(\bs,t)=-\frac{30\pi^3}{49}\cos\left(\frac{\pi t}{7}\right)\bA_1(\bs)$ and $\partial_t^2\bpartial_\bs^2\mu_2(\bs,t) =\frac{90\pi^4}{49}\cos\left(\frac{\pi t}{7}\right)\bA_2(\bs)$, where the terms
	$\bA_1(s)=\left(\begin{smallmatrix}\cos(3\pi s_x)\cos(3\pi s_y)\\-\sin(3\pi s_x)\sin(3\pi s_y)\end{smallmatrix}\right)$ and $\bA_2(s)=\left(\begin{smallmatrix}\sin(3\pi s_x)\cos(3\pi s_y) & \cos(3\pi s_x)\sin(3\pi s_y)\\\cos(3\pi s_x)\sin(3\pi s_y) & \sin(3\pi s_x)\cos(3\pi s_y)\end{smallmatrix}\right)$. Setting up experiments, we vary $N_s=\{30, 50, 100\}$ and $N_t = \{3,6,9\}$ with 10 replications for each setting. These settings are chosen to record the effect of varying spatial and temporal samples on the quality of inference. % for the spatiotemporal differential processes. 
	
	\paragraph{Bayesian model fitting}\label{sec:bayes-mod-fit}
	
	We fit the model in \cref{eq:hier-spt-model} with an intercept only to learn the functional spatiotemporal patterns in the synthetic response. The following hyper-parameters are used: $a_{\phi_s}=a_{\phi_t}=0.01$, $b_{\phi_s}=b_{\phi_t}=30$, $a_\sigma = a_\tau = 2$, $b_\sigma=1$, $b_\tau=0.1$, $\bmu_\bbeta = \0_p$, $\Sigma_\bbeta=10^6\bI_p$ and $a_\nu=0.01$, $b_\nu=10$. These choices produce weakly informative priors. Generally, a Uniform$(2,3)$ prior can be placed on $\nu$ and this features in our implementation. However, the scales for our data generating patterns are chosen using model fit criteria. {  Patterns 1 and 2 are favorable to a Mat\'ern with $\nu = \frac{5}{2}$ or a Gaussian kernel. %, while patterns 3 and 4 favor $\nu=\frac{3}{2}$. % Resulting estimates from the model fit and the spatiotemporal differential processes are detailed in the Supplement, Table~S1--S4. 
		% Figs.~S93--94 show convergence diagnostics for the MCMC sampler. 
		Figs.~S10~\&~S11 compare model fit for the differential processes for Pattern 2 at $t=3$, a time point of considerable change in $y(\bs,t)$. Tables~S5--S8 present numerical results.}
	
	%Point estimates for $\Theta$ are computed using the posterior median, interval estimates are provided using highest posterior density (HPD) intervals. 
	% Posterior samples of $\beta$ and $Z(\bs,t)$ are generated one-for-one for posterior samples of $\Theta$. 
	A grid spanning the unit square is overlaid for every $t$. Posterior predictive inference on the spatiotemporal differential processes is performed at each grid location following developments in \Cref{sec:STDP}. Medians and 95\% highest posterior density (HPD) intervals \citep[][]{coda} are used to summarize inference. % Figs.~S104--S105 show convergence diagnostics for posterior samples. 
	For each replicate in the simulation setting, the quality of inference on local geometry of the space-time surface is evaluated using coverage probabilities (CP) and root mean squared errors (RMSE). 
	
	\paragraph{Spatiotemporal Surface Wombling} {  We perform wombling on spatiotemporal surfaces by specifying planar curves in advance that serve the triangulation process}. We aim to locate the curves that track rapid spatial change at each time point on the simulated surfaces. For instance, consider the space-time surfaces generated using Pattern 1. We begin with curves as they allow us to triangulate the wombling surface. Subsequently, we proceed to evaluate the posterior distribution of total or average wombling measure for the surface and assess their significance. In the absence of pre-specified curves they can be selected using \emph{level curves}: $C_{y_0^{(t)}} = \left\{(\bs,t):Y(\bs,t) = y_0^{(t)}\right\}$, for $t = 1,\ldots,N$. \Cref{fig:womb-surfaces} demonstrates this on surfaces generated from Pattern 1. Alternatively, working with \emph{annotated curves}, points selected by the investigator are used to generate a smooth curve, for example, a Be\'zier spline, specifying regions that possibly contain wombling boundaries within the space-time domain. We consider three level curves for our purposes: closed curves enclosing (A) troughs, (B) peaks and (C) an open curve; see \Cref{fig:womb-curve-1}.
	
	\begin{figure}[t]
		\centering
		\begin{subfigure}{.3\textwidth} 
			\centering
			\includegraphics[scale = 0.27]{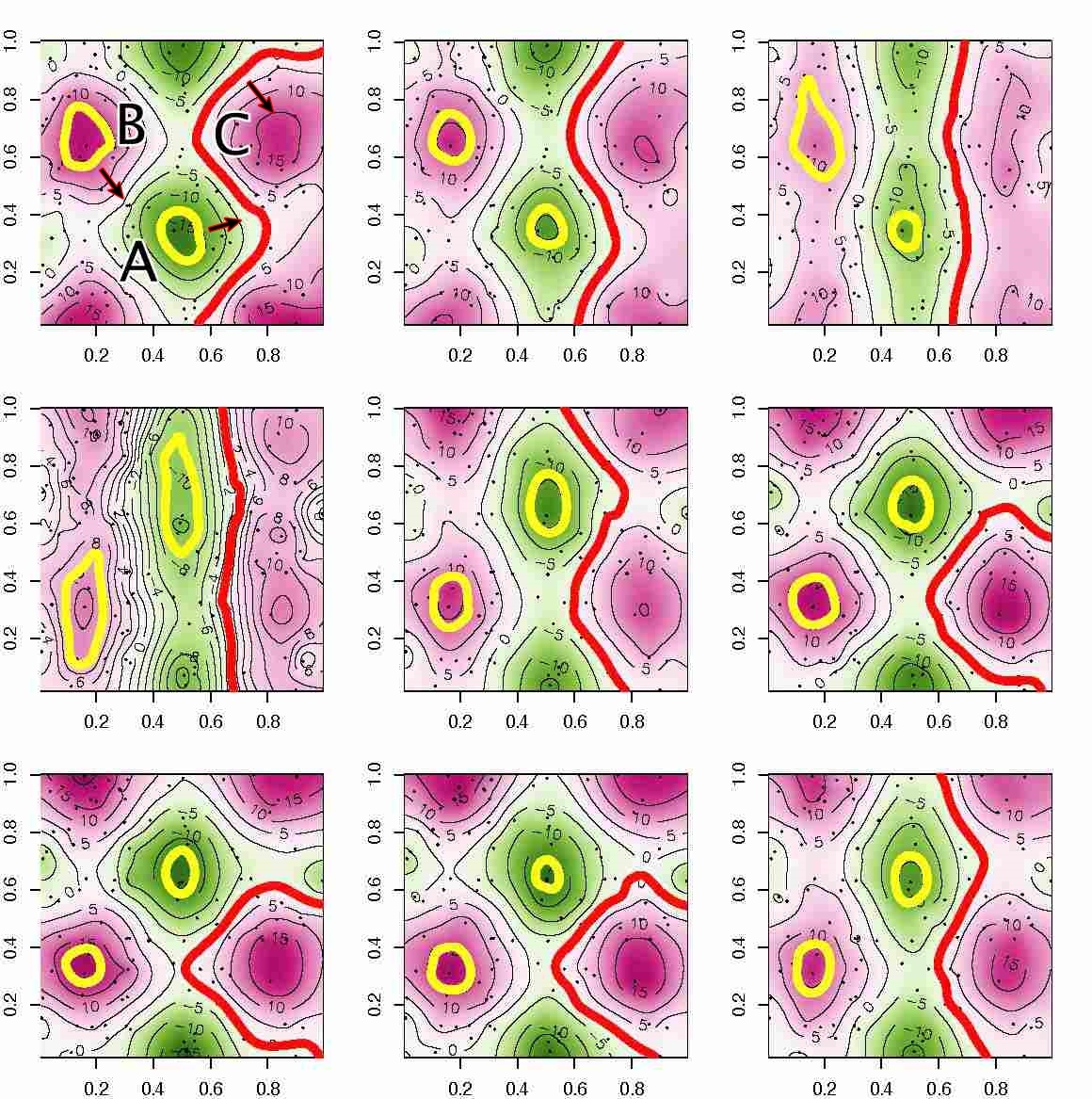}
			\caption{}\label{fig:womb-curve-1}
		\end{subfigure}%
		\begin{subfigure}{.3\textwidth}
			\centering
			\includegraphics[scale = 0.3]{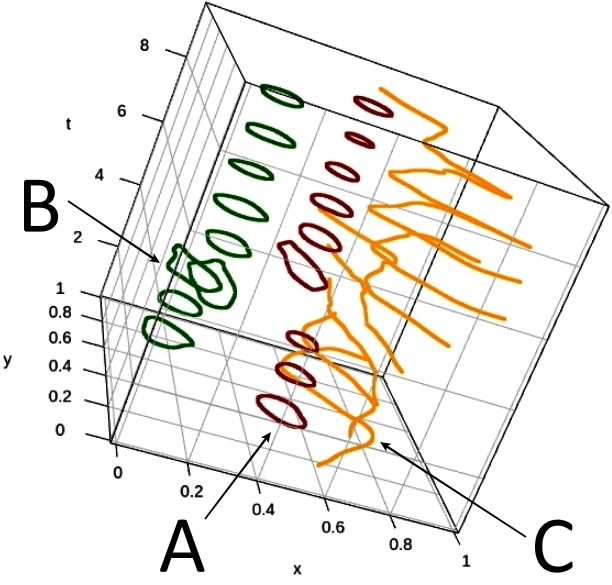}
			\caption{}\label{fig:womb-curve-2}
		\end{subfigure}%
		\begin{subfigure}{.3\textwidth} 
			\centering
			\includegraphics[scale = 0.3]{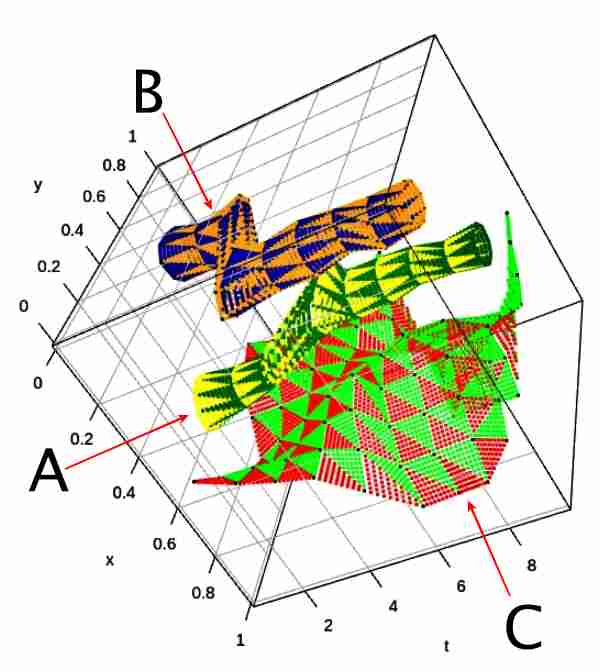}
			\caption{}\label{fig:womb-partition}
		\end{subfigure}%
		\caption{(a) Contours selected--``A" and ``B", marked in yellow, are closed curves and ``C", marked in red, is an open curve. (b) 3--D plot of the planar curves collected over time. (c) Triangulation of the \emph{wombling surfaces} associated with the curves using $n_\omega=n_\upsilon=10$.}\label{fig:womb-surfaces}
	\end{figure}
	
	Spatiotemporal surface wombling is pursued following the methods described in \Cref{sec:spt-womb}. We use a partition with $n_\omega = n_\upsilon = 10$ to triangulate the wombling surface. For each pair of triangular regions, we compute areas and outward pointing normals (see \Cref{sec:tri-surf}). The collection of all pairs represents the wombling surface, $\widetilde{\C}$. Algorithm~4 in the Supplement formulates %
	efficient computation for $\bGamma(\widetilde{\C})$ using two-dimensional quadrature (see \cref{eq:cross-cov,sec:var-womb}). Posterior samples of $\bGamma(\widetilde{\C})$ are obtained using \cref{eq:joint-dist-womb-measure}. For surfaces $A$, $B$ and $C$ the chosen partitions yield sufficiently small norms ($1.5\times 10^{-3}$, $1.4\times 10^{-3}$ and $2.2\times 10^{-3}$ respectively) thereby guaranteeing the accuracy of wombling measures (see Thm.~2 in the Supplement). The median (HPD interval) of posterior samples for $\bGamma(\widetilde{\C})$ serves as our point (interval) estimate. The design of our experiment allows us to compute, $\bGamma_{\rm true}(\C)=\iint_{\C}\L^*_{n_t,\bn_\bs}\mu_1(\bs,t)\;{\rm d}\A$. Coverage of HPD intervals and significance is assessed at triangular region, time-interval and overall levels.  
	% @Didong making the commented changes will also allow us to save a line here--quickly running out of space
	\begin{table}[t]
		\centering
		\caption{True and estimated average wombling measures for surfaces A, B and C (see \Cref{fig:womb-surfaces}) accompanied by HPD intervals. Statistically significant estimates are marked in bold.}\label{tab:wm-fit}
		\resizebox{\linewidth}{!}{
			\begin{tabular}{l|@{\extracolsep{45pt}}*{6}{>{\bfseries}c}@{}}
				\hline\hline
				\multirow{2}{*}{$\overline{\bGamma}(\C^*)$} & \multicolumn{2}{c}{A} & \multicolumn{2}{c}{B} & \multicolumn{2}{c}{C} \\ 
				\cline{2-3}\cline{4-5}\cline{6-7}
				& \normalfont{Estimate} & \normalfont{True} &  \normalfont{Estimate} & \normalfont{True}  &  \normalfont{Estimate} & \normalfont{True}  \\ 
				\hline
				\multirow{2}{*}{$\bpartial_\bs$} & 2.23 & \multirow{2}{*}{2.25} & -2.26 & \multirow{2}{*}{-2.21} & 0.98 & \multirow{2}{*}{1.04} \\ 
				& (1.65, 2.82) &  & (-2.80, -1.73) &  & (0.48, 1.49) \\ 
				\multirow{2}{*}{$\bpartial_\bs^2$} & 40.58 & \multirow{2}{*}{41.09} & -36.66 & \multirow{2}{*}{-36.26} & \normalfont{ 
					-0.14} & \multirow{2}{*}{\normalfont{-1.43}} \\
				&  (35.19, 46.01) & & (-42.43, -30.77) &  & \normalfont{(-4.48, 4.28)} &  \\\hline
				\multirow{2}{*}{$\partial_t$} & 0.16 & \multirow{2}{*}{0.15} & 0.33 & \multirow{2}{*}{0.32} & -0.19 & \multirow{2}{*}{-0.19} \\
				& (0.03, 0.30) &  & (0.20, 0.47) &  & (-0.29, -0.08) &  \\ 
				\multirow{2}{*}{$\partial_t\bpartial_\bs$} & 1.45 & \multirow{2}{*}{1.43} & 0.59 & \multirow{2}{*}{0.60} & -0.59 & \multirow{2}{*}{-0.64} \\
				& (0.94, 1.94) &  & (0.11, 1.04) &  & (-1.03, -0.18) &  \\
				\multirow{2}{*}{$\partial_t\bpartial_\bs^2$} & \normalfont{ -7.92} & \multirow{2}{*}{\normalfont{-7.00}} & -13.56 & \multirow{2}{*}{-12.07} & 11.07 & \multirow{2}{*}{10.56} \\ 
				& \normalfont{(-16.73, 0.95)} &  & (-23.38, -3.88) &  & (3.63, 18.68) &  \\\hline
				\multirow{2}{*}{$\partial_t^2$} & 0.79 & \multirow{2}{*}{0.74} & \normalfont{ -0.29} & \multirow{2}{*}{\normalfont{-0.70}} & \normalfont{ 0.34} & \multirow{2}{*}{\normalfont{0.16}} \\ 
				& (0.27, 1.34) &  & \normalfont{(-0.85, 0.29)} &  & \normalfont{(-0.17, 0.85)} &  \\
				\multirow{2}{*}{$\partial_t^2\bpartial_\bs$} & \normalfont{ -3.90} & \multirow{2}{*}{\normalfont{-4.23}} & \normalfont{\normalfont{3.62}} & \multirow{2}{*}{\normalfont{4.53}} & \normalfont{-0.22} & \multirow{2}{*}{\normalfont{-0.82}} \\ 
				& \normalfont{(-9.12, 1.35)} &  & \normalfont{(-2.15, 9.28)} &  & \normalfont{(-4.81, 4.51)} &  \\
				\multirow{2}{*}{$\partial_t^2\bpartial_\bs^2$} & \normalfont{ -21.61} & \multirow{2}{*}{\normalfont{-23.44}} & \normalfont{ 3.60} & \multirow{2}{*}{\normalfont{25.75}} & \normalfont{-37.27} & \multirow{2}{*}{\normalfont{-6.72}} \\ 
				& \normalfont{(-76.30, 36.11)} &  & \normalfont{(-65.64, 72.73)} &  & \normalfont{(-96.01, 19.40)} &  \\
				\hline\hline
			\end{tabular}
		}
	\end{table}
	Surface $A$ ($B$) encloses evolving troughs (peaks) and local minimums (maximums) while surface $C$ delineates a separating region between troughs and peaks (see \Cref{fig:womb-curve-1})%and corresponding locations in figs.~S74--91)
	. For surfaces $A$ and $C$ a positive spatial directional gradient is expected along the normal, while $B$ is expected to house negative %spatial
	gradients. They %would
	yield significant wombling measures. % corresponding to spatial gradients. 
	
	Surfaces $A$ and $B$ lie in formative regions (troughs and peaks) associated with evolving curvature, thereby expected to produce significant wombling measures associated with curvature, while surface $C$ lies in a comparatively flat region. %(see spatiotemporal Laplacians in pp. 10--11 of the Supplement and figs.~S74--91).
	With regard to temporal gradients---surface $A$ encloses troughs which flatten by $t=4$, deepen over $t=6$ before flattening again. A trough flattening over time implies positive temporal gradients; surface $B$ shows a similar trend with peaks flattening and reforming; surface $C$ encloses regions where peaks mostly flatten yielding negative temporal gradients. For surfaces $A$ and $B$ there is an overall decrease in spatial curvature over time, while for surface $C$, specifically referring to $t=6,7$ and $8$, the spatial curvature increases over time. These interpretations align with findings presented in \Cref{tab:wm-fit}, which assesses the quality of estimation for spatiotemporal wombling measures. Analysis for the overall temporal curvature is similar. A more detailed analysis follows from time-intervals or triangular-region specific considerations (see Section~S15). Fig.~S14 shows 3-D plots displaying significance at a triangular-region level for the wombling surfaces we consider. The magnitude and algebraic sign of the wombling measure helps in discerning the nature and differentiating between boundaries. For example, the rate of curvature %, inferred from temporal gradient of spatial curvature (row 5, \Cref{tab:wm-fit}),
	decreases at a slower rate on surface $A$ compared to $B$. Similar conclusions follow from Tables~S9~and~S10. 
	
	\paragraph{Additional Experiments} The online supplement contains auxiliary experiments and results. Tables~S5--S8 present goodness of fit measures for assessing spatiotemporal differential processes using RMSE, standard deviations and CPs for 95\% HPD intervals addressing Patterns 1 \& 2. Tables~S5~\&~S6 use \emph{non-separable} kernels while S7~\&~S8 use \emph{separable} kernels. Increased spatial (temporal) locations result in decreased RMSE for spatial (temporal) derivatives. For instance, comparing results for $N_s \times N_t=50\times 6$ with $N_s \times N_t=100\times 3$, we find lower (higher) RMSEs for spatial (temporal) differential processes in the latter. Separable covariance kernels provide fast inference on rates of change and wombling, albeit their several shortcomings \citep[][]{stein2005space}. Figs.~S10~\&~S11 compare true and estimated differential processes for Patterns 2 using a Mat\'ern $(\nu=5/2)$ kernel at $t=3$. %; Figs.~S38--S55 and~S56--S73 address Patterns 3 \& 4 % (compare with simulation in \cite{quick2015bayesian}) 
	% respectively using a Mat\'ern $(\nu=3/2)$ kernel. % They demonstrate our framework's ability to learn geometry under varying process smoothness specifications. 
	{  We are able to learn the induced spatiotemporal differential processes under varying smoothness in process specification.}
	Figs.~S12~\&~S13 demonstrate inference for differential geometric quantities--Gaussian curvature, divergence and Laplacian, that capture change. %(see pp. 10--11, Supplement).
	
	\section{Application: Neuroimaging}\label{sec:app}
	
	\begin{figure}[t]
		\centering
		\includegraphics[width = \linewidth, page = 1]{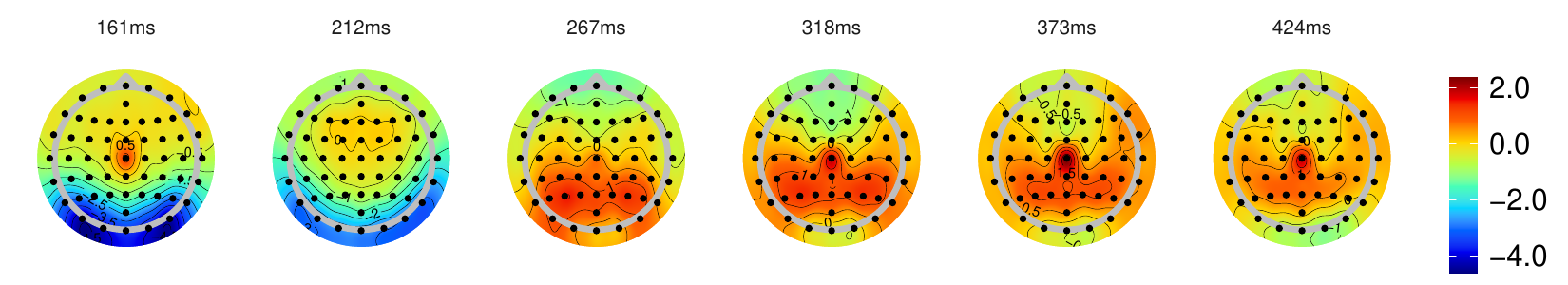}
		
		\vspace*{-.5cm}
		
		\includegraphics[width = \linewidth, page = 2]{plots/eeg_response.pdf}
		\caption{ERP (measured in \si{\micro\volt}) surfaces for alcoholic (top row) and control (bottom row) subjects.}\label{fig:eeg}
	\end{figure}
	
	\begin{table}[t]
		\caption{Posterior estimates accompanied by HPD intervals for $\Theta$ from the ERP analyses.}\label{tab:erp-model}
		\begin{minipage}{.5\linewidth}
			\resizebox{\linewidth}{!}{
				\centering
				\begin{tabular}{l|@{\extracolsep{80pt}}cc@{}}
					\hline\hline
					\multirow{2}{*}{Parameter} & \multicolumn{2}{c}{Estimate} \\
					\cline{2-3}
					& Alcoholic & Control\\
					\hline
					\multirow{2}{*}{$\sigma^2$} & 5.55 & 14.55 \\
					& (1.57, 13.22) & (5.90, 30.43) \\ 
					\multirow{2}{*}{$\tau^2$} & 0.10 & 0.07 \\
					& (0.08, 0.11) & (0.05, 0.08) \\ 
					\multirow{2}{*}{$\phi_s$} & 1.25 & 1.53 \\
					& (0.86, 1.63) & (1.27, 1.90) \\
					\hline
					\hline
				\end{tabular}
			}
		\end{minipage}%
		\begin{minipage}{.5\linewidth}
			\resizebox{\linewidth}{!}{
				\centering
				\begin{tabular}{l|@{\extracolsep{80pt}}cc@{}}
					\hline\hline
					\multirow{2}{*}{Parameter} & \multicolumn{2}{c}{Estimate} \\
					\cline{2-3}
					& Alcoholic & Control\\
					\hline
					\multirow{2}{*}{$\phi_t$} & 0.42 & 0.36 \\
					& (0.34, 0.55) & (0.30, 0.43)\\
					\multirow{2}{*}{$\beta_0$} & -0.23 & 0.82 \\
					& (-0.27, -0.20) & (0.80, 0.85) \\
					&  & \\
					&  & \\
					\hline
					\hline
				\end{tabular}
			}
		\end{minipage}
	\end{table}
	
	\begin{figure}[t]
		\centering
		\includegraphics[width = \linewidth]{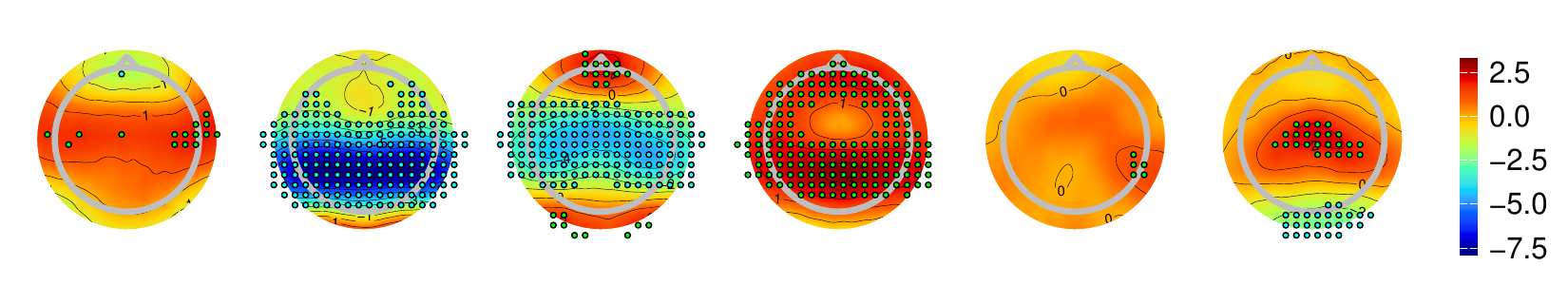}
		
		\vspace*{-.5cm}
		
		\includegraphics[width = \linewidth]{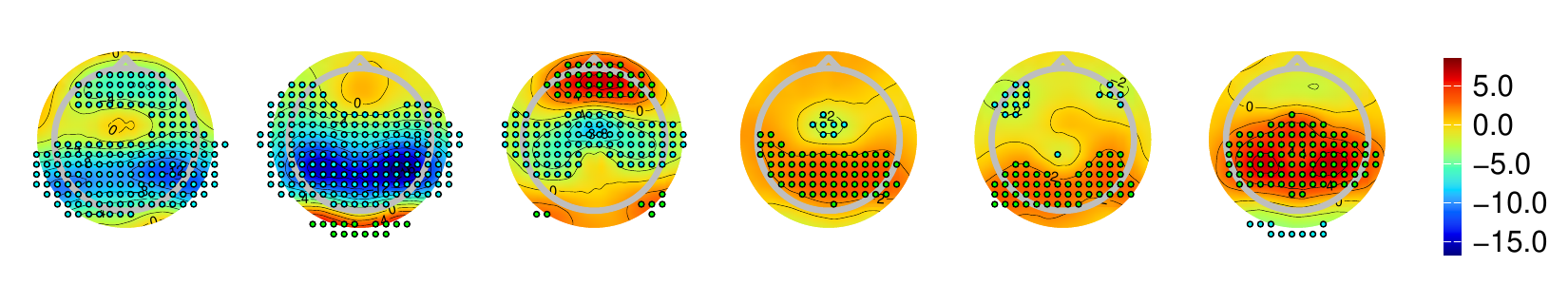}
		
		\vspace*{-.5cm}
		
		\includegraphics[width = \linewidth,]{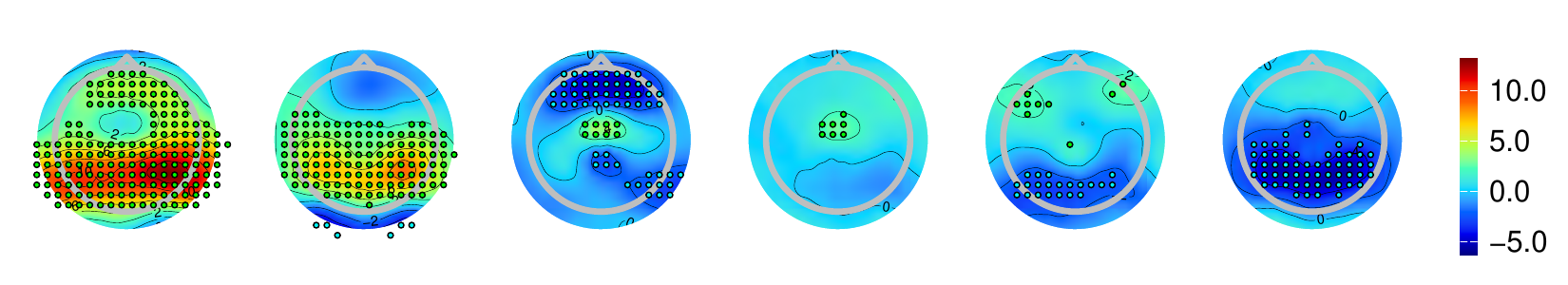}
		\caption{Comparing {\color{purple} surfaces generated from spatial point estimates of} $\partial_t\partial_{s_{y}}$ in ERPs for alcoholic (first row) versus control (second row) and their difference (third row). Significant locations are color-coded: positive/negative (\texttt{green}/\texttt{cyan}).}\label{fig:gradsyt}
	\end{figure}
	
	%\noindent{\bf Neural Activity and Alcoholism:} 
	
	We study neural activity from EEG sessions designed to test genetic predisposition to alcoholism \citep[][]{misc_eeg_database_121}. Subjects are classified into two groups: alcoholic and control. A standard set of pictures %\citep[see][]{snodgrass1980standardized} 
	serve as stimulus generating neural activity recorded by the EEG cap through ERPs. \Cref{fig:eeg} shows interpolated spatiotemporal surfaces for average ERPs over six time points.
	%Scientific interest lies in comparing and identifying active regions of the brain for the two groups \citep[][]{hu2015local}. Our advances leverage spatiotemporal wombling in this regard.
	{ Heavy alcohol users show reduced activity in the occipital and temporal regions \citep[see, e.g.,][]{lew2020altered}. These regions are demarcated by curves on the scalp (see Fig. S15). We are interested in comparing neural activity in these regions between the two groups. The entirety of these regions are not constantly active throughout the measurement period (see \Cref{fig:eeg}). 
		Therefore, a space-time surface displaying regions that exhibit significant changes in neural activity for the two groups is of scientific interest \citep[see, e.g.][]{lenartowicz2014electroencephalography}. Spatiotemporal wombling can answer these questions.}
	
	\emph{Projected locations} of 57 electrodes on the EEG cap are used as coordinates (marked with dots in \Cref{fig:eeg}) to model the ERP using \cref{eq:hier-spt-model}. We use a Mat\'ern kernel with fractal parameter, $\nu=5/2$ that ensures the existence and validity of all spatiotemporal processes within $\L^*Z(\bs,t)$. Weakly informative priors with hyper-parameter settings outlined in \Cref{sec:bi&c} are used to conduct Bayesian inference. The resulting posterior estimates are shown in \Cref{tab:erp-model}. A larger value of the variance ratio, $\frac{\sigma^2}{\sigma^2+\tau^2}$ indicates a larger proportion of the total variance being attributed to spatiotemporal change. This implies the presence of strong {\em spatiotemporal {  variation}} in the data. It is stronger in the control subjects compared to the alcoholic subjects, given $\frac{\widehat{\sigma^2}}{\widehat{\sigma^2}+\widehat{\tau^2}} = \frac{14.55}{14.55 + 0.07}=99.52\%$ and $\frac{\widehat{\sigma^2}}{\widehat{\sigma^2}+\widehat{\tau^2}} = \frac{5.55}{5.55 + 0.10} = 98.23\%$ of the variation being spatiotemporal respectively. % Comparing the spatial and temporal range parameters we observe that both groups display more temporal compared to spatial dependence. 
	\Cref{fig:gradsyt} compares the {  surfaces generated from spatial point estimates of} mixed spatial-temporal gradient along the $y$-axis for the two groups and their difference. We see significant differences at times 1, 2, 5 and 6. { For a location $\bs$, $\partial_t\partial_{s_y}Z(\bs) < 0$ implies increased spatial change over time towards the occipital region. This is seen in the control group (second row, \Cref{fig:gradsyt}) for time points $t=1,2$. The third row tests the difference of ERPs between alcoholic and control groups. A significant positive value implies a higher and quicker rate of change in neural activation of the occipital lobe for the control subjects.} Evidently, spatiotemporal gradient analysis aids in discerning differences between groups. 
	
	\begin{figure}[t]
		\centering
		\includegraphics[width = \linewidth]{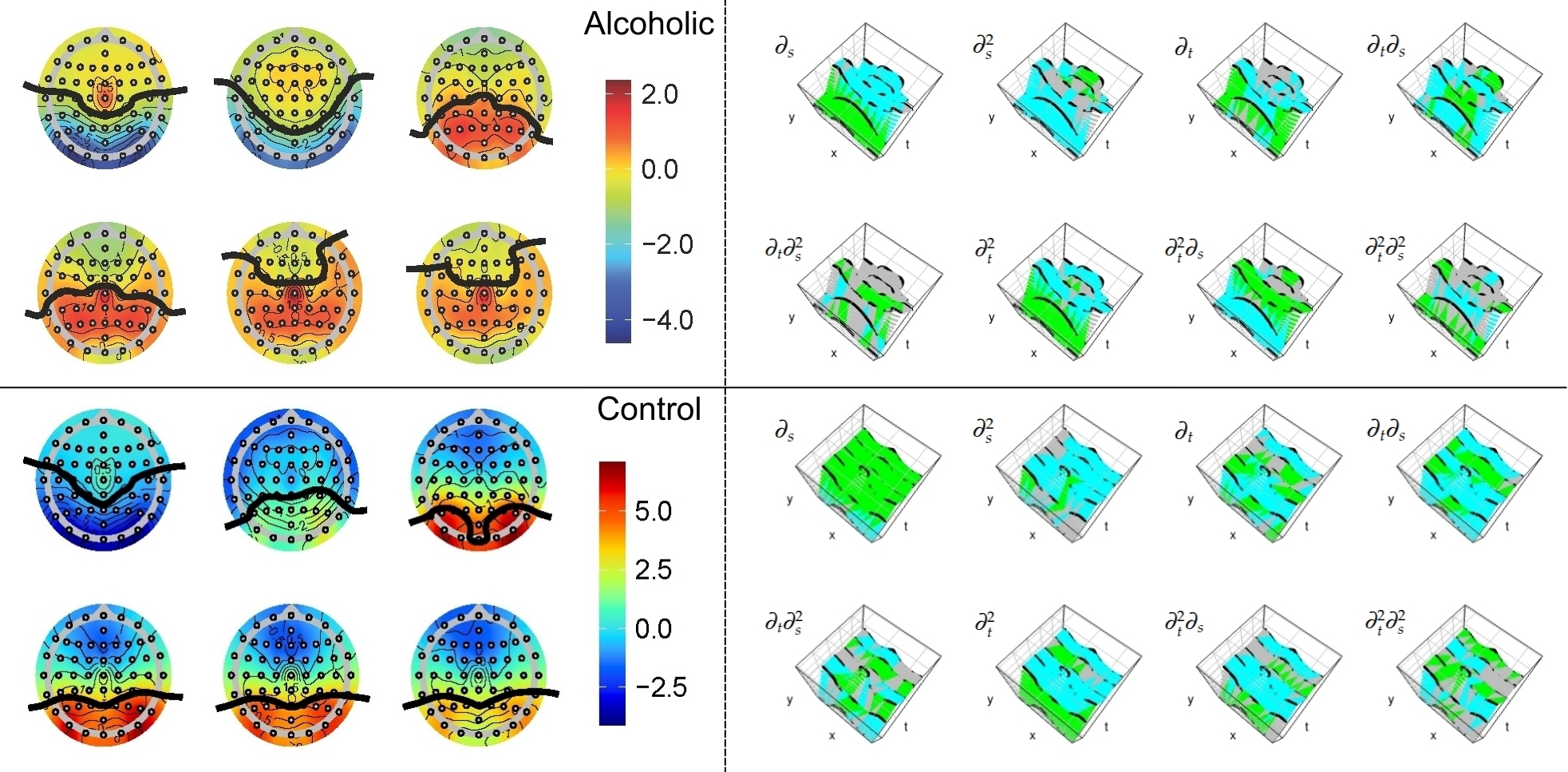}
		\caption{Curves and resulting wombling surfaces for ERPs from alcoholic and control subjects. Significant triangular regions are color-coded: positive/negative (\texttt{green}/\texttt{cyan}), not significant (\texttt{grey}).}\label{fig:wm-eeg}%
		
		\vspace*{-0.3cm}
		
		\centering
		\includegraphics[width=\linewidth, page = 1]{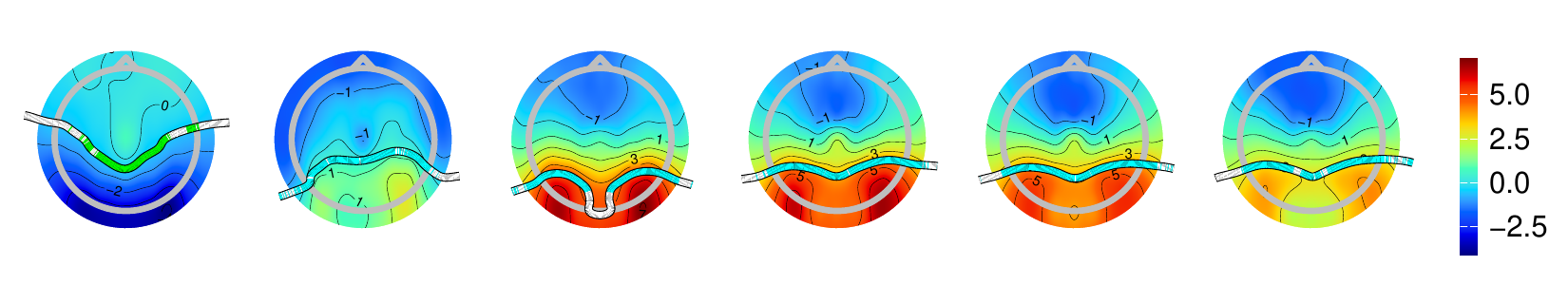}
		
		\vspace*{-0.8cm}
		
		\includegraphics[width=\linewidth, page = 2]{plots/SIG_SPWOMB_CON_EEG.pdf}
		
		\vspace*{-0.5cm}
		
		\caption{Significance assessments for (top row) spatial gradients ($\bpartial_\bs$) (bottom row) spatial curvature ($\bpartial_\bs^2$) over curves used for the control group performing purely spatial wombling. Significant regions are color coded, while white regions indicate no significance.}\label{fig:significance-sp-wombling}
	\end{figure}
	
	% Overall, the surfaces show similar patterns with higher magnitudes and larger number of significant regions detected in the control group (see figs.~S117--120 for further differential process assessment).
	
	Spatiotemporal wombling is executed on the residual surfaces, $Z(\bs,t)$, {  which is not a requirement, but a better choice as it displays variation that is free from the effects explained by covariates, if present}. { To generate the surface we select the 0 contour. It acts as a separator between regions of high and low activity. The regions enclosed by the curves overlap the occipital and temporal regions.}
	%For both groups, we select \emph{open level curves}, $C_t$, encompassing locations which show a minimal positive ERP (see \Cref{fig:wm-eeg}).
	{ The resulting wombling surface}, generated by rotating the topographic head-plots anti-clockwise, enclose spatiotemporally active locations. Triangulation (see \Cref{sec:tri-surf}) yields the wombling surface approximation. { The top and bottom right quadrant of \Cref{fig:wm-eeg} shows the triangulated approximation of the wombling surface enclosing active regions generated for each group}. It also shows the triangular region level significance for wombling measures. Comparing the top and bottom right quadrants of \Cref{fig:wm-eeg} reveals a larger proportion of triangular regions with statistically significant wombling measures for the control group than the alcoholic group for all eight measures, except perhaps for $\bpartial_\bs$. Comparing average wombling measures in \Cref{tab:wm-eeg}, the alcoholic subjects show significant wombling measures associated only with $\bpartial_\bs^2$ and $\partial_t\bpartial_\bs$. The control group reveal significant measures, with the exception of $\partial_t\bpartial_\bs^2$ and $\partial_t^2\bpartial_\bs^2$. {  This implies there were significant spatial and temporal changes in the neural activity, however the  overall temporal rates of spatial change, measured by wombling measures corresponding to $\partial_t\bpartial_\bs^2$ and $\partial_t^2\bpartial_\bs^2$ (see \Cref{tab:wm-eeg}), stayed relatively constant over time. This is a scientific insight of interest in EEG and alcoholism \citep[see, for e.g.,][]{durazzo_regional_2024}.} A topographic shift occurs after $t=2$ in both groups, %which is seen more vividly in 
	as seen from the time-interval level analysis in Tables~S12~\&~S13, in the Supplement. { The spatiotemporal wombling exercise provides concrete statistical evidence of higher neural activity in the occipital and temporal regions for the control group.}
	
	\begin{table}[t]
		\caption{Average wombling measures (WMs) for wombling surfaces selected for the two groups.}\label{tab:wm-eeg}
		\resizebox{\linewidth}{!}{
			\begin{tabular}{c|@{\extracolsep{2pt}}*{8}{>{\bfseries}c}@{}}
				\hline\hline
				\multirow{2}{*}{Group} &\multicolumn{8}{c}{$\overline{\bGamma}(\C)\times 10^2$}\\
				\cline{2-9}
				& $\bpartial_\bs$ & $\bpartial_\bs^2$ & $\partial_t$ & $\partial_t\bpartial_\bs$ & $\partial_t \bpartial_\bs^2$ & $\partial_t^2$ & $\partial_t^2\bpartial_\bs$ & $\partial_t^2\bpartial_\bs^2$ \\ 
				\hline
				\multirow{2}{*}{Alcoholic}& \normalfont{-7.90} & -105.94 & \normalfont{0.09} & -1.52 & \normalfont{1.68} & \normalfont{-0.01} & \normalfont{-0.17} & \normalfont{-1.36} \\ 
				& \normalfont{(-17.93, 2.64)} & (-170.77, -40.15) & \normalfont{(-0.35, 0.52)} & (-2.74, -0.32) & \normalfont{(-7.30, 11.07)} & \normalfont{(-0.15, 0.13)} & \normalfont{(-0.78, 0.42)} & \normalfont{(-4.81, 2.45)} \\
				\hline
				\multirow{2}{*}{Control} &  151.42  & -307.12 & -0.86 & -7.43 & \normalfont{-8.57} & 0.39 & -0.83 & {\normalfont -1.56}\\
				& (135.52, 167.30) & (-419.33, -192.83) & (-1.23, -0.47) & (-8.51, -6.33) & \normalfont{(-20.32, 2.23)} & (0.26, 0.54) & (-1.32, -0.34) & {\normalfont (-5.01, 2.13)} \\ 
				\hline\hline
			\end{tabular}
		}
		
		\medskip
		
		\centering
		\caption{Average spatial WMs for gradient, $\Gamma_1(C_t)$ and curvature, $\Gamma_2(C_t)$ for $C_t$ in \Cref{fig:significance-sp-wombling}.}\label{tab:wm-sp-con-eeg}
		\resizebox{\linewidth}{!}{
			\begin{tabular}{c|@{\extracolsep{2pt}}*{6}{c}@{}}
				\hline\hline
				Wombling & \multicolumn{6}{c}{Time}\\ 
				\cline{2-7}
				Measures & $t=1$ & $t=2$ & $t = 3$ & $t=4$ & $t=5$ & $t=6$\\
				\hline
				\multirow{2}{*}{$\Gamma_1(C_t)$} & {\bf 2.21}  & 0.05  & -1.81 & {\bf -3.74}  & {\bf -3.80}  & {\bf -3.41}\\ 
				& {\bf (0.41, 4.18)} & (-2.40, 2.49) &  (-4.82, 1.03) & {\bf (-7.31, -0.83)} & {\bf (-6.67, -0.67)} & {\bf (-5.80, -1.02)}\\
				\multirow{2}{*}{$\Gamma_2(C_t)$} & -7.86  & -23.30   & -23.11  & -1.05  & 1.58   & 4.16 \\
				& (-25.51, 9.79) & (-46.73, 2.43) & (-48.72, 1.27) & (-33.43, 32.20) & (-28.05, 30.80) & (-19.94, 27.40)\\
				\hline\hline
			\end{tabular}
		}
	\end{table}     
	We also compared spatiotemporal wombling with purely spatial wombling. Owing to higher neural activity we consider the control group for our comparison. Estimates and HPD intervals for average wombling measures for spatial gradient $(\Gamma_1(C_t))$ and curvature $(\Gamma_2(C_t))$ are shown in \Cref{tab:wm-sp-con-eeg}. Wombling is performed on the same set of planar curves as seen in the bottom left quadrant of \Cref{fig:wm-eeg}. Statistically significant curve segments are shown in \Cref{fig:significance-sp-wombling}. Purely spatial wombling does not yield a significant spatial curvature which is in contrast with the findings from its spatiotemporal counterpart. % (see \Cref{tab:wm-eeg,tab:wm-sp-con-eeg}).
	The proportion of significant regions on the wombling surface for spatial gradients % in spatiotemporal wombling 
	is also larger compared to the same for curve segments. %in spatial wombling 
	% (compare \Cref{fig:wm-eeg,fig:significance-sp-wombling}). 
	The temporal % and spatial-temporal mixed 
	rates of change are unaccounted for in purely spatial wombling, for example, the increase in $\Gamma_1(C_t)$, from $t=1$ to $t=2$. Thus, in the presence of temporal variation in the data, purely spatial wombling fails to account for the extent of variability, thereby undermining the quality of inference. % for wombling.
	
	\section{Discussion}\label{sec:diss}
	
	We have developed a fully model-based Bayesian framework for inference on differential processes that quantify rates of change in spatiotemporal data. These processes are used to formalize wombling surfaces---a landmark in the existing literature for wombling. The resulting difference surface boundaries and measures aid in discerning zones of rapid spatiotemporal change. % and comparison between subgroups. 
	GPs prove crucial for calibrating inference for differential processes and wombling. { Section~S13 in the supplement presents detailed comparisons of wombling with related approaches such as clustering, image segmentation, hot-spot detection and level set modeling}. %A detailed comparison can be found in the Supplement, Section~S13.
	Wombling surfaces are 2-dimensional compact Riemannian manifolds which can be triangulated in multiple non-unique ways (cf. Rad{\'o}'s theorem). As a result, they can also be constructed using space curves that are interpreted as temporal wombling boundaries which evolve in space.
	
	% This framework is broadly applicable---selected scenarios of interest have been showcased.   
	Future work can proceed along several directions. % to produce wombling surfaces. 
	Extending Bayesian wombling beyond Euclidean domains using GPs on manifolds \citep[see, e.g.,][]{li2023manifold,li_statistical_2026} is especially relevant for modern neuroimaging studies. % brain imaging applications, where we work with 3-D coordinates for electrodes with due regard to the geometry of the human scalp---a 2-D compact Riemannian manifold. 
	Spatiotemporal energetics is an emerging area related to streaming data from wearable devices that holds significant promise \citep[][]{alaimo2023bayesian}. Exploring differential activity can provide scientific insights into human movement and energy patterns. % Finally, extending \cite{wang2018process} to model temporal rates of change in aspect and principal curvature, defined as the directions of maximum gradient and curvature respectively, can explore connections with projected GPs \citep[see. e.g.,][]{wang2014projected}.
	Open-source statistical software development can expand on \cite{halder2025nimblewomble}. Finally, embedding wombling within the predictive stacking inferential framework \citep[see, for e.g.,][]{zhang_bayesian_2025} has several theoretical and computational advantages.
	
	\section{Acknowledgments and Funding}
	AH acknowledges funding from the National Science Foundation, Division of Mathematical Sciences (DMS 2515897), the National Institutes of Health, National Institute of Environmental Health Sciences (R56 ES036250), and institutional funds at the Dornsife School of Public Health and the Urban Health Collaborative at Drexel University. DL acknowledges funding from the National Science Foundation, Division of Mathematical Sciences (DMS 2515899), and the National Institutes of Health (P42 ES031007, P50 CA058223, P30 ES010126, R01 LM014407). SB acknowledges funding from the National Science Foundation, Division of Mathematical Sciences (DMS 2515898), and the National Institute of Health, National Institute of General Medical Sciences (R01 GM148761).
	
	\section{Conflict of Interest Disclosure Statement}
	The authors declare no conflicts of interest at this time.

	\title{\bf  Supplementary Materials for\\ {\em ``Bayesian Spatiotemporal Wombling"}}
	\date{}
	
	\maketitle
	
	\beginsupplement
	
	\section{Tensor reshaping: Vectorization}\label{sec::vec-tensor}
	
	We supplement the discussion in Section 2.1 of the manuscript with details concerning reshaping of tensor differential operators. We are interested in the vectorization of tensors--reshaping tensors to vectors. The order of listing for elements in the tensor operator $\bpartial^{\otimes r}$ operating on $f:\mathfrak{R}^{d+1}\to\mathfrak{R}$ needs specification. We begin with an explicit definition of $\bpartial^{\otimes r}$ in terms of its coordinates.
	\begin{equation}\label{eq:nablaor}
		\bpartial^{\otimes r}f(\bx,t) = \left(\begin{smallmatrix}
			\frac{\partial}{\partial x_1}\\
			\vdots\\
			\frac{\partial}{\partial x_d}\\
			\frac{\partial}{\partial t}
		\end{smallmatrix}\right)^{\otimes r}f(\bx,t),
	\end{equation}
	where $\bx = (x_1,\ldots, x_d)^{\T}$. We illustrate the operation using $d=2$ and $r = 1, 2, 3$ and 4. For $r=1$, $\bpartial^{\otimes 1}f(\bx,t) = \left(\begin{smallmatrix}
		\frac{\partial}{\partial x_1}\\
		\frac{\partial}{\partial x_2}\\
		\frac{\partial}{\partial t}
	\end{smallmatrix}\right)f(\bx,t) = \left(\begin{smallmatrix}
		\frac{\partial}{\partial x_1}f(\bx,t)\\
		\frac{\partial}{\partial x_2}f(\bx,t)\\
		\frac{\partial}{\partial t}f(\bx,t)
	\end{smallmatrix}\right) = \bpartial f(\bx,t)$. For $r = 2$,
	\begin{equation}\label{eq::vechess}
		\begin{split}
			\bpartial^{\otimes 2}f(\bx,t) &= \left(\begin{smallmatrix}
				\frac{\partial}{\partial x_1}\\
				\frac{\partial}{\partial x_2}\\
				\frac{\partial}{\partial t}
			\end{smallmatrix}\right)^{\otimes 2}f(\bx,t) = \left(\begin{smallmatrix}
				\frac{\partial}{\partial x_1}\\
				\frac{\partial}{\partial x_2}\\
				\frac{\partial}{\partial t}
			\end{smallmatrix}\right) \otimes \bpartial f(\bx,t)\\
			& = \left(
			\frac{\partial^2}{\partial x_1^2},
			\frac{\partial^2}{\partial x_1\partial x_2},
			\frac{\partial^2}{\partial x_1\partial t},
			\frac{\partial^2}{\partial x_2\partial x_1},
			\frac{\partial^2}{\partial x_2^2},
			\frac{\partial^2}{\partial x_2\partial t},
			\frac{\partial^2}{\partial t\partial x_1},
			\frac{\partial^2}{\partial t\partial x_2},
			\frac{\partial^2}{\partial t^2}
			\right)^{\T}f(\bx,t)\\
			&= \left(
			\frac{\partial^2}{\partial x_1^2}f(\bx,t),
			\frac{\partial^2}{\partial x_1\partial x_2}f(\bx,t),
			\frac{\partial^2}{\partial x_1\partial t}f(\bx,t),
			\frac{\partial^2}{\partial x_2\partial x_1}f(\bx,t),
			\frac{\partial^2}{\partial x_2^2}f(\bx,t),\right.\\&\myquad[4]\left.
			\frac{\partial^2}{\partial x_2\partial t}f(\bx,t),
			\frac{\partial^2}{\partial t\partial x_1}f(\bx,t),
			\frac{\partial^2}{\partial t\partial x_2}f(\bx,t),
			\frac{\partial^2}{\partial t^2}f(\bx,t)
			\right)^{\T},
		\end{split}
	\end{equation}
	This is the $9 \times 1$ vectorized Hessian. { For $r= 3$, $ \bpartial^{\otimes 3} f(\bx,t)=\left(\begin{smallmatrix}
			\frac{\partial}{\partial x_1}\\
			\frac{\partial}{\partial x_2}\\
			\frac{\partial}{\partial t}
		\end{smallmatrix}\right)^{\otimes 3}f(\bx,t) = \left(\begin{smallmatrix}
			\frac{\partial}{\partial x_1}\\
			\frac{\partial}{\partial x_2}\\
			\frac{\partial}{\partial t}
		\end{smallmatrix}\right) \otimes \bpartial^{\otimes 2}f(\bx,t)\\$ produces the vectorized 3-order tensor containing all pure and partial derivatives. Similar calculations follow for $r=4$.}

	\section{Details for calculating the covariance between derivative processes, \texorpdfstring{$\partial_t^{j_1}\bpartial_\bs^{r_1-j_1}Z(\bs,t)$}{Zr1} and \texorpdfstring{$\partial_t^{j_2}\bpartial_\bs^{r_2-j_2}Z(\bs,t)$}{Zr2}}
	
	For ease of notation, we write $\frac{\partial^r}{\partial_t^{j}\partial x_1^{i_1}\ldots\partial x_d^{i_d}}Z(\bs,t)$, $\sum\limits_{l=1}^{d}i_l=r-j$ as $\partial_t^{j}\bpartial_\bs^{r-j}Z(\bs,t)$, for $j = 0,\ldots, r$. We aim to derive the covariance between two spatiotemporal derivative processes of possibly different orders.
	
	\begin{proposition}\label{prop:1}
		Let $\partial_t^{j_1}\bpartial_\bs^{r_1-j_1}Z(\bs,t)$ and $\partial_t^{j_2}\bpartial_\bs^{r_2-j_2}Z(\bs,t)$ be two spatiotemporal derivative processes of order $r_1$ and $r_2$. Then,
		
		\begin{equation*}
			\Cov\left(\partial_t^{j_1}\bpartial_\bs^{r_1-j_1}Z(\bs,t), \partial_t^{j_2}\bpartial_\bs^{r_2-j_2}Z(\bs +\bDelta,t + \delta)\right)=(-1)^{r_1}\partial_t^{j_1+j_2}\bpartial_\bs^{r_1+r_2-(j_1+j_2)}\;K(\bDelta,\delta),
		\end{equation*}
		for every $j_1 = 0,\ldots, r_1$ and $j_2=0,\ldots,r_2$.
	\end{proposition}
	
	{\allowdisplaybreaks
		\begin{proof}
			We express the covariance between the two processes as a limit of the covariance between its finite difference approximations as,
			% do not delete this space--latex puts extra gap between equation and text
			
			\begin{equation*}
				\Cov\left(\partial_t^{j_1}\bpartial_\bs^{r_1-j_1}Z(\bs,t), \partial_t^{j_2}\bpartial_\bs^{r_2-j_2}Z(\bs +\bDelta,t + \delta)\right) = \lim_{k\to 0}\lim_{h\to 0}\Cov\left(Z^{(j_1,r_1)}_{k,\bI_d,1}(\bs,t),\;Z^{(j_2,r_2)}_{h,\bI_d,1}(\bs+\bDelta,t+\delta)\right),
			\end{equation*}
			where $Z^{(j_1,r_1)}_{k,\bI_d,1}(\bs,t)=k^{-r_1}\prod\limits_{l=1}^{d}(D_{k,\be_l,0}-1)^{i_l}(D_{k,0,1}-1)^{j_1}Z(\bs,t)$, $Z^{(j_2,r_2)}_{h,\bI_d,1}(\bs+\bDelta,t+\delta)=h^{-r_2}\prod\limits_{l'=1}^{d}(D_{h,\be_{l'},0}-1)^{l'}(D_{h,0,1}-1)^{j_2}Z(\bs+\bDelta,t+\delta)$, $\sum\limits_{l=1}^{d}i_l=r_1-j_1$ and $\sum\limits_{l'=1}^{d}i_{l'}=r_2-j_2$. Using these expressions and performing binomial expansion on the spatiotemporal shift operator we obtain
			
			\begin{align*}
				%\begin{split}
				&\lim_{k\to 0}\lim_{h\to 0}\Cov\left(Z^{(j_1,r_1)}_{k,\bI_d,1}(\bs,t),\;Z^{(j_2,r_2)}_{h,\bI_d,1}(\bs+\bDelta,t+\delta)\right) =\\
				&\lim_{k\to 0}\lim_{h\to 0}\;k^{-r_1}\;h^{-r_2}\;\Cov\left(\;\prod\limits_{l=1}^{d}\;\sum\limits_{i=1}^{i_l}\sum\limits_{j=1}^{j_1}{\binom{i_l}{i}}\;{\binom{j_1}{j}}\;(-1)^{(i_l+j_1)-(i+j)}\;Z(\bs+ik\be_l,t+jk),\right.\\&\myquad[8]\left.\prod\limits_{l'=1}^{d}\;\sum\limits_{i'=1}^{i_{l'}}\sum\limits_{j'=1}^{j_2}{\binom{i_{l'}}{i'}}\;{\binom{j_2}{j'}}\;(-1)^{(i_{l'}+j_2)-(i'+j')}\;Z(\bs+\bDelta+i'h\be_{l'},t+\delta+j'h)\right),\\
				% &=\lim_{h\to 0}\;\lim_{k\to 0}\;k^{-r_1}\;h^{-r_2}\;\prod\limits_{l=1}^{d}\;\prod\limits_{l'=1}^{d}\;\sum\limits_{i=1}^{i_l}\sum\limits_{j=1}^{j_1}\sum\limits_{i'=1}^{i_{l'}}\sum\limits_{j'=1}^{j_2}{\binom{i_l}{i}}\;{\binom{j_1}{j}}\;{\binom{i_{l'}}{i'}}\;{\binom{j_2}{j'}}\;(-1)^{(i_l+i_{l'}+j_1+j_2)-(i+j+i'+j')}\cdot\\
				% &\myquad[8] K\left(\bDelta+i'he_{l'}-ike_l,\delta+(j'h-jk)\right),\\
				% &=\lim_{h\to 0}\;\lim_{k\to 0}\;k^{-r_1}\;h^{-r_2}\prod\limits_{l=1}^{d}\;\prod\limits_{l'=1}^{d}\;\sum\limits_{i=1}^{i_l}\sum\limits_{j=1}^{j_1}\sum\limits_{i'=1}^{i_{l'}}\sum\limits_{j'=1}^{j_2}{\binom{i_l}{i}}\;{\binom{j_1}{j}}\;{\binom{i_{l'}}{i'}}\;{\binom{j_2}{j'}}\;(-1)^{(i_l+i_{l'}+j_1+j_2)-(i+j+i'+j')}\cdot\\
				% &\myquad[8] D_{i'h,e_{l'},0}\;D_{-ik,e_l,0}\;D_{j'h-jk,0,1}\;K(\bDelta,\delta),\\
				% &=\lim_{h\to 0}\;\lim_{k\to 0}\;k^{-r_1}\;h^{-r_2}\prod\limits_{l=1}^{d}\;\prod\limits_{l'=1}^{d}\;\sum\limits_{i=1}^{i_l}\sum\limits_{j=1}^{j_1}\sum\limits_{i'=1}^{i_{l'}}\sum\limits_{j'=1}^{j_2}{\binom{i_l}{i}}\;{\binom{j_1}{j}}\;{\binom{i_{l'}}{i'}}\;{\binom{j_2}{j'}}\;(-1)^{(i_l+i_{l'}+j_1+j_2)-(i+j+i'+j')}\cdot\\
				% &\myquad[8] D_{h,e_{l'},0}^{i'}\;D_{h,0,1}^{j'}\;D_{-k,e_l,0}^{i}\;D_{-k,0,1}^{j}\;K(\bDelta,\delta),\\
				% &=(-1)^{r_1}\;\lim_{h\to 0}\;\lim_{k\to 0}\;(-k)^{-r_1}h^{-r_2}\prod\limits_{l=1}^{d}\;(D_{-k,e_l,0}-1)^{i_l}\;(D_{-k,0,1}-1)^{j_1}\cdot\\&\myquad[15]\prod\limits_{l'=1}^{d}\;(D_{h,e_{l'},0}-1)^{i_{l'}}\;(D_{h,0,1}-1)^{j_2}K(\bDelta,\delta),\\
				&=(-1)^{r_1}\partial_t^{j_1+j_2}\bpartial_\bs^{r_1+r_2-(j_1+j_2)}\;K(\bDelta,\delta),
				%\end{split}
			\end{align*}
			for every $j_1 = 0,\ldots, r_1$ and $j_2 = 0,\ldots, r_2$. { The last equality is obtained after some algebra.}
		\end{proof}
	}
	The algebraic sign depends on the order of the first derivative process. \Cref{sec:cross-cov} provides the cross-covariance matrix using the above result for the process, $\left(\begin{smallmatrix}
		Z(\bs,t)\\ \L^*Z(\bs,t)
	\end{smallmatrix}\right)$. Subsequently, \Cref{sec:covfns} provides the required expressions for our kernel choices.
	
	%%%%%%%%%%%%%%%%%%%%%%%%%%%%%%%%%%%%%%%%%%%%%%%%%%%%%%%%%%%%%%%%%
	% Spatiotemporal r-order derivative process using lag operators %
	%%%%%%%%%%%%%%%%%%%%%%%%%%%%%%%%%%%%%%%%%%%%%%%%%%%%%%%%%%%%%%%%%
	\section{Validity of the process \texorpdfstring{$\L_r$}{Lr}, its consequences and distributions for the divergence and Laplacian}\label{sec:valid}
	
	Let $Z(\bs,t)$ be a zero-centered process with $\Cov(Z(\bs,t),Z(\bs',t'))=K(\bDelta,\delta)$, $\bDelta = \bs-\bs'$ and $\delta=t-t'$. Consider the process $(Z(\bs,t),\L_1Z(\bs,t))^{\T}=(Z(\bs,t),\bpartial_\bs Z(\bs,t)^{\T}, \partial_t Z(\bs,t))^{\T}$. Define the lower triangular matrix $\G_{h}=\left(\begin{smallmatrix}
		1 &  & \\
		-\frac{1}{h}\1_d & \frac{1}{h}\bI_d & \\
		-\frac{1}{h} &  &\frac{1}{h}
	\end{smallmatrix}\right)$ and the $(d+2)$-vector $\bW_{h}=\left(\begin{smallmatrix}
		1\\D_{h,\be_1,0}\\\vdots\\D_{h,\be_d,0}\\D_{h,0,1}
	\end{smallmatrix}\right)$. The associated finite difference process is $\bU^{(1)}_{h}(\bs,t)=\G_{h} \bW_{h}Z(\bs,t)$. The process $\bU^{(1)}_{h}(\bs,t)$ is valid. It is zero-centered, since, $\E \bU^{(1)}_{h}(\bs,t)=\G_{h} \bW_{h}\E Z(\bs,t)=\0_{1+\binom{d+1}{1}}$ and  $\Cov(\bU^{(1)}_{h}(\bs,t),\bU^{(1)}_{h}(\bs',t'))=\G_{h} \bW_{h}\bK(\bDelta,\delta)\left(\G_{h} \bW_{h}\right)^{\T}$ for $h\ne0$. The next lemma generalizes this for any order $r>0$.
	
	\begin{lemma}\label{lemma:1}
		For $h\ne0$ and any $r>0$, let $\widetilde{\G}_{h}^{\otimes r}$ be a matrix containing finite difference coefficients and $\widetilde{\bW}_{h}^{\otimes r}$ be a vector of spatiotemporal shift operators corresponding to unique spatiotemporal derivative processes up to order $r$. The finite difference process, $\bU^{(r)}_{h}(\bs,t)=\widetilde{\G}_{h}^{\otimes r} \widetilde{\bW}_{h}^{\otimes r}Z(\bs,t)$ corresponding to $(Z(\bs,t), \L_rZ(\bs,t)^{\T})^{\T}$ is valid.
	\end{lemma}
	\begin{proof}
		The process is valid if for every $r>0$, $\widetilde{\G}_{h}^{\otimes r}$ is non-singular. We rely on an inductive argument. For $r=1$, set $\widetilde{\G}_{h}^{\otimes 1}=\G_h$ and $\widetilde{\bW}_{h}^{\otimes 1}=\bW_h$. Since ${\rm det}\;\G_{h}=h^{-\left\{0\cdot 1+1\cdot\binom{d+1}{1}\right\}}\ne0$, this results in a non-singular linear transformation of $\bW_{h}Z(\bs,t)$ for every $h\ne 0$. For $r=2$, $(Z(\bs,t), \L_2 Z(\bs,t)^{\T})^{\T}=(Z(\bs,t),\bpartial_\bs Z(\bs,t)^{\T}, \partial_t Z(\bs,t), \Partials^2Z(\bs,t)^{\T}, \partial_t\bpartial_\bs Z(\bs,t)^{\T}, \partial_t^2Z(\bs,t))^{\T}$. The associated finite difference process is $\bU^{(2)}_{h}(\bs,t)=\widetilde{\G}_{h}^{\otimes 2} \widetilde{\bW}_{h}^{\otimes 2}Z(\bs,t)$, where $\widetilde{\bW}_{h}^{\otimes 2}$ contains only unique elements from $\bW_{h}^{\otimes 2}$ and $\widetilde{\G}_{h}^{\otimes 2}$ is constructed from $\G_{h}^{\otimes 2}$ containing finite differences coefficients corresponding to the elements in $\widetilde{\bW}_{h}^{\otimes 2}$. The vector $\widetilde{\bW}_{h}^{\otimes 2}$ is comprised of $\bW_h$ and $\left\{D_{h,0,1}^{i_{d+1}}\prod\limits_{j=1}^{d}D_{h,\be_j,0}^{i_j}\right\}$, $\sum\limits_{j=1}^{d+1}i_j=2$, $i_j\geq0$. The number of integer solutions to the equation $\sum\limits_{j=1}^{d+1}i_j=2$, $i_j\geq0$ is $\binom{d+2}{2}$. The matrix $G_h^{\otimes 2}$ (and $\widetilde{\G_h}^{\otimes 2}$) is lower triangular since $\G_h$ is lower triangular. The second order finite difference processes in $\widetilde{\bW}_h^{\otimes 2}Z(\bs,t)$ are constructed from first order finite difference processes in $\bW_hZ(\bs,t)$. Also, $\widetilde{\G}_h^{\otimes 2}$ is non-singular since ${\rm det}\;\widetilde{\G_h}^{\otimes 2}=h^{-\left\{0\cdot1+1\cdot\binom{d+1}{1}+2\cdot\binom{d+2}{2}\right\}}\ne 0$. Hence, $\bU_{h}^{(2)}(\bs,t)$ is a non-singular linear transformation of $\widetilde{\bW_{h}}^{\otimes 2} Z(\bs,t)$. For example, if $d=2$, then 
		
		\begin{equation*}
			\begin{split}
				\bW_{h}^{\otimes 2}&=\left(\bW_h^{\T},\; D_{h,\be_1,0},\; D^2_{h,\be_1,0},\; D_{h,\be_1,0}D_{h,\be_2,0},\; D_{h,\be_1,0}D_{h,0,1},\; D_{h,\be_2,0},\; D_{h,\be_2,0}D_{h,\be_1,0},\right.\\& \myquad[6]\left.D_{h,\be_2,0}^2,\; D_{h,0,1}D_{h,\be_2,0},\; D_{h,0,1},\; D_{h,0,1}D_{h,\be_1,0},\; D_{h,0,1}D_{h,\be_2,0},\; D_{h,0,1}^2\right)^{\T},
			\end{split}
		\end{equation*}
		and $\widetilde{\bW_{h}}^{\otimes 2}=\left(
		\bW_h^{\T}, D^2_{h,\be_1,0}, D_{h,\be_1,0}D_{h,\be_2,0}, D_{h,\be_1,0}D_{h,0,1}, D_{h,\be_2,0}^2, D_{h,0,1}D_{h,\be_2,0}, D_{h,0,1}^2
		\right)^{\T}$
		contains only the unique spatiotemporal shift operators. Next, $\G_h^{\otimes 2}=\left(\begin{smallmatrix}
			\G_h & & & \\
			-\frac{1}{h}\G_h &  \frac{1}{h}\G_h &  & \\
			-\frac{1}{h}\G_h &  & \frac{1}{h}\G_h & \\
			-\frac{1}{h}\G_h &  &  & \frac{1}{h}\G_h\\
		\end{smallmatrix}\right)$. Using column operations followed by row-column elimination matrices on $\G_h^{\otimes 2}$ produces
		
		\begin{equation*}
			\widetilde{\G_h}^{\otimes 2}=\left(\begin{smallmatrix}
				\G_h & &  & \\
				\G_{21} & \frac{1}{h^2}\bI_3 &  & \\
				\G_{31} &  & \frac{1}{h^2}\bI_2 & \\
				\G_{41} & & & \frac{1}{h^2}\\
			\end{smallmatrix}\right), \quad \G_{21}=\left(\begin{smallmatrix}
				\frac{1}{h^2} & -\frac{2}{h^2} & 0 & 0\\
				\frac{1}{h^2} & -\frac{1}{h^2} & -\frac{1}{h^2} & 0\\
				\frac{1}{h^2} & -\frac{1}{h^2} & 0 & -\frac{1}{h^2}\\
			\end{smallmatrix}\right), \quad \G_{31} = \left(\begin{smallmatrix}
				\frac{1}{h^2} & 0 & -\frac{2}{h^2} & 0\\
				\frac{1}{h^2} & 0 & -\frac{1}{h^2} & -\frac{1}{h^2}\\
			\end{smallmatrix}\right) \mbox{ and }  \G_{41} = \left(\begin{smallmatrix}
				\frac{1}{h^2}\\ 0\\ 0\\ -\frac{2}{h^2}
			\end{smallmatrix}\right)^{\T},
		\end{equation*}
		which is lower triangular. The row and column elimination matrices are $\bE_{r2}=\bI_3\oplus \bO_1 \oplus \bI_2 \oplus \bO_2 \oplus \bI_1$ and $\bE_{c2}=\bE_{r2}^{\T}$ respectively. For $r=k-1$, let $\bU^{(k-1)}_{h}(\bs,t)=\widetilde{\G}_{h}^{\otimes k-1} \widetilde{\bW}_{h}^{\otimes k-1}Z(\bs,t)$ be the finite difference process associated with $\left(Z(\bs,t), \L_{k-1}Z(\bs,t)^{\T}\right)^{\T}$, where $\widetilde{\bW}_{h}^{\otimes k-1}$ is comprised of shift operators required for finite difference processes up to order $k-1$ and $\widetilde{\G}_{h}^{\otimes k-1}$ is a lower-triangular non-singular matrix of finite difference coefficients of order $k-1$ with, ${\rm det}\;\widetilde{\G_h}^{\otimes (k-1)}=h^{-\sum\limits_{i=0}^{k-1}i\binom{d+i}{i}}$. Then, for $r=k$, the process $\left(Z(\bs,t), \L_{k}Z(\bs,t)^{\T}\right)^{\T}$ has the associated finite difference process $\bU^{(k)}_{h}(\bs,t)=\widetilde{\G}_{h}^{\otimes k} \widetilde{\bW}_{h}^{\otimes k}Z(\bs,t)$, where $\widetilde{\bW}_{h}^{\otimes k}$ is composed of $\widetilde{\bW}_{h}^{\otimes k-1}$ and $\left\{D_{h,0,1}^{i_{d+1}}\prod\limits_{j=1}^{d}D_{h,\be_j,0}^{i_j}\right\}$, $\sum\limits_{j=1}^{d+1}i_j=k$, $i_j\geq0$. The number of integer solutions to this equation is $\binom{d+k}{k}$. Observe that $\G_{h}^{\otimes k}$ (and $\widetilde{\G}_{h}^{\otimes k}$) is lower triangular since  $\G_{h}^{\otimes k-1}$ (and $\widetilde{\G}_{h}^{\otimes k-1}$) is lower triangular and the $k$-order finite difference processes are constructed using the $(k-1)$-order processes in $\widetilde{\bW}_{h}^{\otimes k-1}Z(\bs,t)$. The matrix $\widetilde{\G}_{h}^{\otimes k}$ is non-singular since, ${\rm det}\;\widetilde{\G_h}^{\otimes k}= {\rm det}\;\widetilde{\G_h}^{\otimes k-1}\;h^{-k\binom{d+k}{k}}$ which is non-zero for $h\ne0$ implying that $\bU^{(k)}_{h}(\bs,t)$ is a non-singular transformation of $\widetilde{\bW}_{h}^{\otimes k}Z(\bs,t)$. 
	\end{proof}
	
	The process $(Z(\bs,t), \L_rZ(\bs,t)^{\T})^{\T}$ is valid if the associated mean vector and cross-covariance matrix are well-defined. The next theorem elucidates the necessary conditions. 
	\begin{theorem}\label{thm:1}
		For $r>0$, if $\partial_t^{2j}\bpartial_\bs^{2(r-j)}K(\0_d,0)$ exists for every $j=0,\ldots,r$, then $(Z(\bs,t), \L_rZ(\bs,t)^{\T})^{\T}$ is a valid zero-centered process. 
	\end{theorem}
	\begin{proof}
		We use a limiting argument on the finite difference process, $\bU^{(r)}_h(\bs,t)$ described in \Cref{lemma:1}. Observe that $\E \bU^{(r)}_{h}(\bs,t)=\widetilde{\G}_{h}^{\otimes r} \widetilde{\bW}_{h}^{\otimes r}\E Z(\bs,t)=\0_{m}$ and $\Cov\{\bU^{(r)}_{h}(\bs,t), \bU^{(r)}_{h}(\bs',t')\} =  \Cov\{\widetilde{\G}_{h}^{\otimes r} \widetilde{\bW}_{h}^{\otimes r}Z(\bs,t), \widetilde{\G}_{h}^{\otimes r} \widetilde{\bW}_{h}^{\otimes r}Z(\bs',t')\}=\widetilde{\G}_{h}^{\otimes r} \widetilde{\bW}_{h}^{\otimes r}K(\bDelta,\delta)(\widetilde{\G}_{h}^{\otimes r} \widetilde{\bW}_{h}^{\otimes r})^{\T}$, where $m=\sum\limits_{i=0}^{r}\binom{d+i}{i}$, for every $r> 0$ and $h\ne0$. $\bU^{(r)}_{h}(\bs,t)$ is a well-defined zero-centered process. The elements of $\bU^{(r)}_h(\bs,t)$ are $Z(\bs,t)$ and $\left\{\frac{1}{h^{k}}(D_{h,0,1}-1)^{i_{d+1}}\prod\limits_{j=1}^{d}(D_{h,\be_j,0}-1)^{i_j}\right\}Z(\bs,t)$, $\sum\limits_{j=1}^{d+1}i_j=k$, $i_j\geq0$, for every $k=1,\ldots,r$. Taking the limit, $\lim\limits_{h\downarrow0}\frac{1}{h^{k}}(D_{h,0,1}-1)^{i_{d+1}}\prod\limits_{j=1}^{d}(D_{h,\be_j,0}-1)^{i_j}Z(\bs,t)=\partial_t^{i_{d+1}}\bpartial_\bs^{k-i_{d+1}}Z(\bs,t)$, $0\leq i_{d+1}\leq k$. Hence, $\lim\limits_{h\downarrow0}\bU^{(r)}_{h}(\bs,t)=\lim\limits_{h\downarrow0}\widetilde{\G}_{h}^{\otimes r} \widetilde{\bW}_{h}^{\otimes r}Z(\bs,t)=\left(\begin{smallmatrix}
			Z(\bs,t)\\\L_rZ(\bs,t)
		\end{smallmatrix}\right)$, where the limit operates element-wise. Clearly, $\E\left(\begin{smallmatrix} Z(\bs,t)\\\L_rZ(\bs,t) \end{smallmatrix}\right)=\lim\limits_{h\downarrow0}\E \bU^{(r)}_{h}(\bs,t)=\lim\limits_{h\downarrow0}\widetilde{\G}_{h}^{\otimes r} \widetilde{\bW}_{h}^{\otimes r}\E Z(\bs,t)=0_{m}$. Similarly, $\Cov\left\{\left(\begin{smallmatrix} Z(\bs,t)\\\L_rZ(\bs,t) \end{smallmatrix}\right), \left(\begin{smallmatrix} Z(\bs',t')\\\L_rZ(\bs',t') \end{smallmatrix}\right)\right\}=\lim\limits_{h\downarrow0}\widetilde{\G}_{h}^{\otimes r} \widetilde{\bW}_{h}^{\otimes r}K(\bDelta,\delta)\left(\widetilde{\G}_{h}^{\otimes r} \widetilde{\bW}_{h}^{\otimes r}\right)^{\T}=\{(-1)^{r_1}\partial_t^{j_1+j_2}\bpartial_\bs^{r_1+r_2-j_1-j_2}K(\bDelta,\delta)\}$, where $0\leq r_1,r_2\leq r$, $0\leq j_i\leq r_i$, $i=1,2$. The last equality is obtained using \Cref{prop:1} and the limit operates on matrices. The covariance matrix is well-defined if all the above entries exist including the variances, obtained as $(\bDelta,\delta)\to (\0_d,0)$. % On closer inspection, 
		The existence of $\partial_t^{2j}\bpartial_\bs^{2(r-j)}K(\0_d,0)$, for $0\leq j\leq r$, is necessary.
	\end{proof}
	Constructing $\L^*$, we only need $j=r-j=2$ and hence require the existence of $\partial_t^{4}\bpartial_\bs^{4}K(\0_d,0)$. Here onwards, we assume $Z(\bs,t)\sim GP(0,K)$. The next corollary provides some immediate consequences of this assumption.
	
	\begin{corollary}\label{cor:1}
		If $Z_1(\bs,t)\sim GP(0, K_1(\cdot,\cdot))$ and $Z_2(\bs,t) \sim GP(0, K_2(\cdot,\cdot))$ independently, then (i) $\L_{r_1} Z_1(\bs,t)$ and $\L_{r_2} Z_2(\bs,t)$ are independent GPs, (ii) for $a_1, a_2\in \mathfrak{R}$, $\L_r (a_1Z_1(\bs,t)+a_2Z_2(\bs,t))=a_1\L_r Z_1(\bs,t)+a_2\L_rZ_2(\bs,t)$ is a GP, and (iii) any sub-vector of $\L_r Z(\bs,t)$ is a GP. 
	\end{corollary}
	
	\begin{proof}
		Observe that from \Cref{thm:1}, $\L_{r_i} Z_i(\bs,t)$ is a zero-centered GP with cross-covariance matrix $\bV_{\L_{r_i}Z}(\bDelta,\delta)=\{(-1)^{k_1}\partial_t^{j_1+j_2}\bpartial_\bs^{k_1+k_2-j_1-j_2}K(\bDelta,\delta)\}$, $0< k_1,k_2\leq r_i$, $0\leq j_i\leq k_i$, for $i=1,2$ and $ \Cov(\L_{r_1} Z_1(\bs,t),\L_{r_2} Z_2(\bs,t))=\lim\limits_{h\downarrow0}\Cov\left(\widetilde{\G}_{h}^{\otimes r_1} \widetilde{\bW}_{h}^{\otimes r_1}Z_1(\bs,t),\widetilde{\G}_{h}^{\otimes r_2} \widetilde{\bW}_{h}^{\otimes r_2}Z_2(\bs,t)\right)=\lim_{h\downarrow0}\widetilde{\G}_{h}^{\otimes r_1} \widetilde{\bW}_{h}^{\otimes r_1}\;\cdot0\;\cdot \left(\widetilde{\G}_{h}^{\otimes r_2} \widetilde{\bW}_{h}^{\otimes r_2}\right)^{\T}=0$
		% \begin{equation*}
			%     \begin{split}
				%         \Cov(\L_{r_1} Z_1(\bs,t),\L_{r_2} Z_2(\bs,t))&=\lim_{h\downarrow0}\Cov\left(\widetilde{G}_{h}^{\otimes r_1} \widetilde{W}_{h}^{\otimes r_1}Z_1(\bs,t),\widetilde{G}_{h}^{\otimes r_2} \widetilde{W}_{h}^{\otimes r_2}Z_2(\bs,t)\right),\\
				%         &=\lim_{h\downarrow0}\widetilde{G}_{h}^{\otimes r_1} \widetilde{W}_{h}^{\otimes r_1}\;\cdot0\;\cdot \left(\widetilde{G}_{h}^{\otimes r_2} \widetilde{W}_{h}^{\otimes r_2}\right)^{\T}=0,
				%     \end{split}
			% \end{equation*}
		for $r_1,r_2>0$. Consider $\L_r (a_1Z_1(\bs,t)+a_2Z_2(\bs,t))=\lim\limits_{h\downarrow 0}\widetilde{\G}_{h}^{\otimes r} \widetilde{\bW}_{h}^{\otimes r}(a_1Z_1(\bs,t)+a_2Z_2(\bs,t))=a_1\lim\limits_{h\downarrow 0}\widetilde{\G}_{h}^{\otimes r} \widetilde{\bW}_{h}^{\otimes r}Z_1(\bs,t)+a_2\lim\limits_{h\downarrow 0}\widetilde{\G}_{h}^{\otimes r} \widetilde{\bW}_{h}^{\otimes r}Z_2(\bs,t)=a_1\L_rZ_1(\bs,t)+a_2\L_rZ_2(\bs,t)$. Finally, since $Z_1(\bs,t)$ and $Z_2(\bs,t)$ are independent GPs, $a_1\L_rZ_1(\bs,t)+a_2\L_rZ_2(\bs,t)$ is also a GP with zero mean and cross-covariance $a_1\bV_{\L_{r_1}}(\bDelta,\delta)+a_2\bV_{\L_{r_2}}(\bDelta,\delta)$. Lastly, since $\L_rZ(\bs,t)$ is a GP, any sub-vector is also a zero-centered GP with the cross-covariance being a sub-matrix of $\bV_{\L_r}(\bDelta,\delta)$ \citep[see, e.g.,][]{rasmussen_gaussian_2005}. 
	\end{proof}

	\section{Cross-Covariance Matrix for \texorpdfstring{$(Z, \L^*Z^{\T})^{\T}$}{Lstar}}\label{sec:cross-cov}
	Assuming $d = 2$ and suppressing the dependence on $\bDelta=\bs-\bs'$ and $\delta=t-t'$, the resulting cross-covariance matrix for the process, $(Z(\bs,t),\L^*Z(\bs,t)^{\T})^{\T}$ denoted by $\bV_{Z,\L^* Z}$, is
	\begin{equation}\label{eq:full-cross-cov}
		\scalemath{0.7}{
			\left[
			\begin{array}{*{3}{c};{2pt/2pt}*{3}{c};{2pt/2pt}*{3}{c}}
				\bK & \bpartial_\bs \bK^{\T} & \Partials^2 \bK^{\T} & \partial_t \bK & \partial_t\bpartial_\bs \bK^{\T} & \partial_t\Partials^2 \bK^{\T} & \partial_t^2 \bK & \partial_t^2\bpartial_\bs \bK^{\T} & \partial_t^2\Partials^2 \bK^{\T}\\
				-\bpartial_\bs \bK & {\color{blue}-\bpartial_\bs^2 \bK}_{2\times 2} & -\Partials^3 \bK_{2\times 3} & -\partial_t\bpartial_\bs \bK^{\T}_{2\times 1} & -\partial_t\bpartial_\bs^2 \bK^{\T}_{2\times 2} & -\partial_t\Partials^3 \bK^{\T}_{2\times 3} & -\partial_t^2\bpartial_\bs \bK^{\T}_{2\times 1} & {\color{blue}-\partial_t^2\bpartial_\bs^2 \bK}^{\T}_{2\times 2} & -\partial_t^2\Partials^3 \bK^{\T}_{2\times 3}\\
				\Partials^2 \bK & \Partials^3 \bK_{3\times 2} & {\color{blue}\Partials^4 \bK}_{3\times 3} & \partial_t\Partials^2 \bK_{3\times 1} & \partial_t\Partials^3 \bK^{\T}_{3\times 2} & \partial_t\Partials^4  \bK^{\T}_{3\times 3} & {\color{blue}\partial_t^2\Partials^2 \bK}^{\T}_{3\times 1} & \partial_t^2\Partials^3 \bK^{\T}_{3\times 2} & {\color{blue}\partial_t^2\Partials^4 \bK}^{\T}_{3\times 3}\\\hdashline[2pt/2pt]
				-\partial_tK & -\partial_t\bpartial_\bs \bK_{1\times 2} & -\partial_t\Partials^2 \bK^{\T}_{1\times 3} & {\color{blue} -\partial_t^2 \bK}_{1\times 1} & -\partial_t^2\bpartial_\bs \bK^{\T}_{1\times 2} & {\color{blue}-\partial_t^2\Partials^2 \bK}^{\T}_{1\times 3} & -\partial_t^3 \bK_{1\times 1} & -\partial_t^3\bpartial_\bs \bK^{\T}_{1\times 2} & -\partial_t^3\Partials^2 \bK^{\T}_{1\times 3}\\
				\partial_t\bpartial_\bs \bK & \partial_t\bpartial_\bs^2 \bK_{2\times 2} & \partial_t\Partials^3 \bK_{2\times 3} & \partial_t^2\bpartial_\bs \bK_{2\times 1} & {\color{blue}\partial_t^2\bpartial_\bs^2 \bK}_{2\times 2} &  \partial_t^2\Partials^3 \bK^{\T}_{2\times 3} & \partial_t^3\bpartial_\bs \bK^{\T}_{2\times 1} & \partial_t^3\bpartial_\bs^2 \bK^{\T}_{2\times 2} & \partial_t^3\Partials^3 \bK^{\T}_{2\times 3}\\
				-\partial_t\Partials^2 \bK & -\partial_t\Partials^3 \bK_{3\times 2} & -\partial_t\Partials^4 \bK_{3\times 3} & {\color{blue}-\partial_t^2\Partials^2 \bK}_{3\times 1} & -\partial_t^2\Partials^3 \bK_{3\times 2} & {\color{blue}-\partial_t^2\Partials^4 \bK}_{3\times 3} & -\partial_t^3\Partials^2 \bK_{3\times 1} & -\partial_t^3\Partials^3 \bK^{\T}_{3\times 2} & -\partial_t^3\Partials^4 \bK^{\T}_{3\times 3}\\\hdashline[2pt/2pt]
				\partial_t^2 \bK & \partial_t^2\bpartial_\bs \bK_{1\times 2} & {\color{blue}\partial_t^2\Partials^2 \bK}_{1\times 3} & \partial_t^3 \bK_{1\times 1} & \partial_t^3\bpartial_\bs \bK_{1\times 2} & \partial_t^3\Partials^2 \bK^{\T}_{1\times 3} & {\color{blue}\partial_t^4 \bK}_{1\times 1} & \partial_t^4\bpartial_\bs \bK^{\T}_{1\times 2} & {\color{blue}\partial_t^4\Partials^2 \bK}^{\T}_{1\times 3}\\
				-\partial_t^2\bpartial_\bs \bK & {\color{blue}-\partial_t^2\bpartial_\bs^2 \bK}_{2\times 2} & -\partial_t^2\Partials^3 \bK_{2\times 3} & -\partial_t^3\bpartial_\bs \bK_{2\times 1} & -\partial_t^3\bpartial_\bs^2 \bK_{2\times 2} & -\partial_t^2\Partials^3 \bK_{2\times 3} & -\partial_t^4\bpartial_\bs \bK_{2\times 1} & {\color{blue}-\partial_t^4\bpartial_\bs^2 \bK}_{2\times 2} & -\partial_t^4\Partials^3 \bK^{\T}_{2\times 3}\\
				\partial_t^2\Partials^2 \bK & \partial_t^2\Partials^3 \bK_{3\times 2} & {\color{blue}\partial_t^2\Partials^4 \bK}_{3\times 3} & \partial_t^3\Partials^2 \bK_{3\times 1} & \partial_t^3\Partials^3  \bK_{3\times 2} & \partial_t^3\Partials^4 \bK_{3\times 3} & {\color{blue}\partial_t^4\Partials^2 \bK}_{3\times 1} & \partial_t^4\Partials^3 \bK_{3\times 2} & {\color{blue}\partial_t^4\Partials^4 \bK}_{3\times 3}
			\end{array}
			\right],
		}
	\end{equation}
	where $\Var{\bpartial_\bs Z(\bs,t)} = \Cov\{\bpartial_\bs Z(\bs,t),\bpartial_\bs Z(\bs',t')\}= \bpartial_\bs^2\bK(\bDelta,\delta)=\bH_K^{11}(\bDelta,\delta)$ is the variance of $\bpartial_\bs Z(\bs,t)$ of order $d\times d$, $\Cov\{\bpartial_\bs Z(\bs,t),\Partials^2Z(\bs',t')\} = \Partials^3 \bK(\bDelta,\delta) = \bH_K^{12}(\bDelta,\delta) = {\bH_K^{21}}^{\T}(\bDelta,\delta)$ is of order $d\times \frac{d(d+1)}{2}$ and $\Var{\Partials^2Z(\bs,t)}= \Cov\{\Partials^2Z(\bs,t),\Partials^2Z(\bs',t')\} = \Partials^4 \bK(\bDelta,\delta) = \bH_K^{22}(\bDelta,\delta)$ of order, $\frac{d(d+1)}{2}\times \frac{d(d+1)}{2}$. The entries that remain non-zero as $(\bDelta,\delta)\to(\0_{d},0)$ in the $\left(d+\frac{d(d+1)}{2}\right)\times \left(d+\frac{d(d+1)}{2}\right)$ sub-matrix corresponding to $\Cov\{\L^{-Z}_4Z(\bs,t), \L^{-Z}_4Z(\bs',t')\}$ are marked in \textcolor{blue}{blue}. In comparison to the spatiotemporal gradient theory \citep[see,][p. 577]{quick2015bayesian}, as $(\bDelta,\delta)\to(\0_{d},0)$, the resulting cross-covariance matrix % for the differential processes 
	is no longer diagonal.
	
	\paragraph{Spatiotemporal Divergence \& Laplacian} Using \Cref{cor:1} we are able to devise inference for common differential geometric operators that capture change in geometry for spatiotemporal surfaces. We consider two quantities--divergence and Laplacian. We define \emph{spatiotemporal divergence} as ${\rm div} Z(\bs,t)=\left(\1^{\T}\bpartial_\bs Z(\bs,t),\1^{\T}\partial_t\bpartial_\bs Z(\bs,t), \1^{\T}\partial_t^2\bpartial_\bs Z(\bs,t)\right)^{\T}$. Assume $d=2$. First, observe that $\left(\begin{smallmatrix}\bpartial_\bs Z(\bs,t)\\\partial_t\bpartial_\bs Z(\bs,t)\\ \partial_t^2\bpartial_\bs Z(\bs,t)\end{smallmatrix}\right)=(\bI_2\oplus \bO_4 \oplus \bI_2\oplus \bO_4\oplus \bI_2\oplus \bO_3)\L^*Z(\bs,t)=\bE_{\rm div}\L^*Z(\bs,t)$ forms a subset of $\L^*Z(\bs,t)$ and is therefore is a zero-centered GP with cross-covariance, $\bE_{\rm div}\bV_{\L_*}(\bDelta,\delta)\bE_{\rm div}^{\T}$. This is obtained using \Cref{cor:1}, part (iii). Hence, $(\1^{\T}\oplus\1^{\T}\oplus\1^{\T})\bE_{\rm div}\L_*Z(\bs,t)$ is also a zero-centered GP with cross-covariance, $(\1^{\T}\oplus\1^{\T}\oplus\1^{\T})\bE_{\rm div}\bV_{\L_*}(\bDelta,\delta)\left((\1^{\T}\oplus\1^{\T}\oplus\1^{\T})\bE_{\rm div}\right)^{\T}$. We define the \emph{spatiotemporal Laplacian} as, ${\rm Lap} Z(\bs,t)=\left(\sum_{i=1}^{d}\frac{\partial^2}{\partial s_i^2}Z(\bs,t),\sum_{i=1}^{d}\frac{\partial}{\partial t}\frac{\partial^2}{\partial s_i^2}Z(\bs,t), \sum_{i=1}^{d}\frac{\partial^2}{\partial t^2}\frac{\partial^2}{\partial s_i^2}Z(\bs,t)\right)^{\T}$. Similar arguments relying on \Cref{cor:1} (part (iii)) allows ${\rm Lap} Z(\bs,t)$ to be a GP. We show the resulting figures in \Cref{figsec:diff-geo}. Next, we show details for $\bV_{Z, \L ^*Z}(\bDelta,\delta)$.
	
	\section{Simplifications under isotropy}\label{sec:isotropy}
	
	As is evident from \cref{eq:full-cross-cov} inference for spatiotemporal gradients (and directional spatiotemporal gradients) requires the derivatives of $K(\bDelta,\delta)$. Commonly used covariance kernels are isotropic (see \Cref{sec:covfns}), i.e. $K(\bDelta,\delta)=\tildeK(||\bDelta||,|\delta|)$. We derive simplifications assuming isotropy---first considering the resulting expressions for spatiotemporal gradients followed by the same for directional spatiotemporal gradients.
	
	\subsection{Spatiotemporal Gradients}\label{sec:simple-spt-grad}
	
	\begin{proposition}\label{prop:2}
		Under isotropy, $K=K(\bDelta,\delta)=\tildeK(||\bDelta||,|\delta|)=\tildeK$, simplifications arise in the resulting expressions for derivatives of the covariance. They are as follows, 
		
		\begin{align*}
			\bpartial_\bs K &=d_{||\bDelta||}\tildeK \frac{\bP_{1,1}}{||\bDelta||},
			\partial_t K =d_{|\delta|}\tildeK \frac{\delta}{|\delta|},
			\bpartial_\bs^2 K =d_{||\bDelta||}\tildeK\frac{\bP_{2,1}}{||\bDelta||}+\left(d_{||\bDelta||}^2\tildeK -\frac{d_{||\bDelta||}\tildeK}{||\bDelta||}\right)\frac{\bP_{2,2}}{||\bDelta||^2},\\
			\partial_t\bpartial_\bs K &=d_{|\delta|} d_{||\bDelta||}\tildeK \frac{\bP_{1,1}}{||\bDelta||}\frac{\delta}{|\delta|},
			\partial_t^2 K =d_{|\delta|}^2\tildeK,\\
			\partial_t\bpartial_\bs^2 K &=d_{|\delta|}d_{||\bDelta||}\tildeK\frac{\bP_{2,1}}{||\bDelta||}\frac{\delta}{|\delta|}+\left(d_{|\delta|}d_{||\bDelta||}^2\tildeK -\frac{d_{|\delta|}d_{||\bDelta||}\tildeK }{||\bDelta||}\right)\frac{\bP_{2,2}}{||\bDelta||^2}\frac{\delta}{|\delta|},
			\partial_t^2\bpartial_\bs K=d_{|\delta|}^2d_{||\bDelta||}\tildeK\frac{\bP_{1,1}}{||\bDelta||},\\
			\bpartial_\bs^4K&=\left(d_{||\bDelta||}^2\tildeK-\frac{d_{||\bDelta||}\tildeK}{||\bDelta||}\right)\left(-3\frac{\bP_{4,1}}{||\bDelta||^4}+\frac{\bP_{4,2}}{||\bDelta||^2}+15\frac{\bP_{4,3}}{||\bDelta||^6}-3\frac{\bP_{4,4}}{||\bDelta||^4}\right)+\\
			&\myquad[6]d_{||\bDelta||}^3\tildeK\left(\frac{\bP_{4,1}}{||\bDelta||^3}-6\frac{\bP_{4,3}}{||\bDelta||^5}+\frac{\bP_{4,4}}{||\bDelta||^3}\right)+ d_{||\bDelta||}^4\tildeK\frac{\bP_{4,3}}{||\bDelta||^4},\\
			\partial_t^2\bpartial_\bs^2 K&=d_{|\delta|}^2d_{||\bDelta||}\tildeK\frac{\bP_{2,1}}{||\bDelta||}+\left(d_{|\delta|}^2d_{||\bDelta||}^2\tildeK-\frac{d_{|\delta|}^2d_{||\bDelta||}\tildeK}{||\bDelta||}\right)\frac{\bP_{2,2}}{||\bDelta||^2},~\partial_t^4K=d_{|\delta|}^4\tildeK,\\
			\partial_t^2\bpartial_\bs^4K&=\left(d_{|\delta|}^2d_{||\bDelta||}^2\tildeK-\frac{d_{|\delta|}^2d_{||\bDelta||}\tildeK}{||\bDelta||}\right)\left(-3\frac{\bP_{4,1}}{||\bDelta||^4}+\frac{\bP_{4,2}}{||\bDelta||^2}+15\frac{\bP_{4,3}}{||\bDelta||^6}-3\frac{\bP_{4,4}}{||\bDelta||^4}\right)+\\
			&\myquad[6]d_{|\delta|}^2d_{||\bDelta||}^3\tildeK\left(\frac{\bP_{4,1}}{||\bDelta||^3}-6\frac{\bP_{4,3}}{||\bDelta||^5}+\frac{\bP_{4,4}}{||\bDelta||^3}\right)+ d_{|\delta|}^2d_{||\bDelta||}^4\tildeK\frac{\bP_{4,3}}{||\bDelta||^4},\\
			\partial_t^4\bpartial_\bs^2 K&=\left\{d_{|\delta|}^4d_{||\bDelta||}\tildeK\frac{\bP_{2,1}}{||\bDelta||}+\left(d_{|\delta|}^4d_{||\bDelta||}^2\tildeK-\frac{d_{|\delta|}^4d_{||\bDelta||}\tildeK}{||\bDelta||}\right)\frac{\bP_{2,2}}{||\bDelta||^2}\right\},\\
			\partial_t^4\bpartial_\bs^4K&=\left(d_{|\delta|}^4d_{||\bDelta||}^2\tildeK-\frac{d_{|\delta|}^4d_{||\bDelta||}\tildeK}{||\bDelta||}\right)\left(-3\frac{\bP_{4,1}}{||\bDelta||^4}+\frac{\bP_{4,2}}{||\bDelta||^2}+15\frac{\bP_{4,3}}{||\bDelta||^6}-3\frac{\bP_{4,4}}{||\bDelta||^4}\right)+\\
			&\myquad[6]d_{|\delta|}^4 d_{||\bDelta||}^3\tildeK\left(\frac{\bP_{4,1}}{||\bDelta||^3}-6\frac{\bP_{4,3}}{||\bDelta||^5}+\frac{\bP_{4,4}}{||\bDelta||^3}\right)+ d_{|\delta|}^4d_{||\bDelta||}^4\tildeK\frac{\bP_{4,3}}{||\bDelta||^4},
		\end{align*}
		where $\bP_{1,1}=\bDelta$, $\bP_{2,1}=\vecc \bI_d$, $\bP_{2,2}=(\bI_d\otimes\bDelta)\bDelta$, $\bP_{3,1}= \bP_{3,11} + \bP_{3,12} + \bP_{3,13}$, $\bP_{3,11}=\bDelta\otimes\vecc \bI_d$, $\bP_{3,12}= (\bI_d\otimes \bU_2)(\vecc \bI_d\otimes \bI_d)\bU_1\bDelta$, $\bP_{3,13}=\bI_d\otimes(\bI_d\otimes\bDelta)\vecc \bI_d$, $\bP_{3,2}=\bDelta\otimes(\bI_d\otimes\bDelta)\bDelta$, $\bP_{4,1}=\bDelta\otimes \bP_{3,1}$, $\bP_{4,2}=\bP_{4,21}+\bP_{4,22}+\bP_{4,23}$, $\bP_{4,21}=\vecc \bI_d\otimes\vecc \bI_d$, $\bP_{4,22}=\left(\bI_d\otimes(\bI_d\otimes \bU_2)(\vecc \bI_d\otimes \bI_d)\bU_1\right)\vecc \bI_d$, $\bP_{4,23}= (\bI_d\otimes \bU_3)\left\{\left((\bI_d\otimes \bU_2)(\vecc \bI_d\otimes \bI_d)\bU_1\right)\otimes \bI_d\right\}\bU_2\vecc \bI_d,~ \bP_{4,3}=\bDelta\otimes \bP_{3,2}$ and $\bP_{4,4}=[\vecc \bI_d\otimes(\bI_d\otimes\bDelta)+(\bI_d\otimes \bU_3)\left\{\left((\bI_d\otimes \bU_2)(\vecc \bI_d\otimes \bI_d)\bU_1\right)\otimes \bDelta\right\}\bU_1]\bDelta+\left\{\bI_d\otimes(\bDelta\otimes(\bI_d\otimes\bDelta))\right\}\vecc \bI_d$, $\bU_i$ is a permutation matrix of order $d^i\times d^i$, $i=1,2,3$ and $\bP_{i,\cdot}$ is a $d^i\times 1$ vector for $i=1,2,3,4$. The operators, $d^r_{||\bDelta||} = \frac{d^r}{d||\bDelta||^r}$ and $d^r_{|\delta|} = \frac{d^r}{d|\delta|^r}$.
	\end{proposition}
	\begin{proof}
		Note that $\frac{\partial}{\partial\delta}\frac{\delta}{|\delta|}=\left(\frac{1}{|\delta|}-\frac{\delta}{|\delta|^2}\frac{\delta}{|\delta|}\right)=0$, $ \frac{\partial}{\partial\bDelta}\frac{d_{||\bDelta||}^r\tildeK}{||\bDelta||^r}=\left(d_{||\bDelta||}^{r+1}\tildeK-\frac{d_{||\bDelta||}^r\tildeK}{||\bDelta||}\right)\frac{\bDelta}{||\bDelta||^{r+1}}-(r-1)\frac{d_{||\bDelta||}^r\tildeK}{||\bDelta||^r}\frac{\bDelta}{||\bDelta||^2}$, and for $r,r'>0$, $ \frac{\partial}{\partial\bDelta}\left(d_{||\bDelta||}^{r+1}\tildeK-\frac{d_{||\bDelta||}^r\tildeK}{||\bDelta||}\right)\frac{1}{||\bDelta||^{r'}}=\left\{\frac{d_{||\bDelta||}^{r+2}\tildeK}{||\bDelta||^{r'+1}}-\left(d_{||\bDelta||}^{r+1}\tildeK-\frac{d_{||\bDelta||}^r\tildeK}{||\bDelta||}\right)\frac{r'+1}{||\bDelta||^{r'+2}}\right\}\bDelta$. { The resulting expressions under isotropy are obtained via successive differentiation}.

	\end{proof}
	Deriving the expressions above requires computing the derivatives of matrices with respect to vectors.  In the traditional sense, this results in tensors, for example, $\frac{\partial}{\partial\bDelta}\bpartial_\bs^2K(\bDelta,\delta)$ produces a 3-order tensor. We use notation in \cite{vetter1970derivative} \citep[also, see, e.g.,][Chapter 6]{graham2018kronecker} to obtain vector representations for tensor-valued derivatives. The order of listing conforms to the reshaping of tensors discussed in \Cref{sec::vec-tensor}. We rely on the following, $\frac{\partial\bDelta}{\partial\bDelta}=(\be_1^{\T},\be_2^{\T},\ldots,\be_d^{\T})^{\T}=\vecc(\bI_d)$, $\frac{\partial \bA\bB}{\partial\bDelta}=\frac{\partial \bA}{\partial\bDelta}(\1\otimes \bB) + (\bI_d\otimes \bA)\frac{\partial \bB}{\partial\bDelta}$ and $\frac{\partial \bA\otimes \bB}{\partial\bDelta}=\frac{\partial \bA}{\partial\bDelta}\otimes \bB + (\bI_d\otimes \bU_1)\left(\frac{\partial \bB}{\partial \bDelta}\otimes \bA\right) (\1\otimes \bU_2)$. In particular, when deriving $\bpartial_\bs^3 K$, we use $\frac{\partial}{\partial\bDelta}\frac{d_{||\bDelta||}\tildeK}{||\bDelta||}\vecc \bI_d=\bigg(d_{||\bDelta||}^2\tildeK-\frac{d_{||\bDelta||}\tildeK}{||\bDelta||}\bigg)\frac{\bDelta}{||\bDelta||^2}\otimes \vecc \bI_d$. The vectors $\be_i$, $i=1,\ldots,d$ constitute the orthonormal basis for $\mathfrak{R}^d$.
	
	\subsection{Directional Spatiotemporal Gradients}\label{sec:simple-dir-spt-gradients}
	
	Directional spatiotemporal gradients and curvatures  involve a direction vector of interest (see Section~2.1 of the manuscript), $(\bu^{\T},v)^{\T} \in \mathfrak{H}^{d+1}=\mathfrak{R}^d\times\mathfrak{R}^{+}$ along which the gradients and curvatures are evaluated. For convenience, we fix $v=1$. We require the covariance between two directional differential processes of possibly different orders. Let $\bDelta(h_s)=\bDelta+h_s\bu$ and $\delta(h_t)=\delta+h_t$. Suppressing dependence on $(\bDelta,\delta)$, we write $g(h_s,h_t)=K\left(\bDelta(h_s),\delta(h_t)\right)=K(\bDelta+h_s\bu,\delta+h_t)$. The following result derives the required covariance.
	
	\begin{proposition}\label{prop:3}
		Let ${\bu^{\otimes (r_1-j_1)}}^{\T}\partial_t^{j_1}\bpartial_\bs^{r_1-j_1}Z(\bs,t)$ and ${\bu^{\otimes (r_2-j_2)}}^{\T}\partial_t^{j_2}\bpartial_\bs^{r_2-j_2}Z(\bs',t')$ be two directional differential processes of orders $r_1$ and $r_2$ respectively. The covariance,
		
		\begin{equation*}
			\Cov\left({\bu^{\otimes (r_1-j_1)}}^{\T}\partial_t^{j_1}\bpartial_\bs^{r_1-j_1}Z(\bs,t),{\bu^{\otimes (r_2-j_2)}}^{\T}\partial_t^{j_2}\bpartial_\bs^{r_2-j_2}Z(\bs',t')\right)=(-1)^{r_1}\partial_t^{j_1+j_2}\bpartial_\bs^{(r_1+r_2-j_1-j_2)} g(\0_d,0), 
		\end{equation*}
		for every $j_i=0,\ldots,r_i$ and $i=1,2$.
	\end{proposition}
	
	{\allowdisplaybreaks
		\begin{proof}
			Following \Cref{prop:1}, we express the covariance between the directional differential processes as a limit of the covariance between the finite difference process,
			
			\begin{equation*}
				\Cov\left({\bu^{\otimes (r_1-j_1)}}^{\T}\partial_t^{j_1}\bpartial_\bs^{r_1-j_1}Z(\bs,t), {\bu^{\otimes (r_2-j_2)}}^{\T}\partial_t^{j_2}\bpartial_\bs^{r_2-j_2}Z(\bs',t')\right) = \lim_{\substack{k\to 0\\h\to0}}\Cov\left(Z^{(j_1,r_1)}_{k,\bu,1}(\bs,t),\;Z^{(j_2,r_2)}_{h,\bu,1}(\bs',t')\right),
			\end{equation*}
			where $Z^{(j_1,r_1)}_{k,\bu,1}(\bs,t)=k^{-r_1}(D_{k,\bu,0}-1)^{r_1-j_1}(D_{k,0,1}-1)^{j_1}Z(\bs,t)$, $Z^{(j_2,r_2)}_{h,\bu,1}(\bs',t')=h^{-r_2}(D_{h,\bu,0}-1)^{r_2-j_2}(D_{h,0,1}-1)^{j_2}Z(\bs',t')$. Using these expressions and performing binomial expansion % on the spatiotemporal shift operator we obtain
			
			\begin{align*}
				%\begin{split}
				&\lim_{\substack{k\to 0\\h\to0}}\Cov\left(Z^{(j_1,r_1)}_{k,\bu,1}(\bs,t),\;Z^{(j_2,r_2)}_{h,\bu,1}(\bs',t')\right) =\\
				&\lim_{\substack{k\to 0\\h\to0}}\;k^{-r_1}\;h^{-r_2}\;\Cov\left(\;\sum\limits_{i=1}^{r_1-j_1}\sum\limits_{j=1}^{j_1}{\binom{r_1-j_1}{i}}\;{\binom{j_1}{j}}\;(-1)^{r_1-(i+j)}\;Z(\bs+ik\bu,t+jk),\right.\\&\myquad[10]\left.\;\sum\limits_{i'=1}^{r_2-j_2}\sum\limits_{j'=1}^{j_2}{\binom{r_2-j_2}{i'}}\;{\binom{j_2}{j'}}\;(-1)^{r_2-(i'+j')}\;Z(\bs'+i'h\bu,t'+j'h)\right),\\
				% &=\lim_{\substack{k\to 0\\h\to0}}\;k^{-r_1}\;h^{-r_2}\;\sum\limits_{i=1}^{r_1-j_1}\sum\limits_{j=1}^{j_1}\sum\limits_{i'=1}^{r_2-j_2}\sum\limits_{j'=1}^{j_2}{\binom{r_1-j_1}{i}}\;{\binom{j_1}{j}}\;{\binom{r_2-j_2}{i'}}\;{\binom{j_2}{j'}}\;(-1)^{(r_1+r_2)-(i+j+i'+j')}\cdot\\
				% &\myquad[18] g((i'h-ik)u,j'h-jk),\\
				% &=\lim_{\substack{k\to 0\\h\to0}}\;k^{-r_1}\;h^{-r_2}\sum\limits_{i=1}^{r_1-j_1}\sum\limits_{j=1}^{j_1}\sum\limits_{i'=1}^{r_2-j_2}\sum\limits_{j'=1}^{j_2}{\binom{r_1-j_1}{i}}\;{\binom{j_1}{j}}\;{\binom{r_2-j_2}{i'}}\;{\binom{j_2}{j'}}\;(-1)^{(r_1+r_2)-(i+j+i'+j')}\cdot\\
				% &\myquad[10] D_{i'h,u,0}\;D_{-ik,u,0}\;D_{j'h,0,1}\;D_{-jk,0,1}\;g(0_d,0),\\
				% &=\lim_{\substack{k\to 0\\h\to0}}\;k^{-r_1}\;h^{-r_2}\sum\limits_{i=1}^{r_1-j_1}\sum\limits_{j=1}^{j_1}\sum\limits_{i'=1}^{r_2-j_2}\sum\limits_{j'=1}^{j_2}{\binom{r_1-j_1}{i}}\;{\binom{j_1}{j}}\;{\binom{r_2-j_2}{i'}}\;{\binom{j_2}{j'}}\;(-1)^{(r_1+r_2)-(i+j+i'+j')}\cdot\\
				% &\myquad[8] D_{h,u,0}^{i'}\;D_{h,0,1}^{j'}\;D_{-k,u,0}^{i}\;D_{-k,0,1}^{j}\;g(0_d,0),\\
				% &=(-1)^{r_1}\;\lim_{\substack{k\to 0\\h\to0}}\;(-k)^{-r_1}h^{-r_2}\;(D_{-k,u,0}-1)^{r_1-j_1}\;(D_{-k,0,1}-1)^{j_1}\cdot\\&\myquad[15](D_{h,u,0}-1)^{r_2-j_2}\;(D_{h,0,1}-1)^{j_2}g(0_d,0),\\
				&=(-1)^{r_1}\partial_t^{j_1+j_2}\bpartial_\bs^{r_1+r_2-(j_1+j_2)}\;g(\0_d,0),
				%\end{split}
			\end{align*}
			for every $j_i = 0,\ldots, r_i$, $i=1,2$. { The last equality is obtained after some algebra.}
		\end{proof}
	}
	We are interested in the directional derivative process required for inference,
	
	\begin{equation*}
		\L^*_{\bu,1} Z(\bs,t)=\left(\L_{\bs,\bu}Z(\bs,t)^{\T}, \partial_tZ(\bs,t), \partial_t\L_{\bs,\bu}Z(\bs,t)^{\T}, \partial_t^2Z(\bs,t), \partial_t^2\L_{\bs,\bu}Z(\bs,t)^{\T}\right)^{\T},
	\end{equation*}
	where $\L_{\bs,\bu}Z(\bs,t)=\left(\bu^{\T}\bpartial_\bs Z(\bs,t), {\bu^{\otimes 2}}^{\T}\bpartial_\bs^2Z(\bs,t)\right)^{\T}$. Evidently, $\partial_t^4\bpartial_\bs^4g(\0_d,0)$ should exist and be well-defined. The next proposition addresses this.
	
	\begin{proposition}\label{prop:4}
		The covariance between ${\bu^{\otimes2}}^{\T}\partial_t^2\bpartial_\bs^2Z(\bs,t)$ and ${\bu^{\otimes2}}^{\T}\partial_t^2\bpartial_\bs^2Z(\bs',t')$ is given by
		
		\begin{equation*}
			\begin{split}
				&\Cov\left({\bu^{\otimes2}}^{\T}\partial_t^2\bpartial_\bs^2Z(\bs,t), {\bu^{\otimes2}}^{\T}\partial_t^2\bpartial_\bs^2Z(\bs',t')\right)\\
				&= \frac{3}{||\bDelta||^2}\left(1-5\frac{(\bu^{\T}\bDelta)^2}{||\bDelta||^2}\right)\left(1-\frac{(\bu^{\T}\bDelta)^2}{||\bDelta||^2}\right)\left(\partial_t^4\bpartial_\bs^2\tildeK(||\bDelta||,|\delta|)-\frac{\partial_t^4\bpartial_\bs\tildeK(||\bDelta||,|\delta|)}{||\bDelta||}\right)+\\
				&\myquad[4]\frac{6}{||\bDelta||}\frac{(\bu^{\T}\bDelta)^2}{||\bDelta||^2}\left(1-\frac{(\bu^{\T}\bDelta)^2}{||\bDelta||^2}\right)\partial_t^4\bpartial_\bs^3\tildeK(||\bDelta||,|\delta|)+\left(\frac{\bu^{\T}\bDelta}{||\bDelta||}\right)^4\partial_t^4\bpartial_\bs^4\tildeK(||\bDelta||,|\delta|),
			\end{split}
		\end{equation*}
		where $\bDelta = \bs-\bs'$ and $\delta=t-t'$.
	\end{proposition}
	
	\begin{proof}
		For convenience, let $\bDelta(h_s)=\bDelta$, $\delta(h_t)=\delta$, $K\left(\bDelta(h_s),\delta(h_t)\right)=K$ and $g(h_s,h_t)=g$. Successive differentiation yields
		
		{\allowdisplaybreaks
			\begin{align*}
				\bpartial_\bs g&=\bDelta'^{\T}\bpartial_\bs K,~\partial_t g=\partial_t K\delta',\bpartial_\bs^2 g=\bDelta''^{\T}\bpartial_\bs K+{\bDelta'^{\otimes 2}}^{\T}\bpartial_\bs^2K,\\
				\partial_t\bpartial_\bs g&=\bDelta'^{\T}\partial_t\bpartial_\bs K\delta',\partial_t^2 g=\delta''\partial_t K+\delta'^2\partial_t^2 K,\\
				\bpartial_\bs^3g&=\bDelta'''^{\T}\bpartial_\bs K+3\left(\bDelta''\otimes \bDelta'\right)^{\T}\bpartial_\bs^2K + {\bDelta'^{\otimes 3}}^{\T}\bpartial_\bs^3K,\;
				\partial_t\bpartial_\bs^2g=\bDelta''^{\T}\partial_t\bpartial_\bs K\delta'+{\bDelta'^{\otimes 2}}^{\T}\partial_t\bpartial_\bs^2K\delta',\\
				\partial_t^2\bpartial_\bs g&=\bDelta'^{\T}\partial_t\bpartial_\bs K\delta''+\bDelta'^{\T}\partial_t^2\bpartial_\bs K\delta'^2,\;
				\partial_t^3g=\delta'''\partial_tK + 3\delta''\delta'\partial_t^2K+\delta'^3\partial_t^3K,\\
				\bpartial_\bs^4g&={\bDelta^{(iv)}}^{\T}\bpartial_\bs K + 4\left(\bDelta'''\otimes \bDelta'\right)^{\T}\bpartial_\bs^2K + 3{\bDelta''^{\otimes 2}}^{\T}\bpartial_\bs^2K+
				6\left(\bDelta''\otimes\bDelta'^{\otimes 2}\right)^{\T}\bpartial_\bs^3K+{\bDelta'^{\otimes 4}}^{\T}\bpartial_\bs^4K,\\
				\partial_t^4\bpartial_\bs^4g&=\left\{{\bDelta^{(iv)}}^{\T}\partial_t\bpartial_\bs K + 4\left(\bDelta'''\otimes \bDelta'\right)^{\T}\partial_t\bpartial_\bs^2K + 3{\bDelta''^{\otimes 2}}^{\T}\partial_t\bpartial_\bs^2K+
				6\left(\bDelta''\otimes\bDelta'^{\otimes 2}\right)^{\T}\partial_t\bpartial_\bs^3K+\right.\\
				&\myquad[4]\left.{\bDelta'^{\otimes 4}}^{\T}\partial_t\bpartial_\bs^4K\right\}\delta^{(iv)}+\left\{{\bDelta^{(iv)}}^{\T}\partial_t^2\bpartial_\bs K + 4\left(\bDelta'''\otimes \bDelta'\right)^{\T}\partial_t^2\bpartial_\bs^2K +\right.\\&\myquad[4]
				3{\bDelta''^{\otimes 2}}^{\T}\partial_t\bpartial_\bs^2K+
				6\left(\bDelta''\otimes\bDelta'^{\otimes 2}\right)^{\T}\partial_t^2\bpartial_\bs^3K+\left.{\bDelta'^{\otimes 4}}^{\T}\partial_t^2\bpartial_\bs^4K\right\}\left(4\delta'''\delta'+3\delta''^2\right)+\\&\myquad[4]6\left\{{\bDelta^{(iv)}}^{\T}\partial_t^3\bpartial_\bs K + 4\left(\bDelta'''\otimes \bDelta'\right)^{\T}\partial_t^3\bpartial_\bs^2K +
				3{\bDelta''^{\otimes 2}}^{\T}\partial_t^3\bpartial_\bs^2K+
				6\left(\bDelta''\otimes\bDelta'^{\otimes 2}\right)^{\T}\partial_t^3\bpartial_\bs^3K+\right.\\
				&\myquad[4]\left.{\bDelta'^{\otimes 4}}^{\T}\partial_t^3\bpartial_\bs^4K\right\}\delta'^2\delta''+\left\{{\bDelta^{(iv)}}^{\T}\partial_t^4\bpartial_\bs K + 4\left(\bDelta'''\otimes \bDelta'\right)^{\T}\partial_t^4\bpartial_\bs^2K +
				3{\bDelta''^{\otimes 2}}^{\T}\partial_t^4\bpartial_\bs^2K+\right.\\&\myquad[4]\left.
				6\left(\bDelta''\otimes\bDelta'^{\otimes 2}\right)^{\T}\partial_t^4\bpartial_\bs^3K+{\bDelta'^{\otimes 4}}^{\T}\partial_t^4\bpartial_\bs^4K\right\}\delta'^4,
			\end{align*}
		}
		{ obtained through successive differentiation.} We note that $\bDelta(0)=\bDelta$, $\delta(0)=\delta$, $\bDelta'(h_s)=\bu$, $\bDelta''(h_s)=\bDelta'''(h_s)={\bDelta^{(iv)}}(h_s)=0$, $\delta'(h_t)=1$ and $\delta''(h_t)=\delta'''(h_t)=\delta^{(iv)}(h_t)=0$, for all $h_s$ and $h_t$. Substituting these values in the expressions derived above, we get $\partial_t^{4}\bpartial_\bs^{4}g(0_d,0)={\bu^{\otimes 4}}^{\T}\partial_t^{4}\bpartial_\bs^4K(0_d,0)$. Under isotropy, $K(\bDelta(h_s),\delta(h_t))=\tildeK(||\bDelta(h_s)||,|\delta(h_t)|)$. We define $\rho_s(h_s)=||\bDelta(h_s)||$, $\rho_t(h_t)=|\delta(h_t)|$, and let $g(h_s,h_t)=\tildeK(\rho_s(h_s),\rho_t(h_t))$. Hence, $\rho'_s(0)=\frac{\bu^{\T}\bDelta}{||\bDelta||}$, $\rho_s''(0)=\frac{1}{||\bDelta||}\left(1-\frac{(\bu^{\T}\bDelta)^2}{||\bDelta||^2}\right)$, $\rho_s'''(0)=-3\frac{\bu^{\T}\bDelta}{||\bDelta||^3}\left(1-\frac{(\bu^{\T}\bDelta)^2}{||\bDelta||^2}\right)$, and $\rho_s^{(iv)}(0)=-\frac{3}{||\bDelta||^3}\left(1-5\frac{(\bu^{\T}\bDelta)^2}{||\bDelta||^2}\right)\left(1-\frac{(\bu^{\T}\bDelta)^2}{||\bDelta||^2}\right)$. Similarly, $\rho_t'(0)=\frac{\delta}{|\delta|}$, but $\rho_t''(0)=\rho_t'''(0)=\rho_t^{(iv)}(0)=0$. On noting that $\rho_t'(0)^4=1$ we obtain $\partial_t^4\bpartial_\bs^4g(0,0) = \rho_s^{(iv)}(0)\partial_t^4\bpartial_\bs\tildeK+\left(4\rho_s'''(0)\rho_s'(0)+3\rho_s''(0)^2\right)\partial_t^4\bpartial_\bs^2\tildeK+6\rho_s''(0)\rho_s'(0)^2\partial_t^4\bpartial_\bs^3\tildeK+\rho_s'(0)^4\partial_t^4\bpartial_\bs^4\tildeK$,
		% \begin{equation*}
			%     \partial_t^4\bpartial_\bs^4g(0,0) = \rho_s^{(iv)}(0)\partial_t^4\bpartial_\bs\tildeK+\left(4\rho_s'''(0)\rho_s'(0)+3\rho_s''(0)^2\right)\partial_t^4\bpartial_\bs^2\tildeK+6\rho_s''(0)\rho_s'(0)^2\partial_t^4\bpartial_\bs^3\tildeK+\rho_s'(0)^4\partial_t^4\bpartial_\bs^4\tildeK,
			% \end{equation*}
		which is the required result.
		% }
\end{proof}
Expressions for lower order spatiotemporal derivatives in the above equations are obtained via similar substitutions. In the next discussion we examine covariance kernels that admit the spatiotemporal derivative processes of interest, viz., $\L^*Z(\bs,t)$ or $\L^*_{n_t,\bn_s}Z(\bs,t)$. We investigate separable and non-separable spatiotemporal covariance kernels in search of a necessary criterion involving spectral moments that ensures the existence of $\L^*Z(\bs,t)$. 

%%%%%%%%%%%%%%%%%%%%
% Spectral Density %
%%%%%%%%%%%%%%%%%%%%

\section{Spectral densities for spatiotemporal covariance kernels}
We turn to spectral theory to characterize these covariance functions. Using Bochner's theorem \citep[see, e.g.,][]{bochner2005harmonic, rasmussen_gaussian_2005}, $K$ is a positive-definite covariance function on $\mathfrak{R}^d\times \mathfrak{R}$ if and only if it is of the form $K(\bDelta,\delta)=\iint\limits_{\mathfrak{R}^{d+1}}\exp(i\bw_\bs^{\T}\bDelta+iw_t\delta)\;F(d\bw_\bs,dw_t)$, where $F$ is the spectral measure for $Z(\bs,t)$. Additionally, if $K(\bDelta,\delta)$ is integrable then $F$ is absolutely continuous and has a spectral density, $f$, with respect to the Lebesgue measure which is the Fourier transform of $K$, thus allowing  $ K(\bDelta,\delta)=\iint\limits_{\mathfrak{R}^{d+1}}\exp(i\bw_\bs^{\T}\bDelta+iw_t\delta)\;f(\bw_\bs,w_t)\;d\bw_\bs\;dw_t$, where $f$ is a continuous, non-negative, integrable function. If, $f(\bw_\bs,w_t)=f_\bs(\bw_\bs)\;f_t(w_t)$, a separable spatiotemporal covariance arises (see \Cref{sec:covfns} for examples). The next proposition provides the necessary conditions from the perspective of spectral theory that ensure the existence of $\L^*Z(\bs,t)$.
\begin{proposition}\label{prop:5}
	The process $\L^*Z(\bs,t)$ is valid if the associated spatiotemporal spectral density possesses fourth moments in space and time.
\end{proposition}
\begin{proof}
	We consider the non-separable case. The proof for the separable case uses $f(\bw_\bs,w_t)=f_\bs(\bw_\bs)f_t(w_t)$. Note that $f$ is symmetric around the origin,
	
	\begin{equation*}
		K(\bDelta,\delta)=\iint\limits_{\mathfrak{R}^{d+1}}e^{i\left(\bw_\bs^{\T}\bDelta+w_t\delta\right)}f(\bw_\bs,w_t)\;d\bw_\bs\;dw_t=\iint\limits_{\mathfrak{R}^{d+1}}\cos(\bw_\bs^{\T}\bDelta+w_t\delta)f(\bw_\bs,w_t)\;d\bw_\bs\;dw_t,%=\int_{\mathfrak{R}}\exp\{iw_t\delta\}g(\bDelta,w_t)dw_t,%
	\end{equation*}
	Since $f$ is integrable, differentiating with respect to $\bDelta$ and $\delta$ we obtain
	
	\begin{equation*}
		\begin{split}
			\bpartial_\bs K(\bDelta,\delta)&=-\iint\limits_{\mathfrak{R}^{d+1}}\bw_\bs\sin(\bw_\bs^{\T}\bDelta+w_t\delta)f(\bw_\bs,w_t)\;d\bw_\bs\;dw_t,\\
			\partial_tK(\bDelta,\delta)&=-\iint\limits_{\mathfrak{R}^{d+1}}w_t\sin(\bw_\bs^{\T}\bDelta+w_t\delta)f(\bw_\bs,w_t)\;d\bw_\bs\;dw_t,
		\end{split}
	\end{equation*}
	Since, $|\bw_\bs\sin(\bw_\bs^{\T}\bDelta+w_t\delta)f(\bw_\bs,w_t)|\leq |\bw_\bs f(\bw_\bs,w_t)|$, $|w_t\sin(\bw_\bs^{\T}\bDelta+w_t\delta)f(\bw_\bs,w_t)|\leq |w_t f(\bw_\bs,w_t)|$, $\iint\limits_{\mathfrak{R}^{d+1}}|\bw_\bs|f(\bw_\bs,w_t)\;d\bw_\bs\;dw_t<\infty$ and $\iint\limits_{\mathfrak{R}^{d+1}}|w_t|f(\bw_\bs,w_t)\;d\bw_\bs\;dw_t<\infty$ differentiation under the integral sign is valid and we have 
	
	\begin{equation*}%\label{eq:diff-spectra}
		\partial_t^{j}\bpartial_\bs^{r-j}K(\bDelta,\delta)=\iint_{\mathfrak{R}^{d+1}}\bw_\bs^{\otimes (r-j)}w_t^{j}\cos\left(\frac{r\pi}{2}+\bw_\bs^{\T}\bDelta+w_t\delta\right)f(\bw_\bs,w_t)\;d\bw_\bs\;dw_t,
	\end{equation*}
	where $r>0$ and for $\bw^{\otimes^{r_s}}$ the integration is defined element-wise.
	% This can be verified by observing that for $r=2$ we have
	% \begin{equation*}
		% \begin{split}
			%      \bpartial_\bs^2 K(\bDelta,\delta) &=-\int_{\mathfrak{R}^d}\int_{\mathfrak{R}}w_s^{\otimes^2}\cos(w_s^{\T}\bDelta+w_t\delta)f(w_s,w_t)dw_sdw_t,\\
			%      \bpartial_\bs\partial_tK(\bDelta,\delta) &=-\int_{\mathfrak{R}^d}\int_{\mathfrak{R}}w_sw_t\cos(w_s^{\T}\bDelta+w_t\delta)f(w_s,w_t)dw_sdw_t,\\ \partial_t^2K(\bDelta,\delta) &=-\int_{\mathfrak{R}^d}\int_{\mathfrak{R}}w_t^2\cos(w_s^{\T}\bDelta+w_t\delta)f(w_s,w_t)dw_sdw_t,
			% \end{split}
		% \end{equation*}
	% and for $r=3$,
	% \begin{equation*}
		% \begin{split}
			%      \bpartial_\bs^3K(\bDelta,\delta) &=\int_{\mathfrak{R}^d}\int_{\mathfrak{R}}w_s^{\otimes^3}\sin(w_s^{\T}\bDelta+w_t\delta)f(w_s,w_t)dw_sdw_t,\\
			%      \bpartial_\bs^2\partial_tK(\bDelta,\delta) &=\int_{\mathfrak{R}^d}\int_{\mathfrak{R}}w_s^{\otimes^2}w_t\sin(w_s^{\T}\bDelta+w_t\delta)f(w_s,w_t)dw_sdw_t,\\
			%      \bpartial_\bs\partial_t^2K(\bDelta,\delta) &=\int_{\mathfrak{R}^d}\int_{\mathfrak{R}}w_s w_t^2\sin(w_s^{\T}\bDelta+w_t\delta)f(w_s,w_t)dw_sdw_t,\\
			%      \partial_t^3K(\bDelta,\delta) &=\int_{\mathfrak{R}^d}\int_{\mathfrak{R}}w_t^3\cos(w_s^{\T}\bDelta+w_t\delta)f(w_s,w_t)dw_sdw_t,
			% \end{split}
		% \end{equation*}
	% and so on.
	Under isotropy, noting that the differentiation is with respect to $||\bDelta||$ and $|\delta|$ \citep[see, for e.g,][Theorem 2.5.2, p. 35]{adler1981geometry}, the above result simplifies to
	
	\begin{equation*}
		\partial_t^{j}\bpartial_\bs^{r-j}\tildeK(||\bDelta||,|\delta|)=\iint\limits_{\mathfrak{R}^2}\bw_\bs^{r-j}w_t^{j}\cos\left(\frac{r\pi}{2}+\bw_\bs||\bDelta||+w_t|\delta|\right)\tilde{f}(\bw_\bs,w_t)\;d\bw_\bs\;dw_t.
	\end{equation*}
	We focus on $\partial_t^4\bpartial_\bs^4K(\bDelta,\delta)$ and pass to the limits. We note that the limits exist and the Dominated Convergence Theorem allows for the limit to be interchanged with the integral sign. We make the following observations,
	
	\begin{equation*}
		\begin{split}
			&\lim_{\substack{||\bDelta||\to 0\\|\delta|\to0}}\left(\partial_t^4\bpartial_\bs^2\tildeK-\frac{\partial_t^4\bpartial_\bs\tildeK}{||\bDelta||}\right)\\
			&=-\lim_{||\bDelta||\to 0}\iint\limits_{\mathfrak{R}^2}\left(\cos\left(\bw_\bs||\bDelta||\right)-\frac{\sin\left(\bw_\bs||\bDelta||\right)}{\bw_\bs||\bDelta||}\right)\bw_\bs^2 w_t^4\tilde{f}(\bw_\bs,w_t)\;d\bw_\bs dw_t=0,\\
			&\lim_{\substack{||\bDelta||\to 0\\|\delta|\to0}}\partial_t^4\bpartial_\bs^3\tildeK=\lim_{||\bDelta||\to0}\iint\limits_{\mathfrak{R}^2}\bw_\bs^3 w_t^4\sin(\bw_\bs||\bDelta||)\tilde{f}(\bw_\bs,w_t)\;d\bw_\bs dw_t=0,\\
			&\lim_{\substack{||\bDelta||\to 0\\|\delta|\to0}}\partial_t^4\bpartial_\bs^4\tildeK=\lim_{||\bDelta||\to0}\iint\limits_{\mathfrak{R}^2}\bw_\bs^4 w_t^4\cos(\bw_\bs||\bDelta||)\tilde{f}(\bw_\bs, w_t)\;d\bw_\bs dw_t\\
			&=\iint\limits_{\mathfrak{R}^2}\bw_\bs^4w_t^4\tilde{f}(\bw_\bs,w_t)\;d\bw_\bs dw_t<\infty.
		\end{split}
	\end{equation*}
	Referring to the expressions in \Cref{prop:2}, % eqs. (\ref{eq:iso-2}) and (\ref{eq:iso-3}) 
	we observe that the factors corresponding to $\Bigg(d_{|\delta|}^4d_{||\bDelta||}^2\tildeK-\frac{d_{|\delta|}^4d_{||\bDelta||}\tildeK}{||\bDelta||}\Bigg)$ and $d_{|\delta|}^4d_{||\bDelta||}^3\tildeK$ are $\left(-3\frac{\bP_{4,1}}{||\bDelta||^4}+\frac{\bP_{4,2}}{||\bDelta||^2}+15\frac{\bP_{4,3}}{||\bDelta||^6}-3\frac{\bP_{4,4}}{||\bDelta||^4}\right)$ and $\bigg(\frac{\bP_{4,1}}{||\bDelta||^3}-6\frac{\bP_{4,3}}{||\bDelta||^5}+\frac{\bP_{4,4}}{||\bDelta||^3}\bigg)$ respectively, remain bounded as $||\bDelta||\to 0$. For the spatiotemporal derivative process to be valid, the spectral density, $\tilde{f}(\bw_\bs,w_t)$, must possess fourth moments in space and time. 
\end{proof}
Working with non-separable spatiotemporal kernels, verifying the necessary condition in \Cref{prop:5} requires the existence of closed-form Fourier inversion for the spectral density associated with the kernel. In a purely spatial setting and/or a separable space-time scenario, where Fourier transform pairs exist, verifying this condition is straightforward \citep[see, e.g.,][]{wang2018process, halder2024bayesian}. If we assume separability, then covariance kernels satisfying this condition are readily available. However, non-separable spatiotemporal covariance functions that have closed-form spectral distributions are not easily available \citep[see, e.g.,][]{gneiting2002nonseparable,fuentes2008class}. %\ahal{Talk about the separable case--how simplifications arise! $\partial_t^{r_t}\bpartial_\bs^{r_s}K(||\bDelta||,|\delta|) = \partial_t^{r_t}\bpartial_\bs^{r_s}K_s(||\bDelta||)\otimes K_t(|\delta|) = \partial_t^{r_t}K_t(|\delta|) \otimes \bpartial_\bs^{r_s}K_s(||\bDelta||)$} 
The Gneiting class \citep[see, e.g.,][]{gneiting2002nonseparable} of stationary non-separable spatiotemporal covariance functions allows for a simpler verification of the required condition. They are expressible as $K(\bDelta,\delta)=\tildeK(||\bDelta||, |\delta|)=\frac{\sigma^2}{\psi(|\delta|^2)^{d/2}}\varphi\left(\frac{||\bDelta||^2}{\psi(|\delta|^2)}\right)$, where $\varphi(x)$ and $\psi(x)$ are complete monotone, and positive with complete monotone derivative on $[0,\infty)$ respectively. 
% Following \cite{cressie1999classes} we have,
% \begin{equation}
	%     \tilde{f}(w_s,w_t)=\sigma^2(2\pi)^{-1}\int_{\mathfrak{R}^+}\exp\left\{-iw_t|\delta|\right\}\frac{h(w_s,|\delta|)}{\psi(|\delta|^2)^{d/2}}d|\delta|,
	% \end{equation}
% where
% \begin{equation}
	%     h(w_s,|\delta|)=(2\pi)^{-d}\int_{\mathfrak{R}^+}\exp\left\{-iw_s||\bDelta||\right\}\varphi\left(\frac{||\bDelta||^2}{\psi(|\delta|^2)}\right)d||\bDelta||
	% \end{equation}
% \begin{myprop}
	%     Considering non-separable covariances, $\tildeK$ of the form above, the existence of $\partial^rK(\bDelta,\delta)=\bpartial_\bs^{r_s}\partial_t^{r_t}K(\bDelta,\delta)$, for $r_s=0,1,\ldots,r$, $r_t=r-r_s$, requires $\varphi$ to be differentiable $r$-times and 
	% \end{myprop}
Particularly, we are interested in the Mat\'ern kernel, $\varphi(x)=\left(2^{\nu-1}\bGamma(\nu)\right)^{-1}(\phi_sx^{1/2})^{\nu}K_\nu(\phi_s x^{1/2})$, $\phi_s>0, \nu>0$, where $K_{\nu}(\cdot)$ is the modified Bessel function of second kind, and $\psi(x)=(\phi_t^2x^\alpha+1)^\beta$, $a>0$, $0<\alpha\leq 1$, $0\leq\beta\leq1$. We fix $\alpha=\beta=1$, producing the covariance kernel, $ \tildeK(||\bDelta||,|\delta|;\btheta)=\frac{\sigma^2}{2^{\nu-1}\bGamma(\nu)A_t^{d/2}}\left(\frac{\phi_s||\bDelta||}{A_t^{1/2}}\right)^\nu K_{\nu}\left(\frac{\phi_s||\bDelta||}{A_t^{1/2}}\right)$, $A_t = \phi_t^2|\delta|^2+1$.
% \begin{equation*}
	%     \tildeK(||\bDelta||,|\delta|;\btheta)=\frac{\sigma^2}{2^{\nu-1}\bGamma(\nu)A_t^{d/2}}\left(\frac{\phi_s||\bDelta||}{A_t^{1/2}}\right)^\nu K_{\nu}\left(\frac{\phi_s||\bDelta||}{A_t^{1/2}}\right), \qquad A_t = \phi_t^2|\delta|^2+1.
	% \end{equation*}
%We fix $d=2$, $\sigma^2=1$. 
Observe that for $\delta=0$, $K(||\bDelta||,0;\btheta)$ is the spatial Mat\'ern and, for a fixed $\delta$, $\tildeK(||\bDelta||,|\delta|;\btheta)$ is a re-scaled spatial Mat\'ern, with a spectral density, $\tilde{f}(\bw_s)=C_1(\btheta,|\delta|)\left(C_2(\btheta,|\delta|)+||\bw_s||^2\right)^{-\nu-d/2}$ that belongs to the $t$-family. The fourth spectral moment exists if $\nu>2$. If $||\bDelta||\to0$, observe that $\varphi(0)=1$, hence, $\tildeK(0,|\delta|;\btheta)=\frac{\sigma^2}{(\phi_t^2|\delta|^2+1)^{d/2}}$. This has a spectral density in the Laplace family, $\tilde{f}(w_t)=\frac{\sigma^2}{\phi_t}\sqrt{\frac{\pi}{2}}\exp\left(-\frac{|w_t|}{\phi_t^2}\right)$, for which the fourth moment always exists. In case $||\bDelta||\ne 0$ fixed, and $|\delta|$ close to 0, $\varphi\left(\frac{||\bDelta||^2}{\psi(|\delta|^2)}\right)$ does not interfere with smoothness of $\psi(|\delta|^2)^{d/2}$, for our choices of $\varphi(\cdot)$ and $\psi(\cdot)$, allowing for $K(||\bDelta||,|\delta|)$ to inherit smoothness properties of $K(0,|\delta|)$ \citep[see for e.g.][pp. 311--312]{stein2005space}. We also consider the case, $\varphi(x)$, $\phi_s>0$ and $\nu\to\infty$ resulting in the squared exponential kernel. In such a scenario, if $|\delta|\ne 0$ is kept constant, the spectral distribution, $\tilde{f}(\bw_s)=\left(\frac{\pi}{\phi_s}\right)^{d/2}\exp\left(-\pi^2\frac{\psi(|\delta|^2)}{\phi_s}||\bw_s||^2\right)$, belongs to the Gaussian family having a finite fourth moment \citep[see, e.g.,][Section 4.2.1, pp. 82--83]{rasmussen_gaussian_2005}. Evidently, in the Gneiting class (with $\varphi(x)$ chosen to be Mat\'ern), spatial smoothness of process realizations is governed by the fractal parameter, $\nu>0$, of the Mat\'ern kernel. We consider only half-integer values of $\nu = \nu_0 +\frac{1}{2}$, $\nu_0 =0, 1, 2$ and $\nu\to\infty$, for which we have closed-form expressions of $K_\nu(\cdot)$. Our choice of $\alpha=1$ allows us to study second order temporal derivatives. However, $\alpha\in (0,1]$ can be used to control the temporal smoothness of process realizations \citep[see, e.g.,][]{gneiting2002nonseparable}. The next discussion focuses on surfaces in $\mathfrak{R}^3$ that have \emph{simple} parameterizations.

\section{Parametric Surfaces}

We focus on parametric surfaces $\C=\left\{ (s_x(\omega,\upsilon),s_y(\omega,\upsilon),t(\upsilon)):(\omega,\upsilon) \in \mathcal{D}_{\omega}\times \mathcal{D}_\upsilon\right\}$. The normal to $\C$ at a point $(\bs(\omega,\upsilon),t(\upsilon))$, $\bs(\omega,\upsilon)=(s_x(\omega,\upsilon),s_y(\omega,\upsilon))^{\T}$ is,
\begin{equation*}
	\begin{split}
		\overline{\bn}(\bs(\omega,\upsilon),\upsilon) &= \left(\frac{\partial s_x(\omega,\upsilon)}{\partial\omega}, \frac{\partial s_y(\omega,\upsilon)}{\partial\omega}, 0\right)^{\T} \times \left(\frac{\partial s_x(\omega,\upsilon)}{\partial\upsilon}, \frac{\partial s_y(\omega,\upsilon)}{\partial\upsilon}, \frac{\partial t(\upsilon)}{\partial\upsilon}\right)^{\T}\\&=\left|\begin{smallmatrix}%{*{3}{c}}
			i & j & k\\
			\frac{\partial s_x(\omega,\upsilon)}{\partial\omega} & \frac{\partial s_y(\omega,\upsilon)}{\partial\omega} & 0\\
			\frac{\partial s_x(\omega,\upsilon)}{\partial\upsilon} & \frac{\partial s_y(\omega,\upsilon)}{\partial\upsilon} & \frac{\partial t(\upsilon)}{\partial\upsilon}\\
		\end{smallmatrix}\right|
		= \frac{\partial s_y(\omega,\upsilon)}{\partial\omega}\frac{\partial t(\upsilon)}{\partial\upsilon}\;\hat{i}-\frac{\partial s_x(\omega,\upsilon)}{\partial\omega}\frac{\partial t(\upsilon)}{\partial\upsilon}\;\hat{j}+\left|\frac{\partial(s_x,s_y)}{\partial(\omega,\upsilon)}\right|\;\hat{k},
	\end{split}
\end{equation*}
where $\times$ is the vector product. The equation of the tangent plane at $(\bs(\omega_0,\upsilon_0), \upsilon_0)$ is, $\left\{(x,y,z):\overline{\bn}(\bs(\omega_0,\upsilon_0), t(\upsilon_0))^{\T}\left(\begin{smallmatrix}
	x-s_x(\omega_0,\upsilon_0)\\
	y-s_y(\omega_0,\upsilon_0)\\
	z-t(\upsilon_0)
\end{smallmatrix}\right)=0\right\}$. For convenience, we assume $t(\upsilon)=\upsilon$. Note that $\overline{\bn}_s(\omega,\upsilon)=\left(\frac{\partial s_y(\omega,\upsilon)}{\partial\omega}, -\frac{\partial s_x(\omega,\upsilon)}{\partial\omega}\right)^{\T}$ is the normal to the parametric planar curve, $C_\upsilon=\{(s_x(\omega,\upsilon),s_y(\omega,\upsilon)):\omega\in \D_\omega\}$. This connects parametric curvilinear wombling \citep[see, e.g.,][]{banerjee2006bayesian, halder2024bayesian} to surface wombling. It is also the reason behind our choice of parameterization for $\C$. The usual parameterization found in the literature $\C=\left\{ (s_x(\omega,\upsilon),s_y(\omega,\upsilon),t(\omega,\upsilon)):(\omega,\upsilon) \in \mathcal{D}_{\omega}\times \mathcal{D}_\upsilon\right\}$ is not suitable; the parameterization of the time coordinate cannot depend on the parameterization of spatial coordinates. Hence, $t(\omega,\upsilon)=t(\upsilon)$. In that case, $\overline{\bn}(\bs(\omega,\upsilon),\upsilon)=\frac{\partial s_y(\omega,\upsilon)}{\partial\omega}\frac{\partial t(\upsilon)}{\partial\upsilon}\;\hat{i}-\frac{\partial s_x(\omega,\upsilon)}{\partial\omega}\frac{\partial t(\upsilon)}{\partial\upsilon}\;\hat{j}+\left|\frac{\partial(s_x,s_y)}{\partial(\omega,\upsilon)}\right|\;\hat{k}$. Further assuming, $t(\upsilon)=\upsilon$, produces the desired normal. The surface $\C$ is regular if the bottom two rows of the determinant in \cref{eq:normal} are linearly independent. Note that our parameterization results in regular surfaces with the exception where $\frac{\partial s_x(\omega,\upsilon)}{\partial\omega} = \frac{\partial s_y(\omega,\upsilon)}{\partial\omega} = 0$ and $||\bn_\bs||=0$. Hence, we define  $\C$ as regular if $||\bn_\bs||\ne0$. Parametric curves were defined to be regular in \cite{banerjee2006bayesian} using the same condition. Finally, turning to static wombling, $\left|\frac{\partial(s_x,s_y)}{\partial(\omega,\upsilon)}\right|$ can be zero even if the curve is evolving over time---if the vectors $\left(\frac{\partial s_x(\omega,\upsilon)}{\partial\omega}, \frac{\partial s_y(\omega,\upsilon)}{\partial\omega}\right)^{\T}$, $\left(\frac{\partial s_x(\omega,\upsilon)}{\partial\upsilon}, \frac{\partial s_y(\omega,\upsilon)}{\partial\upsilon}\right)^{\T}$ are linearly dependent. We exclude such parameterizations. Eventually, working with parametric triangular planes (see Section~4.2 of the manuscript) this is not a concern.

\section{Curves on Surfaces: Gauss's Divergence Theorem}

We provide proofs for the results concerning \emph{spatiotemporal flux} (STF) and the \emph{rate of change in spatiotemporal flux} (RSTF) in Section~3.2 of the manuscript. We assume that $d=2$.

\begin{proposition}\label{prop:6}
	Let the process $Z(\bs,t)$ be thrice differentiable, $\bF_1(\bs,t)=\left(\begin{smallmatrix}F_{11}\\F_{12}\\F_{13}\end{smallmatrix}\right)=\bnabla^{\otimes 1}Z(\bs,t)=\left(\begin{smallmatrix}\bpartial_\bs Z(\bs,t)\\\partial_tZ(\bs,t)\end{smallmatrix}\right)$ and $\bF_2(\bs,t)=\left(\begin{smallmatrix}
		F_{21}\\\vdots\\F_{26}
	\end{smallmatrix}\right)=\bnabla^{\otimes 2}Z(\bs,t)$. The STF and RSTF satisfy 
	
	\begin{equation*}
		\begin{split}
			\oiintctrclockwise_{\C=\partial W} \bn(\bs,t)^{\T} \bF_1(\bs,t)\;\mathrm{d}\A&=\iiint_{W}{\rm div}\; \bF_1\;\mathrm{d}W,\\
			\oiintctrclockwise_{\C=\partial W} \{\bn(\bs,t)^{\otimes 2}\}^{\T} \bF_2(\bs,t)\;\mathrm{d}\A&=\iiint_{W}\left\{{\rm div}\; n_{s_x} \left(\begin{smallmatrix}
				F_{21}\\F_{22}\\F_{23}
			\end{smallmatrix}\right) + {\rm div}\; n_{s_y} \left(\begin{smallmatrix}
				F_{22}\\F_{24}\\F_{25}
			\end{smallmatrix}\right) + {\rm div}\; n_{t} \left(\begin{smallmatrix}
				F_{23}\\F_{25}\\F_{26}
			\end{smallmatrix}\right)\right\}\;\mathrm{d}W,
		\end{split}
	\end{equation*}
	where $W$ is a simple region enclosed by curves, $C_\upsilon$ which are closed for every $\upsilon\in [\upsilon_0,\upsilon_1]$, $\oiintctrclockwise$ denotes the surface integral using the default counter-clockwise orientation of $\C$ which produces outward pointing normals, $\mathrm{d}W=ds_x\;ds_y\;dt$ is the volume element and ${\rm div}\; n_{s_x} \left(\begin{smallmatrix}
		F_{21}\\F_{22}\\F_{23}
	\end{smallmatrix}\right) + {\rm div}\; n_{s_y} \left(\begin{smallmatrix}
		F_{22}\\F_{24}\\F_{25}
	\end{smallmatrix}\right) + {\rm div}\; n_{t} \left(\begin{smallmatrix}
		F_{23}\\F_{25}\\F_{26}
	\end{smallmatrix}\right)=(\frac{\partial}{\partial s_x}n_{s_x}F_{21} + \frac{\partial}{\partial s_y}n_{s_x}F_{22} +\frac{\partial}{\partial t}n_{s_x}F_{23})+(\frac{\partial}{\partial s_x}n_{s_y}F_{22}+\frac{\partial}{\partial s_y}n_{s_y}F_{24}+\frac{\partial}{\partial t}n_{s_y}F_{25})+(\frac{\partial}{\partial s_x}n_{t}F_{23}+\frac{\partial}{\partial s_y}n_{t}F_{25}+\frac{\partial}{\partial t}n_{t}F_{26})$, consisting of elements from $\bnabla^{\otimes 3}Z(\bs,t)$.
\end{proposition}
\begin{proof}
	\begin{figure}[t]
		\centering
		\includegraphics[scale=0.6]{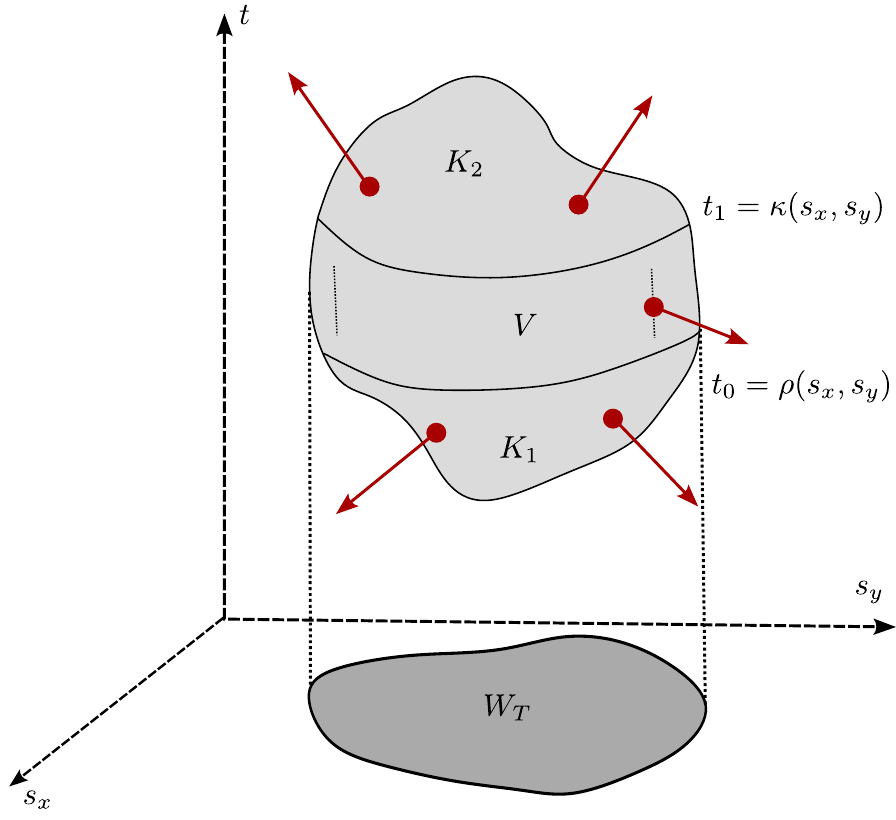}
		\caption{Figure supplementing the proof of \Cref{prop:6}. Normals to the surface are in red.}
		\label{fig:prop-5}
	\end{figure}
	With regard to the first identity for STF, observe that
	
	\begin{equation*}
		\begin{split}
			&\oiintctrclockwise_{\C=\partial W} \bn(\bs,t)^{\T} \bF_1(\bs,t)\;\mathrm{d}\A= \iint\limits_{\C} \frac{\overline{\bn}(\bs(\omega,\upsilon),\upsilon)^{\T}}{||\overline{\bn}(\bs(\omega,\upsilon),\upsilon)||} \bF_1(\bs(\omega,\upsilon),\upsilon)\; ||\overline{\bn}(\bs(\omega,\upsilon),\upsilon)|| \;d\omega\;d\upsilon\\
			&\myquad[2]=\iint\limits_{\C}  F_{11}\frac{\partial s_y(\omega,\upsilon)}{\partial\omega} - F_{12}\frac{\partial s_x(\omega,\upsilon)}{\partial\omega}+ F_{13}\left|\frac{\partial (s_x,s_y)}{\partial(\omega,\upsilon)}\right|\;d\omega\;d\upsilon\\
			&\myquad[2]=\iint\limits_{\C}  F_{11}\frac{\partial s_y(\omega,\upsilon)}{\partial\omega}\;d\omega\;d\upsilon\; + \iint\limits_{\C} F_{12}\left(-\frac{\partial s_x(\omega,\upsilon)}{\partial\omega}\right)\;d\omega\;d\upsilon\; + \iint\limits_{\C}F_{13}\left|\frac{\partial (s_x,s_y)}{\partial(\omega,\upsilon)}\right|\;d\omega\;d\upsilon.
		\end{split}
	\end{equation*}
	Focusing on the third component above, we prove the identity
	
	\begin{equation*}
		\iint\limits_{\C}F_{13}\left|\frac{\partial (s_x,s_y)}{\partial(\omega,\upsilon)}\right|\;d\omega\;d\upsilon=\iiint\limits_{W}\frac{\partial F_{13}}{\partial t}\;ds_x\;ds_y\;dt=\iiint\limits_{W}\frac{\partial F_{13}}{\partial t}\;{\rm d}W.                
	\end{equation*}
	Since $W$ is a simple region ($s_x$-$s_y$ simple), there exists a simple region, $W_T$ in $\mathfrak{R}^2$ and continuous functions $\rho:\mathfrak{R}^2\to \mathfrak{R}$ and $\kappa:\mathfrak{R}^2\to \mathfrak{R}$ such that $\rho(s_x,s_y)<\tau(s_x,s_y)$ for all $(s_x,s_y)\in W_T\setminus\partial W_T$ as seen in \cref{fig:prop-5}. The boundary $\partial W = \C$ for $W$ consists of a top, the graph of $\kappa(s_x,s_y)$, denoted by $K_2$, a bottom, the graph of $\rho(s_x,s_y)$, denoted by $K_1$ and possibly a vertical collar $V$ and $\iint\limits_{\C}F_{13}\left|\frac{\partial (s_x,s_y)}{\partial(\omega,\upsilon)}\right|\;d\omega\;d\upsilon= \iint\limits_{K_1}F_{13}\left|\frac{\partial (s_x,s_y)}{\partial(\omega,\upsilon)}\right|\;d\omega\;d\upsilon+ \iint\limits_{V}F_{13}\left|\frac{\partial (s_x,s_y)}{\partial(\omega,\upsilon)}\right|\;d\omega\;d\upsilon+ \iint\limits_{K_2}F_{13}\left|\frac{\partial (s_x,s_y)}{\partial(\omega,\upsilon)}\right|\;d\omega\;d\upsilon$. Since $K_j$, for $j=1,2$ is a surface, there are simple regions $S_1,\ldots,S_k$ whose interiors, denoted by ${\rm Int}(S_i)$, $i=1,\ldots,k$, are disjoint, such that for injective and smooth functions $s_x:\mathfrak{R}^2\to \mathfrak{R}$ and $s_y:\mathfrak{R}^2\to \mathfrak{R}$ on $U=\bigcup\limits_{i=1}^{k}{\rm Int} (S_i)$, $K_j=\left\{(s_x(\omega,\upsilon)),s_y(\omega,\upsilon),\upsilon):\omega,\upsilon\in S=\bigcup\limits_{i=1}^{k}S_i\right\}$, for $j=1,2$. Since $K_1$ is the bottom of the surface, the third coordinate of $\overline{\bn}(\bs(\omega,\upsilon),\upsilon)$ is negative. Using change of variables, we get
	
	\begin{equation*}
		\begin{split}
			\iint\limits_{K_1}F_{13}\left|\frac{\partial (s_x,s_y)}{\partial(\omega,\upsilon)}\right|\;d\omega\;d\upsilon %&=-\iint\limits_{S}F_{13}(s_x(\omega,\upsilon),s_y(\omega,\upsilon),\upsilon)\left|\frac{\partial (s_x,s_y)}{\partial(\omega,\upsilon)}\right|\;d\omega\;d\upsilon\\
			&=-\iint\limits_{W_T}F_{13}(s_x,s_y,\rho(s_x,s_y))ds_x\;ds_y.
		\end{split}
	\end{equation*}
	Since $F_{13}$ is tangential to the vertical collar $V$, $\iint\limits_{V}F_{13}\left|\frac{\partial (s_x,s_y)}{\partial(\omega,\upsilon)}\right|\;d\omega\;d\upsilon=0$. Note that $K_2$ is the top of the surface and hence the third coordinate of $\overline{\bn}$ would be positive. Using change of variables, $\iint\limits_{K_2}F_{13}\left|\frac{\partial (s_x,s_y)}{\partial(\omega,\upsilon)}\right|\;d\omega\;d\upsilon=\iint\limits_{W_T}F_{13}(s_x,s_y,\kappa(s_x,s_y))\;ds_x\;ds_y$. It follows that
	
	\begin{equation*}
		\begin{split}
			\iint\limits_{\C}F_{13}\left|\frac{\partial (s_x,s_y)}{\partial(\omega,\upsilon)}\right|\;d\omega\;d\upsilon&=\iint\limits_{W_T}F_{13}(s_x,s_y,\kappa(s_x,s_y))\;ds_x\;ds_y-\iint\limits_{W_T}F_{13}(s_x,s_y,\rho(s_x,s_y))\;ds_x\;ds_y,\\
			&=\iint\limits_{W_T}ds_x\;ds_y\int\limits_{\rho(s_x,s_y)}^{\kappa(s_x,s_y)}\frac{\partial F_{13}(s_x,s_y,t)}{\partial t}\;dt=\iiint\limits_{W}\frac{\partial F_{13}}{\partial t}{\rm d} W.
		\end{split}
	\end{equation*}
	Similar results follow for $F_{11}$ and $F_{12}$ observing that $\left|\frac{\partial (s_y,t)}{\partial(\omega,\upsilon)}\right|=\frac{\partial s_y(\omega,\upsilon)}{\partial\omega}$, $\left|\frac{\partial (s_x,t)}{\partial(\omega,\upsilon)}\right|=-\frac{\partial s_x(\omega,\upsilon)}{\partial\omega}$ provided that $\frac{\partial F_{11}}{\partial s_x}$, $\frac{\partial F_{12}}{\partial s_y}$ and $\frac{\partial F_{13}}{\partial t}$ exist, i.e. $Z(\bs,t)$ is twice differentiable.
	
	The second identity concerning the RSTF is proved similarly. Observe that $F_{21}=\frac{\partial^2}{\partial s_x^2}Z(\bs,t)$, $F_{22}=\frac{\partial^2}{\partial s_x\partial s_y}Z(\bs,t)$, $F_{23}=\frac{\partial^2}{\partial s_x\partial t}Z(\bs,t)$, $F_{24}=\frac{\partial^2}{\partial s_y^2}Z(\bs,t)$, $F_{25}=\frac{\partial^2}{\partial s_y\partial t}Z(\bs,t)$ and $F_{26}=\frac{\partial^2}{\partial t^2}Z(\bs,t)$. Hence, $\bnabla^{\otimes 2} Z(\bs,t)= (F_{21},F_{22},F_{23},F_{22},F_{24}, F_{25}, F_{23}, F_{25}, F_{26})^{\T}$ is the $9\times 1$ vector of pure and mixed spatiotemporal derivatives. Let $\overline{\bn}(\bs(\omega,\upsilon),\upsilon)=(\overline{n}_{s_x}(\omega,\upsilon), \overline{n}_{s_y}(\omega,\upsilon),\overline{n}_t(\omega,\upsilon))^{\T}=\left(\frac{\partial s_y(\omega,\upsilon)}{\partial\omega},-\frac{\partial s_x(\omega,\upsilon)}{\partial\omega}, \left|\frac{\partial (s_x,s_y)}{\partial(\omega,\upsilon)}\right|\right)^{\T}$ and similarly $n_{s_x}(\omega,\upsilon)=\frac{\overline{n}_{s_x}(\omega,\upsilon)}{||\overline{\bn}(\bs(\omega,\upsilon),\upsilon)||}$ etc. The left-hand side of the identity can be expressed as
	
	\begin{equation*}
		\oiintctrclockwise_{\C=\partial W} \left(n_{s_x}^2F_{21} + n_{s_y}^2F_{24} + n_{t}^2F_{26}\right) + 2\left(n_{s_x} n_{s_y} F_{22} + n_{s_x} n_{t} F_{23} + n_{s_y} n_{t} F_{25}\right) \;\mathrm{d}\A.
	\end{equation*}
	We show $\oiintctrclockwise_{\C=\partial W} n_{t}^2F_{26}\;\mathrm{d}\A = \iiint\limits_{W}\frac{\partial\;n_{t}F_{26}}{\partial t}{\rm d}W$ and $2\oiintctrclockwise_{\C=\partial W} n_{s_y}n_{t}F_{25}\;\mathrm{d}\A = \iiint\limits_{W}\frac{\partial\;n_{t}F_{25}}{\partial s_y}{\rm d}W+\iiint\limits_{W}\frac{\partial\;n_{s_y}F_{25}}{\partial t}{\rm d}W$.
	% \begin{equation*}
		% \begin{split}
			%     \oiintctrclockwise_{\C=\partial W} n_{t}^2F_{26}\;\mathrm{d}\A &= \iiint\limits_{W}\frac{\partial\;n_{t}F_{26}}{\partial t}{\rm d}W,\\
			%     2\oiintctrclockwise_{\C=\partial W} n_{s_y}n_{t}F_{25}\;\mathrm{d}\A &= \iiint\limits_{W}\frac{\partial\;n_{t}F_{25}}{\partial s_y}{\rm d}W+\iiint\limits_{W}\frac{\partial\;n_{s_y}F_{25}}{\partial t}{\rm d}W.
			% \end{split}
		% \end{equation*}
	Similar identities follow for the other terms. For the first identity,
	
	\begin{equation*}
		\begin{split}
			\iint\limits_{\C}n_{t}F_{26}\left|\frac{\partial (s_x,s_y)}{\partial(\omega,\upsilon)}\right|\;d\omega\;d\upsilon&=\iint\limits_{W_T}n_{t}(s_x,s_y,\kappa(s_x,s_y))\;F_{26}(s_x,s_y,\kappa(s_x,s_y))\;ds_x\;ds_y-\\&\myquad[4]\iint\limits_{W_T}n_{t}(s_x,s_y,\rho(s_x,s_y))\;F_{26}(s_x,s_y,\rho(s_x,s_y))\;ds_x\;ds_y,\\
			&=\iint\limits_{W_T}ds_x\;ds_y\int\limits_{\rho(s_x,s_y)}^{\kappa(s_x,s_y)}\frac{\partial n_{t}(s_x,s_y,t)F_{26}(s_x,s_y,t)}{\partial t}dt=\iiint\limits_{W}\frac{\partial\; n_{t}F_{26}}{\partial t}{\rm d} W.
		\end{split}
	\end{equation*}
	The second equality is proved by observing that $W$ is also $s_x$-$t$ simple and writing $2\oiintctrclockwise\limits_{\C} n_{s_y}n_{t}F_{25}\;\mathrm{d}\A=\oiintctrclockwise\limits_{\C} n_{s_y}F_{25}\left|\frac{\partial (s_x,s_y)}{\partial(\omega,\upsilon)}\right|\;d\omega\; d\upsilon+\oiintctrclockwise\limits_{\C} n_{t}F_{25}\left|\frac{\partial (s_x,t)}{\partial(\omega,\upsilon)}\right|\;d\omega\; d\upsilon$ followed by similar arguments as pursued for the proof for the previous identity. Combining we have 
	
	\begin{equation*}
		\begin{split}
			&\oiintctrclockwise_{\C=\partial W} \{\bn(\bs,t)^{\otimes 2}\}^{\T} F_2(\bs,t)\;\mathrm{d}\A = \iiint\limits_{W}\left\{\left(\mfrac{\partial}{\partial s_x}n_{s_x}F_{21} + \mfrac{\partial}{\partial s_y}n_{s_x} F_{22} +\mfrac{\partial}{\partial t}n_{s_x}F_{23}\right)+\right.\\&\myquad[2]\left.\left(\mfrac{\partial}{\partial s_y}n_{s_y}F_{24}+\mfrac{\partial}{\partial s_x}n_{s_y}F_{22}+\mfrac{\partial}{\partial t}n_{s_y}F_{25}\right)+\left(\mfrac{\partial}{\partial t}n_{t}F_{26}+\mfrac{\partial}{\partial s_x}n_{t}F_{23}+\mfrac{\partial}{\partial s_y}n_{t}F_{25}\right)\right\}\;{\rm d}W.
		\end{split}
	\end{equation*}
	The integrand on the right-hand side is the expanded expression for ${\rm div}\; n_{s_x} \left(\begin{smallmatrix}
		F_{21}\\F_{22}\\F_{23}
	\end{smallmatrix}\right) + {\rm div}\; n_{s_y} \left(\begin{smallmatrix}
		F_{22}\\F_{24}\\F_{25}
	\end{smallmatrix}\right) + {\rm div}\; n_{t} \left(\begin{smallmatrix}
		F_{23}\\F_{25}\\F_{26}
	\end{smallmatrix}\right)$. It is comprised of elements from $\bnabla^{\otimes 3}Z(\bs,t)$ which requires $Z(\bs,t)$ to be thrice differentiable. The arguments above are adapted from the proof of Gauss's Divergence Theorem \citep[see, e.g.,][]{spivak1979comprehensive}.
\end{proof}

\section{Asymptotic Properties of \texorpdfstring{$\bGamma(\widetilde{\C})$}{GCtilde}}
We investigate the asymptotic properties of the wombling measure, $\bGamma(\widetilde{\C})$ computed on the triangulation, $\widetilde{\C}$ of a regular surface, $\C$ using a partition, $\mathscr{P}$ as the norm of the partition, $|\mathscr{P}|_A\to 0$. This supplements the discussion in Section~4.2 of the manuscript.

\begin{theorem}\label{thm:2}
	Let $\C$ be a regular (wombling) surface that can be triangulated using a partition $\mathscr{P}$ (with norm $|\mathscr{P}|_A$) and let $\widetilde{\C}$ be the triangulated approximation. Then $\bGamma_g(\widetilde{\C})\stackrel{a.s.}{\to} \bGamma_g(\C)$ as $|\mathscr{P}|_A\to0$, where $\bGamma_g(\widetilde{\C})$ is the wombling measure for $\widetilde{\C}$ and $\bGamma_g(\C) = \iint\limits_{\C} g\left(\L^*_{n_t,\bn_s}Z(\bs,t)\right)\;{\rm d}\A$ with $g(\cdot)$ being any uniformly continuous (u.c.) function.
\end{theorem}
\begin{proof}
	Over the rectangle $\D_\omega\times \D_{\upsilon} = [\omega_0,\omega_1]\times [\upsilon_0,\upsilon_1]$, the wombling measure $\bGamma_g(\C) = \int\limits_{\upsilon_0}^{\upsilon_1} \int\limits_{\omega_0}^{\omega_1}g\left(\L^*_{n_t,\bn_s}Z(\bs(\omega,\upsilon),\upsilon)\right)\;||\overline{\bn}(\bs(\omega,\upsilon),\upsilon)||\; d\omega\;d\upsilon$ and we denote $g\left(\L^*_{n_t,\bn_s}Z(\bs(\omega,\upsilon),\upsilon)\right) = g(\omega,\upsilon)$ and $\overline{\bn}(\bs(\omega,\upsilon),\upsilon)= \overline{\bn}(\omega,\upsilon)$. Given the partition $\mathscr{P}$, $\omega_0=\omega_1'<\ldots<\omega_{n_\omega}'<\omega_1$ and $\upsilon_0=\upsilon'_1<\ldots<\upsilon'_{n_\upsilon}=\upsilon_1$ used to triangulate the surface, using the definition of a surface integral---for any $\epsilon>0$ there exists a $\delta_0>0$ such that for any $|\mathscr{P}|_A<\delta_0$
	
	\begin{equation*}
		\left|\int\limits_{\upsilon_0}^{\upsilon_1} \int\limits_{\omega_0}^{\omega_1}g\left(\omega,\upsilon\right)\;||\overline{\bn}(\omega,\upsilon)||\; d\omega\;d\upsilon-\sum\limits_{i=1}^{n_\omega-1}\sum\limits_{j=1}^{n_\upsilon-1}\sum\limits_{k=1}^{2}(\upsilon'_{j+1}-\upsilon'_j)(\omega'_{i+1}-\omega'_i)g\left(\omega_i^*,\upsilon_j^*\right)||\overline{\bn}_{ijk}(\omega_i^*,\upsilon_j^*)||\right|<\frac{\epsilon}{2},
	\end{equation*}
	where $\sum\limits_{i=1}^{n_\omega-1}\sum\limits_{j=1}^{n_\upsilon-1}\sum\limits_{k=1}^{2}(\upsilon'_{j+1}-\upsilon'_j)(\omega'_{i+1}-\omega'_i)\;g\left(\omega_i^*,\upsilon_j^*\right)\;||\overline{\bn}_{ijk}(\omega_i^*,\upsilon_j^*)||$ is the \emph{Riemann sum approximation}, $\left(\begin{smallmatrix}
		\omega_i^*\\ \upsilon_j^*
	\end{smallmatrix}\right)=\left(\begin{smallmatrix}
		\omega'_i\\\upsilon'_j
	\end{smallmatrix}\right)+\frac{1}{2}\left(\begin{smallmatrix}
		\omega'_i +\omega'_{i+1}\\\upsilon'_j+\upsilon'_{j+1}
	\end{smallmatrix}\right)$ are the centroids of the rectangles which partition $\bigcup\limits_{k=1}^2T_{ijk}$ and $\overline{\bn}_{ijk}$ is the normal to the triangle $T_{ijk}$, for $k=1,2$. 
	
	Since, $g:\mathfrak{R}^2\to\mathfrak{R}$ is a u.c. function over $\D_\omega\times \D_{\upsilon}$, for $\epsilon>0$, there exists $\delta_1>0$ such that for any $x,y\in \D_\omega\times \D_{\upsilon}$ with $||x-y||<\delta_1$, $x= \left(\begin{smallmatrix}
		x_\omega\\x_\upsilon
	\end{smallmatrix}\right)$, $y= \left(\begin{smallmatrix}
		y_\omega\\y_\upsilon
	\end{smallmatrix}\right)$ we have
	
	\begin{equation*}
		\left| g\left(x_\omega,x_\upsilon\right)||\overline{\bn}_{ijk}(x_\omega,x_\upsilon)|| - g\left(y_\omega,y_\upsilon\right)||\overline{\bn}_{ijk}(y_\omega,y_\upsilon)||\right|<\frac{\epsilon}{2(\omega_1-\omega_0)(\upsilon_1-\upsilon_0)}.
	\end{equation*}
	Setting $\delta = \min\{\delta_0,\delta_1\}$, $|\mathscr{P}|_A<\delta$ and using the Mean Value Theorem we obtain the following inequalities,
	
	\begin{align*}
		&\left|\sum\limits_{i=1}^{n_\omega-1}\sum\limits_{j=1}^{n_\upsilon-1}\sum\limits_{k=1}^{2}\iint\limits_{T_{ijk}}g\left(\omega,\upsilon\right)\;||\overline{\bn}_{ijk}(\omega,\upsilon)||\;d\omega\;d\upsilon\;-\right.\\
		&\myquad[4]\left.\sum\limits_{i=1}^{n_\omega-1}\sum\limits_{j=1}^{n_\upsilon-1}\sum\limits_{k=1}^{2}(\upsilon'_{j+1}-\upsilon'_j)\;(\omega'_{i+1}-\omega'_i)\;g\left(\omega_i^*,\upsilon_j^*\right)\;||\overline{\bn}_{ijk}(\omega_i^*,\upsilon_j^*)||\right|\leq\\
		&\left|\sum\limits_{i=1}^{n_\omega-1}\sum\limits_{j=1}^{n_\upsilon-1}\sum\limits_{k=1}^{2}(\upsilon'_{j+1}-\upsilon'_j)(\omega'_{i+1}-\omega'_i)\sup\limits_{\D_\omega\times\D_\upsilon} g\left(\omega,\upsilon\right)||\overline{\bn}_{ijk}(\omega,\upsilon)||\;-\right.\\
		&\myquad[4]\left.\sum\limits_{i=1}^{n_\omega-1}\sum\limits_{j=1}^{n_\upsilon-1}\sum\limits_{k=1}^{2}(\upsilon'_{j+1}-\upsilon'_j)(\omega'_{i+1}-\omega'_i)g\left(\omega_i^*,\upsilon_j^*\right)||\overline{\bn}_{ijk}(\omega_i^*,\upsilon_j^*)||\right|\leq\\
		&\sum\limits_{i=1}^{n_\omega-1}\sum\limits_{j=1}^{n_\upsilon-1}\sum\limits_{k=1}^{2}(\upsilon'_{j+1}-\upsilon'_j)(\omega'_{i+1}-\omega'_i)\sup\limits_{\D_\omega\times\D_\upsilon} \left|g\left(\omega,\upsilon\right)||\overline{\bn}_{ijk}(\omega,\upsilon)||\;-g\left(\omega_i^*,\upsilon_j^*\right)||\overline{\bn}_{ijk}(\omega_i^*,\upsilon_j^*)||\right|\leq \frac{\epsilon}{2}.
	\end{align*}
	Combining this with the first inequality in this proof---for every $\epsilon>0$, there exists a $\delta>0$ such that for any $|\mathscr{P}|_A<\delta$ we have
	
	\begin{equation}\label{eq:eps-delta}
		\left|\int\limits_{\upsilon_0}^{\upsilon_1} \int\limits_{\omega_0}^{\omega_1}g\left(\omega,\upsilon\right)\;||\overline{\bn}(\omega,\upsilon)||\; d\omega\;d\upsilon - \sum\limits_{i=1}^{n_\omega-1}\sum\limits_{j=1}^{n_\upsilon-1}\sum\limits_{k=1}^{2}\iint\limits_{T_{ijk}}g\left(\omega,\upsilon\right)\;||\overline{\bn}_{ijk}(\omega,\upsilon)||\;d\omega\;d\upsilon\right|<\epsilon.
	\end{equation}
	Our choices for $g$ are all linear, for example, directional derivatives and curvatures which are by definition u.c. Next we show the almost sure convergence of $\bGamma_g(\widetilde{\C})$.
	
	Let $\bGamma_{g,n_\omega,n_\upsilon}(\widetilde{\C}_Z)=\sum\limits_{i=1}^{n_\omega-1}\sum\limits_{j=1}^{n_\upsilon-1}\sum\limits_{k=1}^{2}\iint\limits_{T_{ijk}}g\left(\L^*_{n_t,\bn_s}Z(\bs(\omega,\upsilon),\upsilon)\right)\;||\overline{\bn}(\bs(\omega,\upsilon),\upsilon)||\; d\omega\;d\upsilon$ be the wombling measure approximating $\bGamma_g(\C_Z)$ using a triangulated approximation, $\widetilde{\C}_Z$, of surface $\C_Z$, where the suffix $Z$ underscores the dependence on the parent process $Z(\bs,t)$. Consider sequences $n_\omega,n_\upsilon\to\infty$ such that $|\mathscr{P}|_A=\max\limits_{\substack{1\leq i\leq n_\omega-1\\1\leq j\leq n_\upsilon-1}}(\omega'_{i+1}-\omega'_i)(\upsilon'_{j+1}-\upsilon'_j)\to0$. Using \cref{eq:eps-delta} we obtain $\bGamma_{g,n_\omega,n_\upsilon}(\widetilde{\C}_Z) \stackrel{a.s.}{\to}\bGamma_g(\C_Z)$. Therefore, $\bGamma_{g,n_\omega,n_\upsilon}(\widetilde{\C}_Z)$ is a strongly consistent estimator for $\bGamma_g(\C_Z)$.
\end{proof}

\section{Computational Simplifications}

The choice of parameterization, be it simple parametric planes in the spatiotemporal context or, rectilinear segments in the spatial context, allows us to avoid high dimensional quadrature when computing the variance of wombling measures. Simplifications arise in the purely spatial and spatiotemporal case. Our choice of kernels provides further simplifications. They are the focus of the ensuing discussion. We start with the spatial scenario before moving to the spatiotemporal scene. Within a spatial context \citep[see, e.g.,][Section 4 and Section 3 respectively]{banerjee2006bayesian, halder2024bayesian}, depending on the smoothness of the process, we are interested in performing posterior inference either for wombling measures based on the spatial gradient or, curvature wombling, which seeks posterior inference for wombling measures based jointly on the spatial gradient and curvature processes. Detailing the construction of the cross-covariance matrix associated with such inferential pursuits in the next paragraph, we observe that in the spatial case, no quadrature is required. We consider $d=2$ for ease of presentation. We denote the standard Gaussian cumulative distribution function (cdf) using $\Phi(\cdot)$ and the cdf of a gamma distribution with shape parameter $a$ and scale parameter, $b$ as, $G\left(a,\frac{x}{b}\right)$, where $G(\cdot, \cdot)$ is the regularized gamma function.

\subsection{Spatial}
Consider the purely spatial case found in the discussion after eq. (12), p. 1493, \cite{banerjee2006bayesian} and also eq. (10) in \cite{halder2024bayesian}. Using expressions found in the latter, observe that the curve is parameterized via rectilinear segments. The following lemma details the simplification in the cross-covariance followed by a proposition which uses \Cref{lemma:2} to derive further simplifications specific to kernels in \cite{halder2024bayesian}.

\begin{lemma}\label{lemma:2}
	On each rectilinear segment, $C_{t^*}=\{\bs_0+t\bu:t\in[0,t^*]\}$, where $\bu$ is a unit vector and $\bu^\perp$ is its normal, the terms in the cross-covariance matrix are
	
	\begin{equation*}
		\begin{split}
			k_{ij}(t^*,t^*)&= (-1)^i\int_0^{t^*}\int_0^{t^*}\ba_i(t_1)^{\T}\;\bpartial_\bs^{i+j}K(\bDelta(t_1,t_2))\;\ba_j(t_2)\;||\bs'(t_1)||\;||\bs'(t_2)||\;dt_1\;dt_2,\\
			&=(-1)^i\;t^*\int_{-t^*}^{t^*}\ba_i^{\T}\;\bpartial_\bs^{i+j}K(x\bu)\;\ba_j\;dx, \quad i,j = 1,2,
		\end{split}
	\end{equation*}
	where $\ba_1(t)=\bn(\bs(t))$ is the normal to the segment, $\ba_2(t)=\mathscr{E}_2\;\bn(\bs(t))\otimes \bn(\bs(t))$, where $\mathscr{E}_2= \left(\begin{smallmatrix}
		1 & & & \\ & 1 & 1 & \\ & & & 1
	\end{smallmatrix}\right)$ is an elimination matrix and $\bDelta(t_1,t_2)=\bs_2(t)-\bs_1(t)$.
\end{lemma}
\begin{proof}
	Considering the parametric segment, $C_{t^*}=\{\bs_0+t\bu:t\in[0,t^*]\}$, where $\bu=(u_1,u_2)^{\T}$ is a unit vector, $||\bu||=1$, $\bu^\perp$ is its normal, $\bu^{\T}\;\bu^\perp=0$, $\ba_1(t)=\ba_1=\bu^\perp$ and $\ba_2(t)=\ba_2=\mathscr{E}_2(\bu^\perp\otimes \bu^\perp)$ are free of $t$, and $\bDelta(t_1,t_2)=(t_2-t_1)\bu$. The integrand in $k_{ij}(t^*,t^*)= (-1)^i\int_0^{t^*}\int_0^{t^*}\ba_i^{\T}\;\bpartial_\bs^{i+j}K((t_2-t_1)\bu)\;\ba_j\;dt_1\;dt_2$, depends only on $(t_2-t_1)$. Making a change of variable--define $x=t_2-t_1$, $y=t_2+t_1$, implying $\frac{x+y}{2}=t_2$ and $\frac{y-x}{2}=t_1$. The Jacobian is $\frac{1}{2}$. Hence, $0\leq y+x\leq 2t^*$ and $0\leq y-x\leq 2t^*$. This implies $0 \leq y \leq 2t^*$ and $-t^*\leq x \leq t^*$. Making the substitution above reduces, $k_{ij}(t^*,t^*)= \frac{(-1)^i}{2}\int_0^{2t^*}\int_{-t^*}^{t^*}\ba_i^{\T}\;\bpartial_\bs^{i+j}K(x\bu)\;\ba_j\;dx\;dy=(-1)^i\;t^*\int_{-t^*}^{t^*}\ba_i^{\T}\;\bpartial_\bs^{i+j}K(x\bu)\;\ba_j\;dx$. Setting $i=j=1$, produces the scenario in \cite{banerjee2006bayesian}. 
\end{proof}

\begin{proposition}\label{prop:7}
	For the Mat\'ern kernel, with $\nu = 3/2$, the variance of the gradient based wombling measure is  $k_{11}(t^*,t^*)=2\sqrt{3}\sigma^2\phi_st^*G\left(1,\sqrt{3}\phi_st^*\right)$. In case $\nu=5/2$, the cross-covariance matrix for the wombling measures based on spatial gradient and curvature is constructed using $k_{11}(t^*,t^*)=\frac{2\sqrt{5}}{3}\sigma^2\phi_s\;t^*\left\{G\left(1,\sqrt{5}\phi_st^*\right)+G\left(2,\sqrt{5}\phi_st^*\right)\right\}$, $k_{22}(t^*,t^*)=10\sqrt{5}\sigma^2\phi_s^3\;t^*G(1,\sqrt{5}\phi_st^*)$ and $k_{21}(t^*,t^*)=-k_{12}(t^*,t^*)=0$. For the Gaussian kernel, 
	
	\begin{equation*}
		\bV_{\bGamma}(t^*)=2\sigma^2\sqrt{\pi}\phi_s^{1/2}\;t^*\left\{2\Phi\left(\sqrt{2\phi_s}t^*\right)-1\right\}\left(\begin{smallmatrix} 1 & 0\\0&6\phi_s\end{smallmatrix}\right).
	\end{equation*}
	In both cases $\bV_{\bGamma}(t^*)$ is a variance-covariance matrix.
\end{proposition}
\begin{proof}
	Using \Cref{lemma:2}, if $\nu = 3/2$, then we obtain $k_{11}(t^*,t^*)=3\sigma^2\phi_s^2\;t^*\int_{-t^*}^{t^*}e^{-\sqrt{3}\phi_s |x|}\;dx=2\sqrt{3}\sigma^2\phi_st^*\;G(1,\sqrt{3}\phi_st^*)$. If $\nu=5/2$, then $k_{11}(t^*,t^*)=\frac{5}{3}\sigma^2\phi_s\;t^*\int_{-t^*}^{t^*}(1+\sqrt{5}\phi_s|x|)\;e^{-\sqrt{5}\phi_s |x|}\;dx=\frac{2\sqrt{5}}{3}\sigma^2\phi_s\;t^*\left\{G\left(1,\sqrt{5}\phi_st^*\right)+G\left(2,\sqrt{5}\phi_st^*\right)\right\}$. For terms $k_{21}$ and $k_{12}$, consider $a_1^{\T}\bpartial_\bs^3K(x\bu)\;a_2=\frac{25}{3}\sigma^2\phi_se^{-\sqrt{5}\phi_s |x|}|x|\big\{\ba_1^{\T} \bA_1 \ba_2 -\sqrt{5}\phi_s|x|\; \ba_1^{\T}\; \bA_2 \;\ba_2\big\}$, where 
	
	\begin{equation*}
		\bA_1 = \left(\begin{smallmatrix} 3u_1 & u_2 & u_1\\u_2 & u_1 & 3u_2 \end{smallmatrix}\right), \quad \bA_2=\left(\begin{smallmatrix} u_1^3 & u_1^2\;u_2 & u_2^2\;u_1\\ u_1^2\;u_2 & u_2^2\;u_1 & u_2^3 \end{smallmatrix}\right).
	\end{equation*} 
	The first and the second term both reduce to 0. Hence, $k_{21}(t^*,t^*)=-k_{12}(t^*,t^*)=0$. Next, $\ba_2^{\T} \bpartial_\bs^4K(x\bu)\;\ba_2 = \frac{25}{3}\sigma^2\phi_s^4\;e^{-\sqrt{5}\phi_s |x|}\big\{\ba_2^{\T} \bA_3\; \ba_2 - \sqrt{5}\phi_s\; |x|\; \ba_2^{\T} \bA_4\Big(\sqrt{5}\;\phi_s\;|x|+1\Big) \;\ba_2\big\}$, where 
	
	\begin{equation*}
		\bA_3=\left(\begin{smallmatrix} 3 & 0 & 1\\0 & 1 & 0\\1 & 0 & 3 \end{smallmatrix}\right), \quad \bA_4(x')=\left(\begin{smallmatrix}
			6\;u_1^2-\;x'\; u_1^4 & (3-x'\; u_1^2)\;u_1\;u_2 & 1-x'\;u_1^2\;u_2^2\\ (3-x'\; u_1^2)\;u_1\;u_2 & 1-x'\;u_1^2\;u_2^2 & (3-x'\;u_2^2)\;u_1\;u_2\\1-x'\;u_1^2\;u_2^2 & (3-x'\;u_2^2)\;u_1\;u_2 & 6\;u_2^2-\;x'\; u_2^4
		\end{smallmatrix}\right).
	\end{equation*}
	After some algebra, $\ba_2^{\T}\; \bpartial_\bs^4K(x\bu)\;\ba_2=25\sigma^2\phi_s^4\;e^{-\sqrt{5}\phi_s |x|}\big\{1-2\sqrt{5}\;\phi_s\;|x|\;\bu_1^\perp \;\bu_2^\perp\;(\bu_1\;\bu_2+\bu_1^\perp\; \bu_2^\perp)\big\}$. Observe that $\bu_1^\perp = \bu_2$ and $\bu_2^\perp=-\bu_1$, substituting we get $k_{22}(t^*,t^*)=10\sqrt{5}\sigma^2\phi_s^3\;t^*G(1,\sqrt{5}\phi_st^*)$. For the Gaussian kernel,
	$k_{11}(t^*,t^*)=2\sigma^2\phi_s\;t^*\int_{-t^*}^{t^*}e^{-\phi_s x^2}\;dx=2\sigma^2\sqrt{\pi}\phi_s^{1/2}\;t^*\left\{2\Phi\left(\sqrt{2\phi_s}t^*\right)-1\right\}$. 
	Note that $\ba_1^{\T}\bpartial_\bs^3K(x\bu)\;\ba_2=4\sigma^2\phi_s^2\;e^{-\phi_s x^2}x\left\{\ba_1^{\T} \bA_1 \ba_2 -2\phi_s x^2\; \ba_1^{\T}\; \bA_2 \;\ba_2\right\}$. Again, the first and second terms equate to 0 implying $k_{21}(t^*,t^*)=-k_{12}(t^*,t^*)=0$. Next, $\ba_2^{\T} \bpartial_\bs^4K(x\bu)\;\ba_2 = 4\sigma^2\phi_s^2\;e^{-\phi_s x^2}\big\{\ba_2^{\T} \bA_3\; \ba_2 - 2\phi_s\; x^2\; \ba_2^{\T} \bA_4(2\phi_sx^2) \;\ba_2\big\}$. After some algebra, $\ba_2^{\T}\; \bpartial_\bs^4K(x\bu)\;\ba_2=12\sigma^2\phi_s^2\;e^{-\phi_s x^2}\big\{1-4\;\phi_s\;x^2\;\bu_1^\perp\; \bu_2^\perp\;(\bu_1\;\bu_2+\bu_1^\perp\;\bu_2^\perp)\big\}$. 
	Substituting, we get $k_{22}(t^*,t^*)=12\sigma^2\sqrt{\pi}\phi_s^{3/2}t^*\left\{2\Phi\left(\sqrt{2\phi_s}t^*\right)-1\right\}$ resulting in the required expression for $\bV_{\bGamma}(t^*)$.
\end{proof}
The covariance matrix for the Gaussian case has a simpler form than a Mat\'ern with $\nu = 5/2$. The required expressions for $\bpartial_\bs^{i+j}K(x\bu)$, $i,j=1,2$, used for the above construction can be found in \cref{sssec:m52,sssec:gauss} (also see Section S3, Supplement of \cite{halder2024bayesian}). Finally, considering the covariance between $Z(\bs_i)$, $i=1,\ldots, N$ and the wombling measures---in the Gaussian case, closed-forms are available \citep[see, e.g.,][end of Section 3]{halder2024bayesian}. Therefore, the inferential exercise of spatial wombling does not require quadrature when using a Gaussian kernel. However, only one-dimensional quadrature is required for the same when a Mat\'ern kernel is used.

\subsection{Spatiotemporal}

In spatiotemporal wombling, the surface is partitioned using parametric triangular planes. The wombling measures are computed on each triangular plane. Their variance involves a four-dimensional integral (see eq.~(8) in the manuscript). The integrand depends on $(\omega_1,\omega_2,\upsilon_1,\upsilon_2)$ through $(\omega_2-\omega_1)$ and $(\upsilon_2-\upsilon_1)$. The next lemma extends \Cref{lemma:2}.

\begin{lemma}\label{lemma:3}
	For a wombling surface $\C$ triangulated using a partition $\mathscr{P}$ as in \Cref{thm:2}, the cross-covariance for the wombling measures $\bGamma(\C_{ijk})$ on the triangular plane, $T=\{(\omega,\upsilon):0\leq \omega + \upsilon \leq 1\}$ satisfies the following identity, 
	
	\begin{equation*}
		\begin{split}
			\bK_\bGamma(\C_{ijk},\C_{ijk})&=\iiiint\limits_{T\times T}\bN_{ijk,\bs t}\;\bV_{\L^*Z}(\bDelta_k(\omega_1,\omega_2,\upsilon_1,\upsilon_2),\delta_k(\upsilon_1,\upsilon_2))\;\bN_{ijk,\bs t}^{\T}\;||\overline{\bn}_{ijk}||^2\;d\omega\;d\upsilon\\
			&=\frac{||\overline{\bn}_{ijk}||^2}{2}\int\limits_{-1}^{1}\int\limits_{-1-x_\upsilon}^{1-x_\upsilon} \bN_{ijk,\bs t}\;\bV_{\L^*Z}(\bDelta_k(x_\omega,x_\upsilon),\delta_k(x_\upsilon))\;\bN_{ijk,\bs t}^{\T}\;dx_\omega\;dx_\upsilon,
		\end{split}
	\end{equation*}
	where , $i=1,\ldots,n_\omega-1$, $j=1,\ldots, n_\upsilon-1$, $k=1,2$, $x_\omega = \omega_2-\omega_1$ and $x_\upsilon = \upsilon_2-\upsilon_1$ and the integral acts element-wise on matrices.
\end{lemma}
\begin{proof}
	We make the following change of variables, $x_\omega = \omega_2-\omega_1$, $y_\omega = \omega_2+\omega_1$, $x_\upsilon = \upsilon_2-\upsilon_1$, $y_\upsilon = \upsilon_2+\upsilon_1$. As a result, $\frac{x_\omega+y_\omega}{2}=\omega_2$, $\frac{y_\omega-x_\omega}{2}=\omega_1$, $\frac{x_\upsilon+y_\upsilon}{2}=\upsilon_2$ and $\frac{y_\upsilon-x_\upsilon}{2}=\upsilon_1$. The Jacobian is $\frac{1}{4}$. Since $0\leq\omega_k+\upsilon_k\leq 1$ for $k=1,2$ the above change of variables also results in %$0\leq x_\omega+y_\omega +x_\upsilon+y_\upsilon\leq 2$ and $0\leq y_\omega-x_\omega+y_\upsilon-x_\upsilon\leq 2$, 
	$0\leq y_\omega+y_\upsilon\leq 2$ and $-1\leq x_\omega+x_\upsilon\leq 1$. $\bN_{ijk,\bs t}$ is free of $(\omega_k,\upsilon_k)$ for every $i=1,\ldots,n_\omega-1$, $j=1,\ldots,n_\upsilon-1$ and $k=1,2$ (see Section~4.2 of the manuscript). Making the substitution, $K_\bGamma(\C_{ijk},\C_{ijk})= \frac{||n_{ijk}||^2}{4}\;\int\limits_{0}^{2}\int\limits_{0}^{2-y_\upsilon}\int\limits_{-1}^{1}\int\limits_{-1-x_\upsilon}^{1-x_\upsilon} \bN_{ijk,\bs t}\;\bV_{\L^*Z}(\bDelta_k(x_\omega,x_\upsilon),\delta_k(x_\upsilon))\;\bN_{ijk,\bs t}^{\T}\;dx_\omega\;dx_\upsilon\;dy_\omega\;dy_\upsilon=\frac{||\overline{\bn}_{ijk}||^2}{2}\int\limits_{-1}^{1}\int\limits_{-1-x_\upsilon}^{1-x_\upsilon} \bN_{ijk,\bs t}\;\bV_{\L^*Z}(\bDelta_k(x_\omega,x_\upsilon),\delta_k(x_\upsilon))\;\bN_{ijk,\bs t}^{\T}\;dx_\omega\;dx_\upsilon$.
\end{proof}
The identity in the above lemma reduces an integration in four dimensions to two dimension thereby easing the computational burden associated with high dimensional quadrature. Subsequently, \Cref{lemma:3} finds use for analytic evaluation of the cross-covariance matrix when obtaining posterior samples of the wombling measures (see \Cref{algorithm:wm} for details). Further simplifications for the Gaussian kernel are outlined in the next discussion.

\subsection{Non-separable Gaussian Kernel}
We consider elements of the $8\times 1$ vector, $\bgamma(\C_{ijk})=\left(\begin{smallmatrix}
	\gamma_1\\\vdots\\\gamma_8
\end{smallmatrix}\right)$ computed over the triangular region, $T=\{(\omega,\upsilon):0\leq \omega + \upsilon < 1\}$ for any $i=1,\ldots,n_\omega-1$, $j=1,\ldots,n_\upsilon-1$, $k=1,2$ and an arbitrary point $(\bs_{i'},t_{i'})$, $i'=1,\ldots, N$. We first recap the notation. Let $\bDelta_{i'1}(\omega,\upsilon) = \bDelta_{i'ij}-\omega \bu_{ij}-\upsilon \bv_{ij+1}$, $\bDelta_{i'2}(\omega,\upsilon)=\bDelta_{i'i+1j+1}-\omega \bu_{ij+1}-\upsilon \bv_{ij}$, $\delta_{i'1}(\upsilon)=\delta_{i'j}-\upsilon\;\delta_{j+1j}$ and $\delta_{i'2}(\upsilon)=\delta_{i'j+1}+\upsilon\;\delta_{j+1j}$ for $i'=1,\ldots, N$, where $\bDelta_{i'ij}=\bs_{i'}-\bs(\omega'_i,t_j)$, $\bDelta_{i'i+1j+1}=\bs_{i'}-\bs(\omega'_{i+1},t_{j+1})$, $\delta_{i'j}=t_{i'}-t_j$, $\delta_{i'j+1}=t_{i'}-t_{j+1}$ and $\delta_{jj+1}=t_{j+1}-t_j$.

\begin{proposition}\label{prop:8}
	Let $Z(\bs,t)\sim GP(0,K(\cdot,\cdot;\btheta))$ where $K$ is the Gaussian kernel. For a wombling surface, $\C$, upon triangulation (see Section~4.2), analytic evaluation of $\bG(\C)$ and $\bK_\bGamma(\C,\C)$ (see eq.~(10)) requires only 1-dimensional quadrature.
\end{proposition}

\begin{proof}
	First we note that for a Gaussian kernel $\L^*Z(\bs,t)$ is well-defined. We focus on one triangular plane. The calculation for the other plane is similar and hence mentioned in brief. We refer to \Cref{subsubsec:gaussian} for required expressions of $\partial_t^{j}\bpartial_\bs^{r-j}\tildeK$, $j=0,\ldots,r$, $r=1,2$. 
	\paragraph{Triangular plane 1:} Here $k=1$. The normal to the plane is given by $\overline{\bn}_{ij1}=\left(\begin{smallmatrix}
		\bv_{ij+1}\\1
	\end{smallmatrix}\right)\times \left(\begin{smallmatrix}
		\bu_{ij}\\0
	\end{smallmatrix}\right) = \left(\begin{smallmatrix}
		-u_{ij2}\\u_{ij1}\\ v_{ij+11}u_{ij2}-v_{ij+12}u_{ij1}
	\end{smallmatrix}\right)=\left(\begin{smallmatrix}\bu_{ij}^{\perp}\\-\bv_{ij+1}^{\T}\bu_{ij}^{\perp}\end{smallmatrix}\right)=\left(\begin{smallmatrix}\overline{\bn}_\bs\\\overline{n}_t\end{smallmatrix}\right)$, where $\bu_{ij}^{\perp}$ is the normal to $\bu_{ij}$, $\overline{\bn}_\bs=\bu_{ij}^{\perp}$ and $\overline{n}_t=-\bv_{ij+1}^{\T}\bu_{ij}^{\perp}$. The unit normal is $\bn_{ij1} = \frac{\overline{\bn}_{ij1}}{||\overline{\bn}_{ij1}||}$, where $||\overline{\bn}_{ij1}||={\bu_{ij}^{\perp}}^{\T}\bu_{ij}^{\perp}+\left(\bv_{ij+1}^{\T}\bu_{ij}^{\perp}\right)^{\T}\left(\bv_{ij+1}^{\T}\bu_{ij}^{\perp}\right)={\bu_{ij}^{\perp}}^{\T}\left(\bI_2+\bv_{ij+1}\bv_{ij+1}^{\T}\right)\bu_{ij}^{\perp}$. The following identities are required for subsequent simplifications
	
	{\allowdisplaybreaks
		\begin{align}
			&\bn_s^{\T}\bDelta_{i'1}(\omega,\upsilon) = \frac{{\bu_{ij}^\perp}^{\T}}{||\bn_{ij1}||}\left(\bDelta_{i'ij}-\omega \bu_{ij}-\upsilon \bv_{ij+1}\right)=\frac{{\bu_{ij}^\perp}^{\T}}{||\bn_{ij1}||}\left(\bDelta_{i'ij}-\upsilon \bv_{ij+1}\right) %=\bn_s^{\T}\left(\bDelta_{i'ij}-\upsilon \bv_{ij+1}\right)
			=\bn_s^{\T}\bDelta(\upsilon),\label{eq:id-1}\\
			&||\bDelta_{i'1}(\omega,\upsilon)||^2 = \bDelta_{i'1}(\omega,\upsilon)^{\T}\bDelta_{i'1}(\omega,\upsilon) = ||\bu_{ij}||^2\left(\omega-\frac{\bu_{ij}^{\T}\bDelta(\upsilon)}{||\bu_{ij}||^2}\right)^2 + \bDelta(\upsilon)^{\T}\left(\bI_2-\frac{\bu_{ij}\bu_{ij}^{\T}}{||\bu_{ij}||^2}\right)\bDelta(\upsilon)\nonumber\\
			&\myquad[2]=||\bu_{ij}||^2\left(\omega-\frac{\bu_{ij}^{\T}\bDelta(\upsilon)}{||\bu_{ij}||^2}\right)^2 + \bDelta(\upsilon)^{\T}\frac{\bu_{ij}^{\perp}{\bu_{ij}^{\perp}}^{\T}}{||\bu_{ij}||^2}\bDelta(\upsilon)=||\bu_{ij}||^2\left(\omega-\frac{\bu_{ij}^{\T}\bDelta(\upsilon)}{||\bu_{ij}||^2}\right)^2 + \left(\frac{{\bu_{ij}^{\perp}}^{\T}\bDelta(\upsilon)}{||\bu_{ij}||}\right)^2,\label{eq:id-2}\\
			&||\bDelta_{i'1}(\omega,\upsilon)||^4 = \left\{\bDelta_{i'1}(\omega,\upsilon)^{\T}\bDelta_{i'1}(\omega,\upsilon)\right\}^2\nonumber\\
			&\myquad[2]= \left\{||\bu_{ij}||^2\left(\omega-\frac{\bu_{ij}^{\T}\bDelta(\upsilon)}{||\bu_{ij}||^2}\right)^2 + \bDelta(\upsilon)^{\T}\left(\bI_2-\frac{\bu_{ij}\bu_{ij}^{\T}}{||\bu_{ij}||^2}\right)\bDelta(\upsilon)\right\}^2\nonumber\\
			&\myquad[2]=||\bu_{ij}||^4\left(\omega-\frac{\bu_{ij}^{\T}\bDelta(\upsilon)}{||\bu_{ij}||^2}\right)^4 + 2\left(\omega-\frac{\bu_{ij}^{\T}\bDelta(\upsilon)}{||\bu_{ij}||^2}\right)^2\bDelta(\upsilon)^{\T}\bu_{ij}^{\perp}{\bu_{ij}^{\perp}}^{\T}\bDelta(\upsilon) + \left(\bDelta(\upsilon)^{\T}\frac{\bu_{ij}^{\perp}{\bu_{ij}^{\perp}}^{\T}}{||\bu_{ij}||^2}\bDelta(\upsilon)\right)^2\nonumber\\
			&\myquad[2]=||\bu_{ij}||^4\left(\omega-\frac{\bu_{ij}^{\T}\bDelta(\upsilon)}{||\bu_{ij}||^2}\right)^4 + 2\left(\omega-\frac{\bu_{ij}^{\T}\bDelta(\upsilon)}{||\bu_{ij}||^2}\right)^2\left({\bu_{ij}^{\perp}}^{\T}\bDelta(\upsilon)\right)^2 + \left(\frac{{\bu_{ij}^{\perp}}^{\T}\bDelta(\upsilon)}{||\bu_{ij}||}\right)^4.\label{eq:id-3}
	\end{align}}
	Using \cref{eq:id-1,eq:id-2}, the entry corresponding to the spatial gradient satisfies 
	
	\begin{align}
		\gamma_1&=\iint\limits_T-2\;\sigma^2\frac{\phi_s^2}{A_t^2}\exp\left(-\frac{\phi_s^2}{A_t}||\bDelta_{i'1}(\omega,\upsilon)||^2\right)\;\bn_s^{\T}\;\bDelta_{i'1}(\omega,\upsilon)\;||\overline{\bn}_{ij1}||\;d\omega\;d\upsilon\nonumber\\
		&=\int\limits_0^1 -2\sigma^2\frac{\phi_s^2}{A_t^2}\exp\left\{-\frac{\phi_s^2}{A_t}\left(\frac{{\bu_{ij}^{\perp}}^{\T}\bDelta(\upsilon)}{||\bu_{ij}||}\right)^2\right\}\;\bn_s^{\T}\bDelta(\upsilon)\;||\overline{\bn}_{ij1}||\nonumber\\
		&\myquad[8]\int_0^{1-\upsilon}\exp\left\{-\frac{\phi_s^2}{A_t}||\bu_{ij}||^2\left(\omega-\frac{\bu_{ij}^{\T}\bDelta(\upsilon)}{||\bu_{ij}||^2}\right)^2\right\}\;d\omega\; d\upsilon\nonumber\\
		&=\int\limits_0^1a(\upsilon)\; F_1(\upsilon)\;\bn_s^{\T}\bDelta(\upsilon)\;d\upsilon,\label{eq:gamma-1}
	\end{align}
	where $A_t = 1+\phi_t^2\;|\delta_{i'1}(\upsilon)|^2$, $a(\upsilon) = a(\upsilon;\btheta,\bn_s)= -2\sigma^2\frac{\sqrt{\pi}\phi_s}{A_t^{3/2}}\;\exp\left\{-\frac{\phi_s^2}{A_t}\left(\frac{{\bu_{ij}^\perp}^{\T}\bDelta(\upsilon)}{||\bu_{ij}||}\right)^2\right\}\frac{||\overline{\bn}_{ij1}||}{||\bu_{ij}||}$ and $F_1(\upsilon) = F_1(\upsilon;\phi_s,\phi_t)=
	\Phi(a_2(\upsilon))+\Phi(-a_1(\upsilon))-1$, where $a_2(\upsilon) = a_2(\upsilon;\phi_s,\phi_t)=\frac{\sqrt{2}\phi_s}{A_t^{1/2}}||\bu_{ij}||\left(1-\upsilon-\frac{\bu_{ij}^{\T}\bDelta(\upsilon)}{||\bu_{ij}||^2}\right)$ and $a_1(\upsilon) = a_1(\upsilon;\phi_s,\phi_t)=-\frac{\sqrt{2}\phi_s}{A_t^{1/2}}\frac{\bu_{ij}^{\T}\bDelta(\upsilon)}{||\bu_{ij}||}$. Similarly, for spatial curvature, we have
	
	\begin{align}
		\gamma_2&=\iint\limits_T -2\;\sigma^2\;\frac{\phi_s^2}{A_t^2}\;\exp\left(-\frac{\phi_s^2}{A_t}||\bDelta_{i'1}(\omega,\upsilon)||^2\right)\;\left(||\bn_s||^2-2\frac{\phi_s^2}{A_t}\bn_s^{\T}\bDelta(\upsilon)\bDelta(\upsilon)^{\T}\bn_s\right)\; ||\overline{\bn}_{ij1}||\;d\omega\;d\upsilon\nonumber\\
		&=\int\limits_0^1 a(\upsilon)\; F_1(\upsilon)\;\left(||\bn_s||^2-2\frac{\phi_s^2}{A_t}\bn_s^{\T}\bDelta(\upsilon)\bDelta(\upsilon)^{\T} \bn_s\right)\;d\upsilon.\label{eq:gamma-2}
	\end{align}
	Identities \cref{eq:gamma-1,eq:gamma-2} serve as spatiotemporal extensions for the results found at the end of Section~3 in \cite{halder2024bayesian}. For the temporal derivative, we can write $\partial_t\tildeK=-2\sigma^2\frac{\phi_t^2}{A_t^2}\exp\left(-\frac{\phi_s^2}{A_t}||\bDelta_{i'1}(\omega,\upsilon)||^2\right)\;\delta_{i'1}(\upsilon)+2\sigma^2\frac{\phi_t^2}{A_t^2}\exp\left(-\frac{\phi_s^2}{A_t}||\bDelta_{i'1}(\omega,\upsilon)||^2\right)\frac{\phi_s^2}{A_t}||\bDelta_{i'1}(\omega,\upsilon)||^2\delta_{i'1}(\upsilon)$. The cross-covariance has the following expression
	
	{\allowdisplaybreaks
		\begin{align*}
			\gamma_3&=\iint\limits_T-2\;\sigma^2\;\frac{\phi_t^2}{A_t^2}\;\exp\left(-\frac{\phi_s^2}{A_t}||\bDelta_{i'1}(\omega,\upsilon)||^2\right)\;n_t\;\delta_{i'1}(\upsilon)\;||\bn_{ij1}||\;d\omega\;d\upsilon\;+\\
			&\myquad[2]\iint\limits_T 2\;\sigma^2\;\frac{\phi_t^2}{A_t^2}\;\exp\left(-\frac{\phi_s^2}{A_t}||\bDelta_{i'1}(\omega,\upsilon)||^2\right)\;\frac{\phi_s^2}{A_t}\;||\bDelta_{i'1}(\omega,\upsilon)||^2\;n_t\;\delta_{i'1}(\upsilon)\;||\bn_{ij1}||\;d\omega\;d\upsilon.
	\end{align*}}
	The first integral when evaluated using similar algebra as the preceding calculation produces $\int_0^1 a(\upsilon)\;\frac{\phi_t^2}{\phi_s^2}\;F_1(\upsilon)\;n_t\;\delta_{i'1}(\upsilon)\;d\upsilon$. We evaluate the second integral,
	
	{\allowdisplaybreaks
		\begin{align*}
			&\iint\limits_T 2\;\sigma^2\;\frac{\phi_t^2}{A_t^2}\;\exp\left(-\frac{\phi_s^2}{A_t}||\bDelta_{i'1}(\omega,\upsilon)||^2\right)\;\frac{\phi_s^2}{A_t}\;||\bDelta_{i'1}(\omega,\upsilon)||^2\;n_t\;\delta_{i'1}(\upsilon)\;||\bn_{ij1}||\;d\omega\;d\upsilon\\
			&=\int\limits_0^1 2\sigma^2\frac{\phi_t^2}{A_t^2}\frac{\phi_s^2}{A_t}\left(\frac{{\bu_{ij}^\perp}^{\T}\bDelta(\upsilon)}{||\bu_{ij}||}\right)^2\exp\left\{-\frac{\phi_s^2}{A_t}\left(\frac{{\bu_{ij}^\perp}^{\T}\bDelta(\upsilon)}{||\bu_{ij}||}\right)^2\right\}\;n_t\;\delta_{i'1}(\upsilon)\;||\overline{\bn}_{ij1}||\\
			&\myquad[6]\int\limits_0^{1-\upsilon}\exp\left\{-\frac{\phi_s^2}{A_t}||\bu_{ij}||^2\left(\omega-\frac{\bu_{ij}^{\T}\bDelta(\upsilon)}{||\bu_{ij}||^2}\right)^2\right\}d\omega\;d\upsilon\;+\\
			&\myquad[2]\int\limits_0^12\sigma^2\frac{\phi_t^2}{A_t^2}\exp\left\{-\frac{\phi_s^2}{A_t}\left(\frac{{\bu_{ij}^\perp}^{\T}\bDelta(\upsilon)}{||\bu_{ij}||}\right)^2\right\}\;n_t\;\delta_{i'1}(\upsilon)\;||\overline{\bn}_{ij1}||\\
			&\myquad[6]\int\limits_0^{1-\upsilon}\frac{\phi_s^2}{A_t}||\bu_{ij}||^2\left(\omega-\frac{\bu_{ij}^{\T}\bDelta(\upsilon)}{||\bu_{ij}||^2}\right)^2\exp\left\{-\frac{\phi_s^2}{A_t}||\bu_{ij}||^2\left(\omega-\frac{\bu_{ij}^{\T}\bDelta(\upsilon)}{||\bu_{ij}||^2}\right)^2\right\}d\omega\;d\upsilon\\
			&=-\int\limits_0^1a(\upsilon)\frac{\phi_t^2}{A_t}\left(\frac{{\bu_{ij}^\perp}^{\T}\bDelta(\upsilon)}{||\bu_{ij}||}\right)^2F_1(\upsilon)\;n_t\;\delta_{i'1}(\upsilon)\;d\upsilon\;-\\
			&\myquad[6]\int\limits_0^1a(\upsilon)\frac{\phi_t^2}{\phi_s^2}\;\frac{1}{4}\left\{G\left(\frac{3}{2},\frac{a_2(\upsilon)^2}{2}\right)-G\left(\frac{3}{2},\frac{a_1(\upsilon)^2}{2}\right)\right\}\;n_t\;\delta_{i'1}(\upsilon)\;d\upsilon\\
			&=-\int\limits_0^1a(\upsilon)\;\frac{\phi_t^2}{\phi_s^2}\left\{\frac{\phi_s^2}{A_t}\left(\frac{{\bu_{ij}^\perp}^{\T}\bDelta(\upsilon)}{||\bu_{ij}||}\right)^2F_1(\upsilon)+\frac{1}{4}F_2(\upsilon)\right\}\;n_t\;\delta_{i'1}(\upsilon)\;d\upsilon,
	\end{align*}}
	where $F_2(\upsilon) = F_2(\upsilon;\phi_s,\phi_t)=G\left(\frac{3}{2},\frac{a_2(\upsilon)^2}{2}\right)-G\left(\frac{3}{2},\frac{a_1(\upsilon)^2}{2}\right)$. Combining the two terms yield
	
	\begin{equation}
		\gamma_3=\int\limits_0^1a(\upsilon)\;\frac{\phi_t^2}{\phi_s^2}\;\left[\left\{1-\frac{\phi_s^2}{A_t}\left(\frac{{\bu_{ij}^\perp}^{\T}\bDelta(\upsilon)}{||\bu_{ij}||}\right)^2\right\}F_1(\upsilon)-\frac{1}{4}F_2(\upsilon)\right]\;n_t\;\delta_{i'1}(\upsilon)\;d\upsilon.\label{eq:gamma-3}
	\end{equation} 
	For the term corresponding to the spatial-temporal derivative the covariance is given by $\partial_t\bpartial_\bs\tildeK=4\sigma^2\frac{\phi_s^2\phi_t^2}{A_t^3}\exp\left(-\frac{\phi_s^2||\bDelta||^2}{A_t}\right)\left(2-\frac{\phi_s^2||\bDelta||^2}{A_t}\right)\bDelta\delta$. Hence,
	
	\begin{equation*}
		\begin{split}
			\gamma_4&=\iint\limits_T8\sigma^2\frac{\phi_s^2\phi_t^2}{A_t^3}\exp\left(-\frac{\phi_s^2}{A_t}||\bDelta_{i'1}(\omega,\upsilon)||^2\right)\bn_s^{\T}\bDelta(\upsilon)\;n_t\delta_{i'1}(\upsilon)\;||\overline{\bn}_{ij1}||\;d\omega\;d\upsilon\;-\\
			&\myquad[4]\iint\limits_T4\sigma^2\frac{\phi_s^4\phi_t^2}{A_t^4}||\bDelta_{i'1}(\omega,\upsilon)||^2\exp\left(-\frac{\phi_s^2}{A_t}||\bDelta_{i'1}(\omega,\upsilon)||^2\right)\bn_s^{\T}\bDelta(\upsilon)\;n_t\delta_{i'1}(\upsilon)\;||\overline{\bn}_{ij1}||\;d\omega\;d\upsilon.
		\end{split}
	\end{equation*}
	The first term reduces to $-\int_0^14\;a(\upsilon)\;\frac{\phi_t^2}{A_t^2}\;F_1(\upsilon)\;\bn_s^{\T}\bDelta(\upsilon)\;n_t\;\delta_{i'1}(\upsilon)\;d\upsilon$. The second term is $\int_0^1 2\;\frac{\phi_t^2}{A_t}\;a(\upsilon)\left\{-\frac{\phi_s^2}{A_t}\left(\frac{{\bu_{ij}^\perp}^{\T}\bDelta(\upsilon)}{||\bu_{ij}||}\right)^2F_1(\upsilon)-\frac{1}{4}F_2(\upsilon)\right\}\;\bn_s^{\T}\bDelta(\upsilon)\;n_t\delta_{i'1}(\upsilon)\;d\upsilon$ following the previous calculation for $\gamma_3$. Combining the terms, we have
	
	\begin{equation}
		\gamma_4=\int\limits_0^1 2\;a(\upsilon)\;\frac{\phi_t^2}{A_t}\left[\left\{\frac{\phi_s^2}{A_t}\left(\frac{{\bu_{ij}^\perp}^{\T}\bDelta(\upsilon)}{||\bu_{ij}||}\right)^2-\frac{2}{A_t}\right\}F_1(\upsilon)\;+\frac{1}{4}F_2(\upsilon)\right]\;\bn_s^{\T}\bDelta(\upsilon)\;n_t\delta_{i'1}(\upsilon)\;d\upsilon.\label{eq:gamma-4}
	\end{equation}
	Next, for $\gamma_5$ associated with the temporal derivative of spatial curvature, we require $\partial_t\bpartial_\bs^2\tildeK(\bDelta,\delta) = 4\sigma^2\frac{\phi_s^2\phi_t^2}{A_t^3}\exp\left(-\frac{\phi_s^2}{A_t}||\bDelta||^2\right)\left\{\left(2-\frac{\phi_s^2||\bDelta||^2}{A_t}\right)\bI_2-2\frac{\phi_s^2}{A_t}\left(3-\frac{\phi_s^2||\bDelta||^2}{A_t}\right)\bDelta\bDelta^{\T}\right\}\delta$. Hence,
	
	\begin{align}
		&\gamma_5= \iint\limits_Tn_tn_s^{\T}\partial_t\bpartial_\bs^2\tildeK(\bDelta,\delta)\bn_s\;||\overline{\bn}_{ij1}||\;d\omega\;d\upsilon\nonumber\\
		&=\int\limits_0^1 2\;a(\upsilon)
		\;\frac{\phi_t^2}{A_t}\left[\left\{\left(\frac{\phi_s^2}{A_t}\left(\frac{{\bu_{ij}^{\perp}}^{\T}\bDelta(\upsilon)}{||\bu_{ij}||}\right)^2-\frac{2}{A_t}\right)F_1(\upsilon)+\frac{1}{4}F_2(\upsilon)\right\}||\bn_s||^2-\right.\nonumber\\
		&\myquad[3]\left. \frac{2\phi_s^2}{A_t} \left\{\left(\frac{\phi_s^2}{A_t}\left(\frac{{\bu_{ij}^{\perp}}^{\T}\bDelta(\upsilon)}{||\bu_{ij}||}\right)^2-\frac{3}{A_t}\right)F_1(\upsilon)+\frac{1}{4}F_2(\upsilon)\right\}\bn_s^{\T}\bDelta(\upsilon)\bDelta(\upsilon)^{\T}\bn_s\right]n_t\delta_{i'1}(\upsilon)\;d\upsilon.\label{eq:gamma-5}
	\end{align}
	Next, for $\gamma_6$ associated with the temporal curvature, $\partial_t^2\tildeK=-2\sigma^2\frac{\phi_t^2}{A_t^2}\exp\left(-\frac{\phi_s^2||\bDelta||^2}{A_t}\right)\left(1-\frac{\phi_s^2||\bDelta||^2}{A_t}\right)\;+\;4\;\sigma^2\frac{\phi_t^4\delta^2}{A_t^3}\exp\left(-\frac{\phi_s^2||\bDelta||^2}{A_t}\right)\left(2-4\frac{\phi_s^2||\bDelta||^2}{A_t}+\frac{\phi_s^4||\bDelta||^4}{A_t^2}\right)$ is required with $\gamma_6 = \iint\limits_Tn_t^2\partial_t^2\tildeK(\bDelta,\delta)\;||\overline{\bn}_{ij1}||\;d\omega\;d\upsilon$. Evaluation of the terms in $\gamma_6$ follows a similar algebra as seen in the previous calculations with the last term being an exception. We use \cref{eq:id-1,eq:id-2,eq:id-3} to simplify the last term,
	
	{\allowdisplaybreaks
		\begin{align*}
			&\iint\limits_T4\sigma^2\frac{\phi_t^4}{A_t^3}\;\exp\left(-\frac{\phi_s^2}{A_t}||\bDelta_{i'1}(\omega,\upsilon)||^2\right)\;\frac{\phi_s^4}{A_t^2}\;||\bDelta_{i'1}(\omega,\upsilon)||^4\;n_t^2\;\delta_{i'1}^2(\upsilon)\;||\bn_{ij1}||\;d\omega\;d\upsilon\\
			&=2\int\limits_0^1 2\sigma^2\frac{\phi_t^4}{A_t^3}\frac{\phi_s^4}{A_t^2}\left(\frac{{\bu_{ij}^\perp}^{\T}\bDelta(\upsilon)}{||\bu_{ij}||}\right)^4\exp\left\{-\frac{\phi_s^2}{A_t}\left(\frac{{\bu_{ij}^\perp}^{\T}\bDelta(\upsilon)}{||\bu_{ij}||}\right)^2\right\}\;n_t^2\;\delta_{i'1}(\upsilon)^2\;||\overline{\bn}_{ij1}||\\
			&\myquad[6]\int\limits_0^{1-\upsilon}\exp\left\{-\frac{\phi_s^2}{A_t}||\bu_{ij}||^2\left(\omega-\frac{\bu_{ij}^{\T}\bDelta(\upsilon)}{||\bu_{ij}||^2}\right)^2\right\}d\omega\;d\upsilon\;+\\
			&\myquad[2]4\int\limits_0^12\sigma^2\frac{\phi_t^4}{A_t^3}\exp\left\{-\frac{\phi_s^2}{A_t}\left(\frac{{\bu_{ij}^\perp}^{\T}\bDelta(\upsilon)}{||\bu_{ij}||}\right)^2\right\}\;n_t^2\;\delta_{i'1}(\upsilon)^2\;||\overline{\bn}_{ij1}||\\
			&\myquad[4]\int\limits_0^{1-\upsilon}\frac{\phi_s^4}{A_t^2}\left(\frac{{\bu_{ij}^\perp}^{\T}\bDelta(\upsilon)}{||\bu_{ij}||}\right)^2||\bu_{ij}||^2\left(\omega-\frac{\bu_{ij}^{\T}\bDelta(\upsilon)}{||\bu_{ij}||^2}\right)^2\exp\left\{-\frac{\phi_s^2}{A_t}||\bu_{ij}||^2\left(\omega-\frac{\bu_{ij}^{\T}\bDelta(\upsilon)}{||\bu_{ij}||^2}\right)^2\right\}d\omega\;d\upsilon\;+\\
			&\myquad[2]2\int\limits_0^12\sigma^2\frac{\phi_t^4}{A_t^3}\exp\left\{-\frac{\phi_s^2}{A_t}\left(\frac{{\bu_{ij}^\perp}^{\T}\bDelta(\upsilon)}{||\bu_{ij}||}\right)^2\right\}\;n_t^2\;\delta_{i'1}(\upsilon)^2\;||\overline{\bn}_{ij1}||\\
			&\myquad[4]\int\limits_0^{1-\upsilon}\frac{\phi_s^4}{A_t^2}||\bu_{ij}||^4\left(\omega-\frac{\bu_{ij}^{\T}\bDelta(\upsilon)}{||\bu_{ij}||^2}\right)^4\exp\left\{-\frac{\phi_s^2}{A_t}||\bu_{ij}||^2\left(\omega-\frac{\bu_{ij}^{\T}\bDelta(\upsilon)}{||\bu_{ij}||^2}\right)^2\right\}d\omega\;d\upsilon\;\\
			&=-2\int\limits_0^1a(\upsilon)\frac{\phi_t^4\phi_s^2}{A_t^3}\left(\frac{{\bu_{ij}^\perp}^{\T}\bDelta(\upsilon)}{||\bu_{ij}||}\right)^4 F_1(\upsilon)\;n_t^2\;\delta_{i'1}(\upsilon)^2\;d\upsilon\;-\\
			&\myquad[4]\int\limits_0^1a(\upsilon)\frac{\phi_t^4}{A_t^2}\;\left(\frac{{\bu_{ij}^\perp}^{\T}\bDelta(\upsilon)}{||\bu_{ij}||}\right)^2\;F_2(\upsilon)\;n_t^2\;\delta_{i'1}(\upsilon)^2\;d\upsilon\;-\\
			&\myquad[6]3\int\limits_0^1a(\upsilon)\frac{\phi_t^4}{\phi_s^2A_t}\;\left\{G\left(\frac{5}{2},\frac{a_2(\upsilon)^2}{2}\right)-G\left(\frac{5}{2},\frac{a_1(\upsilon)^2}{2}\right)\right\}\;n_t^2\;\delta_{i'1}(\upsilon)^2\;d\upsilon\\
			&=-\int\limits_0^1a(\upsilon)\frac{\phi_t^4}{\phi_s^2A_t}\left\{\frac{2\phi_s^4}{A_t^2}\left(\frac{{\bu_{ij}^\perp}^{\T}\bDelta(\upsilon)}{||\bu_{ij}||}\right)^4F_1(\upsilon)+\frac{\phi_s^2}{A_t}\left(\frac{{\bu_{ij}^\perp}^{\T}\bDelta(\upsilon)}{||\bu_{ij}||}\right)^2 F_2(\upsilon)+3F_3(\upsilon)\right\}\;n_t^2\delta_{i'1}(\upsilon)^2\;d\upsilon,
	\end{align*}}
	where $F_3(\upsilon) = F_3(\upsilon;\phi_s,\phi_t)=G\left(\frac{5}{2},\frac{a_2(\upsilon)^2}{2}\right)-G\left(\frac{5}{2},\frac{a_1(\upsilon)^2}{2}\right)$. Collecting terms, the final expression for $\gamma_6$ is obtained as follows:
	
	{\allowdisplaybreaks
		\begin{align}
			\gamma_6&=\int\limits_0^1 a(\upsilon)\frac{\phi_t^2}{\phi_s^2}\left[\left\{\left(1-\frac{\phi_s^2}{A_t}\left(\frac{{\bu_{ij}^{\perp}}^{\T}\bDelta(\upsilon)}{||\bu_{ij}||}\right)^2\right)F_1(\upsilon)-\frac{1}{4}F_2(\upsilon)\right\}n_t^2\;+\right.\nonumber\\
			&\myquad[6]\left.\frac{\phi_t^2}{A_t}\left\{\left(-4+8\frac{\phi_s^2}{A_t}\left(\frac{{\bu_{ij}^{\perp}}^{\T}\bDelta(\upsilon)}{||\bu_{ij}||}\right)^2-2\frac{\phi_s^4}{A_t^2}\left(\frac{{\bu_{ij}^{\perp}}^{\T}\bDelta(\upsilon)}{||\bu_{ij}||}\right)^4\right)F_1(\upsilon)\;+ \right.\right.\nonumber\\
			&\myquad[8]\left.\left.\left(2-\frac{\phi_s^2}{A_t}\left(\frac{{\bu_{ij}^\perp}^{\T}\bDelta(\upsilon)}{||\bu_{ij}||}\right)^2 \right)F_2(\upsilon)-3F_3(\upsilon)\right\}n_t^2\delta_{i'1}(\upsilon)^2
			\right]\;d\upsilon.\label{eq:gamma-6}
		\end{align}
	}
	Proceeding similarly for the temporal curvature in the spatial gradient, we have
	
	\begin{align}
		\gamma_7&=\int\limits_0^1 a(\upsilon)\frac{\phi_t^2}{A_t}\left[\left\{\left(-4+2\frac{\phi_s^2}{A_t}\left(\frac{{\bu_{ij}^{\perp}}^{\T}\bDelta(\upsilon)}{||\bu_{ij}||}\right)^2\right)F_1(\upsilon)+\frac{1}{4}F_2(\upsilon)\right\}n_t^2\;+\right.\nonumber\\
		&\myquad[6]\left.\frac{\phi_t^2}{A_t}\left\{\left(24-24\frac{\phi_s^2}{A_t}\left(\frac{{\bu_{ij}^{\perp}}^{\T}\bDelta(\upsilon)}{||\bu_{ij}||}\right)^2+ 4\frac{\phi_s^4}{A_t^2}\left(\frac{{\bu_{ij}^{\perp}}^{\T}\bDelta(\upsilon)}{||\bu_{ij}||}\right)^4\right)F_1(\upsilon)\;+ \right.\right.\nonumber\\
		&\myquad[8]\left.\left.\left(6-2\frac{\phi_s^2}{A_t}\right)F_2(\upsilon)-6F_3(\upsilon)\right\}n_t^2\delta_{i'1}(\upsilon)^2
		\right]\;\bn_s^{\T}\bDelta(\upsilon)\;d\upsilon.\label{eq:gamma-7}
	\end{align}
	Finally, for the spatial-temporal curvature, we note that
	
	\begin{equation*}
		\begin{split}
			\partial_t^2\bpartial_\bs^2\tildeK &= 4\sigma^2\frac{\phi_s^2\phi_t^2}{A_t^3}\exp\left(-\frac{\phi_s^2||\bDelta||^2}{A_t}\right)\left[\left\{\left(2-\frac{\phi_s^2||\bDelta||^2}{A_t}\right)-2\left(6-6\frac{\phi_s^2||\bDelta||^2}{A_t}+\frac{\phi_s^4||\bDelta||^4}{A_t^2}\right)\frac{\phi_t^2\delta^2}{A_t}\right\}\bI_2\;-\right.\\&\left.\qquad\qquad\qquad\qquad2\frac{\phi_s^2}{A_t}\left\{\left(3-\frac{\phi_s^2||\bDelta||^2}{A_t}\right)-2\left(12-8\frac{\phi_s^2||\bDelta||^2}{A_t}+\frac{\phi_s^4||\bDelta||^4}{A_t^2}\right)\frac{\phi_t^2\delta^2}{A_t}\right\}\bDelta\bDelta^{\T}\right].
		\end{split}
	\end{equation*}
	We compute the terms separately using algebra similar to that required for obtaining $\gamma_6$. The first term evaluates to 
	
	\begin{subequations}
		\begin{align}
			&\int\limits_0^1 a(\upsilon)\frac{\phi_t^2}{A_t}\left[\left\{\left(-4+2\frac{\phi_s^2}{A_t}\left(\frac{{\bu_{ij}^{\perp}}^{\T}\bDelta(\upsilon)}{||\bu_{ij}||}\right)^2\right)F_1(\upsilon)+\frac{1}{4}F_2(\upsilon)\right\}n_t^2\;+\right.\nonumber\\
			&\myquad[6]\left.\frac{\phi_t^2}{A_t}\left\{\left(24-24\frac{\phi_s^2}{A_t}\left(\frac{{\bu_{ij}^{\perp}}^{\T}\bDelta(\upsilon)}{||\bu_{ij}||}\right)^2+ 4\frac{\phi_s^4}{A_t^2}\left(\frac{{\bu_{ij}^{\perp}}^{\T}\bDelta(\upsilon)}{||\bu_{ij}||}\right)^4\right)F_1(\upsilon)\;+ \right.\right.\nonumber\\
			&\myquad[8]\left.\left.\left(6-2\frac{\phi_s^2}{A_t}\right)F_2(\upsilon)-6F_3(\upsilon)\right\}n_t^2\delta_{i'1}(\upsilon)^2
			\right]\;||\bn_s||^2\;d\upsilon,\label{eq:gamma-8-a}
		\end{align}
		while for the second term again following similar algebra produces
		
		\begin{align}
			&2\int\limits_0^1 a(\upsilon)\;\frac{\phi_t^2\phi_s^2}{A_t^2}\left[\left\{\left(6-2\frac{\phi_s^2}{A_t}\left(\frac{{\bu_{ij}^{\perp}}^{\T}\bDelta(\upsilon)}{||\bu_{ij}||}\right)^2\right)F_1(\upsilon)-\frac{1}{4}F_2(\upsilon)\right\}n_t^2\;+\right.\nonumber\\
			&\myquad[4]\left.\frac{\phi_t^2}{A_t}\left\{\left(-48+32\frac{\phi_s^2}{A_t}\left(\frac{{\bu_{ij}^{\perp}}^{\T}\bDelta(\upsilon)}{||\bu_{ij}||}\right)^2 - 4\frac{\phi_s^4}{A_t^2}\left(\frac{{\bu_{ij}^{\perp}}^{\T}\bDelta(\upsilon)}{||\bu_{ij}||}\right)^4\right)F_1(\upsilon)\;+ \right.\right.\nonumber\\
			&\myquad[6]\left.\left.\left(-8+2\frac{\phi_s^4}{A_t^4}\right)F_2(\upsilon;\phi_s,\phi_t)+6F_3(\upsilon)\right\}n_t^2\delta_{i'1}(\upsilon)^2\right]\;\bn_s^{\T}\bDelta(\upsilon)\bDelta(\upsilon)^{\T}\bn_s\;d\upsilon.\label{eq:gamma-8-b}
		\end{align}
	\end{subequations}
	Adding \cref{eq:gamma-8-a,eq:gamma-8-b} produces $\gamma_8$. Referring to \cref{eq:gamma-1,eq:gamma-2,eq:gamma-3,eq:gamma-4,eq:gamma-5,eq:gamma-6,eq:gamma-7,eq:gamma-8-a,eq:gamma-8-b}, we note that the evaluation of $\gamma(\C_{ij1})$ requires a 1-dimensional quadrature over $[0,1]$.
	
	\paragraph{Triangular plane 2:} Here $k=2$. The normal to the plane is given by $\overline{\bn}_{ij2}=\left(\begin{smallmatrix}
		-\bu_{ij+1}\\0
	\end{smallmatrix}\right)\times \left(\begin{smallmatrix}
		-\bv_{i+1j}\\1
	\end{smallmatrix}\right) = \left(\begin{smallmatrix}
		-u_{ij+12}\\u_{ij+11}\\ -v_{i+1j1}u_{ij+12}+v_{i+1j2}u_{ij+11}\end{smallmatrix}\right)=\left(\begin{smallmatrix}\bu_{ij+1}^{\perp}\\\bv_{i+1j}^{\T}\bu_{ij+1}^{\perp}\end{smallmatrix}\right)=\left(\begin{smallmatrix}\overline{\bn}_\bs\\\overline{n}_t\end{smallmatrix}\right)$, where $\bu_{ij+1}^{\perp}$ is the normal to $\bu_{ij+1}$, $\overline{\bn}_\bs=\bu_{ij+1}^{\perp}$ and $\overline{n}_t=\bv_{i+1j}^{\T}\bu_{ij+1}^{\perp}$. The unit normal is $\bn_{ij2} = \frac{\overline{\bn}_{ij2}}{||\overline{\bn}_{ij2}||}$, where $||\overline{\bn}_{ij2}||={\bu_{ij+1}^{\perp}}^{\T}\bu_{ij+1}^{\perp}+\left(\bv_{i+1j}^{\T}\bu_{ij+1}^{\perp}\right)^{\T}\left(\bv_{i+1j}^{\T}\bu_{ij+1}^{\perp}\right)={\bu_{ij+1}^{\perp}}^{\T}\left(\bI_2+\bv_{i+1j}\bv_{i+1j}^{\T}\right)\bu_{ij+1}^{\perp}$. The identities required for simplifications are as follows:
	
	{\allowdisplaybreaks
		\begin{align}
			&\bn_s^{\T}\bDelta_{i'2}(\omega,\upsilon) = \frac{{\bu_{ij+1}^\perp}^{\T}}{||n_{ij2}||}\left(\bDelta_{i'i+1j+1}-\omega \bu_{ij+1}-\upsilon v_{ij}\right)=\frac{{\bu_{ij+1}^\perp}^{\T}}{||n_{ij2}||}\left(\bDelta_{i'i+1j+1}-\upsilon v_{ij}\right)\nonumber\\
			&\myquad[10]=\bn_s^{\T}\left(\bDelta_{i'i+1j+1}-\upsilon v_{ij}\right)=\bn_s^{\T}\bDelta(\upsilon),\label{eq:id-4}\\
			&||\bDelta_{i'2}(\omega,\upsilon)||^2 = \bDelta_{i'2}(\omega,\upsilon)^{\T}\bDelta_{i'2}(\omega,\upsilon) = ||\bu_{ij+1}||^2\left(\omega-\frac{\bu_{ij+1}^{\T}\bDelta(\upsilon)}{||\bu_{ij+1}||^2}\right)^2 + \bDelta(\upsilon)^{\T}\left(\bI_2-\frac{\bu_{ij+1}\bu_{ij+1}^{\T}}{||\bu_{ij+1}||^2}\right)\bDelta(\upsilon)\nonumber\\
			&\myquad[4]=||\bu_{ij+1}||^2\left(\omega-\frac{\bu_{ij+1}^{\T}\bDelta(\upsilon)}{||\bu_{ij+1}||^2}\right)^2 + \bDelta(\upsilon)^{\T}\frac{\bu_{ij+1}^{\perp}{\bu_{ij+1}^{\perp}}^{\T}}{||\bu_{ij+1}||^2}\bDelta(\upsilon)\nonumber\\
			&\myquad[4]=||\bu_{ij+1}||^2\left(\omega-\frac{\bu_{ij+1}^{\T}\bDelta(\upsilon)}{||\bu_{ij+1}||^2}\right)^2 + \left(\frac{{\bu_{ij+1}^{\perp}}^{\T}\bDelta(\upsilon)}{||\bu_{ij+1}||}\right)^2,\label{eq:id-5}\\
			&||\bDelta_{i'2}(\omega,\upsilon)||^4 = \left\{||\bu_{ij+1}||^2\left(\omega-\frac{\bu_{ij+1}^{\T}\bDelta(\upsilon)}{||\bu_{ij+1}||^2}\right)^2 + \bDelta(\upsilon)^{\T}\left(\bI_2-\frac{\bu_{ij+1}\bu_{ij+1}^{\T}}{||\bu_{ij+1}||^2}\right)\bDelta(\upsilon)\right\}^2\nonumber\\
			&~=||\bu_{ij+1}||^4\left(\omega-\frac{\bu_{ij+1}^{\T}\bDelta(\upsilon)}{||\bu_{ij+1}||^2}\right)^4 + 2\left(\omega-\frac{\bu_{ij+1}^{\T}\bDelta(\upsilon)}{||\bu_{ij+1}||^2}\right)^2\bDelta(\upsilon)^{\T}\bu_{ij+1}^{\perp}{\bu_{ij+1}^{\perp}}^{\T}\bDelta(\upsilon)\;+\nonumber\\
			&\myquad[22]\left(\bDelta(\upsilon)^{\T}\frac{\bu_{ij+1}^{\perp}{\bu_{ij+1}^{\perp}}^{\T}}{||\bu_{ij+1}||^2}\bDelta(\upsilon)\right)^2\nonumber\\
			&~=||\bu_{ij+1}||^4\left(\omega-\frac{\bu_{ij+1}^{\T}\bDelta(\upsilon)}{||\bu_{ij+1}||^2}\right)^4 + 2\left(\omega-\frac{\bu_{ij+1}^{\T}\bDelta(\upsilon)}{||\bu_{ij+1}||^2}\right)^2\left({\bu_{ij+1}^{\perp}}^{\T}\bDelta(\upsilon)\right)^2 + \left(\frac{{\bu_{ij+1}^{\perp}}^{\T}\bDelta(\upsilon)}{||\bu_{ij+1}||}\right)^4.\label{eq:id-6}
	\end{align}}
	The remainder of the calculation follows triangular plane 1 using \cref{eq:id-4,eq:id-5,eq:id-6}. Overall, the evaluation of $\gamma(\C_{ijk})$ requires 1-dimensional quadrature. % over $[0,1]$ when using the Gaussian kernel. 
\end{proof}
Finally, we conclude by noting that simplifications for the terms in $\bK_\bGamma(\C,\C)$ follow similar algebraic steps and make use of \Cref{lemma:3}. They can be expressed using cdfs for $\chi^2$-distributions with higher degrees of freedom. The next discussion outlines algorithms used for triangulating surfaces, posterior sampling of gradients and wombling measures.

\newpage

\section{Algorithms}

We outline algorithms for obtaining posterior samples of the spatiotemporal differential processes and the wombling measures. We use a grid $G$ over which the differential processes are sampled. The choice of $G$ can be regularly spaced locations spanning the spatial domain at every time point or space-time locations between time points, or extrapolated space-time locations. These algorithms are implemented in the \texttt{R}-statistical environment {\citep[][]{r_core_team_2021}} and available for download at \if1\blind\href{https://github.com/arh926/sptwombling}{https://github.com/arh926/sptwombling}\fi \if0\blind {\url{(REDACTED)}}\fi.

\subsection{Spatiotemporal Gradients}

\begin{algorithm}[htb]
	\caption{Posterior sampling of spatiotemporal differential processes}\label{algorithm:stdp}
	{\bf Input:} (i) Observed process $(\bs_i,t_i)$, $i = 1, \ldots, N$\; 
	(ii) Posterior samples $\btheta_1,\ldots,\btheta_{n_{mcmc}}$ from fitting model in eq.~(11)\;
	(iii) A kernel, $K(\cdot,\btheta)$ of choice (Mat\'ern with $\nu =3/2,5/2$ or $\infty$)\;
	(iv) A grid $G=\{(\bs_{i_g},t_{i_g}), i_g = 1,\ldots,n_G\}$\;
	{\bf Output:} Posterior samples of spatiotemporal differential processes over $G$\;
	\For{$i = 1, \ldots, N$}
	{
		\For{$i_g = 1, \ldots,n_G$}
		{
			\tcc{Compute $\bDelta$ and $\delta$ for observed locations, $(\bs_i,t_i)$ and grid $G$}
			$(\bDelta[i,i_g], \delta[i,i_g]) = (\bs_{i_g}-\bs_i, t_{i_g}-t_i)$\;
		}
	}
	\For{$i_\btheta =1,\ldots,n_{mcmc}$}
	{
		
		$\bK=\bK(\cdot;\btheta_{i_\btheta})$, $\bK^{-1} = \left(\bK(\cdot;\btheta_{i_\btheta})\right)^{-1}$ \tcp*{$O(N^3)$ step--compute once}
		$\bM = \bK^{-1} + \tau^{-2} \bI_N$,~$\m= \tau^{-2}\bM^{-1}\Y$,~$\Z \sim N(\m, \bM^{-1})$\tcp*{Posterior sample of $Z$}
		$\bV_0 = \bV_{\L^*}(\0_2,0;\btheta_{i_\btheta})$\tcp*{see \Cref{sec:covfns}}
		\For{$i_g = 1,\ldots,n_G$}
		{
			${\calK}_0 = \L^* \bK(\bDelta[,i_g], \delta[,i_g];\btheta_{i_\btheta})$\;
			\BlankLine
			\tcc{Posterior mean and cross-covariance matrix}
			$\bmu_{\L^*} = -{\calK}_0^{\T}\bK^{-1}\Z$,~$\Sigma_{\L^*} = \bV_0 - {\calK}_0^{\T}\bK^{-1}{\calK}_0$\;
			$\L^* Z[i_\btheta,i_g] \sim \N_{18}(\bmu_{\L^*},\Sigma_{\L^*})$
		}
	}
	\Return $\L^* Z$
\end{algorithm}

\subsection{Spatiotemporal Wombling Measures}
We focus on obtaining posterior samples for $\bGamma(\C)$. It relies on \Cref{algorithm:triangle-surf}, which triangulates $\C$. Alternatively, posterior samples of the spatiotemporal differential processes can be used to generate posterior samples of $\bGamma(\C)$. \Cref{algorithm:riemann} utilizes posterior samples of $\L^*Z$ and the Riemann sum approximation from \Cref{thm:2} to generate posterior samples of $\bGamma(\C)$ and the induced GP over parametric triangular planes. Finally, \Cref{algorithm:wm} uses the posterior in eq.~(10). If derivative estimation and wombling are performed simultaneously, \Cref{algorithm:stdp,algorithm:wm} are combined to avoid repeated  inversion of $K(\cdot,\cdot;\btheta)$.

\begin{algorithm}
	\caption{Triangulation of a wombling surface}\label{algorithm:triangle-surf}
	{\bf Input:} Wombling surface $\C$ and partitioned set of planar curves $\widetilde{C}_{t_j}$, $j=1,\ldots,N$ (see Section~4.2) on $\C$ with the length of partition $(n_\omega,n_\upsilon)$\;
	{\bf Output:} Set of points constituting parametric triangles with normals
	
	\For{$i = 1,\ldots,n_\omega-1$}
	{
		\For{$j=1,\ldots,n_\upsilon-1$}
		{   
			$T = \{(\omega,\upsilon):0\leq\omega+\upsilon< 1\}$\tcp*{Triangle in $\mathfrak{R}^2$}
			\tcc{Parametric Plane 1}
			$\C_{ij1}= \left\{\bigl(\begin{smallmatrix}\bs(\omega'_i,t_j)\\ t_j \end{smallmatrix}\bigr)+\omega\;\bigl(\begin{smallmatrix}\bu_{ij}\\0 \end{smallmatrix}\bigr)+\upsilon\;\bigl(\begin{smallmatrix}\bv_{ij+1}\\t_{j+1}-t_{j} \end{smallmatrix}\bigr):(\omega,\upsilon)\in T\right\}$\;
			\tcc{Parametric Plane 2}
			$\C_{ij2} = \left\{\bigl(\begin{smallmatrix}\bs(\omega'_{i+1},t_{j+1})\\ t_{j+1} \end{smallmatrix}\bigr)+\omega\;\bigl(\begin{smallmatrix}\bu_{ij+1}\\ 0 \end{smallmatrix}\bigr)+\upsilon\;\bigl(\begin{smallmatrix}\bv_{i+1j}\\ t_j-t_{j+1} \end{smallmatrix}\bigr):(\omega,\upsilon)\in T\right\}$\;
			$\overline{\bn}_{ij1} = \left(\begin{smallmatrix} \bv_{ij+1}\\t_{j+1}-t_j\end{smallmatrix}\right)\times\left(\begin{smallmatrix} \bu_{ij}\\0\end{smallmatrix}\right)$\tcp*{Normal to Plane 1}
			$\overline{\bn}_{ij2} = \left(\begin{smallmatrix} -\bu_{ij+1}\\0\end{smallmatrix}\right) \times \left(\begin{smallmatrix} -\bv_{i+1j}\\t_{j+1}-t_j\end{smallmatrix}\right)$\tcp*{Normal to Plane 2}
		}
	}
	\Return $\widetilde{\C}=\bigcup\limits_{i=1}^{n_\omega-1}\bigcup\limits_{j=1}^{n_\upsilon-1}\bigcup\limits_{k=1}^{2}\C_{ijk}$ and $\left\{\overline{\bn}_{ijk}: i = 1,\ldots,n_\omega-1, j = 1,\ldots,n_\upsilon-1, k= 1,2\right\}$
\end{algorithm}

\begin{algorithm}
	\caption{Posterior samples of wombling measures using Riemann sums}\label{algorithm:riemann}
	{\bf Input:} Triangulated Surface $\widetilde{\C}$ and Normals (using \Cref{algorithm:triangle-surf})\;
	{\bf Output:} Posterior samples of $\bGamma(\widetilde{\C})$\;
	{\bf Step 1:} Use \Cref{algorithm:stdp} to obtain posterior samples of $\L^*Z$ for $\widetilde{\C}$\;
	\For{$i_{\btheta}=1,\ldots,n_{mcmc}$}
	{
		\tcc{Riemann sum approximation}
		$\bGamma(\widetilde{\C})[i_\btheta] = \sum\limits_{i=1}^{n_\omega-1}\sum\limits_{j=1}^{n_\upsilon-1}\sum\limits_{k=1}^{2}(\upsilon'_{j+1}-\upsilon'_j)\;(\omega'_{i+1}-\omega'_i)\;g\left(\omega_i^*,\upsilon_j^*\right)\;||\overline{\bn}_{ijk}(\omega_i^*,\upsilon_j^*)||$\;
	}
	\Return $\bGamma(\widetilde{\C})$
\end{algorithm}
\noindent In \Cref{algorithm:riemann}, $g$s are linear functions in $\L^*_{n_t,\bn_s}$---spatiotemporal directional derivatives and curvatures and $\omega^*$s and $\upsilon^*$s are mid-points of rectangles (see \Cref{thm:2} p. 27).

\begin{algorithm}
	\caption{Posterior samples of wombling measures using eq.~(10)}\label{algorithm:wm}
	{\bf Input:} (i) Observed process $(\bs_i,t_i)$, $i = 1, \ldots, N$\; 
	(ii) Posterior samples $\btheta_1,\ldots,\btheta_{n_{mcmc}}$ from fitting model in eq.~(11)\;
	(iii) A kernel, $K(\cdot,\btheta)$ of choice (Mat\'ern with $\nu =3/2,5/2$ or $\infty$)\;
	(iv) Triangulated Surface $\widetilde{\C}$ and Normals (using \Cref{algorithm:triangle-surf})\;
	{\bf Output:} Posterior samples of $\bGamma(\widetilde{\C})$\;
	\For{$i_\btheta =1,\ldots,n_{mcmc}$}
	{
		
		$\bK=\bK(\cdot;\btheta_{i_\btheta})$, $\bK^{-1} = \left(\bK(\cdot;\btheta_{i_\btheta})\right)^{-1}$ \tcp*{$O(N^3)$ step--compute once}
		$\bM = \bK^{-1} + \tau^{-2} \bI_N$,~$\m= \tau^{-2}\bM^{-1}\Y$,~$\Z \sim N(\m, \bM^{-1})$\tcp*{Posterior sample of $Z$}
		\BlankLine
		\tcc{computed using 2-D quadrature}
		$\bG(\widetilde{\C}) =  \iint\limits_{\widetilde{\C}}\bN_{\bs t}{\calK}_0(\bDelta(\omega,\upsilon), \delta(\upsilon))||\overline{\bn}||\;d\omega\;d\upsilon$\;
		\tcc{computed using \Cref{lemma:3} and 2-D quadrature}
		$\bK_\bGamma(\widetilde{\C},\widetilde{\C}) = ||\overline{\bn}||^2\idotsint\limits_{\widetilde{\C}\times\widetilde{\C}}\bN_{\bs t}\bV_{\L^*Z}(\bDelta(\omega_1,\omega_2,\upsilon_1,\upsilon_2),\delta(\upsilon_1,\upsilon_2))\;\bN_{\bs t}^{\T}\;d\omega_1\;d\omega_2\;d\upsilon_1\;d\upsilon_2$\;
		\BlankLine
		\tcc{Posterior mean and cross-covariance matrix}
		$\bmu_{\bGamma} = -\bG(\widetilde{\C})^{\T}\bK^{-1}\Z$,~$\Sigma_{\bGamma} = \bK_\bGamma(\widetilde{\C},\widetilde{\C}) - \bG(\widetilde{\C})^{\T}\bK^{-1}\bG(\widetilde{\C})$\;
		$\bGamma(\widetilde{\C})[i_\btheta] \sim \N_8(\bmu_{\bGamma},\Sigma_{\bGamma})$
	}
	\Return $\bGamma(\widetilde{\C})$
\end{algorithm}

% \noindent Next, we list covariance kernels and expression for required derivatives.

\section{Spatiotemporal Covariance Functions}\label{sec:covfns}

\subsection{Separable:}\label{subsec:sep}

The required derivatives for the separable space-time kernels are computed using $\partial_t^{j}\bpartial_\bs^{r-j}K(\bDelta,\delta) = \partial_t^{j}\bpartial_\bs^{r-j}\left\{K_t(\delta)\otimes K_s(\bDelta)\right\} = \partial_t^{j}K_t(\delta) \otimes \bpartial_\bs^{r-j}K_s(\bDelta)$. The derivatives are listed using $x$ which can either be $\delta$ or, $\bDelta$. In case $x=\delta$, $||x||=|\delta|$, then $\phi=\phi_t$ and $\partial=\partial_t$, otherwise, $\partial = \bpartial_\bs$, $||x||=||\bDelta||$, and $\phi=\phi_s$, where $||\cdot||$ stands for the Euclidean norm. For a Mat\'ern with $\nu=3/2$ we have $\partial\tildeK(x)=-3\sigma^2\phi^2e^{-\sqrt{3}\phi||x||}x$ and $\partial^2\tildeK(x) = -3\sigma^2\phi^2e^{-\sqrt{3}\phi||x||}\left(1-\sqrt{3}\phi\frac{xx^{\T}}{||x||}\right)$.

\subsubsection{Mat\'ern Covariance \texorpdfstring{($\nu=5/2$)}{nu2}}\label{sssec:m52}
We have the first derivative, $\partial\tildeK(x)=-\frac{5}{3}\sigma^2\phi^2e^{-\sqrt{5}\phi||x||}\left(1+\sqrt{5}\phi||x||\right)x$ and the second, $\partial^2\tildeK(x)=-\frac{5}{3}\sigma^2\phi^2e^{-\sqrt{5}\phi||x||}\left(1+\sqrt{5}\phi||x||-5\phi^2xx^{\T}\right)$. The third and fourth derivatives are

\begin{equation*}
	\partial^3\tildeK(x) = \begin{cases}
		\frac{25}{3}\sigma^2\phi^4e^{-\sqrt{5}\phi||x||}\left(3-\sqrt{5}\phi\frac{x_i^2}{||x||}\right)x_i & iii\\
		\frac{25}{3}\sigma^2\phi^4e^{-\sqrt{5}\phi||x||}\left(1-\sqrt{5}\phi\frac{x_i^2}{||x||}\right)x_j & iij\\
		\frac{25\sqrt{5}}{3}\sigma^2\phi^5e^{-\sqrt{5}\phi||x||}\frac{x_i x_j x_k}{||x||} & ijk
	\end{cases}
\end{equation*}

\begin{equation*}
	\partial^4\tildeK(x) = \begin{cases}
		\frac{25}{3}\sigma^2\phi^4e^{-\sqrt{5}\phi||x||}\left\{3-6\sqrt{5}\phi\frac{x_i^2}{||x||}+\sqrt{5}\phi\left(\sqrt{5}\phi+\frac{1}{||x||}\right)\frac{x_i^4}{||x||^2}\right\}& iiii\\
		-\frac{25\sqrt{5}}{3}\sigma^2\phi^5e^{-\sqrt{5}\phi||x||}\left\{3-\sqrt{5}\phi\frac{x_i^2}{||x||}-\frac{x_i}{||x||^2}\right\}\frac{x_ix_j}{||x||} & iiij\\
		\frac{25}{3}\sigma^2\phi^4e^{-\sqrt{5}\phi||x||}\left\{\left(1-\sqrt{5}\phi\frac{x_i^2}{||x||}\right)\left(1-\sqrt{5}\phi\frac{x_j^2}{||x||}\right)+\sqrt{5}\phi\frac{x_i^2x_j^2}{||x||^3}\right\} & iijj\\
		-\frac{25\sqrt{5}}{3}\sigma^2\phi^5e^{-\sqrt{5}\phi||x||}\left\{1-\frac{x_i}{||x||}\left(\sqrt{5}\phi+\mfrac{1}{||x||}\right)\right\}\frac{x_jx_k}{||x||} & iijk\\
		\frac{25\sqrt{5}}{3}\sigma^2\phi^5e^{-\sqrt{5}\phi||x||}\left(\sqrt{5}\phi+\mfrac{1}{||x||}\right)\frac{x_ix_jx_kx_l}{||x||^2} & ijkl
	\end{cases}
\end{equation*}

\subsubsection{Gaussian}\label{sssec:gauss}
We have $\partial\tildeK(x)=-2\sigma^2\phi^2 e^{-\phi^2||x||^2}x$ and $\partial^2\tildeK(x)=-2\sigma^2\phi^2 e^{-\phi^2||x||^2}(1-2\phi^2 xx^{\T})$. Next,

\begin{equation*}
	\scalemath{0.95}{
		\partial^3\tildeK(x) = \begin{cases}
			4\sigma^2\phi^4e^{-\phi^2||x||^2}(3-2\phi^2 x_i^2)x_i, & iii\\
			4\sigma^2\phi^4e^{-\phi^2||x||^2}(1-2\phi^2 x_i^2)x_j, & iij\\
			-8\sigma^2\phi^6e^{-\phi^2||x||^2}x_i x_j x_k, & ijk
		\end{cases},
		\partial^4\tildeK(x) = \begin{cases}
			4\sigma^2\phi^4e^{-\phi^2||x||^2}(3-12\phi^2 x_i^2+4\phi^4x_i^4), & iiii\\
			-8\sigma^2\phi^6e^{-\phi^2||x||^2}(3-2\phi^2 x_i^2)x_i x_j, & iiij\\
			4\sigma^2\phi^4e^{-\phi^2||x||^2}(1-2\phi^2 x_i^2)(1-2\phi^2 x_j^2), & iijj\\
			8\sigma^2\phi^6e^{-\phi^2||x||^2}(1-2\phi^2 x_i^2)x_jx_k, & iijk\\
			16\sigma^2\phi^8e^{-\phi^2||x||^2} x_i x_j x_k x_l, & ijkl
		\end{cases}
	}
\end{equation*}

\subsubsection{Inverse}
Particularly, for comparing results with non-separable kernels, presented in the next discussion, we let $\tildeK_t(|\delta|)=\sigma^2(\phi_t^2|\delta|^2+1)^{-1}=\sigma^2A_t^{-1}$. Note that $\frac{\partial}{\partial\delta}A_t=2\phi_t^2\delta$. The required derivatives are, $\partial_t\tildeK(\delta)=-\frac{2\phi_t^2\delta}{A_t^2}$, $\partial_t^2\tildeK(\delta)=-\frac{2\phi_t^2}{A_t^3}(1-3\phi_t^2\delta^2)$, $\partial_t^3\tildeK(\delta)=\frac{24\phi_t^4\delta}{A_t^4}(1-\phi_t^2\delta^2)$ and $\partial_t^4\tildeK(\delta)=\frac{24\phi_t^4}{A_t^5}(1-12\phi_t^2\delta^2+5\phi_t^4\delta^4)$. 

For comparison with the non-separable case, the resulting cross-covariance using Mat\'ern with $\nu=3/2$, $\nu=5/2$ and the Gaussian case for the spatial covariance and the inverse for the temporal covariance is shown below. For a Mat\'ern with $\nu=3/2$, $\L^*Z(\bs_0,t_0)=\left(Z(s_{\cal S},t_{\cal T}),\bpartial_\bs Z(\bs_0,t_0)^{\T},\partial_tZ(\bs_0,t_0),\partial_t\bpartial_\bs Z(\bs_0,t_0)^{\T}\right)^{\T}$, at an arbitrary space-time location, $(\bs_0,t_0)$

\begin{equation}\label{eq:matern-1-cov-sep}
	\bV_{\L^*}=\sigma^2\begin{pmatrix}
		\Sigma_Z & \bpartial_\bs\bK_0^{\T} & \partial_t\bK_0&\partial_t\bpartial_\bs\bK_0^{\T}\\
		-\bpartial_\bs\bK_0 & 3\phi_s^2\bI_d & 0 & 0\\
		-\partial_t\bK_0 & 0 & 2\phi_t^2 & 0\\
		\partial_t\bpartial_\bs\bK_0 & 0 & 0 & 6\phi_s^2\phi_t^2\bI_d
	\end{pmatrix},
\end{equation}
In case of a Mat\'ern with $\nu=5/2$ and Gaussian, the full cross-covariance matrix is obtained as $(||\bDelta||,|\delta|)\to 0$, excepting the terms $\bpartial_\bs K,\partial_tK,\Partials^2K,\partial_t\bpartial_\bs K$, $\partial_t^2K,\partial_t\Partials^2K$, $\partial_t^2\bpartial_\bs K$ and $\partial_t^2\Partials^2K$. The resulting cross-covariance matrices are
\begin{equation*}
	\bV_{\L^*}=\scalemath{0.7}{\sigma^2\left[
		\begin{array}{c;{2pt/2pt}cc;{2pt/2pt}ccc;{2pt/2pt}ccc}
			\Sigma_Z & \bpartial_\bs \bK_0^{\T} & \Partials^2 \bK_0^{\T} & \partial_t \bK_0 & \partial_t\bpartial_\bs \bK_0^{\T} & \partial_t\Partials^2 \bK_0^{\T} & \partial_t^2 \bK_0 & \partial_t^2\bpartial_\bs \bK_0^{\T} & \partial_t^2\Partials^2 \bK_0^{\T}\\\hdashline[2pt/2pt]
			-\bpartial_\bs \bK_0 & \frac{5}{3}\phi_s^2I_d & 0 & 0 & 0 & 0 & 0 & -\frac{10}{3}\phi_s^2\phi_t^2\bI_d & 0\\
			\Partials^2 \bK_0 & 0 & \frac{25}{3}\phi_s^4\bA & 0 & 0 & 0 & \frac{10}{3}\phi_s^2\phi_t^2\tilde{I}_d & 0 & -\frac{50}{3}\phi_s^4\phi_t^2\bA\\\hdashline[2pt/2pt]
			-\partial_t \bK_0 & 0 & 0 & 2\phi_t^2 & 0 & -\frac{10}{3}\phi_s^2\phi_t^2\tilde{I}_d^{\T} & 0 & 0 & 0\\
			\partial_t\bpartial_\bs \bK_0 & 0 & 0 & 0 & \frac{10}{3}\phi_s^2\phi_t^2I_d & 0 & 0 & 0 & 0\\
			-\partial_t\Partials^2 \bK_0 & 0 & 0 & -\frac{10}{3}\phi_s^2\phi_t^2\tilde{I}_d & 0 & \frac{50}{3}\phi_s^4\phi_t^2\bA & 0 & 0 & 0\\\hdashline[2pt/2pt]
			\partial_t^2 \bK_0 & 0 & \frac{10}{3}\phi_s^2\phi_t^2\tilde{I}_d^{\T} & 0 & 0 & 0 & 24\phi_t^4 & 0 & -40\phi_s^2\phi_t^4\tilde{I}_d^{\T}\\
			-\partial_t^2\bpartial_\bs \bK_0 & -\frac{10}{3}\phi_s^2\phi_t^2\bI_d & 0 & 0 & 0 & 0 & 0 & 40\phi_s^2\phi_t^4\bI_d & 0\\
			\partial_t^2\Partials^2 \bK_0 & 0 & -\frac{50}{3}\phi_s^4\phi_t^2\bA & 0 & 0 & 0 & -40\phi_s^2\phi_t^4\tilde{I}_d & 0 & 200\phi_s^4\phi_t^4\bA
		\end{array}
		\right]},
\end{equation*}
and,

\begin{equation*}
	\bV_{\L^*}=\scalemath{0.7}{\sigma^2\left(
		\begin{array}{ccc;{2pt/2pt}ccc;{2pt/2pt}ccc}
			\Sigma_Z & \bpartial_\bs \bK_0^{\T} & \Partials^2 \bK_0^{\T} & \partial_t \bK_0 & \partial_t\bpartial_\bs \bK_0^{\T} & \partial_t\Partials^2 \bK_0^{\T} & \partial_t^2 \bK_0 & \partial_t^2\bpartial_\bs \bK_0^{\T} & \partial_t^2\Partials^2 \bK_0^{\T}\\
			-\bpartial_\bs \bK_0 & 2\phi_s^2\bI_d & 0 & 0 & 0 & 0 & 0 & -4\phi_s^2\phi_t^2\bI_d & 0\\
			\Partials^2 \bK_0 & 0 & 4\phi_s^4\bA & 0 & 0 & 0 & 4\phi_s^2\phi_t^2\tilde{I}_d & 0 & -8\phi_s^4\phi_t^2\bA\\\hdashline[2pt/2pt]
			-\partial_t \bK_0 & 0 & 0 & 2\phi_t^2 & 0 & -4\phi_s^2\phi_t^2\tilde{I}_d^{\T} & 0 & 0 & 0\\
			\partial_t\bpartial_\bs \bK_0 & 0 & 0 & 0 & 4\phi_s^2\phi_t^2\bI_d & 0 & 0 & 0 & 0\\
			-\partial_t\Partials^2 \bK_0 & 0 & 0 & -4\phi_s^2\phi_t^2\tilde{I}_d & 0 & 8\phi_s^4\phi_t^2\bA & 0 & 0 & 0\\\hdashline[2pt/2pt]
			\partial_t^2 \bK_0 & 0 & 4\phi_s^2\phi_t^2\tilde{I}_d^{\T} & 0 & 0 & 0 & 24\phi_t^4 & 0 & -48\phi_s^2\phi_t^4\tilde{I}_d^{\T}\\
			-\partial_t^2\bpartial_\bs \bK_0 & -4\phi_s^2\phi_t^2\bI_d & 0 & 0 & 0 & 0 & 0 & 48\phi_s^2\phi_t^4\bI_d & 0\\
			\partial_t^2\Partials^2 \bK_0 & 0 & -8\phi_s^4\phi_t^2\bA & 0 & 0 & 0 & -48\phi_s^2\phi_t^4\tilde{I}_d & 0 & 96\phi_s^4\phi_t^4\bA
		\end{array}
		\right)},
\end{equation*}
respectively, where $\bA = \bA_{d(d+1)/2}=\left(\begin{smallmatrix}
	3 & 0 & 0 & \cdots & 1 & 0 & \cdots & 1 & 0 & 1\\
	0 & 1 & 0 & \cdots & 0 & 0 & \cdots & 0 & 0 & 0\\
	\vdots & \vdots &\vdots & \ddots & \vdots & \vdots & \ddots & \vdots &\vdots &\vdots \\
	1 & 0 & 0 & \cdots & 3 & 0 & \cdots & 1 & 0 & 1\\
	\vdots & \vdots &\vdots & \ddots & \vdots & \vdots & \ddots & \vdots &\vdots &\vdots \\
	1 & 0 & 0 & \cdots & 1 & 0 & \cdots & 1 & 0 & 3
\end{smallmatrix}\right)$. For $d=2$, $\bA=\bA_3=\left(\begin{smallmatrix}
	3 & 0  & 1\\ 
	0 & 1 & 0\\
	1 & 0 & 3
\end{smallmatrix}\right)$. Observe that the \emph{variance for spatiotemporal derivatives is lower} compared to the cross-covariance matrices for the non-separable kernels (for, e.g., compare \cref{eq:matern-1-cov-sep,eq:matern-1-cov}). 

\subsection{Non-separable}\label{subsec:kernels-2}
We provide examples of non-separable stationary space-time covariance. We focus on the Mat\'ern family of covariance functions \citep[see, for e.g,][]{gneiting2002nonseparable, quick2015bayesian}. Let $K = \tildeK(||\bDelta||,|\delta|) = \tildeK(||\bDelta||,|\delta|;\btheta)$, where $\btheta = \{\phi_s,\phi_t,\sigma^2\}$ denote the covariance kernel. The Mat\'ern family has an additional fractal parameter, $\nu$ determining smoothness of process realizations. We consider three kernels, $\nu=3/2, 5/2$ and $\infty$ (Gaussian or, squared exponential kernel). Let $A_t=\phi_t^2|\delta|^2+1$, then $\displaystyle\frac{\partial A_t}{\partial\delta}=2\phi_t^2\delta$. Substituting $\delta = 0$ in the expressions derived below produces spatial derivatives (see \Cref{subsec:sep}).

\subsubsection{Mat\'ern Covariance \texorpdfstring{($\nu=3/2$)}{nu1}}\label{subsubsec:matern1}
The covariance kernel is 

\begin{equation*}
	\tildeK(||\bDelta||,|\delta|)=\tildeK=\frac{\sigma^2}{A_t}\left(1+\sqrt{3}\frac{\phi_s||\bDelta||}{A_t^{1/2}}\right)\exp\left(-\sqrt{3}\frac{\phi_s||\bDelta||}{A_t^{1/2}}\right).
\end{equation*}
It admits two derivatives in space-time respectively, i.e. $\partial^2_t\bpartial_\bs^2\tildeK<\infty$. Let $ c(\btheta)=c(||\bDelta||,|\delta|;\btheta)=\frac{\sigma^2}{A_t^2}e^{-\sqrt{3}\frac{\phi_s||\bDelta||}{A_t^{1/2}}}$. Partial derivatives with respect to $\delta$ and $\bDelta$ are, $\frac{\partial}{\partial\delta} c(\btheta)=\frac{c(\btheta)}{A_t}\left(\sqrt{3}\frac{\phi_s||\bDelta||}{A_t^{1/2}}-4\right)\phi_t^2\delta$, and $\frac{\partial}{\partial\bDelta} c(\btheta)=-c(\btheta)\;\left(\sqrt{3}\frac{\phi_s}{A_t^{1/2}}\frac{\bDelta}{||\bDelta||}\right)$. We compute the required derivatives. $\bpartial_\bs\tildeK=-3\phi_s^2c(\btheta)\bDelta$, $\partial_t\tildeK =-2\phi_t^2c(\btheta)\;f_{10}\;\delta$, where $f_{10} = 1+\sqrt{3}\frac{\phi_s||\bDelta||}{A_t^{1/2}}-\frac{3}{2}\frac{\phi_s^2||\bDelta||^2}{A_t}$, $\bpartial_\bs^2\tildeK =-3\phi_s^2c(\btheta)\left(\bI_d-\sqrt{3}\frac{\phi_s}{A_t^{1/2}}\frac{\bDelta\bDelta^{\T}}{||\bDelta||}\right)$, $\partial_t\bpartial_\bs\tildeK =6\phi_s^2\phi_t^2\frac{c(\btheta)}{A_t}\;f_{11}\;\bDelta\delta$, where $f_{11}=2-\frac{\sqrt{3}}{2}\frac{\phi_s||\bDelta||}{A_t^{1/2}}$, $\partial^2_{t}\tildeK=-2\phi_t^2c(\btheta)\left\{f_{10}-f_{20}\frac{\phi_t^2\delta^2}{A_t}\right\}$, where $f_{20} % 4+4\sqrt{3}\frac{\phi_s||\bDelta||}{A_t^{1/2}}-12\frac{\phi_s^2||\bDelta||^2}{A_t}+\frac{3\sqrt{3}}{2}\frac{\phi_s^3||\bDelta||^3}{A_t^{3/2}} 
= -\left\{\left(\sqrt{3}\frac{\phi_s||\bDelta||}{A_t^{1/2}}-4\right)f_{10}+\frac{A_t}{\phi_t^2\delta}\frac{\partial}{\partial\delta}f_{10}\right\}$, $\partial_t\bpartial_\bs^2\tildeK=3\phi_s^2\phi_t^2\frac{c(\btheta)}{A_t}\;f_{12}\;\delta$, where $f_{12} = 4\bI_d-5\sqrt{3}\frac{\phi_s}{A_t^{1/2}}\frac{\bDelta\bDelta^{\T}}{||\bDelta||}-\sqrt{3}\frac{\phi_s||\bDelta||}{A_t^{1/2}}\bI_d+3\frac{\phi_s^2}{A_t}\bDelta\bDelta^{\T}$, $\partial_t^2\bpartial_\bs\tildeK=6\phi_s^2\phi_t^2\frac{c(\btheta)}{A_t}\left\{f_{11}-f_{21}\frac{\phi_t^2\delta^2}{A_t}\right\}\bDelta$, where $f_{21} = 12-\frac{11\sqrt{3}}{2}\frac{\phi_s||\bDelta||}{A_t^{1/2}}+\frac{3}{2}\frac{\phi_s^2||\bDelta||^2}{A_t}= -\left\{\left(\sqrt{3}\frac{\phi_s||\bDelta||}{A_t^{1/2}}-6\right)f_{11}+\frac{A_t}{\phi_t^2\delta}\frac{\partial}{\partial\delta}f_{11}\right\}$ and finally $\partial_t^2\bpartial_\bs^2\tildeK= 3\phi_s^2\phi_t^2\frac{c(\btheta)}{A_t}\left\{f_{12}-f_{22}\frac{\phi_t^2\delta^2}{A_t}\right\}$, where $f_{22} = 24\bI_d-35\sqrt{3}\frac{\phi_s}{A_t^{1/2}}\frac{\bDelta\bDelta^{\T}}{||\bDelta||}-11\sqrt{3}\frac{\phi_s||\bDelta||}{A_t^{1/2}}\bI_d+39\frac{\phi_s^2}{A_t}\bDelta\bDelta^{\T}+3\frac{\phi_s^2||\bDelta||^2}{A_t}\bI_d-3\sqrt{3}\frac{\phi_s^3||\bDelta||}{A_t^{3/2}}\bDelta\bDelta^{\T} = -\left\{\left(\sqrt{3}\frac{\phi_s||\bDelta||}{A_t^{1/2}}-6\right)f_{12}+\frac{A_t}{\phi_t^2\delta}\frac{\partial}{\partial\delta}f_{12}\right\}$. These recurrence relations for $f$ above are increasingly useful when computing higher order derivatives as is seen in the following subsections. 
	Statistical inference is sought on the spatiotemporal differential process, $\L^*Z(\bs_0,t_0)=\left(Z(s_{\cal S},t_{\cal T}),\bpartial_\bs Z(\bs_0,t_0)^{\T},\partial_tZ(\bs_0,t_0),\partial_t\bpartial_\bs Z(\bs_0,t_0)^{\T}\right)^{\T}$, at an arbitrary space-time location, $(\bs_0,t_0)$. This requires evaluation of the derivatives pertaining to the variance-covariance sub-matrix as $(||\bDelta||,|\delta|)\to 0$. They combine to produce
	
	\begin{equation}\label{eq:matern-1-cov}
		\bV_{Z,\L^*Z}=\sigma^2\begin{pmatrix}
			K & \bpartial_\bs K^{\T} & \partial_tK&\partial_t\bpartial_\bs K^{\T}\\
			-\bpartial_\bs K & 3\phi_s^2\bI_d & 0 & 0\\
			-\partial_tK & 0 & 2\phi_t^2 & 0\\
			\partial_t\bpartial_\bs K & 0 & 0 & 12\phi_s^2\phi_t^2\bI_d
		\end{pmatrix},
	\end{equation}
	which is the required cross-covariance matrix for obtaining posterior samples of $\L^* Z(\bs_0,t_0)$. The variance-covariance matrix pertaining to the differential processes is a \emph{diagonal matrix}.
	
	\subsubsection{Mat\'ern Covariance \texorpdfstring{($\nu=5/2$)}{nu2}}\label{subsubsec:matern2}
	The covariance kernel is $\tildeK=\frac{\sigma^2}{A_t}\left(1+\sqrt{5}\frac{\phi_s||\bDelta||}{A_t^{1/2}}+\frac{5}{3}\frac{\phi_s^2||\bDelta||^2}{A_t}\right)\exp\left(-\sqrt{5}\frac{\phi_s||\bDelta||}{A_t^{1/2}}\right)$. { The required derivatives are obtained via successive differentiation.} %\citep[see][for detailed calculations]{halder_bayesian_2024}.} 
We require them to compute $\bV_{Z,\L^*Z}$ in \cref{eq:full-cross-cov} at an arbitrary space-time location $(\bs_0,t_0)$ for the spatiotemporal differential process, 

\begin{equation*}
	\begin{split}\L^* Z(\bs,t)&=\left(Z(s_{\cal S},t_{\cal T}),\bpartial_\bs Z(\bs_0,t_0)^{\T},\partial_tZ(\bs_0,t_0),\Partials^2Z(\bs_0,t_0)^{\T},\bpartial_\bs\partial_t Z(\bs_0,t_0)^{\T},\right.\\&\qquad~\left.\partial_t^2Z(\bs_0,t_0), \Partials^2\partial_tZ(\bs_0,t_0)^{\T},\bpartial_\bs\partial_t^2Z(\bs_0,t_0)^{\T}, \Partials^2\partial_t^2Z(\bs_0,t_0)^{\T}\right)^{\T}.
	\end{split}
\end{equation*} 
Considering the cross-covariance matrix, excepting the terms $\bpartial_\bs K,\partial_tK,\Partials^2K,\partial_t\bpartial_\bs K$, $\partial_t^2K,\partial_t\Partials^2K$, $\partial_t^2\bpartial_\bs K$ and $\partial_t^2\Partials^2K$, all other terms are evaluated as $(||\bDelta||,|\delta|)\to 0$ resulting in the following cross-covariance matrix. Note that compared to Mat\'ern($\nu=3/2$) it is not a diagonal matrix.

\begin{equation*}
	\bV_{Z,\L^*Z}=\scalemath{0.7}{\sigma^2\left[
		\begin{array}{c;{2pt/2pt}cc;{2pt/2pt}ccc;{2pt/2pt}ccc}
			\Sigma_Z & \bpartial_\bs \tildeK^{\T} & \Partials^2 \tildeK^{\T} & \partial_t \tildeK & \partial_t\bpartial_\bs \tildeK^{\T} & \partial_t\Partials^2 \tildeK^{\T} & \partial_t^2 \tildeK & \partial_t^2\bpartial_\bs \tildeK^{\T} & \partial_t^2\Partials^2 \tildeK^{\T}\\\hdashline[2pt/2pt]
			-\bpartial_\bs \tildeK & \frac{5}{3}\phi_s^2\bI_d & 0 & 0 & 0 & 0 & 0 & -\frac{20}{3}\phi_s^2\phi_t^2\bI_d & 0\\
			\Partials^2 \tildeK & 0 & \frac{25}{3}\phi_s^4\bA & 0 & 0 & 0 & \frac{20}{3}\phi_s^2\phi_t^2\tilde{I}_d & 0 & -50\phi_s^4\phi_t^2\bA\\\hdashline[2pt/2pt]
			-\partial_t \tildeK & 0 & 0 & 2\phi_t^2 & 0 & -\frac{20}{3}\phi_s^2\phi_t^2\tilde{I}_d^{\T} & 0 & 0 & 0\\
			\partial_t\bpartial_\bs \tildeK & 0 & 0 & 0 & \frac{20}{3}\phi_s^2\phi_t^2\bI_d & 0 & 0 & 0 & 0\\
			-\partial_t\Partials^2 \tildeK & 0 & 0 & -\frac{20}{3}\phi_s^2\phi_t^2\tilde{I}_d & 0 & 50\phi_s^4\phi_t^2\bA & 0 & 0 & 0\\\hdashline[2pt/2pt]
			\partial_t^2 \tildeK & 0 & \frac{20}{3}\phi_s^2\phi_t^2\tilde{I}_d^{\T} & 0 & 0 & 0 & 24\phi_t^4 & 0 & -120\phi_s^2\phi_t^4\tilde{I}_d^{\T}\\
			-\partial_t^2\bpartial_\bs \tildeK & -\frac{20}{3}\phi_s^2\phi_t^2\bI_d & 0 & 0 & 0 & 0 & 0 & 120\phi_s^2\phi_t^4\bI_d & 0\\
			\partial_t^2\Partials^2 \tildeK & 0 & -50\phi_s^4\phi_t^2\bA & 0 & 0 & 0 & -120\phi_s^2\phi_t^4\tilde{I}_d & 0 & 1200\phi_s^4\phi_t^4\bA
		\end{array}
		\right]}.
\end{equation*}
% The algebraic burden while deriving the above derivatives can be reduced by leveraging the following recurrence relations:$\left(\sqrt{5}\frac{\phi_s||\bDelta||}{A_t^{1/2}} - 4 - 2(k-1)\right)f_{k-10} + \frac{A_t}{\phi_t^2\delta}\frac{\partial}{\partial\delta}f_{k-10} = -f_{k0}$, $\left(\sqrt{5}\frac{\phi_s||\bDelta||}{A_t^{1/2}} - 6 - 2(k-1)\right)f_{k-11} + \frac{A_t}{\phi_t^2\delta}\frac{\partial}{\partial\delta}f_{k-11} = -f_{k1}$, $\left(\sqrt{5}\frac{\phi_s||\bDelta||}{A_t^{1/2}} - 6 - 2(k-1)\right)f_{k-12}^{(1)} + \frac{A_t}{\phi_t^2\delta}\frac{\partial}{\partial\delta}f_{k-12}^{(1)} = -f_{k2}^{(1)}$, $\left(\sqrt{5}\frac{\phi_s||\bDelta||}{A_t^{1/2}} - 8 - 2(k-1)\right)f_{k-12}^{(1)} + \frac{A_t}{\phi_t^2\delta}\frac{\partial}{\partial\delta}f_{k-12}^{(1)} = -f_{k2}^{(1)}$,\qquad $\left(\sqrt{5}\frac{\phi_s||\bDelta||}{A_t^{1/2}} - 8 - 2(k-1)\right)f_{k-13}^{(l_1)} + \frac{A_t}{\phi_t^2\delta}\frac{\partial}{\partial\delta}f_{k-13}^{(l_1)} = -f_{k3}^{(l_1)}$, $\left(\sqrt{5}\frac{\phi_s||\bDelta||}{A_t^{1/2}} - 9 - 2(k-1)\right)f_{k-13}^{(3)} + \frac{A_t}{\phi_t^2\delta}\frac{\partial}{\partial\delta}f_{k-13}^{(3)} = -f_{k2}^{(3)}$, $(\sqrt{5}\frac{\phi_s||\bDelta||}{A_t^{1/2}} - 8 - 2(k-1))f_{k-14}^{(l_2)} + \frac{A_t}{\phi_t^2\delta}\frac{\partial}{\partial\delta}f_{k-14}^{(l_2)} = -f_{k4}^{(l_2)}$ and $\left(\sqrt{5}\frac{\phi_s||\bDelta||}{A_t^{1/2}} - 9 - 2(k-1)\right)f_{k-14}^{(l_3)} + \frac{A_t}{\phi_t^2\delta}\frac{\partial}{\partial\delta}f_{k-14}^{(l_3)} = -f_{k4}^{(l_3)}$, where $k=2,3,4$, $l_1 = 1,2$, $l_2 = 1, 3$ and $l_3 = 2,4,5$.

\subsubsection{Gaussian Covariance}\label{subsubsec:gaussian}
In case $\nu\to\infty$, $\displaystyle\tildeK=\frac{\sigma^2}{A_t}\exp\left(-\frac{\phi_s^2||\bDelta||^2}{A_t}\right)$. { The required derivatives are obtained via successive differentiation.} %\citep[see][for detailed calculations]{halder_bayesian_2024}.} 
We require to compute $\bV_{Z, \L^*Z}$ at an arbitrary space-time location, $(\bs_0,t_0)$ for the spatiotemporal differential process, 
\begin{equation*}
\begin{split}\L^* Z(\bs,t)&=\left(Z(s_{\cal S},t_{\cal T}),\bpartial_\bs Z(\bs_0,t_0)^{\T},\partial_tZ(\bs_0,t_0),\Partials^2Z(\bs_0,t_0)^{\T},\partial_t\bpartial_\bs Z(\bs_0,t_0)^{\T},\right.\\&\qquad~\left.\partial_t^2Z(\bs_0,t_0), \partial_t\Partials^2Z(\bs_0,t_0)^{\T},\partial_t^2\bpartial_\bs Z(\bs_0,t_0)^{\T}, \partial_t^2\Partials^2Z(\bs_0,t_0)^{\T}\right)^{\T}.
\end{split}
\end{equation*} 
Referring to the full covariance matrix, excepting the terms $\bpartial_\bs K^{\T},\partial_tK,\Partials^2K,\bpartial_\bs\partial_tK^{\T}, \partial_t^2K$, $\Partials^2\partial_tK^{\T},\bpartial_\bs\partial_t^2K^{\T}$ and $\Partials^2\partial_t^2K^{\T}$, the other terms are evaluated as $(||\bDelta||,|\delta|)\to 0$ resulting in

\begin{equation*}
\bV_{Z,\L^*Z}=\scalemath{0.7}{\sigma^2\left(
\begin{array}{ccc;{2pt/2pt}ccc;{2pt/2pt}ccc}
	\Sigma_Z & \bpartial_\bs \tildeK^{\T} & \Partials^2 \tildeK^{\T} & \partial_t \tildeK & \partial_t\bpartial_\bs \tildeK^{\T} & \partial_t\Partials^2 \tildeK^{\T} & \partial_t^2 \tildeK & \partial_t^2\bpartial_\bs \tildeK^{\T} & \partial_t^2\Partials^2 \tildeK^{\T}\\
	-\bpartial_\bs \tildeK & 2\phi_s^2\bI_d & 0 & 0 & 0 & 0 & 0 & -8\phi_s^2\phi_t^2\bI_d & 0\\
	\Partials^2 \tildeK & 0 & 4\phi_s^4\bA & 0 & 0 & 0 & 8\phi_s^2\phi_t^2\tilde{I}_d & 0 & -24\phi_s^4\phi_t^2\bA\\\hdashline[2pt/2pt]
	-\partial_t \tildeK & 0 & 0 & 2\phi_t^2 & 0 & -8\phi_s^2\phi_t^2\tilde{I}_d^{\T} & 0 & 0 & 0\\
	\partial_t\bpartial_\bs \tildeK & 0 & 0 & 0 & 8\phi_s^2\phi_t^2\bI_d & 0 & 0 & 0 & 0\\
	-\partial_t\Partials^2 \tildeK & 0 & 0 & -8\phi_s^2\phi_t^2\tilde{I}_d & 0 & 24\phi_s^4\phi_t^2\bA & 0 & 0 & 0\\\hdashline[2pt/2pt]
	\partial_t^2 \tildeK & 0 & 8\phi_s^2\phi_t^2\tilde{I}_d^{\T} & 0 & 0 & 0 & 24\phi_t^4 & 0 & -144\phi_s^2\phi_t^4\tilde{I}_d^{\T}\\
	-\partial_t^2\bpartial_\bs \tildeK & -8\phi_s^2\phi_t^2\bI_d & 0 & 0 & 0 & 0 & 0 & 144\phi_s^2\phi_t^4\bI_d & 0\\
	\partial_t^2\Partials^2 \tildeK & 0 & -24\phi_s^4\phi_t^2\bA & 0 & 0 & 0 & -144\phi_s^2\phi_t^4\tilde{I}_d & 0 & 576\phi_s^4\phi_t^4\bA
\end{array}
\right)}.
\end{equation*}

\section{Comparison with clustering, image segmentation, hot-spot detection techniques and level set modeling}

\begin{figure}[t]
\centering
\begin{subfigure}{.5\textwidth}
\centering
\hspace*{-0.6cm}    
\includegraphics[scale = 0.5]{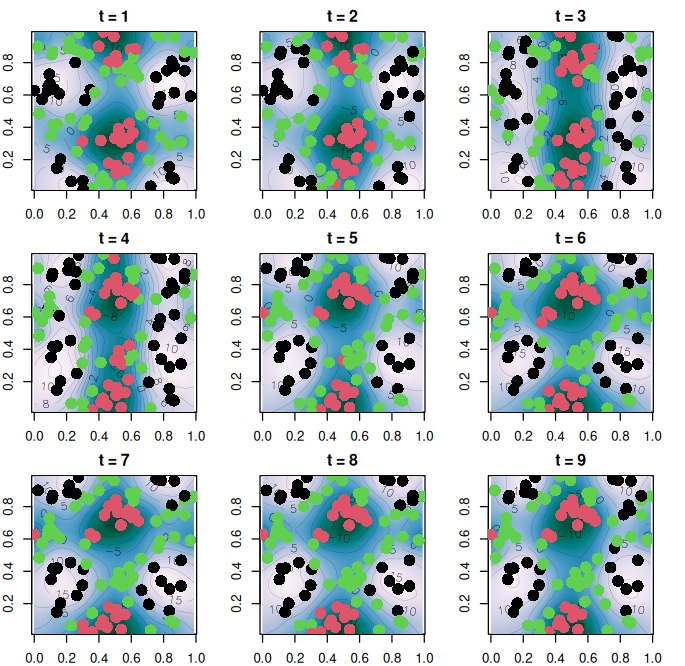}
\caption{$k$-Means with 3 clusters}\label{fig:kmeans}
\end{subfigure}%
\begin{subfigure}{.5\textwidth}
\centering
\includegraphics[scale = 0.5]{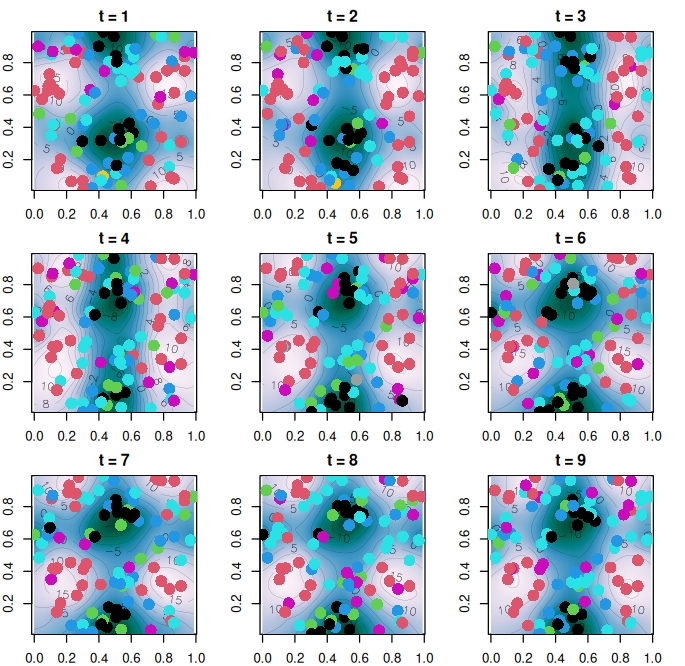}
\caption{Dirichlet Process Mixture Model fit}\label{fig:dpfit}
\end{subfigure}\\
\begin{subfigure}{.5\textwidth}
\centering
\includegraphics[scale = 0.17]{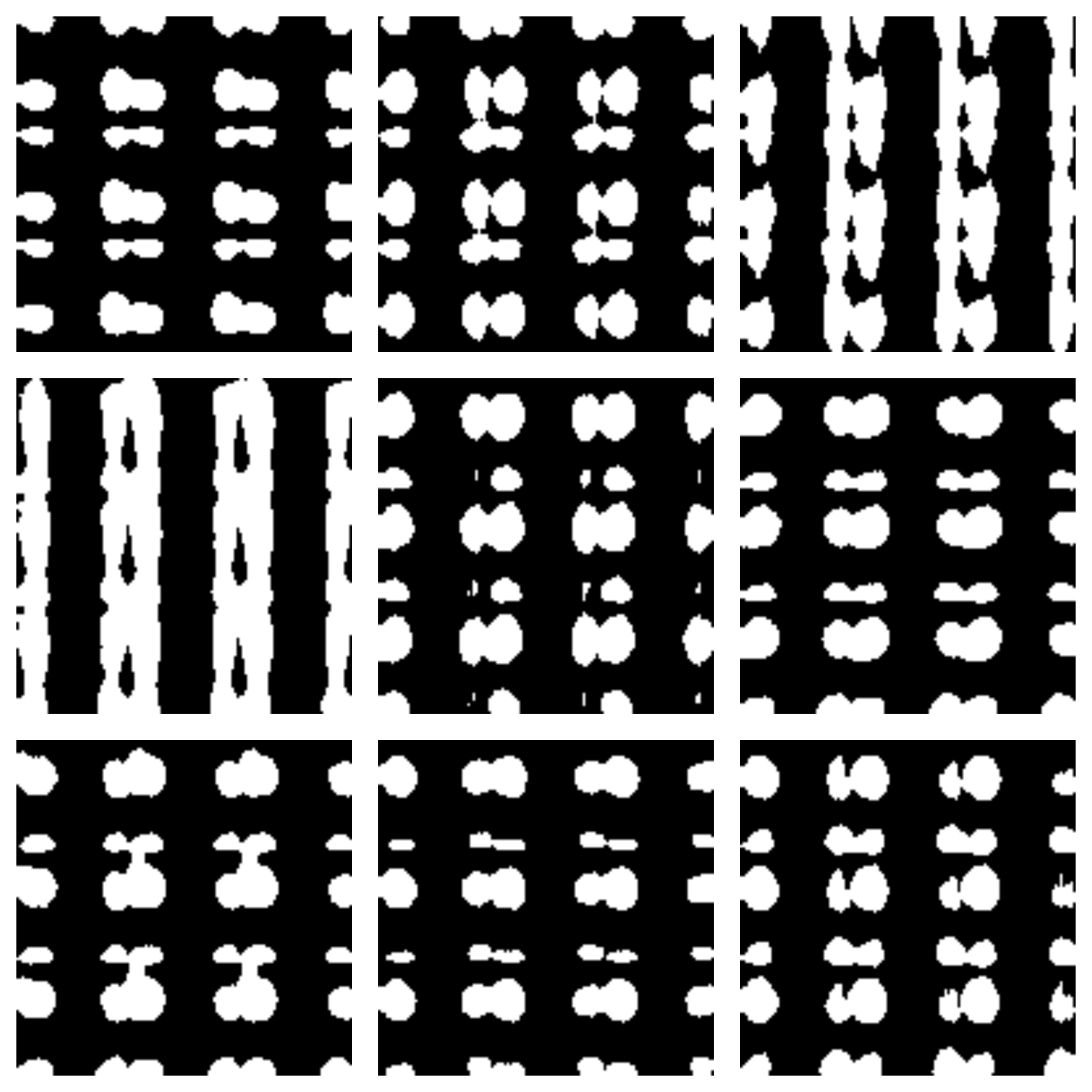}
\caption{Segmentation masks generated by U-Net.}\label{fig:unet}
\end{subfigure}%
\begin{subfigure}{.5\textwidth}
\centering
\hspace*{-0.6cm}    
\includegraphics[scale = 0.4]{plots/wombling-curves.jpg}
\caption{Curves for wombling surfaces}\label{fig:womb-curves}
\end{subfigure}
\caption{Comparison between (a) $k$-means clusters, (b) DPMM clusters (c) segmentation produced by U-Net and (d) the curves chosen for spatiotemporal wombling on simulated data from Pattern 1.}
\end{figure}

{
Wombling can be contrastingly compared to spatiotemporal clustering. While clustering aims to find homogeneous regions and provide predictive probabilities for cluster membership for a new point, wombling seeks predictive process-based inference on rates of change over curvilinear or surface boundaries. Substantive interest focuses on the boundary itself and what distinguishes the regions on either side, rather than on any particular region.

Although several approaches exist for spatiotemporal clustering for example, the spatiotemporal density based scan statistics (ST-DBSCAN) \citep[see, e.g.,][]{birant_stdbscan_2007} or, the spatiotemporal $k$-means \citep[see, e.g.,][]{dorabiala_spatiotemporal_2024}, we were unable to find any dependable open-source software for the \texttt{R}-statistical environment. The situation is similar for spatiotemporal DP approaches \citep[see, e.g.,][]{gelfand2007bayesian, duan2007generalized, kottas_modeling_2008}. In the broader AI/ML scope of research, image segmentation is also related to boundary analysis and can be used as a precursor to provide boundaries that are of possible interest to the user. We select U-Net \citep[see, e.g.,][]{ronneberger_unet_2015}, a supervised learning framework, to compare and demonstrate our ideas. It uses a convolutional neural network (CNN). There also exist other approaches that use residual learning \cite[see, e.g.,][]{he_deep_2016}. 

Furthering the contrastive comparison featuring the approaches discussed, we demonstrate results on the simulated dataset referred to as Pattern 1 in the manuscript using (a) a standard $k$-means algorithm with Euclidean distance for the triplet $(\bs_i,t_i,y(\bs_i,,t_i))^{\T}$, $\bs=(s_x,s_y)$ using three clusters that identify peaks, flats and troughs. We use the \texttt{kmeans} subroutine in base \texttt{R} (b) a DP mixture model (DPMM) using a Gaussian-Inverse-Gamma base measure featured in the \texttt{dirichletprocess} package in \texttt{R} \citep[][]{ross2018dirichletprocess} to cluster $y(\bs_i,t_i)$ and (c) we trained the CNN for 500 epochs using the binary accuracy and dice metric \citep[see, e.g.,][]{zou_statistical_2004} on binary segmentation masks isolating peaks and troughs within interpolated surface plot images. We achieved a validation accuracy of 84\%. The images for each time point in Pattern 1 were then supplied to the trained algorithm for boundary detection. 

The clusters detected using $k$-means are shown in \cref{fig:kmeans} using colored points $(\bs_i,t_i)$ over an interpolated plot for $y(\bs_i,t_i)$. The DPMM detected 6 clusters which are shown in \cref{fig:dpfit}. This is the generic end result of a clustering procedure. The segmented images from U-Net are shown in \cref{fig:unet} where {\em white} regions indicate zones of rapid change. For the sake of comparison we show the curves used for constructing wombling surfaces in \cref{fig:womb-curves}. The boundaries we pick in our simulation overlap with those located by U-Net. Comparing the outputs from clustering, the \emph{region between two clusters} can house possible \emph{wombling boundaries}. Uncertainty quantification along such boundaries is essential. Although quick and easy to implement, clustering is unable to provide any uncertainty quantification on curves (or surfaces) that separate clusters. Image segmentation methods like U-Net require careful training and are unable to provide uncertainty quantification for boundaries. Spatiotemporal wombling provides model-based uncertainty quantification on surface boundaries generated by curves A, B and C. The curves A and B are \emph{contours} within clusters detected by both clustering procedures since they enclose peaks and troughs. However, the curve C is arbitrary. In real world scenarios, as seen in our data examples, when studying precipitation in Northern California, the curve of interest is the central Sierra Nevada Mountain range; in the neuroimaging example, curves delineate regions of the human skull, e.g. the frontal lobe. Lastly, there is no way to characterize clusters, for example, in \cref{fig:kmeans}, the difference between the cluster in red from the cluster in black. Spatiotemporal gradients enable inference on the geometry of the surface indicating that the red points constitute troughs while the black points are peaks.

\begin{figure}[t]
\centering
\begin{subfigure}{.5\textwidth}
	\centering
	\includegraphics[scale = 0.43]{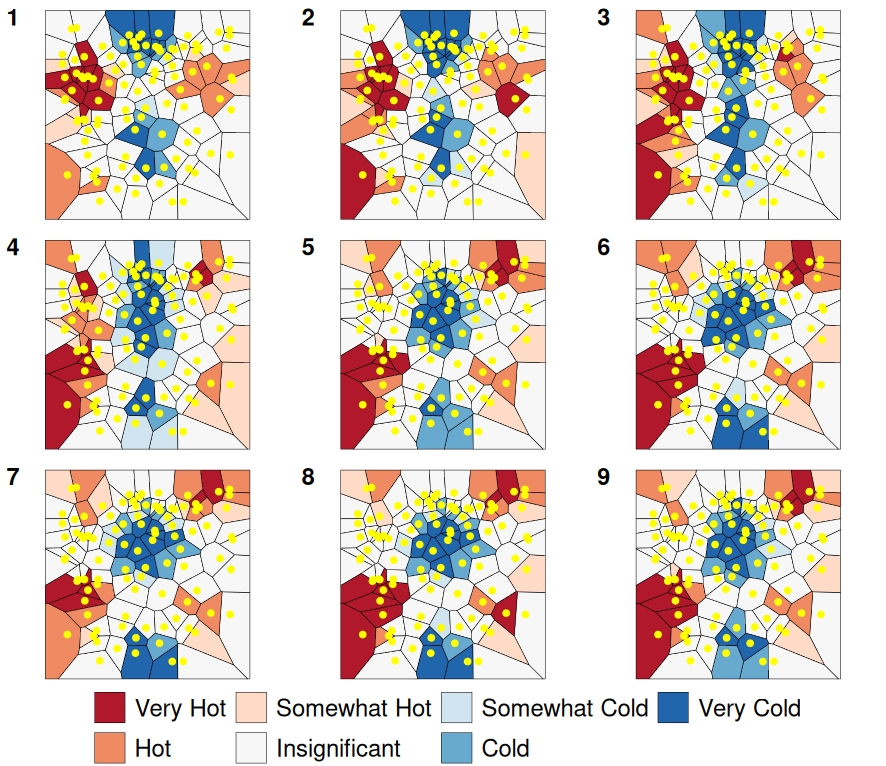}
	\caption{Hot \& Cold spots}\label{fig:getis-ord}
\end{subfigure}%
\begin{subfigure}{.5\textwidth}
	\centering
	\hspace*{-0.6cm}    
	\includegraphics[scale = 0.4]{plots/wombling-curves.jpg}
	\caption{Curves for wombling surfaces}\label{fig:womb-curves-1}
\end{subfigure}
\caption{Comparison between (a) hot \& cold spot detection using local Getis-Ord statistic and (b) the curves chosen for spatiotemporal wombling on simulated data from Pattern 1.}
\end{figure}

We also did a contrastive comparison with local hot/cold spot detection approaches. The results are shown in \cref{fig:getis-ord}. We show the selected curves for spatiotemporal wombling on the side in figure \cref{fig:womb-curves-1} for reference purposes. For the detection of hot \& cold spots we use the local Getis-Ord statistic \citep[see, e.g.,][]{getis_analysis_1992, ord_local_1995, bivand_comparing_2018}. We were able to find literature on spatiotemporal hot-spot detection \citep[see, e.g.,][]{fanaee2015eigenspace, di2018spatiotemporal, butt2020spatio} but no stable implementations for the \texttt{R}-statistical environment. Hot-spot analysis is typically done on areal data. We use Voronoi triangulation (see any plot in \cref{fig:getis-ord}) on the random co-ordinates generated for the simulated data in Pattern 1 to obtain polygons corresponding to an areal representation of the data. We used the \texttt{sfdep} package in \texttt{R} to obtain corresponding $p$-values for the Getis-Ord statistic. Breaks of $(0.05, 0.1]$, $(0.01, 0.05]$ and $(0,0.01]$ were used on the $p$-values to construct the levels ``Somewhat Hot", ``Hot" and ``Very Hot" (similarly for cold). The resulting plots are shown in \cref{fig:getis-ord}. Comparing the hot and cold spots we see that surfaces produced by curves A and B in our spatiotemporal wombling experiment track peaks and troughs respectively and fall in the hot/cold-spot regions.

The convex hull obtained from the hot/cold-spot polygons can be used to construct a crude contour approximation for use in developing wombling surfaces. Wombling estimates spatiotemporal gradients along the surface generated using edges of these polygons to determine whether it forms a spatiotemporal boundary. This is not possible using a hotspot detection approach. The surface generated using curve C is arbitrary and does not correspond to any particular spot but rather lies in the ``Insignificant" zone separating regions of rapid change. % The reviewer is absolutely correct in making this observation.

\begin{figure}[H]
\centering
\begin{subfigure}{.5\textwidth}
	\centering
	\includegraphics[scale = 0.15]{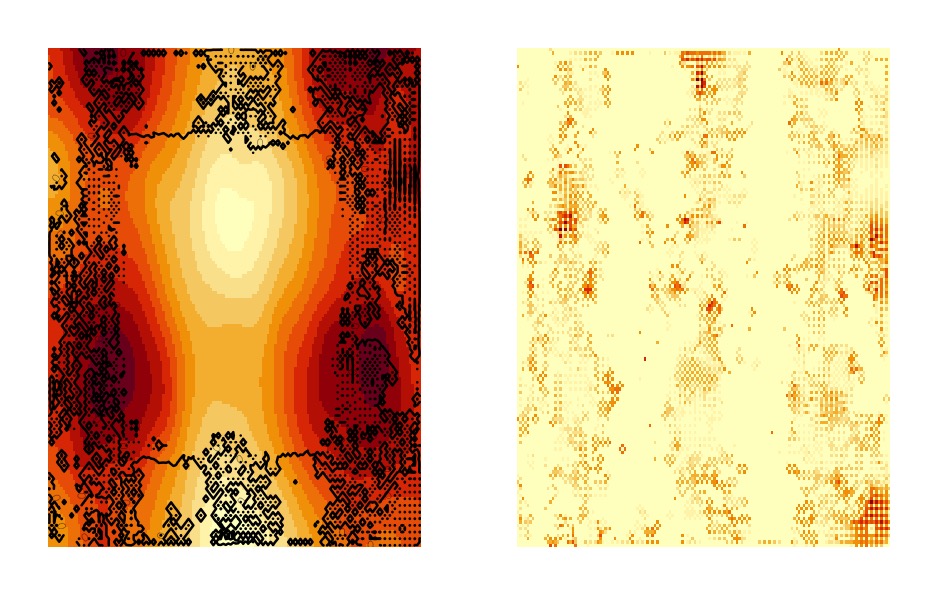}
\end{subfigure}%
\begin{subfigure}{.5\textwidth}
	\centering
	\includegraphics[scale = 0.15]{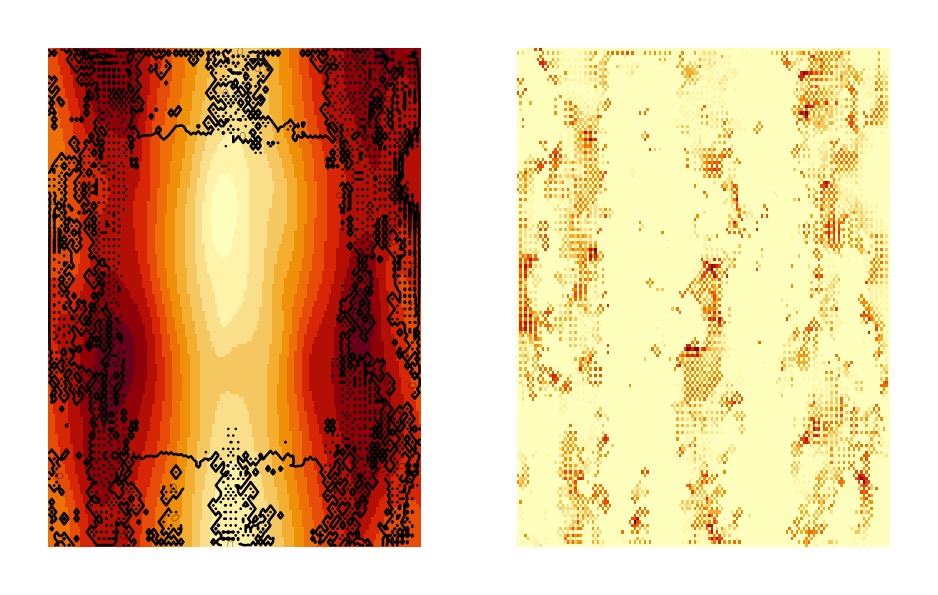}
\end{subfigure}\\
\begin{subfigure}{.5\textwidth}
	\centering
	\includegraphics[scale = 0.15]{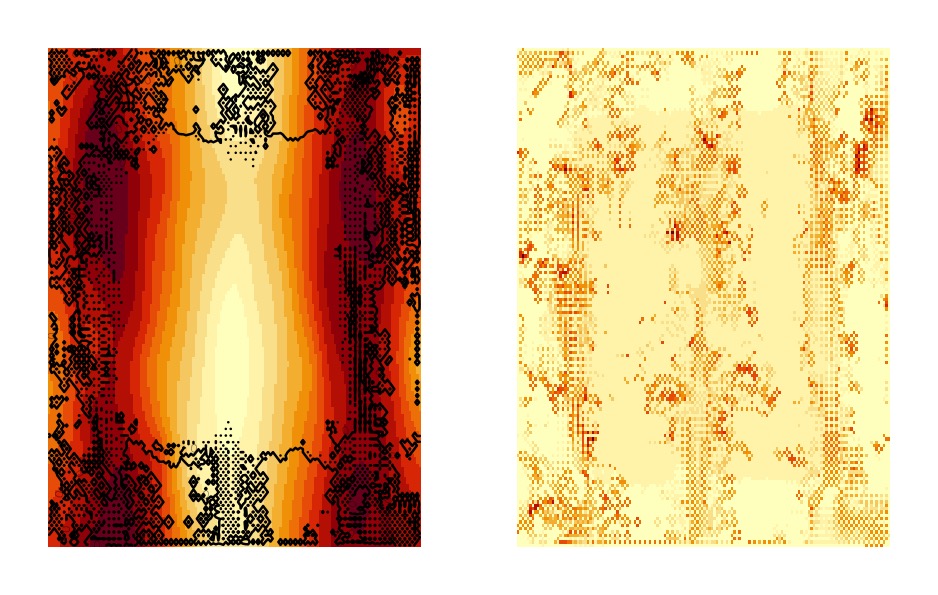}
\end{subfigure}%
\begin{subfigure}{.5\textwidth}
	\centering
	\includegraphics[scale = 0.15]{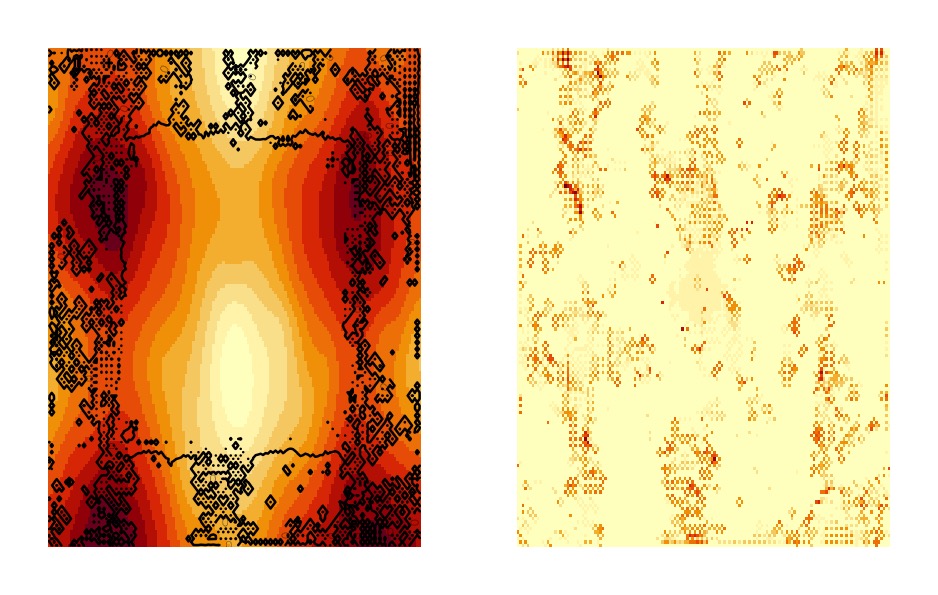}
\end{subfigure}\\
\begin{subfigure}{.5\textwidth}
	\centering
	\includegraphics[scale = 0.15]{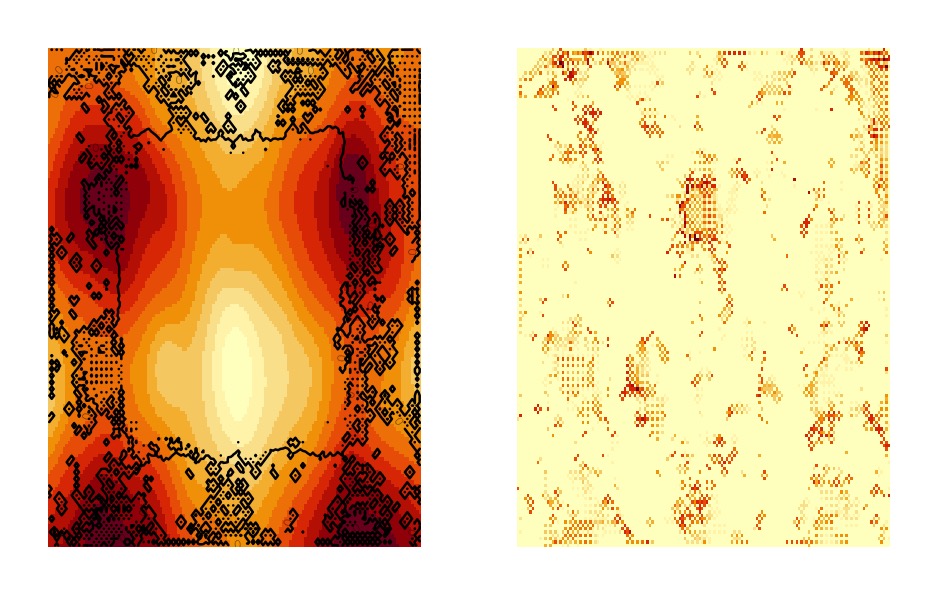}
\end{subfigure}%
\begin{subfigure}{.5\textwidth}
	\centering
	\includegraphics[scale = 0.15]{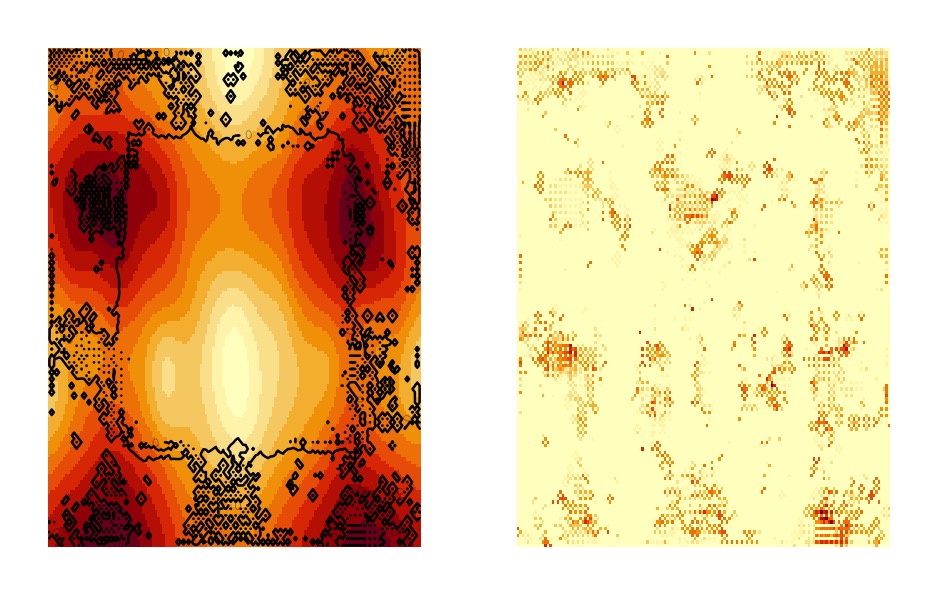}
\end{subfigure}\\
\begin{subfigure}{.5\textwidth}
	\centering
	\includegraphics[scale = 0.15]{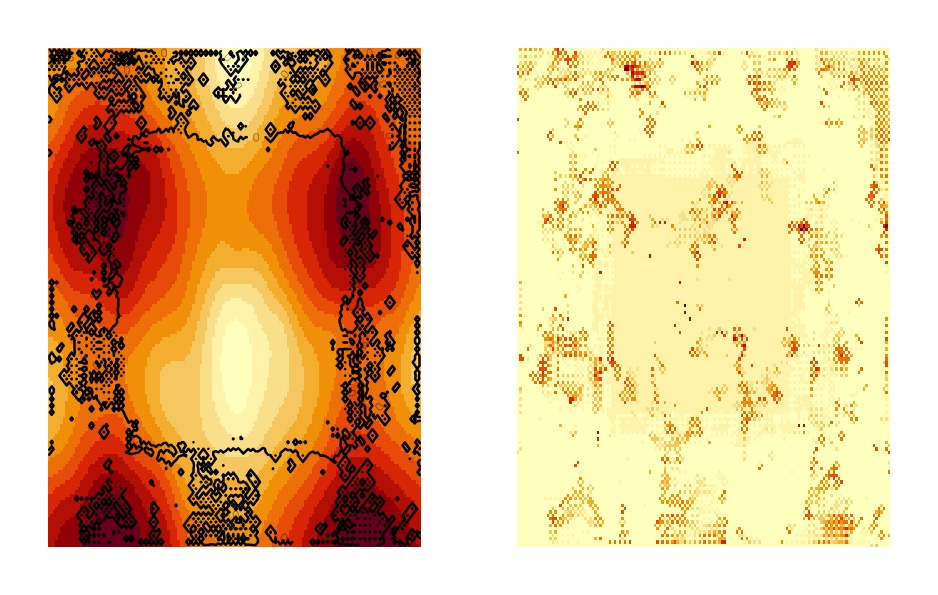}
\end{subfigure}% 
\begin{subfigure}{.5\textwidth}
	\centering
	\includegraphics[scale = 0.15]{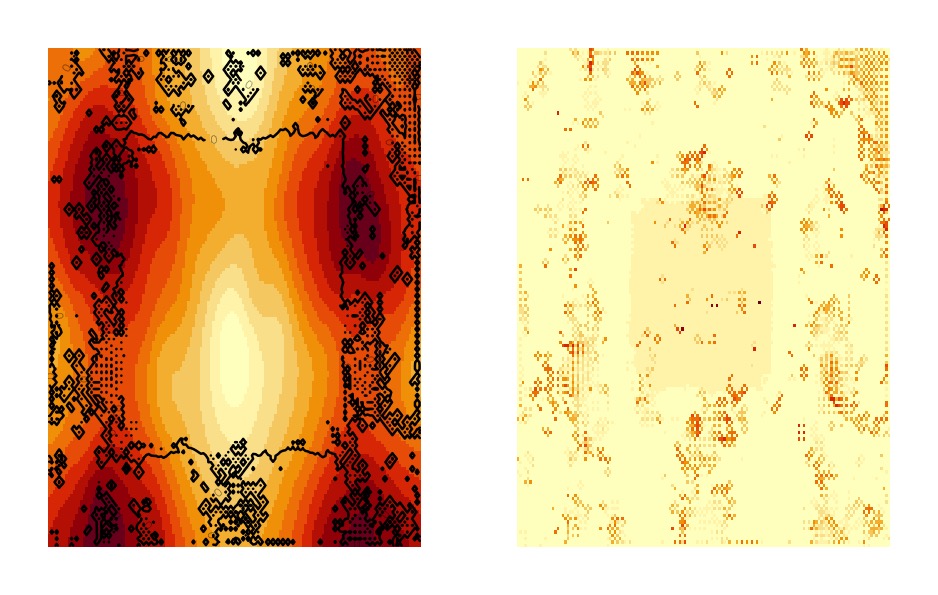}
\end{subfigure}\\
\begin{subfigure}{.5\textwidth}
	\centering
	\includegraphics[scale = 0.15]{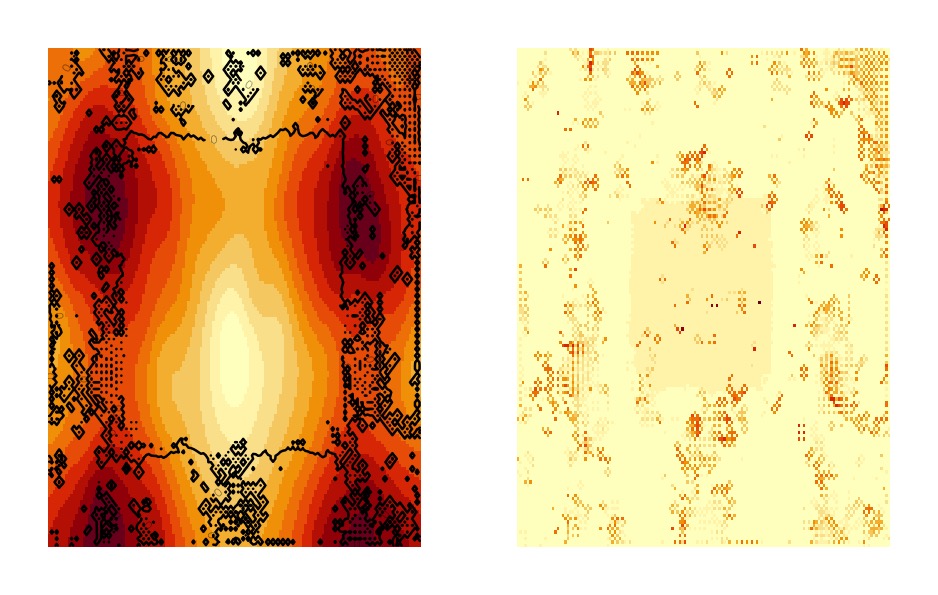}
\end{subfigure}%
\begin{subfigure}{.5\textwidth}
	\centering
	\hspace*{-0.6cm}    
	\includegraphics[scale = 0.25]{plots/wombling-curves.jpg}
\end{subfigure}
\caption{Level set modeling for Pattern 1. Each row (with the exception of the last) consists of the output generated for two time points. For each time point the first picture shows the level set overlay and the second picture shows the values of the level set function, $\phi$. The last plot shows curves chosen for wombling.}\label{fig:level}
\end{figure}
{ 
\noindent {\bf Level Set Modeling:} Level set modeling has a rich history \citep[see, e.g.,][]{osher_fronts_1988,malladi_shape_1995,osher_level_2001,lie2006binary,biswas_state---art_2022,saad-falcon_level_2024}. We provide a brief outline before making remarks. The level set method requires a Lipschitz continuous function often chosen to be the signed distance function (also called level set function) $\phi$ defined for a boundary $\C$ for $\Omega\subseteq \mathfrak{R}^d$ as $\phi(\bs) = -||\bs -\C||_2\;I(\bs\in \Omega^{-})+||\bs -\C||_2\;I(\bs\in \Omega^{+})+0\;I(\bs\in \C)$, where $\Omega^{-}$ and $\Omega^{+}$ indicate regions inside and outside the boundary respectively and $||\cdot||_2$ is the Euclidean norm. The normal to $\C$ and the mean curvature are $\bn(\bs) = \nabla\phi(\bs)/|\nabla\phi(\bs)|$ and $\kappa(\bs) = \nabla\cdot\bn(\bs)$ respectively. The curve $\C$ evolves under a velocity field, $\bv=(u,v,w)^\T$ according to the equation $\frac{\partial\phi}{\partial t}+\bv\cdot \bn\; |\nabla \phi|=0$. Solutions to this equation \citep[for details see, e.g.,][]{sussman1994level, sussman1998improved} provides us with the required level sets.

The level set method is used primarily for shape modeling, computer experiments and image classification. We would like to make some remarks and some possible directions for future work in the following point-wise discussion. 
\begin{enumerate}
	\item The end result of level set modeling comprises the detected level sets and the value of the level set function. There is no uncertainty quantification performed on the level sets detected. Alternatively, Bayesian wombling assumes that the investigator has specific curves that are of scientific interest and performs model-based uncertainty quantification on them. There is precedent for automatic construction and detection of boundaries using significant gradients in the response \citep[see, e.g., Sec. 5,][]{banerjee2005boundary}. Future work can also explore a joint criteria involving both gradient and curvature in the response. Albeit a time-consuming approach, another alternative is to generate the boundaries using the level set method followed by uncertainty quantification on boundary segments using Bayesian wombling. 
	\item Computationally, the level set method is implemented using finite differences. Finite differences produce unstable inference for smaller step-sizes \cite[for more discussion, see,][]{banerjee2003directional}. A model based approach like Bayesian wombling offers computationally stable and exact closed form posterior inference on rates of change.
	\item Specifically for a spatiotemporal response, the level set method requires a space-time level set function $\phi(\bs,t)$. As expected, the ensuing computation is sufficiently complex \citep[see, e.g.,][]{boledi_level-set_2022}. % We feel that implementing it is beyond the scope of this revision and kindly request that the reviewers and the associate editor to understand our predicament. 
	For the sake of comparison we perform a brief simulation that employs the spatial version for every time point independently. We used the simulated data in Pattern 1 from Sec. 5 of the manuscript. The results are shown in \Cref{fig:level}. For comparison, we also show the curves used for our simulation. Evidently, there is a high overlap between the level sets detected and our choice of curves.
\end{enumerate}

We would like to make some comments that might possibly help with clarification of the proposed developments. Our manuscript (and \cite{banerjee2006bayesian,halder2024bayesian}) deals with statistical methods for discerning whether segments of a \emph{predetermined} curve (can be a level set) forms a boundary and \emph{not} methods for explicit construction of the curve which are the focus of methods like U-Net and level set modeling. 

}
}

{ 
\section{Static Wombling}\label{sec:stat-womb}
In certain applications, spatiotemporal wombling can be regarded as a fairly straightforward extension of purely spatial wombling. For example, investigators may be interested in evaluating whether a geographic feature (e.g., a mountain range or a river) delineates {  a region} of rapid rates of change. Here, one seeks to estimate the wombling measures over a curve that stays fixed over time. This special case is derived by setting $n_t=0$. Hence, the directional spatiotemporal differential process $\L^*_{0,\bn_\bs}Z(\bs,t) = \widetilde{\bn}_{\bs,-1} \L_{\bs}Z(\bs,t)$ is $2\times 1$ consisting of the directional spatial derivative and curvature of $Z(\bs,t)$. Consequently, the wombling measures only capture the total and average spatial derivative and curvature. Flux in the spatiotemporal derivative and curvature processes are muted. Using $\L^*_{1,\bn_\bs} Z(\bs,t)$ is an alternative to produce spatiotemporal derivative and curvature wombling measures. However, this alters the normal, $\overline{\bn}(\bs,t)$, thereby removing the geometric interpretation associated with the wombling measures. We provide some examples.

The parametric line segment $C=\{\bs_0+\omega \bu:||\bu||=1,\omega\in[0,1]\}$ ranges over $\{t_0+\upsilon :\upsilon\in[0,1]\}$ to form $\C$ containing the co-planar points, $\left\{\bigl(\begin{smallmatrix}
\bs_0\\ t_0
\end{smallmatrix}\bigr), \bigl(\begin{smallmatrix}
\bs_0 + \bu \\ t_0
\end{smallmatrix}\bigr), \bigl(\begin{smallmatrix}
\bs_0\\ t_0+1
\end{smallmatrix}\bigr), \bigl(\begin{smallmatrix}
\bs_0 + \bu\\ t_0 + 1
\end{smallmatrix}\bigr)\right\}$. The normal to $\C$ is $\overline{\bn} = \left(\begin{smallmatrix}
\overline{\bn}_\bs\\0
\end{smallmatrix}\right)$, where $\overline{\bn}_\bs=(u_y,-u_x)^{\T}=\bu^\perp$. It is independent of $(\omega,\upsilon)$ and $||\overline{\bn}||=1$. Hence, $\C=\left\{(x,y,z):\overline{\bn}^{\T}\left(\begin{smallmatrix}
x\\y\\z
\end{smallmatrix}\right)=\overline{\bn}^{\T}\left(\begin{smallmatrix}
s_{0x}\\s_{0y}\\t_0
\end{smallmatrix}\right)\right\}$. The vector of directional differential processes is $\L^*_{0,\bu^\perp} Z(\bs,t)$. Then, we have $\bGamma(\C)=\int\limits_{t_0}^{t_0+1}\int\limits_0^1\L^*_{0,\bu^\perp} Z(\bs(\omega),\upsilon)\;d\omega\; d\upsilon=\left(\int\limits_{t_0}^{t_0+1}\int\limits_0^1{\bu^{\perp}}^{\T}\bpartial_\bs Z(\bs(\omega),\upsilon)\;d\omega\; d\upsilon,\int\limits_{t_0}^{t_0+1}\int\limits_0^1\{\widetilde{{\bu^\perp}}^{\otimes 2}\}^{\T}\Partials^2Z(\bs(\omega),\upsilon)\;d\omega\; d\upsilon\right)^{\T}$. Next, consider a half-quadrant of radius $r$, $C = \{(r\cos \omega, r\sin \omega):\omega\in[0,\pi/2]\}$ ranging over $\{t_0+\upsilon :\upsilon\in[0,1]\}$, generating the quarter-circular sheet $\C=\{(r\cos \omega, r\sin \omega,t_0+\upsilon):\omega\in[0,\pi/2],\upsilon\in[0,1]\}$. The normal to $\C$ is $\overline{\bn}(\bs(\omega),\upsilon) = \left(\begin{smallmatrix}\overline{\bn}_\bs(\omega)\\0\end{smallmatrix}\right)$, where $\overline{\bn}_\bs(\omega)=(r\cos\omega,-r\sin\omega)^{\T}$ with, $||\overline{\bn}(\bs(\omega),\upsilon)||=r$. The area, $\A(\C)=\int\limits_{t_0}^{t_0+1}\int\limits_0^{\pi/2}r\;d\omega\;d\upsilon=\frac{r\pi}{2}$ and  $\overline{\bGamma}(\C)=\frac{2}{\pi}\int\limits_{t_0}^{t_0+1}\int\limits_0^{\pi/2}\L^*_{0,\bn_\bs} Z(\bs(\omega),\upsilon)\;d\omega\; d\upsilon=\left(\frac{2}{\pi}\int\limits_{t_0}^{t_0+1}\int\limits_0^{\pi/2}{\bn_\bs(\omega)}^{\T}\bpartial_\bs Z(\bs(\omega),\upsilon)\;d\omega\; d\upsilon,\frac{2}{\pi}\int\limits_{t_0}^{t_0+1}\int\limits_0^{\pi/2}\{\widetilde{{\bn_\bs}}(\omega)^{\otimes 2}\}^{\T}\Partials^2Z(\bs(\omega),\upsilon)\;d\omega\; d\upsilon\right)^{\T}$. In the foregoing examples, using $\L^*_{1,\bn_\bs} Z(\bs,t)$ produces the complete vector of eight wombling measures. Lastly, although the curve evolves over time, $n_t=n_t(\omega,\upsilon)$ can still be zero---if the vectors $\left(\begin{smallmatrix}
\frac{\partial s_x(\omega,\upsilon)}{\partial \upsilon}\\ \frac{\partial s_y(\omega,\upsilon)}{\partial \upsilon}
\end{smallmatrix}\right)$ and $\left(\begin{smallmatrix}\frac{\partial s_x(\omega,\upsilon)}{\partial \omega}\\ \frac{\partial s_y(\omega,\upsilon)}{\partial \omega}\end{smallmatrix}\right)$ are linearly dependent. We exclude such pathological cases.
}
\section{Additional Applications}
% The inferential frameworks developed for spatiotemporal differential process assessment and surface wombling are applied to three unique scenarios. We briefly describe and motivate our choice of scenarios. The first scenario examines regions of the Sierra Nevada Conservancy for possible wombling boundaries within a spatiotemporal model for precipitation. The second application focuses on detecting wombling boundaries for ambient particle concentration (PM\textsubscript{2.5}) surfaces in the northern Mid-Atlantic and New England region of the United States (US) during the Canada wildfires, 2023. The third scenario veers away from geophysical processes and concerns surfaces arising from event-related potentials (ERP) recorded using electroencephalography (EEG) to investigate predisposition to alcoholism. We explore differences in wombling boundaries between alcoholic subjects and control. The first scenario features boundaries that stay fixed over time requiring methods in \Cref{sec:stat-womb}. For the remaining scenarios we use methods in \Cref{sec:dyn-womb}. Further details regarding each scenario are housed in the following paragraphs analyzing them separately. The response in each scenario is modeled using  the spatiotemporal Bayesian hierarchical model in \cref{eq:hier-spt-model}. We use a Mat\'ern kernel with fractal parameter, $\nu=5/2$ that ensures the existence and validity of all spatiotemporal processes within $\L^*Z(\bs,t)$. Weakly informative priors with hyper-parameter settings outlined in \Cref{sec:bi&c} are used to setup the Bayesian inference.

\subsection{Precipitation in Northern California}
The state of California receives 90\% of its rainfall during the months of October to April. The Sierra Nevada mountain range, part of the American Cordillera, is a prominent topographic feature of the state. The presence of a mountain range affects precipitation in the surrounding areas, causing clouds to form over the windward side and rain-shadow regions on the leeward side, producing the \emph{orographic effect} \citep[see, e.g.,][]{marra2021orographic}. %houze2012orographic, 
Topographic complexity of the feature effects geophysical manifestation of the orographic effect. \Cref{fig:prcp-sep-dec} captures this in the northern California region during the year 2022. We explore further statistical quantification of the orographic effect using spatiotemporal rates of change and wombling.

\begin{figure}[t]
\centering
\hspace*{-.2in}
\includegraphics[scale = 0.55]{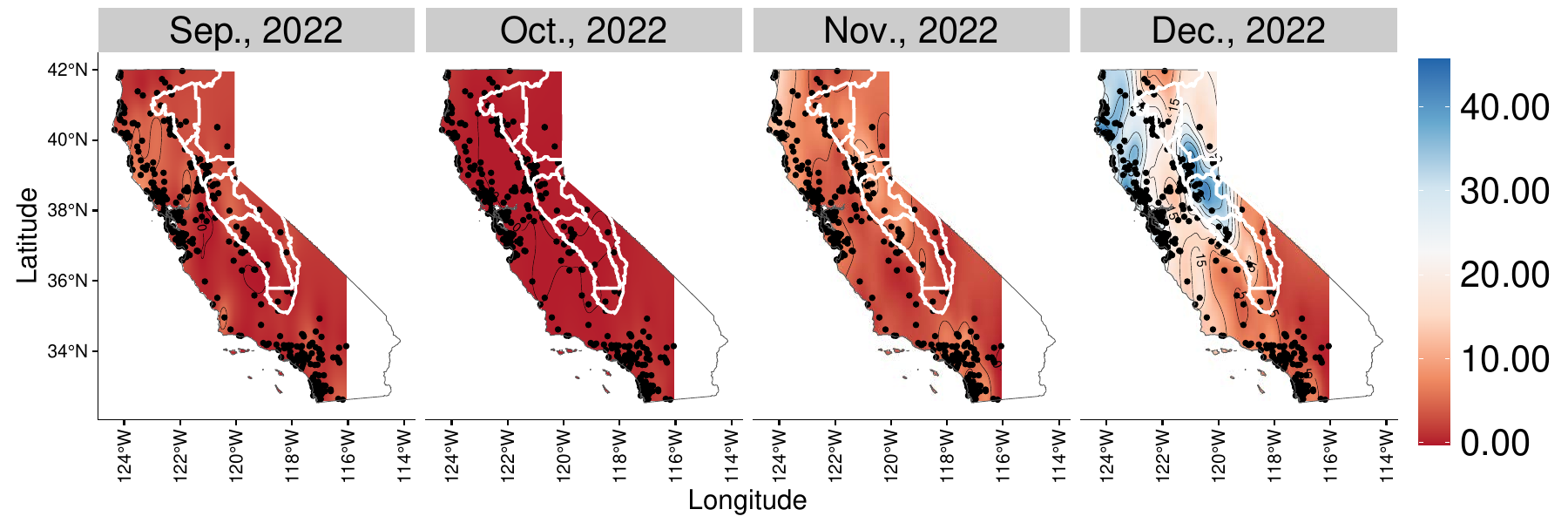}
\caption{Monthly precipitation (in cms.) for the state of California from September to December, 2022. The Sierra Nevada Conservancy sub-regions are marked in \texttt{white}.}\label{fig:prcp-sep-dec}
\end{figure}

\begin{figure}[t]
	\centering
	\includegraphics[scale=0.5]{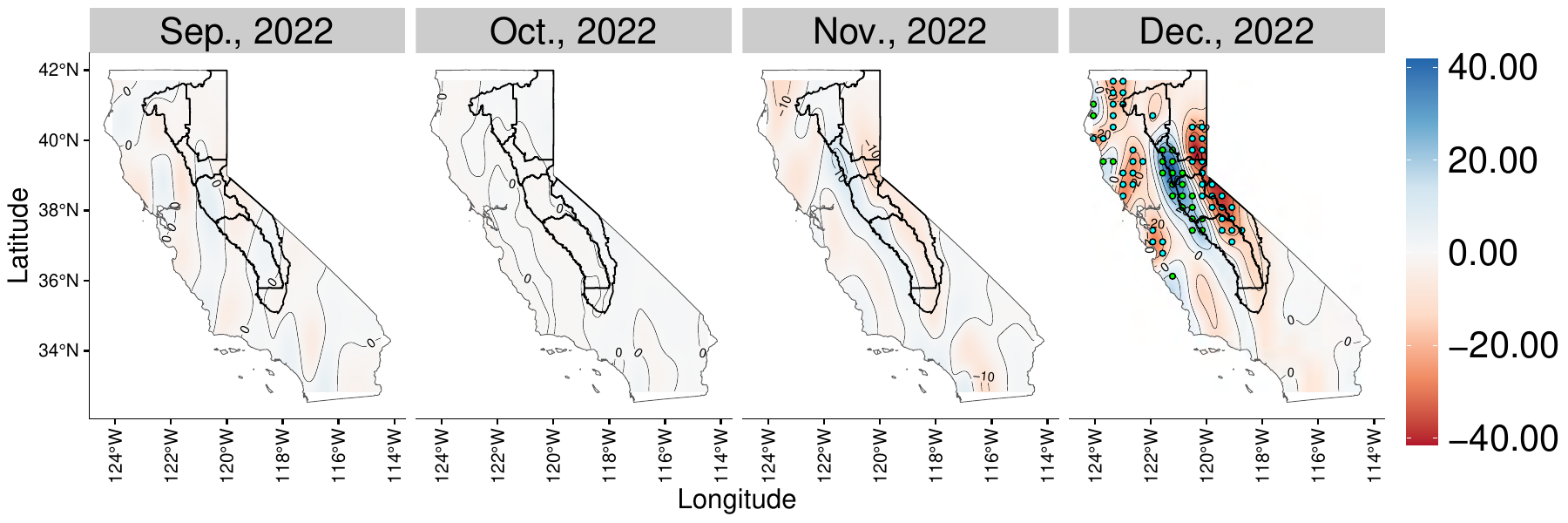}
	\includegraphics[scale=0.5]{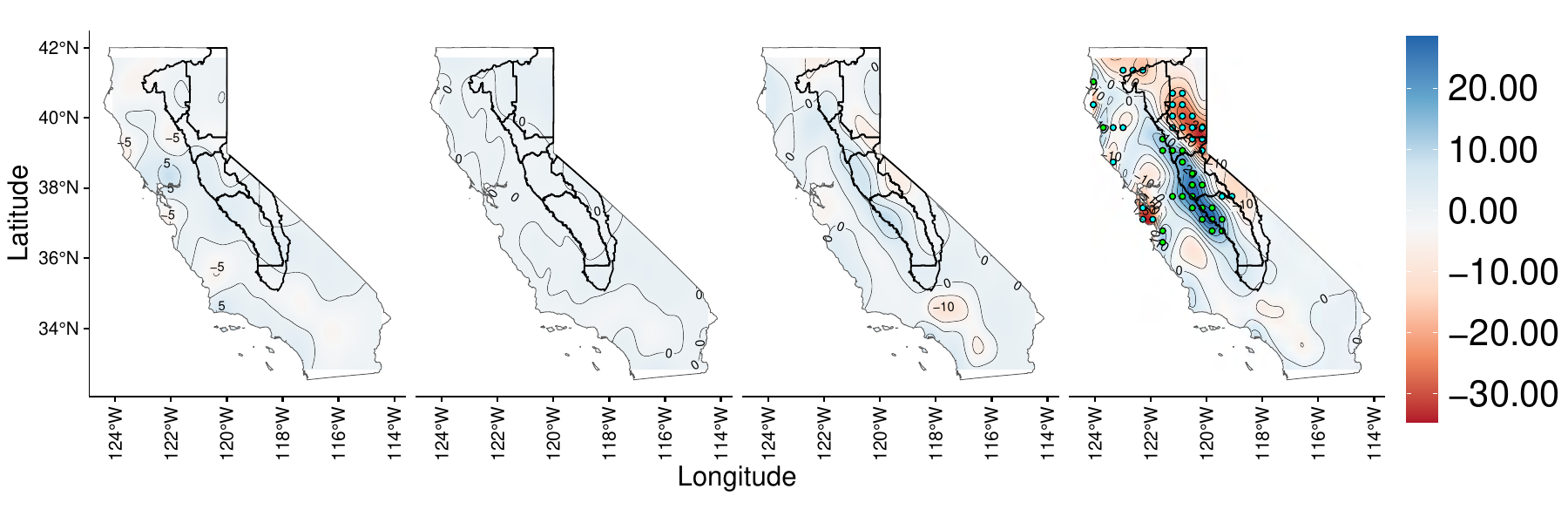}
	\caption{Showing $\partial_{s_x}$ and $\partial_{s_y}$ in precipitation with color-coded grid locations for significant estimates: positive/negative (\texttt{green}/\texttt{cyan}).}\label{fig:gradsy}
\end{figure}

\begin{figure}[t]
\centering
\includegraphics[width=\linewidth]{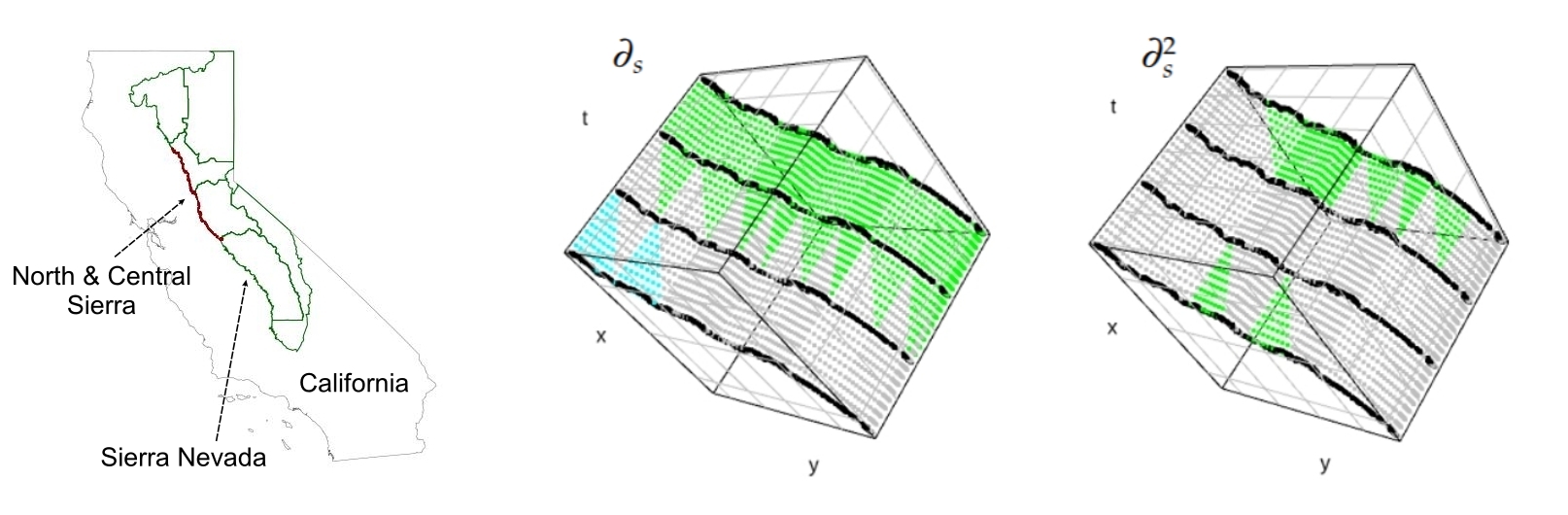}
\caption{Wombling and significance analysis on the North and Centra Sierra. Plots (left) showing the north and central Sierra to be tested for a wombling boundary (center and right) significance at the triangular region level for the \emph{static wombling} surface.}\label{fig:wm-prcp}
\end{figure}

\begin{table}[t]
\caption{Posterior estimates accompanied by HPD intervals for process parameters from analyzing precipitation in California during September--December, 2022. }\label{tab:precip-model}
\begin{minipage}{.5\linewidth}
\resizebox{\linewidth}{!}{
	\centering
	\begin{tabular}{l|@{\extracolsep{150pt}}c@{}}
		\hline\hline
		Parameter & Estimate \\
		\hline
		\multirow{2}{*}{$\sigma^2$ (Space-time Var.)} & 74.40 \\
		& (58.77, 89.76) \\ 
		\multirow{2}{*}{$\tau^2$ (Error Var.)} & 11.15\\
		& (10.28, 12.18) \\ 
		\multirow{2}{*}{$\phi_s$ (Spatial)} & 1.32 \\
		& (1.15, 1.44) \\
		\hline
		\hline
	\end{tabular}
}
\end{minipage}%
\begin{minipage}{.5\linewidth}
\resizebox{\linewidth}{!}{
	\centering
	\begin{tabular}{l|@{\extracolsep{150pt}}c@{}}
		\hline\hline
		Parameter & Estimate \\
		\hline
		\multirow{2}{*}{$\phi_t$ (Temporal)} & 4.74 \\
		& (2.38, 9.06)\\ 
		\multirow{2}{*}{$\beta_0$ (Intercept)} & 6.82 \\
		& (6.64, 7.01) \\ 
		\multirow{2}{*}{$\beta_1$ (Elevation)} &  $3.44 \times 10^{-3}$\\
		& (-0.05, 7.08)$\times 10^{-3}$\\
		\hline
		\hline
	\end{tabular}
}
\end{minipage}
\end{table}

Precipitation data is obtained from the National Ocean and Atmospheric Administration (NOAA) for 411 stations strewn across the state recording daily measurements (in tenths of mms.). They are aggregated to produce monthly totals in cms. \Cref{fig:prcp-sep-dec} shows the interpolated surfaces obtained. Elevation (in meters) of the station is used as a covariate. \Cref{tab:precip-model} shows posterior estimates of $\btheta$ resulting from hierarchical modeling of precipitation. Strength of the spatiotemporal story is measured using $\frac{\widehat{\sigma^2}}{\widehat{\sigma^2}+\widehat{\tau^2}} = \frac{74.40}{74.40 + 11.15}\times 100\approx 87\%$. Comparing spatial and temporal ranges---the spatial variation is stronger compared to temporal variation. \Cref{fig:gradsy} shows the spatial gradient along the latitude and longitude. Owing to high precipitation in Dec. we see significant gradients towards the western (windward) side of Sierra Nevada followed by negative gradients on the eastern (leeward) side---a consequence of the orographic effect. % Further analysis is available in figs.~S106--109. 

{ \emph{Static wombling} is performed on the residual surfaces, $Z(\bs,t)$ using the north and south central Sierra as a static boundary as seen in the left plot for \cref{fig:wm-prcp}. Performing static wombling only yields measures for the spatial derivatives and curvature (see Section 3.3 of the manuscript). The north and south central Sierra shows a significant overall spatial wombling measure (HPD interval) of 1.00 (0.37, 1.68) but no significant curvature estimated (HPD interval) at 1.18 (-1.16, 3.52). A time-interval level analysis is shown in \cref{tab:wmbl-timesplit-prcp}. The intervals Oct.--Nov. and Nov.--Dec. show significant spatial wombling measures as can also be seen in the significance plots in \cref{fig:wm-prcp}.

\begin{table}[ht!]
\centering
\caption{Estimated average wombling measure split by time intervals for the precipitation analysis. They are accompanied by HPD intervals. Estimates that are significant are marked in bold.}\label{tab:wmbl-timesplit-prcp}
\begin{tabular}{l|c|@{\extracolsep{10pt}}*{1}{c}@{}}
	\hline\hline
	\multirow{2}{*}{$\overline{\bGamma}(\C^*)$} & \multirow{2}{*}{Time}  & \multirow{2}{*}{Estimate}\\
	&&\\
	\hline
	\multirow{6}{*}{$\bpartial_\bs$} & \multirow{2}{*}{Sep.--Oct.} & -0.24\\ 
	& & (-1.65, 1.06) \\
	& \multirow{2}{*}{Oct.--Nov.} & {\bf 1.81}\\ 
	& & {\bf (0.34, 3.41)}\\
	& \multirow{2}{*}{Nov.--Dec.} & {\bf 7.36}\\
	& & {\bf (4.59, 10.55)}\\\hline 
	\multirow{6}{*}{$\bpartial_\bs^2$} & \multirow{2}{*}{Sep.--Oct.} & 2.56\\ 
	& & (-3.78, 9.30) \\
	& \multirow{2}{*}{Oct.--Nov.} & 1.33\\ 
	& & (-5.62, 7.83)\\
	& \multirow{2}{*}{Nov.--Dec.} & 6.67\\ 
	& & (-0.97, 14.27)\\\hline
	\hline
\end{tabular}
\end{table}
}

\subsection{\bf Ambient particle concentration (\texorpdfstring{PM\textsubscript{2.5}})) and the Canada Wildfires}

\begin{figure}[t]
\centering
\begin{subfigure}{.5\textwidth}
\centering
\includegraphics[scale = 0.25]{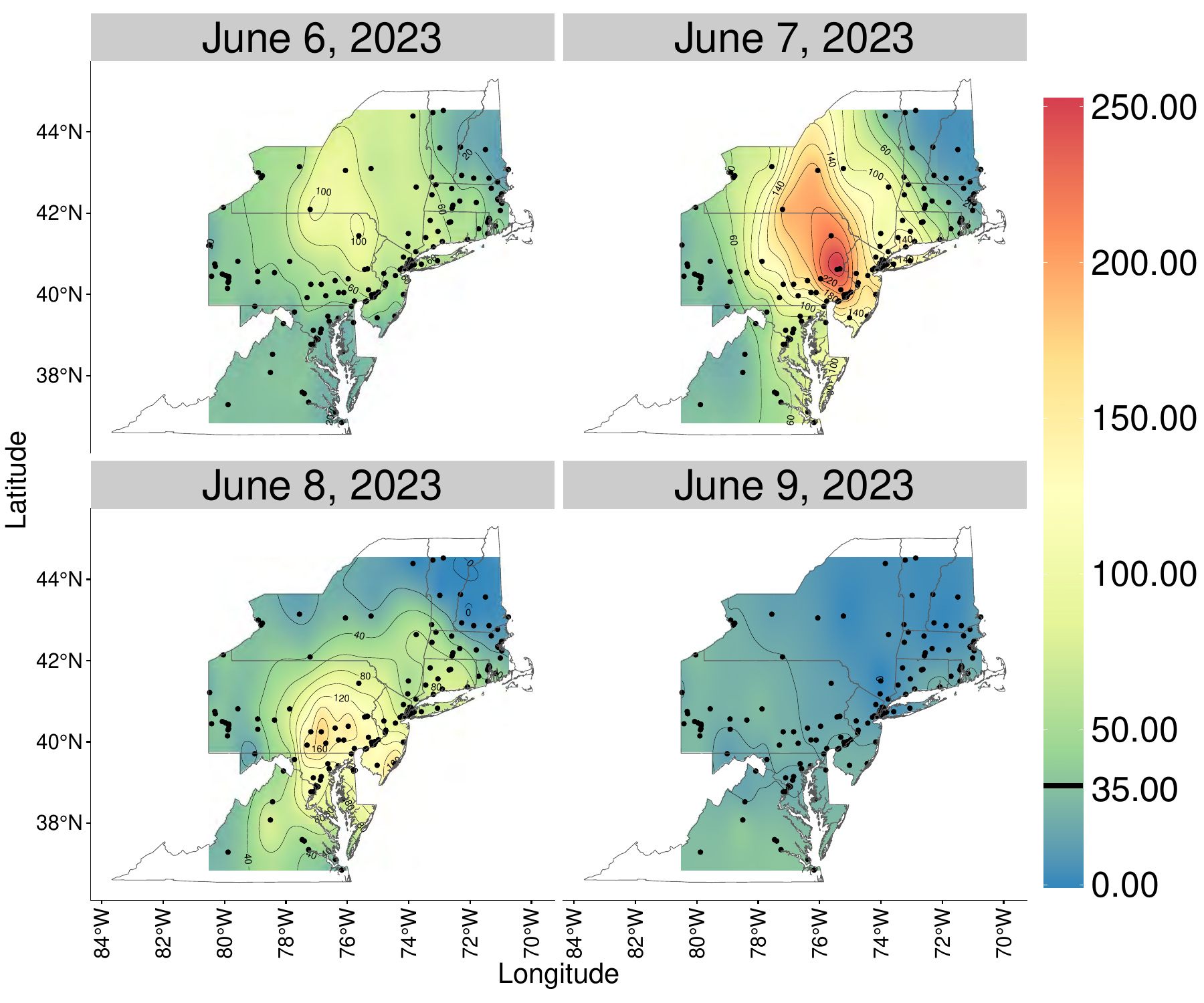}
\caption{June 6 -- 9}\label{fig:pm2.5jun}
\end{subfigure}%
\begin{subfigure}{.5\textwidth}
\centering
\includegraphics[scale = 0.25]{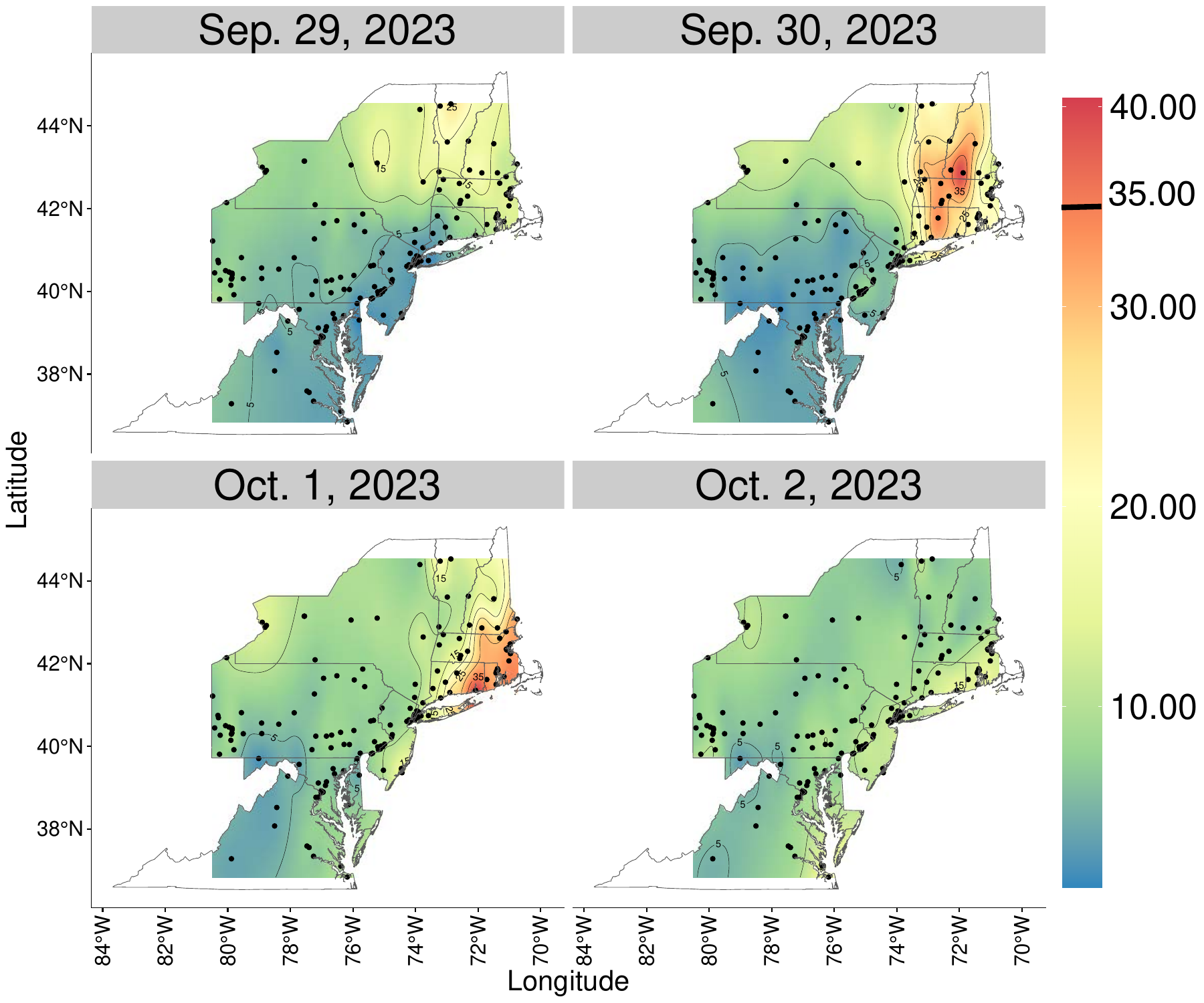}
\caption{Sep. 29 -- Oct. 2}\label{fig:pm2.5sep-oct}
\end{subfigure}
\caption{Spatiotemporal PM\textsubscript{2.5} (in \si{\micro\gram\meter\textsuperscript{-3}}) surfaces during the Canada wildfires, 2023. The \href{https://www.epa.gov/pm-pollution/final-reconsideration-national-ambient-air-quality-standards-particulate-matter-pm}{national daily standard} for PM\textsubscript{2.5} is set at 35\si{\micro\gram\meter^{-3}} by the US Environmental Protection Agency (EPA).}
\end{figure} 

Record setting wildfires in the year 2023 burned 5\% of the entire forest area of Canada. Waves of smoke descended into the US increasing PM\textsubscript{2.5} in the air, resulting in increased hospitalizations and evacuations. During June 6--8, New York City saw record rises in PM\textsubscript{2.5} (117\si{\micro\gram\meter^{-3}}) from a smoke wave resulting in massive respiratory hospitalizations \citep[see, e.g.,][]{chen2023canadian}. The city awoke to a smog-filled overcast sky on Oct. 1 %\citep[][]{lin2023}
as another smoke wave passed through. Interpolated spatiotemporal PM\textsubscript{2.5} surface plots shown in \cref{fig:pm2.5jun,fig:pm2.5sep-oct} during these time periods serve as supporting evidence. We use spatiotemporal wombling to detect boundaries delineating zones of high PM\textsubscript{2.5} levels in the neighboring regions.

% for the two time-periods June 6--9 and Sept. 29--Oct. 2.}
\begin{table}[t]
\caption{Posterior estimates accompanied by HPD intervals for $\btheta$ from the PM\textsubscript{2.5} analyses.}\label{tab:pm2.5-model}
\begin{minipage}{.5\linewidth}
\resizebox{\linewidth}{!}{
\centering
\begin{tabular}{l|@{\extracolsep{55pt}}cc@{}}
	\hline\hline
	\multirow{2}{*}{Parameter} & \multicolumn{2}{c}{Estimate} \\
	\cline{2-3}
	& Jun. 6--9 & Sep. 29 -- Oct. 2\\
	\hline
	\multirow{2}{*}{$\sigma^2$} & 2057.18 & 55.59 \\
	& (1316.47, 3224.98) & (42.36, 85.54) \\ 
	\multirow{2}{*}{$\tau^2$} & 128.78 & 2.99 \\
	& (109.93, 143.82) & (2.40, 3.40) \\ 
	\multirow{2}{*}{$\phi_s$} & 0.51 & 0.60 \\
	& (0.45, 0.61) & (0.50, 0.73) \\
	\hline
	\hline
\end{tabular}
}
\end{minipage}%
\begin{minipage}{.5\linewidth}
\resizebox{\linewidth}{!}{
\centering
\begin{tabular}{l|@{\extracolsep{70pt}}cc@{}}
	\hline\hline
	\multirow{2}{*}{Parameter} & \multicolumn{2}{c}{Estimate} \\
	\cline{2-3}
	&Jun. 6--9 & Sep. 29 -- Oct. 2\\
	\hline
	\multirow{2}{*}{$\phi_t$} & 0.90 & 0.73 \\
	& (0.66, 1.29) & (0.60, 0.90)\\ 
	\multirow{2}{*}{$\beta_0$} & 59.27 & 10.36 \\
	& (58.22, 60.26) & (10.20, 10.53) \\ 
	\multirow{2}{*}{$\beta_1 \times 10^{-2}$ } & -3.31 & -0.42 \\
	& (-3.81, -2.88) & (-0.48, -0.35)\\
	\hline
	\hline
\end{tabular}
}
\end{minipage}

\medskip

\caption{Average wombling measures for wombling surfaces selected for the two waves.}\label{tab:wm-pm25}
\resizebox{\linewidth}{!}{
\begin{tabular}{l|@{\extracolsep{1pt}}*{8}{>{\bfseries}c}@{}}
\hline\hline
\multirow{2}{*}{Waves} &\multicolumn{8}{c}{$\overline{\bGamma}(\C)$}\\
\cline{2-9}
& $\bpartial_\bs$ & $\bpartial_\bs^2$ & $\partial_t$ & $\partial_t\bpartial_\bs$ & $\partial_t \bpartial_\bs^2$ & $\partial_t^2$ & $\partial_t^2\bpartial_\bs$ & $\partial_t^2\bpartial_\bs^2$ \\ 
\hline
\multirow{2}{*}{June}& -26.33 & -28.08 & -36.84 & -13.19 & 12.12 & -100.18 & 20.41 & 19.19 \\ 
& (-29.50, -23.03) & (-32.21, -23.59) & (-41.57, -32.59) & (-15.69, -10.27) & (8.12, 16.14) & (-107.07, -91.79) & (14.38, 27.97) & (11.51, 26.56) \\
\hline
\multirow{2}{*}{Sep.--Oct.} & -4.57 & -1.34 & -6.35 & -0.50 & 1.89 & -4.15 & 4.44 & {\normalfont 0.30}\\
& (-5.13, -4.04) & (-2.10,-0.52) & (-6.87, -5.80) & (-0.86, -0.13) & (1.19, 2.63) & (-4.96, -3.37) & (3.72, 5.18) & {\normalfont (-0.91, 1.57)} \\ 
\hline\hline
\end{tabular}
}
\end{table}
Daily PM\textsubscript{2.5} data from 144 and 154 monitoring stations is obtained from the EPA for June and Sep.-Oct. respectively. We use elevation of the station as a covariate. Posterior results from modeling PM\textsubscript{2.5} are shown in \cref{tab:pm2.5-model}. Strength of the spatiotemporal story in both waves is evident, given $\frac{\widehat{\sigma^2}}{\widehat{\sigma^2}+\widehat{\tau^2}} = \frac{2057.18}{2057.18 + 128.78} \times 100 = 94.11\%$ and $\frac{\widehat{\sigma^2}}{\widehat{\sigma^2}+\widehat{\tau^2}} = \frac{55.59}{55.59 + 2.99} \times 100 = 94.90\%$ of the variation being spatiotemporal in nature, respectively. \Cref{fig:gradsxt} shows the mixed spatial-temporal derivative along the longitude. % The complete differential process assessment can be found in figs.~S110-115 in the Supplement. % Comparing the magnitudes of estimated gradients for the two smoke waves, the first wave had a higher intensity therefore, producing severe detrimental effects on public health. 

\begin{figure}[t]
	\centering
	\includegraphics[scale=0.5]{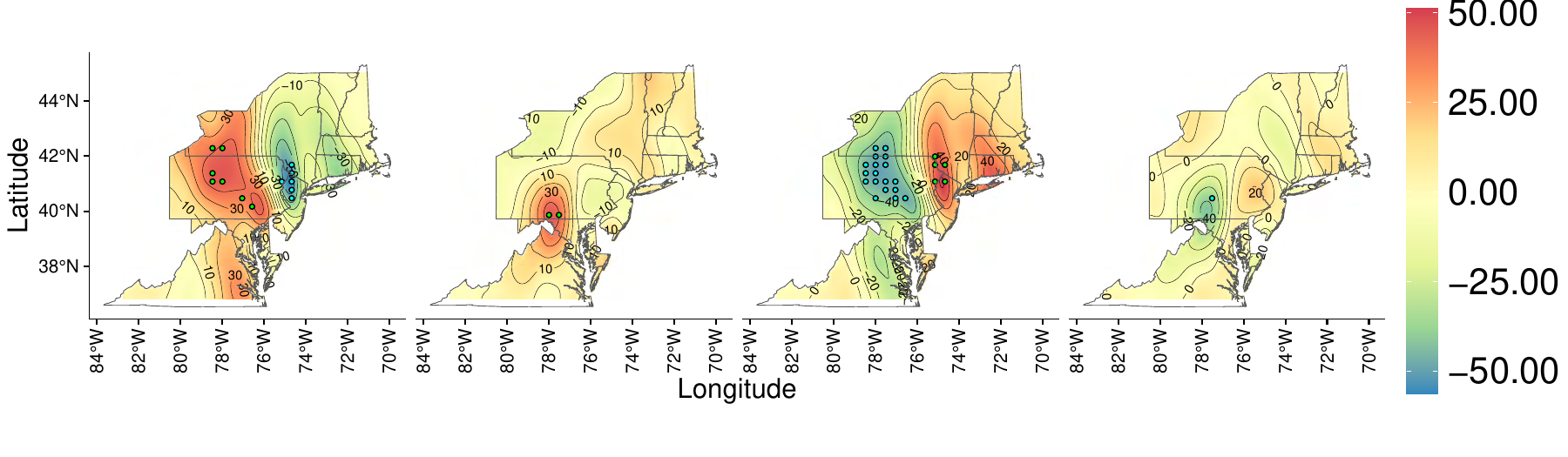}
	\includegraphics[scale=0.5]{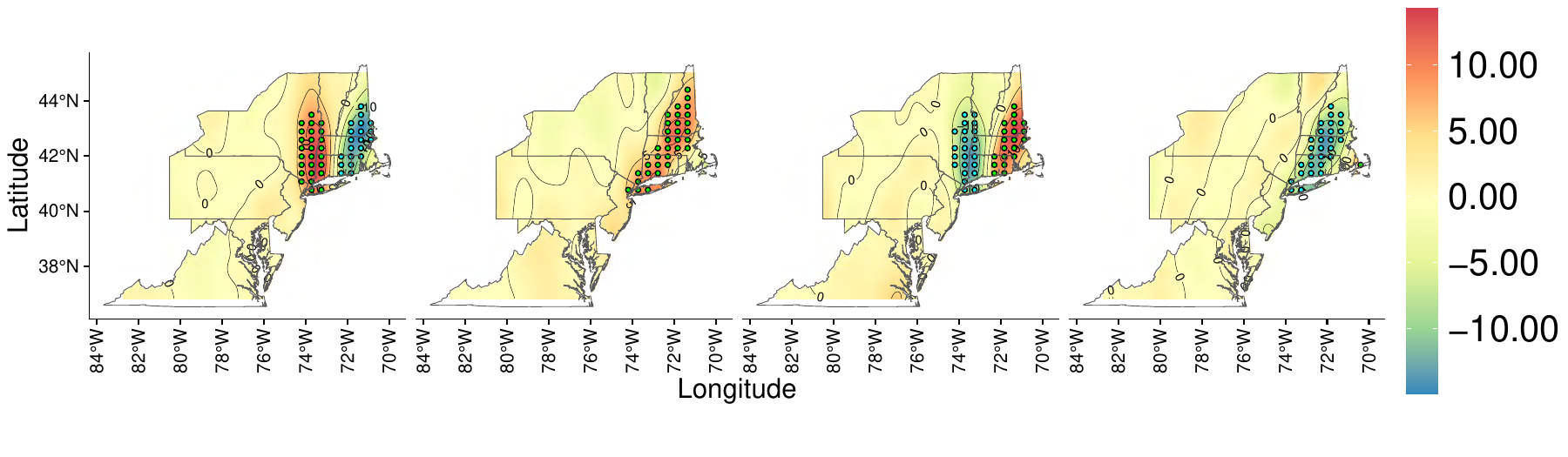}
	\caption{Comparing $\partial_t\partial_{s_x}$ for Jun 6--9 (row 1) versus Sep. 29 -- Oct. 2 (row 2). Significant grid locations are color-coded: positive/negative (\texttt{green}/\texttt{cyan}).}\label{fig:gradsxt}
\end{figure}

\begin{figure}[t]
\centering
\includegraphics[width = \linewidth]{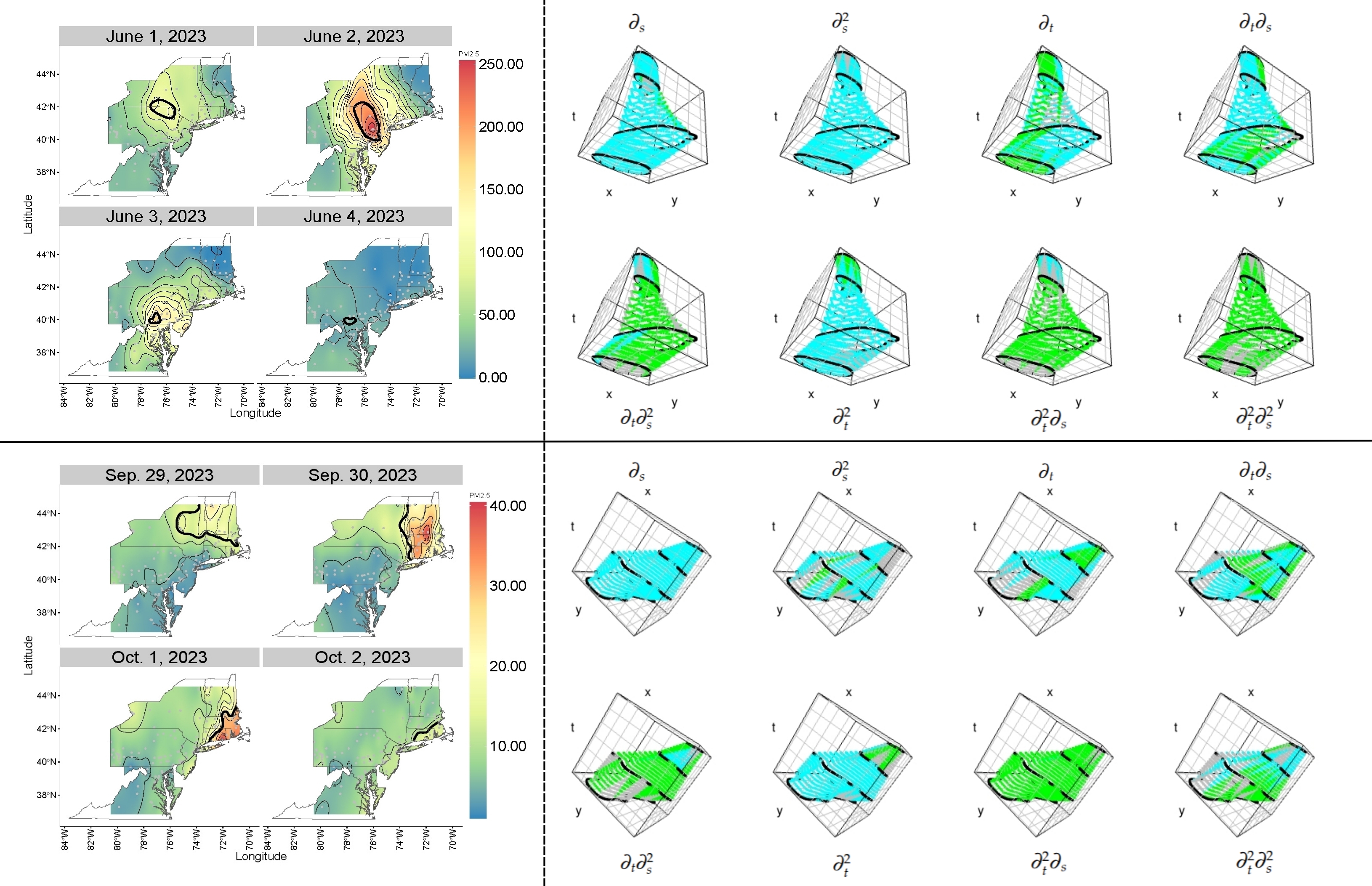}
\caption{Wombling and significance analysis for Jun 6--9 (row 1) versus Sep. 29 -- Oct. 2 (row 2).}\label{fig:wm-pm25}
\end{figure}

The residual spatiotemporal surface, $Z(\bs,t)$ is used for wombling. The wombling surface for June 6--9 is constructed using closed level curves while for Sep.--Oct. we use open curves (see \cref{fig:wm-pm25}). Wombling measures for the two surfaces are shown in \cref{tab:wm-pm25}. Comparing intensities for the two waves, the June 6--9 wave was severe, as is evident from the differences in magnitudes and significance of the overall average wombling measures for the two waves. The rapidity of spatial-temporal changes in PM\textsubscript{2.5} levels is captured by the differences between mixed derivatives. These changes translated to rapidly deteriorating air-quality resulting in massive respiratory hospitalizations on June 7, 2023 in New York City. The total wombling measures summarized by time-intervals, as seen in Table~S11 of the Supplement, serve as further evidence. 

\section{Spatiotemporal Surface Wombling: time interval and triangular region level analysis}

We supplement the discussion in Section~5 with time-interval and triangular-region level analysis. For the time-interval level analysis we refer to \cref{tab:wmbl-timesplit-1,tab:wmbl-timesplit-2} to verify and draw conclusions. The triangular level analysis is carried out through assessing significance for surfaces $A$, $B$ and $C$ as shown in \cref{fig:sig-wombl-A-B,fig:sig-wombl-C}. We interpret the wombling measures (WMs) associated with temporal gradient. The others follow a similar pattern. During time intervals (1--2) and (3--4) troughs (peaks) enclosed by surface $A$ (surface $B$) flatten yielding positive (negative) WMs, while surface $C$ does not experience much change yielding insignificant WMs. Interval (3--4) is one of change where planar curves generating $A$ and $B$ shift spatially and the surface flattens overall. During intervals (4--5), (5--6) and (6--7) the troughs (peaks) re-appear for surface $A$ (surface $B$), while surface $C$ experiences decreasing temporal gradients before returning to insignificant WMs. Significance at the triangular region level is available in \cref{fig:sig-wombl-A-B,fig:sig-wombl-C} third plot from the left in row 1. In the following subsections we show tables and figures associated with simulation experiments and applications in the manuscript.

\newgeometry{margin=1cm} % modify this if you need more space
\begin{landscape}
\thispagestyle{empty}
\section{Tables}
\subsection{Non-separable Kernels}
\begin{table}[ht]
\caption{Goodness-of-fit for the spatiotemporal differential processes for simulation experiments performed using Pattern 1.}\label{tab::sim-1}
\resizebox{\linewidth}{!}{
\centering
\begin{tabular}{*{2}{c}|c|@{\extracolsep{5pt}}*{17}{c}@{}}
	\hline
	\hline
	\multicolumn{2}{c|}{\multirow{2}{*}{$N$}} & \multirow{2}{*}{$\tau^2$} & \multicolumn{17}{c}{\multirow{2}{*}{RMSE}} \\
	&&&&&&&&&&&&&&&&&\\\cline{1-2}\cline{4-20}
	\multirow{2}{*}{$N_s$} & \multirow{2}{*}{$N_t$} & \multirow{2}{*}{(Truth = 1.00)}  & \multirow{2}{*}{$\partial_{s_x}$} & \multirow{2}{*}{$\partial_{s_y}$} & \multirow{2}{*}{$\partial^2_{s_{xx}}$} & \multirow{2}{*}{$\partial^2_{s_{xy}}$} & \multirow{2}{*}{$\partial^2_{s_{yy}}$} & \multirow{2}{*}{$\partial_t$} & \multirow{2}{*}{$\partial_t\partial_{s_x}$} & \multirow{2}{*}{$\partial_t\partial_{s_y}$} & \multirow{2}{*}{$\partial_t\partial^2_{s_{xx}}$} & \multirow{2}{*}{$\partial_t\partial^2_{s_{xy}}$} & \multirow{2}{*}{$\partial_t\partial^2_{s_{yy}}$} & \multirow{2}{*}{$\partial^2_t$} & \multirow{2}{*}{$\partial^2_t\partial_{s_x}$} & \multirow{2}{*}{$\partial^2_t\partial_{s_y}$} & \multirow{2}{*}{$\partial^2_t\partial^2_{s_{xx}}$} & \multirow{2}{*}{$\partial^2_t\partial^2_{s_{xy}}$} & \multirow{2}{*}{$\partial^2_t\partial^2_{s_{yy}}$} \\
	&&&&&&&&&&&&&&&&&&&\\
	\hline
	\multirow{6}{*}{30} & \multirow{2}{*}{3} & 0.74 & 38.09 & 21.07 & 308.59 & 117.41 & 238.24 & 1.34 & 6.56 & 10.32 & 70.33 & 43.70 & 124.38 & 0.79 & 4.07 & 6.11 & 45.61 & 19.76 & 67.83 \\ 
	& & (0.25, 1.37) & (5.17), (0.8) & (2.54), (0.9) & (45.30), (0.9) & (11.71), (1.0) & (22.21), (1.0) & (0.17), (0.9) & (0.94), (0.9) & (1.03), (0.8) & (15.25), (1.0) & (7.44), (1.0) & (9.90), (1.0) & (0.10), (0.8) & (1.31), (1.0) & (0.66), (1.0) & (16.30), (1.0) & (9.00), (1.0) & (8.11), (1.0)\\
	& \multirow{2}{*}{6} & 0.87 & 34.84 & 19.94 & 263.03 & 119.92 & 240.95 & 1.15 & 5.59 & 7.72 & 55.30 & 38.07 & 103.23 & 0.56 & 2.13 & 4.36 & 22.70 & 13.98 & 49.86 \\ 
	& & (0.61, 1.17) & (4.06), (0.8) & (1.86), (0.9) & (37.91), (0.9) & (8.87), (1.0) & (21.46), (0.9) & (0.13), (0.9) & (0.71), (0.8) & (0.86), (0.8) & (10.19), (1.0) & (3.01), (1.0) & (10.86), (0.9) & (0.05), (0.6) & (0.41), (1.0) & (0.53), (0.8) & (6.10), (1.0) & (2.75), (1.0) & (5.09), (0.9) \\ 
	& \multirow{2}{*}{9} & 0.88 & 33.48 & 20.81 & 277.93 & 128.89 & 235.40 & 1.05 & 5.18 & 7.75 & 51.83 & 35.52 & 98.28 & 0.57 & 2.76 & 4.55 & 28.54 & 17.71 & 53.36 \\ 
	& & (0.67, 1.15) & (5.26), (0.8) & (1.94), (0.9) & (42.67), (0.9) & (8.26), (0.9) & (19.05), (0.9) & (0.17), (0.9) & (0.37), (0.8) & (0.68), (0.7) & (5.65), (1.0) & (2.99), (1.0) & (8.33), (0.9) & (0.05), (0.7) & (0.30), (1.0) & (0.20), (0.9) & (4.32), (1.0) & (2.25), (1.0) & (2.40), (0.9) \\ 
	\hline
	\multirow{6}{*}{50} & \multirow{2}{*}{3} & 0.86 & 25.33 & 16.52 & 209.54 & 112.17 & 207.27 & 1.08 & 6.40 & 7.86 & 62.90 & 39.11 & 104.08 & 0.81 & 4.60 & 5.38 & 44.39 & 21.22 & 60.49 \\ 
	& & (0.54 1.29) & (2.79), (0.9) & (0.85), (1.0) & (28.11), (1.0) & (15.26), (1.0) & (12.48), (0.9) & (0.14), (1.0) & (1.11), (0.8) & (1.27), (0.8) & (18.48), (1.0) & (7.77), (1.0) & (11.58), (0.9) & (0.09), (0.7) & (1.11), (1.0) & (0.58), (1.0) & (13.35), (1.0) & (6.12), (1.0) & (6.30), (1.0) \\ 
	& \multirow{2}{*}{6} & {\bf 0.88} & {\bf 24.03} & {\bf 15.12} & {\bf 190.88} & {\bf 105.09} & {\bf 185.69} & {\bf 0.68} & {\bf 4.92} & {\bf 6.02} & {\bf 46.68} & {\bf 33.10} & {\bf 84.06} & {\bf 0.45} & {\bf 2.27} & {\bf 3.29} & {\bf 22.89} & {\bf 12.91} & {\bf 41.20} \\ 
	& & (0.70, 1.10) & (2.72), (0.9) & (1.84), (1.0) & (14.76), (1.0) & (6.66), (1.0) & (11.59), (1.0) & (0.15), (1.0) & (0.51), (0.8) & (1.07), (0.7) & (5.21), (1.0) & (3.15), (1.0) & (7.37), (0.9) & (0.05), (0.7) & (0.34), (1.0) & (0.34), (0.9) & (3.47), (1.0) & (2.23), (1.0) & (3.02), (0.9) \\
	& \multirow{2}{*}{9} & 0.96 & 23.84 & 15.31 & 194.11 & 110.39 & 183.62 & 0.69 & 4.42 & 5.58 & 41.45 & 30.15 & 75.19 & 0.44 & 2.39 & 3.43 & 23.57 & 15.29 & 43.10 \\ 
	& & (0.78, 1.16) & (4.78), (0.9) & (1.97), (1.0) & (29.37), (0.9) & (7.58), (1.0) & (16.72), (0.9) & (0.11), (1.0) & (0.41), (0.8) & (0.44), (0.8) & (3.38), (1.0) & (2.07), (1.0) & (5.69), (0.9) & (0.05), (0.7) & (0.33), (1.0) & (0.17), (0.9) & (2.47), (1.0) & (1.28), (1.0) & (1.95), (0.9)  \\ \hline
	\multirow{6}{*}{100} & \multirow{2}{*}{3} & {\bf 0.90} & {\bf 14.20} & {\bf 12.32} & {\bf 140.74} & {\bf 96.92} & {\bf 166.74} & {\bf 0.78} & {\bf 5.57} & {\bf 6.21} & {\bf 53.58} & {\bf 32.80} & {\bf 85.01} & {\bf 0.72} & {\bf 4.64} & {\bf 4.59} & {\bf 43.58} & {\bf 16.59} & {\bf 50.75} \\
	& & (0.70, 1.15) & (2.16), (1.0) & (1.23), (1.0) & (17.63), (1.0) & (7.94), (0.9) & (16.94), (0.9) & (0.07), (1.0) & (0.68), (0.8) & (0.44), (0.7) & (7.65), (1.0) & (4.03), (1.0) & (5.54), (0.9) & (0.09), (0.7) & (0.72), (1.0) & (0.38), (1.0) & (6.78), (1.0) & (2.52), (1.0) & (4.23), (1.0) \\
	& \multirow{2}{*}{6} & 0.97 & 13.37 & 12.77 & 132.80 & 94.34 & 176.47 & 0.50 & 3.80 & 5.20 & 40.19 & 26.07 & 74.91 & 0.39 & 2.05 & 2.58 & 19.27 & 10.88 & 35.63 \\ 
	& & (0.82, 1.17) & (2.03), (1.0) & (0.73), (1.0) & (9.29), (1.0) & (6.95), (1.0) & (8.41), (0.9) & (0.06), (1.0) & (0.43), (0.8) & (0.52), (0.7) & (3.23), (1.0) & (2.56), (1.0) & (2.51), (0.9) & (0.04), (0.7) & (0.13), (1.0) & (0.36), (1.0) & (1.27), (1.0) & (1.16), (1.0) & (3.49), (0.9) \\ 
	& \multirow{2}{*}{9} & 0.96 & 10.76 & 11.34 & 118.83 & 81.59 & 166.53 & 0.46 & 3.48 & 4.61 & 34.53 & 24.01 & 68.69 & 0.36 & 2.04 & 3.00 & 19.51 & 11.55 & 38.77 \\ 
	& & (0.84, 1.08) & (1.54), (1.0) & (1.15), (1.0) & (8.17), (1.0) & (4.61), (1.0) & (5.24), (0.9) & (0.04), (1.0) & (0.32), (0.8) & (0.48), (0.7) & (2.52), (1.0) & (1.98), (1.0) & (2.89), (0.9) & (0.02), (0.7) & (0.15), (1.0) & (0.15), (1.0) & (1.78), (1.0) & (0.85), (1.0) & (1.49), (0.8) \\  
	\hline
	\hline
\end{tabular}
}

\medskip

\caption{Goodness-of-fit for the spatiotemporal differential processes for simulation experiments performed using Pattern 2.}\label{tab::sim-2}
\resizebox{\linewidth}{!}{
\centering
\begin{tabular}{*{2}{c}|c|@{\extracolsep{5pt}}*{17}{c}@{}}
	\hline
	\hline
	\multicolumn{2}{c|}{\multirow{2}{*}{$N$}} & \multirow{2}{*}{$\tau^2$} & \multicolumn{17}{c}{\multirow{2}{*}{RMSE}} \\
	&&&&&&&&&&&&&&&&&\\\cline{1-2}\cline{4-20}
	\multirow{2}{*}{$N_s$} & \multirow{2}{*}{$N_t$} & \multirow{2}{*}{(Truth = 1.00)}  & \multirow{2}{*}{$\partial_{s_x}$} & \multirow{2}{*}{$\partial_{s_y}$} & \multirow{2}{*}{$\partial^2_{s_{xx}}$} & \multirow{2}{*}{$\partial^2_{s_{xy}}$} & \multirow{2}{*}{$\partial^2_{s_{yy}}$} & \multirow{2}{*}{$\partial_t$} & \multirow{2}{*}{$\partial_t\partial_{s_x}$} & \multirow{2}{*}{$\partial_t\partial_{s_y}$} & \multirow{2}{*}{$\partial_t\partial^2_{s_{xx}}$} & \multirow{2}{*}{$\partial_t\partial^2_{s_{xy}}$} & \multirow{2}{*}{$\partial_t\partial^2_{s_{yy}}$} & \multirow{2}{*}{$\partial^2_t$} & \multirow{2}{*}{$\partial^2_t\partial_{s_x}$} & \multirow{2}{*}{$\partial^2_t\partial_{s_y}$} & \multirow{2}{*}{$\partial^2_t\partial^2_{s_{xx}}$} & \multirow{2}{*}{$\partial^2_t\partial^2_{s_{xy}}$} & \multirow{2}{*}{$\partial^2_t\partial^2_{s_{yy}}$} \\
	&&&&&&&&&&&&&&&&&&&\\
	\hline
	\multirow{6}{*}{30} & \multirow{2}{*}{3} & 0.44 & 26.71 & 20.19 & 265.40 & 216.34 & 257.71 & 1.43 & 15.57 & 12.58 & 166.88 & 133.97 & 172.29 & 0.92 & 12.42 & 11.71 & 211.41 & 112.00 & 215.27 \\ 
	& & (0.01, 1.12) & (2.44), (0.8) & (2.40), (0.9) & (60.50), (1.0) & (31.36), (1.0) & (63.81), (1.0) & (0.10), (0.8) & (1.39), (0.9) & (1.73), (1.0) & (58.17), (1.0) & (29.09), (1.0) & (64.63), (1.0) & (0.24), (1.0) & (4.67), (1.0) & (6.95), (1.0) & (195.87), (1.0) & (119.65), (1.0) & (228.83), (1.0)\\
	& \multirow{2}{*}{6} & 0.92 & 24.65 & 18.13 & 233.04 & 200.04 & 224.00 & 1.18 & 13.14 & 9.81 & 127.17 & 108.50 & 123.88 & 0.45 & 7.14 & 4.22 & 57.40 & 42.74 & 58.77 \\ 
	& & (0.61, 1.29) & (2.44), (0.8) & (2.11), (0.9) & (21.50), (1.0) & (19.05), (0.9) & (21.28), (1.0) & (0.18), (0.7) & (1.37), (0.7) & (1.23), (0.8) & (14.07), (1.0) & (11.16), (0.9) & (14.09), (1.0) & (0.05), (0.9) & (0.58), (0.9) & (0.66), (1.0) & (14.10), (1.0) & (7.00), (1.0) & (14.81), (1.0) \\ 
	& \multirow{2}{*}{9} & 0.89 & 25.43 & 18.86 & 255.28 & 206.19 & 241.77 & 0.99 & 11.37 & 8.56 & 115.37 & 95.07 & 108.49 & 0.51 & 7.64 & 4.62 & 61.88 & 50.69 & 60.57 \\ 
	& & (0.67, 1.18) & (2.87), (0.8) & (2.95), (0.9) & (28.42), (1.0) & (28.46), (0.9) & (29.95), (1.0) & (0.13), (0.8) & (1.14), (0.7) & (1.25), (0.8) & (10.32), (1.0) & (11.67), (0.9) & (12.52), (1.0) & (0.05), (0.8) & (0.38), (0.8) & (0.46), (0.9) & (6.28), (1.0) & (4.67), (1.0) & (8.59), (1.0) \\  
	\hline
	\multirow{6}{*}{50} & \multirow{2}{*}{3} & 0.73 & 20.43 & 14.42 & 222.15 & 173.34 & 205.63 & 1.10 & 12.72 & 9.79 & 145.37 & 109.78 & 144.18 & 0.81 & 12.00 & 9.37 & 147.28 & 96.60 & 155.16 \\ 
	& & (0.36, 1.17) & (1.23), (0.9) & (1.24), (1.0) & (22.46), (1.0) & (9.88), (1.0) & (21.28), (1.0) & (0.07), (0.9) & (0.64), (0.9) & (0.81), (0.9) & (11.49), (1.0) & (7.39), (1.0) & (13.75), (1.0) & (0.13), (0.9) & (1.14), (1.0) & (1.20), (1.0) & (18.93), (1.0) & (11.98), (1.0) & (20.97), (1.0) \\
	& \multirow{2}{*}{6} & {\bf 0.89} & {\bf 18.25} & {\bf 12.91} & {\bf 198.32} & {\bf 155.60} & {\bf 189.54} & {\bf 0.80} & {\bf 9.79} & {\bf 6.65} & {\bf 103.26} & {\bf 82.24} & {\bf 95.38} & {\bf 0.33} & {\bf 7.78} & {\bf 3.14} & {\bf 53.41} & {\bf 35.49} & {\bf 50.46} \\ 
	& & (0.66, 1.18) & (1.66), (0.9) & (1.53), (1.0) & (22.96), (1.0) & (12.90), (0.9) & (20.68), (1.0) & (0.10), (0.9) & (0.82), (0.7) & (0.79), (0.9) & (11.88), (1.0) & (6.86), (0.9) & (11.30), (1.0) & (0.03), (0.9) & (0.61), (0.8) & (0.42), (1.0) & (11.25), (1.0) & (4.51), (1.0) & (11.06), (1.0) \\
	& \multirow{2}{*}{9} & 1.01 & 20.98 & 13.93 & 217.75 & 168.32 & 210.82 & 0.79 & 9.53 & 6.73 & 99.50 & 78.72 & 99.39 & 0.44 & 8.16 & 3.87 & 59.44 & 43.55 & 56.19 \\ 
	& & (0.79, 1.24) & (1.82), (0.9) & (1.27), (0.9) & (22.24), (1.0) & (16.71), (0.9) & (18.56), (1.0) & (0.07), (0.9) & (0.72), (0.7) & (0.57), (0.8) & (8.58), (1.0) & (6.95), (0.9) & (9.20), (1.0) & (0.03), (0.8) & (0.30), (0.8) & (0.30), (1.0) & (6.48), (1.0) & (3.60), (1.0) & (6.44), (1.0) \\ \hline
	\multirow{6}{*}{100} & \multirow{2}{*}{3} & {\bf 0.87} & {\bf 15.34} & {\bf 10.45} & {\bf 182.46} & {\bf 134.34} & {\bf 165.86} & {\bf 0.86} & {\bf 10.42} & {\bf 7.63} & {\bf 116.20} & {\bf 89.68} & {\bf 107.22} & {\bf 0.83} & {\bf 12.21} & {\bf 8.26} & {\bf 113.92} & {\bf 74.32} & {\bf 115.24} \\
	& & (0.66, 1.12) & (0.99), (0.9) & (0.81), (1.0) & (16.50), (1.0) & (7.59), (0.9) & (12.39), (1.0) & (0.07), (1.0) & (0.54), (0.8) & (0.55), (0.9) & (10.78), (1.0) & (3.70), (1.0) & (6.55), (1.0) & (0.11), (0.8) & (1.17), (0.9) & (0.98), (1.0) & (17.24), (1.0) & (8.76), (1.0) & (13.30), (1.0) \\
	& \multirow{2}{*}{6} & 0.93 &  14.53 & 9.65 & 171.10 & 129.04 & 153.40 & 0.57 & 7.65 & 5.07 & 86.70 & 66.41 & 76.32 & 0.28 & 8.43 & 2.77 & 47.06 & 29.86 & 44.09 \\ 
	& & (0.79, 1.10) & (0.87), (0.9) & (0.58), (1.0) & (9.32), (1.0) & (6.37), (0.9) & (8.41), (1.0) & (0.04), (1.0) & (0.52), (0.7) & (0.26), (0.8) & (5.10), (1.0) & (4.42), (0.9) & (4.71), (1.0) & (0.02), (0.9) & (0.38), (0.8) & (0.17), (1.0) & (4.83), (1.0) & (1.98), (1.0) & (2.85), (1.0) \\ 
	& \multirow{2}{*}{9} & 0.95 & 15.25 & 9.87 & 172.23 & 127.89 & 161.00 & 0.56 & 7.22 & 4.73 & 77.06 & 59.96 & 69.60 & 0.35 & 8.64 & 3.07 & 46.50 & 33.66 & 43.01 \\ 
	& & (0.82, 1.09) & (0.79), (0.9) & (0.59), (1.0) & (9.55), (1.0) & (5.36), (0.9) & (8.40), (1.0) & (0.04), (1.0) & (0.33), (0.7) & (0.30), (0.8) & (4.70), (1.0) & (2.02), (0.9) & (4.11), (1.0) & (0.03), (0.8) & (0.20), (0.7) & (0.25), (1.0) & (3.56), (1.0) & (1.69), (1.0) & (3.87), (1.0) \\  
	\hline
	\hline
\end{tabular}
}
\end{table}

\end{landscape}

\newpage

\newgeometry{margin=1cm} % modify this if you need even more space
\begin{landscape}
\thispagestyle{empty}

\subsection{Separable Kernels}

\begin{table}[ht]
\caption{Goodness-of-fit for the spatiotemporal differential processes for simulation experiments performed using Pattern 1.}\label{tab::sim-3}
\resizebox{\linewidth}{!}{
\centering
\begin{tabular}{*{2}{c}|c|@{\extracolsep{5pt}}*{17}{c}@{}}
	\hline
	\hline
	\multicolumn{2}{c|}{\multirow{2}{*}{$N$}} & \multirow{2}{*}{$\tau^2$} & \multicolumn{17}{c}{\multirow{2}{*}{RMSE}} \\
	&&&&&&&&&&&&&&&&&\\\cline{1-2}\cline{4-20}
	\multirow{2}{*}{$N_s$} & \multirow{2}{*}{$N_t$} & \multirow{2}{*}{(Truth = 1.00)}  & \multirow{2}{*}{$\partial_{s_x}$} & \multirow{2}{*}{$\partial_{s_y}$} & \multirow{2}{*}{$\partial^2_{s_{xx}}$} & \multirow{2}{*}{$\partial^2_{s_{xy}}$} & \multirow{2}{*}{$\partial^2_{s_{yy}}$} & \multirow{2}{*}{$\partial_t$} & \multirow{2}{*}{$\partial_t\partial_{s_x}$} & \multirow{2}{*}{$\partial_t\partial_{s_y}$} & \multirow{2}{*}{$\partial_t\partial^2_{s_{xx}}$} & \multirow{2}{*}{$\partial_t\partial^2_{s_{xy}}$} & \multirow{2}{*}{$\partial_t\partial^2_{s_{yy}}$} & \multirow{2}{*}{$\partial^2_t$} & \multirow{2}{*}{$\partial^2_t\partial_{s_x}$} & \multirow{2}{*}{$\partial^2_t\partial_{s_y}$} & \multirow{2}{*}{$\partial^2_t\partial^2_{s_{xx}}$} & \multirow{2}{*}{$\partial^2_t\partial^2_{s_{xy}}$} & \multirow{2}{*}{$\partial^2_t\partial^2_{s_{yy}}$} \\
	&&&&&&&&&&&&&&&&&&&\\
	\hline
	\multirow{6}{*}{30} & \multirow{2}{*}{3} & 1.33 & 37.56 & 21.29 & 305.51 & 129.57 & 248.14 & 1.44 & 6.36 & 10.34 & 67.82 & 35.76 & 168.31 & 0.87 & 3.84 & 5.98 & 31.62 & 16.37 & 71.05 \\ 
	& & (0.73, 2.10) & (3.56), (0.8) & (2.71), (1.0) & (29.05), (1.0) & (11.66), (1.0) & (32.94), (1.0) & (0.09), (0.9) & (0.56), (0.9) & (0.82), (0.7) & (6.18), (1.0) & (4.61), (1.0) & (5.69), (0.9) & (0.11), (0.9) & (0.59), (1.0) & (0.60), (1.0) & (5.82), (1.0) & (2.46), (1.0) & (3.18), (1.0)\\
	& \multirow{2}{*}{6} & 1.08 & 35.56 & 21.00 & 291.35 & 124.83 & 245.35 & 1.16 & 5.79 & 8.11 & 71.05 & 36.66 & 159.98 & 0.57 & 2.25 & 4.04 & 22.90 & 12.32 & 66.26 \\ 
	& & (0.76, 1.48) & (2.00), (0.9) & (1.30), (1.0) & (23.50), (1.0) & (13.89), (1.0) & (19.94), (1.0) & (0.20), (0.9) & (0.37), (0.8) & (0.81), (0.7) & (7.92), (1.0) & (2.65), (1.0) & (6.75), (0.7) & (0.05), (0.8) & (0.35), (1.0) & (0.46), (0.9) & (3.61), (1.0) & (1.44), (1.0) & (2.73), (0.8) \\ 
	& \multirow{2}{*}{9} & 1.03 & 33.48 & 20.81 & 277.93 & 128.89 & 235.40 & 1.05 & 5.18 & 7.75 & 51.83 & 35.52 & 98.28 & 0.57 & 2.76 & 4.55 & 28.54 & 17.71 & 53.36 \\ 
	& & (0.77, 1.30) & (3.94), (0.8) & (2.09), (1.0) & (31.63), (1.0) & (11.03), (1.0) & (17.01), (1.0) & (0.16), (0.9) & (0.33), (0.9) & (1.07), (0.7) & (7.98), (1.0) & (1.48), (1.0) & (8.84), (0.7) & (0.07), (0.7) & (0.26), (1.0) & (0.39), (0.8) & (4.27), (1.0) & (1.41), (1.0) & (3.59), (0.8) \\ 
	\hline
	\multirow{6}{*}{50} & \multirow{2}{*}{3} & 1.12 & 26.25 & 17.13 & 214.23 & 116.46 & 214.19 & 1.15 & 5.60 & 8.64 & 56.49 & 32.60 & 173.11 & 0.82 & 3.88 & 5.57 & 25.61 & 15.06 & 71.03 \\ 
	& & (0.75 1.58) & (2.88), (0.9) & (2.23), (1.0) & (15.61), (1.0) & (8.72), (1.0) & (25.15), (1.0) & (0.20), (1.0) & (0.77), (0.8) & (0.76), (0.7) & (2.57), (1.0) & (3.22), (1.0) & (3.11), (0.6) & (0.11), (0.9) & (0.64), (1.0) & (0.36), (0.9) & (2.55), (1.0) & (1.88), (1.0) & (1.60), (1.0)  \\ 
	& \multirow{2}{*}{6} & {\bf 1.01} & {\bf 24.18} & {\bf 16.98} & {\bf 203.51} & {\bf 114.78} & {\bf 207.63} & {\bf 0.77} & {\bf 4.69} & {\bf 6.47} & {\bf 58.38} & {\bf 30.86} & {\bf 167.77} & {\bf 0.49} & {\bf 2.10} & {\bf 3.24} & {\bf 19.36} & {\bf 11.00} & {\bf 68.03} \\ 
	& & (0.79, 1.31) & (3.51), (0.9) & (1.62), (1.0) & (25.25), (1.0) & (11.23), (1.0) & (14.40), (1.0) & (0.19), (1.0) & (0.43), (0.8) & (0.98), (0.7) & (4.38), (1.0) & (1.66), (1.0) & (5.85), (0.5) & (0.05), (0.8) & (0.21), (1.0) & (0.38), (1.0) & (2.00), (1.0) & (0.77), (1.0) & (2.11), (0.6) \\
	& \multirow{2}{*}{9} & 0.99 & 23.21 & 16.33 & 189.47 & 109.55 & 208.69 & 0.71 & 4.24 & 6.33 & 52.98 & 28.38 & 151.37 & 0.54 & 2.40 & 4.03 & 23.38 & 13.06 & 74.76 \\ 
	& & (0.80, 1.22) & (3.05), (0.9) & (1.69), (1.0) & (15.57), (1.0) & (8.48), (1.0) & (13.58), (1.0) & (0.07), (1.0) & (0.31), (0.8) & (0.39), (0.7) & (2.64), (1.0) & (2.39), (1.0) & (2.77), (0.6) & (0.04), (0.7) & (0.19), (1.0) & (0.22), (0.9) & (1.31), (1.0) & (1.15), (1.0) & (1.23), (0.5)  \\ \hline
	\multirow{6}{*}{100} & \multirow{2}{*}{3} & {\bf 0.95} & {\bf 13.97} & {\bf 14.17} & {\bf 142.00} & {\bf 95.28} & {\bf 190.14} & {\bf 0.84} & {\bf 5.28} & {\bf 7.69} & {\bf 50.50} & {\bf 28.71} & {\bf 178.57} & {\bf 0.81} & {\bf 4.18} & {\bf 5.76} & {\bf 23.38} & {\bf 14.44} & {\bf 71.50} \\
	& & (0.73, 1.21) & (2.41), (1.0) & (1.31), (1.0) & (17.80), (1.0) & (7.81), (1.0) & (7.43), (0.9) & (0.10), (1.0) & (0.60), (0.8) & (1.00), (0.6) & (3.34), (1.0) & (2.41), (1.0) & (3.99), (0.5) & (0.07), (0.9) & (0.50), (1.0) & (0.45), (0.9) & (2.63), (1.0) & (1.33), (1.0) & (2.52), (0.8)\\
	& \multirow{2}{*}{6} & 0.96 & 15.34 & 14.00 & 146.19 & 95.03 & 192.18 & 0.59 & 4.10 & 5.99 & 53.76 & 27.00 & 174.14 & 0.43 & 2.07 & 2.96 & 17.92 & 10.31 & 69.52 \\ 
	& & (0.81, 1.12) & (2.48), (1.0) & (1.63), (1.0) & (12.21), (1.0) & (8.36), (1.0) & (12.63), (0.9) & (0.17), (1.0) & (0.49), (0.8) & (0.66), (0.7) & (1.63), (1.0) & (2.64), (1.0) & (2.94), (0.4) & (0.09), (0.8) & (0.34), (1.0) & (0.23), (1.0) & (1.14), (1.0) & (0.86), (1.0) & (1.20), (0.5) \\ 
	& \multirow{2}{*}{9} & 0.97 & 13.23 & 13.31 & 129.31 & 88.65 & 194.43 & 0.51 & 3.45 & 5.61 & 47.13 & 22.84 & 158.94 & 0.45 & 2.02 & 3.58 & 20.77 & 10.94 & 77.17 \\ 
	& & (0.86, 1.10) & (3.13), (1.0) & (0.77), (1.0) & (19.66), (1.0) & (6.42), (1.0) & (3.43), (0.9) & (0.03), (1.0) & (0.37), (0.8) & (0.39), (0.6) & (1.38), (1.0) & (2.91), (1.0) & (1.83), (0.4) & (0.03), (0.8) & (0.18), (1.0) & (0.20), (0.9) & (0.85), (1.0) & (1.15), (1.0) & (0.82), (0.4) \\  
	\hline
	\hline
\end{tabular}
}

\medskip

\caption{Goodness-of-fit for the spatiotemporal differential processes for simulation experiments performed using Pattern 2.}\label{tab::sim-4}
\resizebox{\linewidth}{!}{
\centering
\begin{tabular}{*{2}{c}|c|@{\extracolsep{5pt}}*{17}{c}@{}}
	\hline
	\hline
	\multicolumn{2}{c|}{\multirow{2}{*}{$N$}} & \multirow{2}{*}{$\tau^2$} & \multicolumn{17}{c}{\multirow{2}{*}{RMSE}} \\
	&&&&&&&&&&&&&&&&&\\\cline{1-2}\cline{4-20}
	\multirow{2}{*}{$N_s$} & \multirow{2}{*}{$N_t$} & \multirow{2}{*}{(Truth = 1.00)}  & \multirow{2}{*}{$\partial_{s_x}$} & \multirow{2}{*}{$\partial_{s_y}$} & \multirow{2}{*}{$\partial^2_{s_{xx}}$} & \multirow{2}{*}{$\partial^2_{s_{xy}}$} & \multirow{2}{*}{$\partial^2_{s_{yy}}$} & \multirow{2}{*}{$\partial_t$} & \multirow{2}{*}{$\partial_t\partial_{s_x}$} & \multirow{2}{*}{$\partial_t\partial_{s_y}$} & \multirow{2}{*}{$\partial_t\partial^2_{s_{xx}}$} & \multirow{2}{*}{$\partial_t\partial^2_{s_{xy}}$} & \multirow{2}{*}{$\partial_t\partial^2_{s_{yy}}$} & \multirow{2}{*}{$\partial^2_t$} & \multirow{2}{*}{$\partial^2_t\partial_{s_x}$} & \multirow{2}{*}{$\partial^2_t\partial_{s_y}$} & \multirow{2}{*}{$\partial^2_t\partial^2_{s_{xx}}$} & \multirow{2}{*}{$\partial^2_t\partial^2_{s_{xy}}$} & \multirow{2}{*}{$\partial^2_t\partial^2_{s_{yy}}$} \\
	&&&&&&&&&&&&&&&&&&&\\
	\hline
	\multirow{6}{*}{30} & \multirow{2}{*}{3} & 1.44 & 25.37 & 18.80 & 254.07 & 209.71 & 235.54 & 1.40 & 14.72 & 11.27 & 121.13 & 122.14 & 121.24 & 0.70 & 7.58 & 5.36 & 67.89 & 51.24 & 68.25 \\ 
	& & (0.81, 2.33) & (1.83), (0.8) & (2.26), (0.9) & (23.06), (1.0) & (21.55), (1.0) & (25.56), (1.0) & (0.09), (0.8) & (0.87), (0.7) & (0.98), (0.8) & (8.05), (1.0) & (9.03), (1.0) & (7.98), (1.0) & (0.07), (1.0) & (0.57), (1.0) & (0.38), (1.0) & (7.70), (1.0) & (3.36), (1.0) & (7.50), (1.0)\\
	& \multirow{2}{*}{6} & 1.03 & 25.66 & 19.58 & 253.65 & 215.64 & 242.26 & 1.29 & 13.55 & 10.32 & 114.59 & 113.97 & 114.21 & 0.53 & 6.82 & 4.31 & 50.17 & 44.69 & 49.70 \\ 
	& & (0.58, 1.71) & (2.44), (0.9) & (2.46), (1.0) & (22.31), (1.0) & (26.27), (1.0) & (25.41), (1.0) & (0.10), (0.8) & (1.11), (0.7) & (1.03), (0.8) & (14.22), (1.0) & (13.85), (1.0) & (14.30), (1.0) & (0.03), (0.9) & (0.48), (0.9) & (0.52), (1.0) & (6.94), (1.0) & (5.22), (1.0) & (6.86), (1.0) \\ 
	& \multirow{2}{*}{9} & 1.00 & 28.38 & 21.12 & 289.47 & 234.66 & 278.79 & 1.18 & 12.49 & 9.42 & 104.09 & 103.96 & 104.04 & 0.61 & 7.25 & 4.86 & 55.92 & 51.81 & 55.61 \\ 
	& & (0.71, 1.34) & (1.40), (0.8) & (1.32), (0.9) & (20.16), (1.0) & (10.57), (1.0) & (19.08), (1.0) & (0.07), (0.8) & (0.55), (0.7) & (0.50), (0.8) & (4.77), (1.0) & (4.20), (1.0) & (5.07), (1.0) & (0.03), (0.8) & (0.29), (0.8) & (0.16), (0.9) & (2.89), (1.0) & (1.61), (1.0) & (2.69), (1.0) \\  
	\hline
	\multirow{6}{*}{50} & \multirow{2}{*}{3} & 1.21 & 20.29 & 14.44 & 211.83 & 169.11 & 201.95 & 1.21 & 12.81 & 9.59 & 103.66 & 106.34 & 103.38 & 0.88 & 9.40 & 6.86 & 70.72 & 56.68 & 71.60 \\ 
	& & (0.79, 1.77) & (1.88), (0.9) & (1.35), (1.0) & (20.24), (1.0) & (14.61), (1.0) & (20.83), (1.0) & (0.12), (0.9) & (1.05), (0.7) & (0.77), (0.8) & (9.43), (1.0) & (7.56), (1.0) & (8.91), (1.0) & (0.20), (0.9) & (1.42), (1.0) & (1.55), (1.0) & (14.18), (1.0) & (9.40), (1.0) & (14.41), (1.0) \\
	& \multirow{2}{*}{6} & {\bf 0.93} & {\bf 20.99} & {\bf 14.78} & {\bf 226.43} & {\bf 177.90} & {\bf 209.76} & {\bf 0.99} & {\bf 11.25} & {\bf 7.90} & {\bf 86.48} & {\bf 93.25} & {\bf 86.49} & {\bf 0.47} & {\bf 7.70} & {\bf 3.65} & {\bf 40.47} & {\bf 37.52} & {\bf 40.64} \\ 
	& & (0.73, 1.18) & (2.30), (0.9) & (1.69), (1.0) & (22.22), (1.0) & (14.31), (1.0) & (24.48), (1.0) & (0.16), (0.9) & (1.27), (0.7) & (1.03), (0.8) & (13.22), (1.0) & (8.46), (0.9) & (12.95), (1.0) & (0.05), (0.9) & (0.33), (0.9) & (0.52), (1.0) & (6.21), (1.0) & (3.42), (1.0) & (6.42), (1.0) \\
	& \multirow{2}{*}{9} & 0.97 & 23.59 & 16.58 & 246.29 & 199.07 & 240.90 & 0.95 & 10.63 & 7.58 & 83.52 & 88.77 & 83.32 & 0.55 & 7.86 & 4.33 & 47.88 & 46.01 & 47.87 \\ 
	& & (0.77, 1.19) & (1.40), (0.9) & (1.53), (1.0) & (18.00), (1.0) & (14.80), (1.0) & (22.47), (1.0) & (0.07), (0.9) & (0.56), (0.7) & (0.63), (0.8) & (5.50), (1.0) & (5.96), (0.9) & (5.64), (1.0) & (0.04), (0.8) & (0.32), (0.8) & (0.32), (1.0) & (3.86), (1.0) & (2.80), (1.0) & (3.78), (1.0) \\ \hline
	\multirow{6}{*}{100} & \multirow{2}{*}{3} & {\bf 0.99} & {\bf 15.72} & {\bf 9.97} & {\bf 188.76} & {\bf 136.03} & {\bf 162.48} & {\bf 0.93} & {\bf 10.64} & {\bf 7.62} & {\bf 85.59} & {\bf 89.24} & {\bf 85.39} & {\bf 1.02} & {\bf 11.62} & {\bf 8.55} & {\bf 71.17} & {\bf 69.37} & {\bf 71.97} \\
	& & (0.72, 1.26) & (1.09), (0.9) & (0.74), (1.0) & (14.06), (1.0) & (6.61), (1.0) & (12.92), (1.0) & (0.07), (1.0) & (0.65), (0.7) & (0.39), (0.8) & (5.08), (1.0) & (4.65), (1.0) & (4.94), (1.0) & (0.18), (0.9) & (1.60), (0.9) & (1.69), (1.0) & (11.70), (1.0) & (12.32), (1.0) & (11.63), (1.0) \\
	& \multirow{2}{*}{6} & 0.98 & 15.48 & 10.49 & 187.59 & 136.70 & 167.15 & 0.67 & 8.14 & 5.36 & 67.31 & 67.87 & 67.22 & 0.44 & 8.46 & 3.05 & 28.61 & 28.52 & 28.42 \\ 
	& & (0.81, 1.18) & (0.96), (1.0) & (0.68), (1.0) & (17.55), (1.0) & (7.98), (1.0) & (23.27), (1.0) & (0.05), (1.0) & (0.47), (0.7) & (0.28), (0.8) & (3.28), (1.0) & (3.93), (0.9) & (3.34), (1.0) & (0.07), (0.9) & (0.38), (0.8) & (0.23), (1.0) & (1.62), (1.0) & (1.94), (1.0) & (1.77), (1.0) \\ 
	& \multirow{2}{*}{9} & 0.96 & 17.23 & 11.12 & 201.00 & 145.78 & 176.16 & 0.68 & 8.04 & 5.27 & 67.68 & 66.18 & 67.53 & 0.50 & 8.48 & 3.62 & 38.68 & 37.19 & 38.73 \\ 
	& & (0.82, 1.11) & (1.95), (0.9) & (0.91), (1.0) & (12.29), (1.0) & (11.79), (1.0) & (12.47), (1.0) & (0.06), (1.0) & (0.57), (0.7) & (0.30), (0.8) & (7.36), (1.0) & (4.73), (0.9) & (7.26), (1.0) & (0.03), (0.8) & (0.24), (0.7) & (0.19), (1.0) & (1.85), (1.0) & (1.51), (0.9) & (1.75), (1.0) \\  
	\hline
	\hline
\end{tabular}
}
\end{table}

\end{landscape}
\restoregeometry

\newpage
\clearpage
\begin{table}[h!]
\centering
\caption{True and estimated \emph{total wombling measure} split by time-intervals. They are accompanied by HPD intervals. Estimates that are not significant are marked in bold.}\label{tab:wmbl-timesplit-1}
\resizebox{\linewidth}{!}{
\begin{tabular}{l|c|@{\extracolsep{40pt}}*{6}{c}@{}}
\hline\hline
\multirow{2}{*}{$\bGamma(\C^*)$} & \multirow{2}{*}{Time}  & \multicolumn{2}{c}{A} & \multicolumn{2}{c}{B} &   \multicolumn{2}{c}{C} \\ 
\cline{3-4}\cline{5-6}\cline{7-8}
& & Estimate & True & Estimate & True & Estimate & True\\
\hline
\multirow{16}{*}{$\bpartial_\bs$} & \multirow{2}{*}{1--2} & 1.02 &  \multirow{2}{*}{1.03} & -1.31 &  \multirow{2}{*}{-1.30} & 0.88 &  \multirow{2}{*}{0.99} \\ 
& & (0.78, 1.26) & & (-1.60, -1.04) & & (0.21, 1.55) & \\
& \multirow{2}{*}{2--3} & 0.51 &  \multirow{2}{*}{0.49} & -1.14 &  \multirow{2}{*}{-1.09} & {\bf 0.25} &  \multirow{2}{*}{0.23} \\ 
& & (0.36, 0.67) & & (-1.41, -0.90) & & (-0.21, 0.69) & \\
& \multirow{2}{*}{3--4} & 1.68 & \multirow{2}{*}{1.75} & -1.91 & \multirow{2}{*}{-1.93} & {\bf 0.02} & \multirow{2}{*}{0.03} \\
& & (1.17, 2.21) & & (-2.52, -1.32) & & (-0.67, 0.76) & \\
& \multirow{2}{*}{4--5} & 1.90 & \multirow{2}{*}{1.98} & -1.86 & \multirow{2}{*}{-1.89} & {\bf -0.05} & \multirow{2}{*}{-0.01} \\
& & (1.29, 2.52) & & (-2.44, -1.27) & & (-0.93, 0.83) & \\
& \multirow{2}{*}{5--6} & 1.52 & \multirow{2}{*}{1.55} & -1.63 & \multirow{2}{*}{-1.55} & 1.02 & \multirow{2}{*}{1.16} \\ 
& & (1.21, 1.82) & & (-1.92, -1.34) & & (0.37, 1.71) & \\
& \multirow{2}{*}{6--7} & 1.17 & \multirow{2}{*}{1.17} & -1.30 & \multirow{2}{*}{-1.23} & 3.97 & \multirow{2}{*}{4.16} \\ 
& & (0.95, 1.40) & & (-1.52, -1.06) & & (3.35, 4.56) & \\
& \multirow{2}{*}{7--8} & 0.58 & \multirow{2}{*}{0.57} & -0.86 & \multirow{2}{*}{-0.82} & 2.79 & \multirow{2}{*}{2.89} \\ 
& & (0.43, 0.73) & & (-1.03, -0.69) & & (2.29, 3.27) & \\
& \multirow{2}{*}{8--9} & 0.72 & \multirow{2}{*}{0.68} & -1.03 & \multirow{2}{*}{-1.00} & 0.85 & \multirow{2}{*}{0.87} \\
& & (0.53, 0.89) & & (-1.25, -0.82) & & (0.35, 1.35) & \\\hline 
\multirow{16}{*}{$\bpartial_\bs^2$} & \multirow{2}{*}{1--2} & 15.24 & \multirow{2}{*}{16.41} & -19.02 & \multirow{2}{*}{-18.50} & {\bf 1.20} & \multirow{2}{*}{1.49} \\ 
& & (12.48, 18.07) & & (-22.32, -15.66) & & (-4.77, 7.24) & \\
& \multirow{2}{*}{2--3} & 8.57 & \multirow{2}{*}{9.11} & -14.27 & \multirow{2}{*}{-14.11} & {\bf 2.43} & \multirow{2}{*}{2.37} \\ 
& & (6.94, 10.33) & & (-17.26, -11.28) & & (-1.51, 6.31) & \\
& \multirow{2}{*}{3--4} & 36.40 & \multirow{2}{*}{36.84} & -40.26 & \multirow{2}{*}{-39.24} & {\bf 1.13} & \multirow{2}{*}{-0.52} \\ 
& & (32.34, 40.3) & & (-45.45, -35.24) & & (-5.12, 7.31) & \\
& \multirow{2}{*}{4--5} & 43.87 & \multirow{2}{*}{44.34} & -42.07 & \multirow{2}{*}{-41.26} & {\bf -2.65} & \multirow{2}{*}{-5.46} \\ 
& & (38.98, 48.86) & & (-47.98, -35.68) & & (-9.67, 5.01) & \\
& \multirow{2}{*}{5--6} & 20.53 & \multirow{2}{*}{20.47} & -20.04 & \multirow{2}{*}{-19.99} & -7.01 & \multirow{2}{*}{-10.18} \\ 
& & (17.56, 23.55) & & (-23.43, -16.58) & & (-11.6, -2.2) & \\
& \multirow{2}{*}{6--7} & 17.04 & \multirow{2}{*}{16.92} & -17.11 & \multirow{2}{*}{-17.02} & {\bf 3.87} & \multirow{2}{*}{2.37} \\ 
& & (14.66, 19.39) & & (-19.87, -14.33) & & (-1.82, 9.57) & \\
& \multirow{2}{*}{7--8} & 11.84 & \multirow{2}{*}{11.69} & -13.07 & \multirow{2}{*}{-13.11} & {\bf 2.44} & \multirow{2}{*}{0.59} \\ 
& & (10.31, 13.34) & & (-15.12, -10.95) & & (-2.34, 7.34) & \\
& \multirow{2}{*}{8--9} & 12.31 & \multirow{2}{*}{12.10} & -13.35 & \multirow{2}{*}{-14.00} & {\bf -2.80} & \multirow{2}{*}{-4.79} \\ 
& & (10.48, 14.15) & & (-15.99, -10.71) & & (-7.46, 1.77) & \\\hline
\multirow{16}{*}{$\partial_t$} & \multirow{2}{*}{1--2} & 0.11 & \multirow{2}{*}{0.09} & -0.09 & \multirow{2}{*}{-0.10} & {\bf -0.06} & \multirow{2}{*}{-0.07}\\ 
& & (0.07, 0.14) & & (-0.13, -0.05) & & (-0.16, 0.04) & \\
& \multirow{2}{*}{2--3} & 0.04 & \multirow{2}{*}{0.04} & -0.06 & \multirow{2}{*}{-0.06} & {\bf 0.01} & \multirow{2}{*}{0.01}\\ 
& & (0.03, 0.05) & & (-0.08, -0.03) & & (-0.03, 0.05) & \\
& \multirow{2}{*}{3--4} & 1.05 & \multirow{2}{*}{1.00} & 1.25 & \multirow{2}{*}{1.22} & -1.51 & \multirow{2}{*}{-1.55}\\
& & (0.78, 1.32) & & (0.88, 1.60) & & (-1.95, -1.02) & \\
& \multirow{2}{*}{4--5} & -0.42 & \multirow{2}{*}{-0.38} & 0.38 & \multirow{2}{*}{0.36} & -0.32 & \multirow{2}{*}{-0.26}\\
& & (-0.59, -0.26) & & (0.22, 0.56) & & (-0.58, -0.05) & \\
& \multirow{2}{*}{5--6} & -0.04 & \multirow{2}{*}{-0.04} & 0.05 & \multirow{2}{*}{0.04} & 0.04 & \multirow{2}{*}{0.03}\\
& & (-0.05, -0.03) & & (0.03, 0.06) & & (0.00, 0.07) & \\
& \multirow{2}{*}{6--7} & {\bf -0.01} & \multirow{2}{*}{-0.02} & 0.02 & \multirow{2}{*}{0.02} & {\bf 0.03} & \multirow{2}{*}{0.02}\\
& & (-0.03, 0.00) & &  (0.01, 0.03) & & (-0.01, 0.06) & \\
& \multirow{2}{*}{7--8} & 0.02 & \multirow{2}{*}{0.01} & {\bf -0.02} & \multirow{2}{*}{-0.02} & {\bf 0.02} & \multirow{2}{*}{0.02}\\ 
& & (0.01, 0.03) & & (-0.03, 0.00) & & (-0.02, 0.06) & \\
& \multirow{2}{*}{8--9} & -0.07 & \multirow{2}{*}{-0.08} & 0.09 & \multirow{2}{*}{0.10} & {\bf -0.06} & \multirow{2}{*}{-0.11}\\ 
& & (-0.09, -0.06) & & (0.07, 0.11) & & (-0.12, 0.00) & \\\hline
\multirow{16}{*}{$\partial_t\bpartial_\bs$} & \multirow{2}{*}{1--2} & -0.20 & \multirow{2}{*}{-0.27} & 0.23 & \multirow{2}{*}{0.34} & -0.80 & \multirow{2}{*}{-0.95} \\ 
& & (-0.37, -0.02) & & (0.02, 0.42) & & (-1.26, -0.31) & \\ 
& \multirow{2}{*}{2--3} & -0.13 & \multirow{2}{*}{-0.11} & 0.27 & \multirow{2}{*}{0.20} & 0.23 & \multirow{2}{*}{0.16} \\ 
& & (-0.20, -0.07) & & (0.18, 0.37) & & (0.07, 0.40) & \\ 
& \multirow{2}{*}{3--4} & 3.56 & \multirow{2}{*}{3.62} & 3.95 & \multirow{2}{*}{4.15} & -9.42 & \multirow{2}{*}{-9.54} \\ 
& & (2.76, 4.37) & & (2.87, 4.95) & & (-11.35, -7.72) & \\
& \multirow{2}{*}{4--5} & 1.98 & \multirow{2}{*}{2.03} & -0.95 & \multirow{2}{*}{-1.14} & 4.13 & \multirow{2}{*}{3.98} \\ 
& & (1.26, 2.67) & & (-1.61, -0.32) & & (3.11, 5.13) & \\ 
& \multirow{2}{*}{5--6} & 0.34 & \multirow{2}{*}{0.28} & -0.26 & \multirow{2}{*}{-0.26} & {\bf 0.07} & \multirow{2}{*}{0.08} \\ 
& & (0.28, 0.39) & & (-0.31, -0.21) & & (-0.03, 0.18) & \\  
& \multirow{2}{*}{6--7} & 0.11 & \multirow{2}{*}{0.05} & {\bf -0.03} & \multirow{2}{*}{-0.05} & 0.24 & \multirow{2}{*}{0.19} \\ 
& & (0.05, 0.17) & & (-0.08, 0.03) & & (0.09, 0.38) & \\  
& \multirow{2}{*}{7--8} & {\bf -0.01} & \multirow{2}{*}{-0.03} & 0.06 & \multirow{2}{*}{0.05} & -0.32 & \multirow{2}{*}{-0.23} \\ 
& & (-0.08, 0.05) & & (0.00, 0.13) & & (-0.52, -0.12) & \\ 
& \multirow{2}{*}{8--9} & 0.26 & \multirow{2}{*}{0.29} & -0.40 & \multirow{2}{*}{-0.37} & {\bf 0.01} & \multirow{2}{*}{-0.00} \\ 
& & (0.15, 0.36) & & (-0.52, -0.28) & & (-0.27, 0.28) \\  \hline\hline
\end{tabular}
}
\end{table}

\newpage
\clearpage
\begin{table}[ht!]
\centering
\caption{True and estimated total wombling measure split by time intervals. They are accompanied by HPD intervals. Estimates that are not significant are marked in bold. Estimates not containing true values are marked in \texttt{teal}.}\label{tab:wmbl-timesplit-2}
\resizebox{\linewidth}{!}{
\begin{tabular}{l|c|@{\extracolsep{40pt}}*{6}{c}@{}}
\hline\hline
\multirow{2}{*}{$\bGamma(\C^*)$} & \multirow{2}{*}{Time}  & \multicolumn{2}{c}{A} & \multicolumn{2}{c}{B} & \multicolumn{2}{c}{C} \\ 
\cline{3-4}\cline{5-6}\cline{7-8}
& & Estimate & True & Estimate & True & Estimate & True\\
\hline
\multirow{16}{*}{$\partial_t\bpartial_\bs^2$} & \multirow{2}{*}{1--2} & -5.08 & \multirow{2}{*}{-4.33} & 4.15 & \multirow{2}{*}{4.64} & {\bf 6.58} & \multirow{2}{*}{6.32} \\ 
& & (-8.33, -2.2) & & (0.45, 7.53) & & (-0.11, 13.49) & \\
& \multirow{2}{*}{2--3} & {\bf -0.82} & \multirow{2}{*}{-0.71} & {\bf -0.69} & \multirow{2}{*}{0.51} & {\bf -2.68} & \multirow{2}{*}{-4.40} \\ 
& & (-2.08, 0.47) & & (-2.92, 1.34) & & (-5.68, 0.51) & \\
& \multirow{2}{*}{3--4} & -27.29 & \multirow{2}{*}{-27.77} & -51.48 & \multirow{2}{*}{-49.29} & 80.35 & \multirow{2}{*}{80.17} \\
& & (-42.43, -12.22) & & (-73.22, -29.60) & & (48.53, 114.35) & \\
& \multirow{2}{*}{4--5} & {\bf -2.97} & \multirow{2}{*}{-0.82} & {\bf -13.32} & \multirow{2}{*}{-8.96} & 21.83 & \multirow{2}{*}{18.22} \\ 
& & (-14.25, 8.92) & & (-27.19, 0.38) & & (3.10, 40.14) & \\
& \multirow{2}{*}{5--6} & 1.73 & \multirow{2}{*}{2.28} & -2.35 & \multirow{2}{*}{-2.37} & {\bf -1.97} & \multirow{2}{*}{-0.96} \\ 
& & (0.61, 2.90) & & (-3.61, -1.06) &  &(-4.07, 0.03) & \\
& \multirow{2}{*}{6--7} & {\bf 0.02} & \multirow{2}{*}{0.45} & {\bf -0.29} & \multirow{2}{*}{-0.52} & {\bf 1.05} & \multirow{2}{*}{-0.12} \\ 
& & (-1.05, 1.13) & & (-1.55, 0.95) & & (-1.99, 4.00) & \\
& \multirow{2}{*}{7--8} & {\bf -0.56} & \multirow{2}{*}{-0.59} & {\bf 1.14} & \multirow{2}{*}{0.68} & {\bf 1.04} & \multirow{2}{*}{-0.46} \\ 
& & (-1.69, 0.56) & & (-0.36, 2.62) & & (-2.64, 4.54) & \\
& \multirow{2}{*}{8--9} & 2.61 & \multirow{2}{*}{2.88} & -3.44 & \multirow{2}{*}{-3.70} & {\bf 3.33} & \multirow{2}{*}{5.63} \\ 
& & (0.85, 4.33) & & (-5.87, -1.12) & & (-1.27, 7.73) & \\\hline
\multirow{16}{*}{$\partial_t^2$} & \multirow{2}{*}{1--2} & 0.21 & \multirow{2}{*}{0.30} & -0.35 & \multirow{2}{*}{-0.33} & {\color{teal} -0.53} & \multirow{2}{*}{{\color{teal}-0.20}} \\
& & (0.13, 0.30) & & (-0.46, -0.26) & & {\color{teal}(-0.77, -0.28)} & \\
& \multirow{2}{*}{2--3} & 0.03 & \multirow{2}{*}{0.05} & -0.07 & \multirow{2}{*}{-0.07} & {\color{teal}-0.05} & \multirow{2}{*}{{\color{teal} -0.00}} \\ 
& & (0.02, 0.05) & & (-0.09, -0.04) & & {\color{teal} (-0.09, -0.01)} & \\
& \multirow{2}{*}{3--4} & {\bf 0.91} & \multirow{2}{*}{0.89} & {\bf 0.88} & \multirow{2}{*}{-1.39} & {\bf 1.87} & \multirow{2}{*}{0.79} \\ 
& & (-0.62, 2.50) & & (-1.22, 3.08) & & (-1.86, 5.65) & \\
& \multirow{2}{*}{4--5} & 1.74 & \multirow{2}{*}{1.43} & -1.52 & \multirow{2}{*}{-1.27} & {\color{teal} 1.92} & \multirow{2}{*}{{\color{teal} 0.67}} \\ 
& & (1.30, 2.21) & & (-1.97, -1.05) & & {\color{teal} (1.10, 2.70)} & \\
& \multirow{2}{*}{5--6} & 0.05 & \multirow{2}{*}{0.04} & -0.04 & \multirow{2}{*}{-0.04} & {\color{teal} 0.01} & \multirow{2}{*}{{\color{teal}-0.01}} \\ 
& & (0.04, 0.05) & & (-0.05, -0.04) & & {\color{teal} (0.00, 0.02)} & \\
& \multirow{2}{*}{6--7} & 0.08 & \multirow{2}{*}{0.08} & -0.08 & \multirow{2}{*}{-0.08} & -0.08 & \multirow{2}{*}{-0.08} \\ 
& & (0.08, 0.09) & & (-0.09, -0.08) & & (-0.10, -0.06) & \\
& \multirow{2}{*}{7--8} & 0.11 & \multirow{2}{*}{0.11} & -0.11 & \multirow{2}{*}{-0.11} & 0.13 & \multirow{2}{*}{0.16} \\ 
& & (0.10, 0.13) & & (-0.12, -0.09) & & (0.08, 0.18) & \\
& \multirow{2}{*}{8--9} & 0.10 & \multirow{2}{*}{0.11} & -0.12 & \multirow{2}{*}{-0.14} & {\bf \color{teal}0.08} & \multirow{2}{*}{\color{teal} 0.27} \\ 
& & (0.06, 0.13) & & (-0.17, -0.08) & & {\color{teal}(-0.03, 0.18)} & \\\hline
\multirow{16}{*}{$\partial_t^2\bpartial_\bs$} & \multirow{2}{*}{1--2} & -1.01 & \multirow{2}{*}{-1.22} & 1.28 & \multirow{2}{*}{1.40} & -3.54 & \multirow{2}{*}{-3.73} \\ 
& & (-1.61, -0.43) & & (0.64, 1.98) & & (-5.15, -1.99) & \\
& \multirow{2}{*}{2--3} & {\bf -0.11} & \multirow{2}{*}{-0.09} & 0.16 & \multirow{2}{*}{0.21} & -0.38 & \multirow{2}{*}{-0.20} \\ 
& & (-0.24, 0.02) & & (-0.06, 0.39) & & (-0.73, -0.02) & \\
& \multirow{2}{*}{3--4} & {\bf -7.08} & \multirow{2}{*}{-8.89} & {\bf 9.30} & \multirow{2}{*}{13.25} & {\bf 1.92} & \multirow{2}{*}{-3.03} \\ 
& & (-22.16, 8.22) & & (-12.68, 30.49) & & (-32.64, 37.94) & \\ 
& \multirow{2}{*}{4--5} & -6.69 & \multirow{2}{*}{-5.87} & 5.57 & \multirow{2}{*}{5.84} & {\bf 2.27} & \multirow{2}{*}{1.29} \\ 
& & (-11.7, -1.75) & & (0.79, 10.51) & & (-5.12, 9.58) & \\ 
& \multirow{2}{*}{5--6} & -0.24 & \multirow{2}{*}{-0.23} & 0.27 & \multirow{2}{*}{0.22} & -0.20 & \multirow{2}{*}{-0.21} \\ 
& & (-0.3, -0.17) & & (0.21, 0.33) & & (-0.34, -0.07) & \\
& \multirow{2}{*}{6--7} & -0.24 & \multirow{2}{*}{-0.17} & 0.18 & \multirow{2}{*}{0.16} & -0.76 & \multirow{2}{*}{-0.74} \\ 
& & (-0.33, -0.15) & & (0.10, 0.28) & & (-0.99, -0.51) & \\
& \multirow{2}{*}{7--8} & -0.21 & \multirow{2}{*}{-0.27} & 0.33 & \multirow{2}{*}{0.31} & -1.68 & \multirow{2}{*}{-1.62} \\ 
& & (-0.35, -0.08) & & (0.19, 0.47) & & (-2.15, -1.24) & \\ 
& \multirow{2}{*}{8--9} & -0.35 & \multirow{2}{*}{-0.54} & 0.61 & \multirow{2}{*}{0.73} & {\bf 0.17} & \multirow{2}{*}{0.08}\\ 
& & (-0.58, -0.13) & & (0.29, 0.94) & & (-0.49, 0.89) & \\\hline
\multirow{16}{*}{$\partial_t^2\bpartial_\bs^2$} & \multirow{2}{*}{1--2} & -11.17 & \multirow{2}{*}{-16.36} & 17.66 & \multirow{2}{*}{18.03} & 18.73 & \multirow{2}{*}{14.11} \\ 
& & (-18.47, -4.07) & & (9.25, 26.29) & & (2.98, 33.62) & \\
& \multirow{2}{*}{2--3} & {\bf -0.70} & \multirow{2}{*}{-1.18} & {\bf 0.19} & \multirow{2}{*}{1.19} & {\bf -3.83} & \multirow{2}{*}{-2.15} \\ 
& & (-2.72, 1.15) & & (-3.09, 3.44) & & (-8.88, 1.03) & \\
& \multirow{2}{*}{3--4} & {\bf -38.05} & \multirow{2}{*}{-30.46} & {\bf -68.52} & \multirow{2}{*}{53.02} & {\bf -296.99} & \multirow{2}{*}{-48.13} \\ 
& & (-186.63, 122.86) & & (-317.43, 177.02) & & (-743.29, 129.76) & \\
& \multirow{2}{*}{4--5} & {\bf -23.93} & \multirow{2}{*}{-34.80} & {\bf 54.47} & \multirow{2}{*}{38.03} & {\bf -76.32} & \multirow{2}{*}{-11.36} \\ 
& & (-82.71, 35.55) & & (-15.23, 126.69) & & (-172.6, 20.09) & \\
& \multirow{2}{*}{5--6} & -2.10 & \multirow{2}{*}{-1.92} & 2.39 & \multirow{2}{*}{2.03} & {\bf 0.42} & \multirow{2}{*}{0.95} \\ 
& & (-2.99, -1.27) & & (1.49, 3.36) & & (-1.14, 1.98) & \\
& \multirow{2}{*}{6--7} & -1.54 & \multirow{2}{*}{-1.77} & 2.69 & \multirow{2}{*}{1.98} & {\bf 0.34} & \multirow{2}{*}{0.08} \\ 
& & (-2.71, -0.37) & & (1.37, 4.01) & & (-2.96, 3.58) & \\
& \multirow{2}{*}{7--8} & -4.79 & \multirow{2}{*}{-4.06} & 4.39 & \multirow{2}{*}{4.44} & {\bf -4.30} & \multirow{2}{*}{-4.51} \\ 
& & (-6.57, -3.1) & & (2.2, 6.45) & & (-9.61, 0.93) & \\
& \multirow{2}{*}{8--9} & -6.03 & \multirow{2}{*}{-5.23} & 4.33 & \multirow{2}{*}{7.14} & {\color{teal}\bf -6.75} & \multirow{2}{*}{\color{teal}-15.45} \\ 
& & (-8.95, -3.22) & & (0.59, 8.23) & & {\color{teal}(-14.24, 0.86)} & \\\hline
\hline
\end{tabular}
}
\end{table}

\newpage
\clearpage
\begin{table}[ht!]
\centering
\caption{Estimated total wombling measure split by time intervals for the PM\textsubscript{2.5} analysis. They are accompanied by HPD intervals. Estimates that are not significant are marked in bold.}\label{tab:wmbl-timesplit-pm25}
\resizebox{\linewidth}{!}{
\begin{tabular}{l|c|@{\extracolsep{130pt}}*{2}{c}@{}}
\hline\hline
\multirow{2}{*}{$\bGamma(\C^*)$} & \multirow{2}{*}{Time}  & June 6 -- 9 & Sep. 29 -- Oct. 2 \\ 
\cline{3-3}\cline{4-4}
& & Estimate & Estimate \\
\hline
\multirow{6}{*}{$\bpartial_\bs$} & \multirow{2}{*}{1--2} & -291.63 & -49.28 \\ 
& & (-326.07, -256.89) &  (-56.41, -42.40)\\
& \multirow{2}{*}{2--3} & -253.85  &  -35.56 \\ 
& & (-282.87, -222.63) &  (-39.02, -32.30)\\
& \multirow{2}{*}{3--4} & -9.34 & -14.41\\
& & (-12.76, -5.80) & (-16.02, -12.97)\\\hline 
\multirow{6}{*}{$\bpartial_\bs^2$} & \multirow{2}{*}{1--2} & -283.86 & -16.50\\ 
& & (-328.55, -232.16) & (-25.38, -6.52)\\
& \multirow{2}{*}{2--3} & -282.52 &  -8.37\\ 
& & (-316.84, -248.03) & (-12.64, -3.96)\\
& \multirow{2}{*}{3--4} & -25.28 & -4.31\\ 
& & (-33.33, -17.03) & (-7.65, -0.85)\\\hline
\multirow{6}{*}{$\partial_t$} & \multirow{2}{*}{1--2} & -265.87 & -57.01\\ 
& & (-289.47, -239.92) & (-63.03, -50.45)\\
& \multirow{2}{*}{2--3} & -506.81 & -73.18\\ 
& & (-582.19, -444.05) &  (-77.87, -68.36)\\
& \multirow{2}{*}{3--4} & -3.65 & -7.70\\
& & (-4.41, -2.87) & (-8.28, -7.15)\\\hline
\multirow{6}{*}{$\partial_t\bpartial_\bs$} & \multirow{2}{*}{1--2} & 8.21 & -8.98\\ 
& & (-9.09, 27.37) & (-13.46, -4.26)\\ 
& \multirow{2}{*}{2--3} & -284.20  & -2.83\\ 
& & (-318.73, -242.77) & (-5.60, -0.14)\\ 
& \multirow{2}{*}{3--4} & -1.93  & 0.95\\ 
& & (-2.87, -0.97) & (0.33, 1.60)\\ \hline
\multirow{6}{*}{$\partial_t\bpartial_\bs^2$} & \multirow{2}{*}{1--2} & 107.04 & 10.24\\ 
& & (74.42, 139.73) & (1.81, 18.95)\\
& \multirow{2}{*}{2--3} & 145.63 & 27.88\\ 
& & (96.49, 195.38) &  (22.76, 33.35)\\
& \multirow{2}{*}{3--4} & 2.77 &  2.99\\
& & (0.28, 5.10 ) & (1.17, 4.81)\\ \hline
\multirow{6}{*}{$\partial_t^2$} & \multirow{2}{*}{1--2} & -182.88 & -30.78\\
& & (-217.80, -144.68) & (-39.36, -22.05)\\
& \multirow{2}{*}{2--3} & -1928.07 & -58.78\\ 
& & (-2037.92, -1790.01) & (-67.17, -51.10)\\
& \multirow{2}{*}{3--4} & -0.14 & -0.58\\ 
& & (-0.55, 0.34) & (-1.11, -0.05)\\\hline
\multirow{6}{*}{$\partial_t^2\bpartial_\bs$} & \multirow{2}{*}{1--2} & 126.95 & 28.75\\ 
& & (95.55, 153.08) & (21.03, 36.31)\\
& \multirow{2}{*}{2--3} & 302.13 & 63.95\\ 
& & (206.92, 434.72) & (56.78, 71.52)\\
& \multirow{2}{*}{3--4} & 1.11 &  3.84\\ 
& & (0.56, 1.60) &  (3.01, 4.70)\\\hline
\multirow{6}{*}{$\partial_t^2\bpartial_\bs^2$} & \multirow{2}{*}{1--2} & 86.33 & 10.67\\ 
& & (36.98, 133.26) & (-4.09, 26.08)\\
& \multirow{2}{*}{2--3} & 316.79 &  -5.90\\ 
& & (205.67, 424.09) & (-15.84, 4.17)\\
& \multirow{2}{*}{3--4} & 1.29 &  1.80\\ 
& & (-0.07, 2.44) & (0.14, 3.87)\\\hline
\hline
\end{tabular}
}
\end{table}

\newpage
\clearpage
\begin{table}[ht!]
\centering
\caption{Estimated total wombling measures ($\times 10^2$) split by time intervals for the EEG analysis. They are accompanied by HPD intervals. Estimates that are not significant are in bold.}\label{tab:wmbl-timesplit-eeg-1}
\resizebox{\linewidth}{!}{
\begin{tabular}{l|c|@{\extracolsep{130pt}}*{2}{c}@{}}
\hline\hline
\multirow{2}{*}{$\bGamma(\C^*)$} & \multirow{2}{*}{Time}  & Alcoholic & Control \\ 
\cline{3-3}\cline{4-4}
& & Estimate & Estimate \\
\hline
\multirow{10}{*}{$\bpartial_\bs$} & \multirow{2}{*}{1--2} & 180.29 & 8.35 \\ 
& & (160.08, 201.04) &  (-19.24, 34.95)\\
& \multirow{2}{*}{2--3} & 30.97  &  309.90 \\ 
& & (13.46, 51.05) &  (282.57, 338.02)\\
& \multirow{2}{*}{3--4} & -116.14 &  369.72\\
& & (-133.00, -98.24) & (342.11, 399.08)\\
& \multirow{2}{*}{4--5} & -86.06 &  307.01\\
& & (-103.93, -68.21) & (282.08, 331.66)\\\
& \multirow{2}{*}{5--6} & -79.70 &  258.40\\
& & (-97.02, -62.00) & (234.22, 281.08)\\\hline
\multirow{10}{*}{$\bpartial_\bs^2$} & \multirow{2}{*}{1--2} & -617.55 & -42.90\\ 
& & (-750.63, -477.67) & (-228.8, 155.70)\\
& \multirow{2}{*}{2--3} & -287.54 &  -565.06\\ 
& & (-415.89, -172.75) & (-753.63, -374.55)\\
& \multirow{2}{*}{3--4} & -104.61 & -755.23\\ 
& & (-214.16, 21.59) & (-943.36, -567.42)\\
& \multirow{2}{*}{4--5} & 34.67 & -471.44\\ 
& & (-62.01, 124.43) & (-650.69, -279.69)\\
& \multirow{2}{*}{5--6} & 27.33 & -707.47\\ 
& & (-84.93, 145.26) & (-894.38, -530.20)\\\hline
\multirow{6}{*}{$\partial_t$} & \multirow{2}{*}{1--2} & 1.26& -0.31\\ 
& & (0.71, 1.81) & (-1.76, 1.27)\\
& \multirow{2}{*}{2--3} &  2.29 & -7.64\\ 
& & (0.84, 3.64 ) &  (-8.45, -6.83)\\
& \multirow{2}{*}{3--4} & -1.40 & 1.11\\
& & (-1.87, -0.94) & (0.46, 1.74)\\
& \multirow{2}{*}{4--5} & -1.18 & 0.05\\
& & (-2.47, 0.10) & (-0.02, 0.12)\\
& \multirow{2}{*}{5--6} & -0.14 & -0.31\\
& & (-0.36, 0.07) & (-0.44, -0.18)\\\hline
\multirow{6}{*}{$\partial_t\bpartial_\bs$} & \multirow{2}{*}{1--2} & -5.26 &  -32.27\\ 
& & (-6.99, -3.71) & (-36.45, -27.97)\\ 
& \multirow{2}{*}{2--3} & 1.20  & -21.63\\ 
& & (-2.83, 5.16) & ( -23.80, -19.38)\\ 
& \multirow{2}{*}{3--4} & -1.83  & -6.54\\ 
& & (-3.31, -0.49) & (-8.47, -4.58)\\
& \multirow{2}{*}{4--5} & -7.66  & 0.54\\ 
& & (-10.74, -4.39) & (0.35, 0.75)\\
& \multirow{2}{*}{5--6} & -0.05  & -1.60\\ 
& & (-0.64, 0.61) & (-2.04, -1.19)\\\hline
\hline
\end{tabular}
}
\end{table}

\newpage
\clearpage
\begin{table}[ht!]
\centering
\caption{Estimated total wombling measures ($\times 10^2$) split by time intervals for the EEG analysis. They are accompanied by HPD intervals. Estimates that are not significant are in bold.}\label{tab:wmbl-timesplit-eeg-2}
\resizebox{\linewidth}{!}{
\begin{tabular}{l|c|@{\extracolsep{130pt}}*{2}{c}@{}}
\hline\hline
\multirow{2}{*}{$\bGamma(\C^*)$} & \multirow{2}{*}{Time}  & Alcoholic & Control \\ 
\cline{3-3}\cline{4-4}
& & Estimate & Estimate \\
\hline
\multirow{6}{*}{$\partial_t\bpartial_\bs^2$} & \multirow{2}{*}{1--2} & 2.88 & -65.39\\ 
& & (-10.51, 17.30)  & (-105.14, -29.25)\\
& \multirow{2}{*}{2--3} & -1.56 & 8.12\\ 
& & (-29.08,  27.45) &  (-19.80, 33.41)\\
& \multirow{2}{*}{3--4} & 14.39 &  -11.82\\
& &  (2.64, 26.79) & (-34.85, 9.23)\\
& \multirow{2}{*}{4--5} & -0.68 & 2.52\\ 
& & (-23.07, 22.79) & (0.18, 4.76)\\
& \multirow{2}{*}{5--6} & -0.03 &  -4.34\\
& & (-5.26, 4.66) & (-8.61, 0.35)\\\hline
\multirow{6}{*}{$\partial_t^2$} & \multirow{2}{*}{1--2} & 0.61 & 5.75\\
& & (0.51, 0.73) & (4.88, 6.72)\\
& \multirow{2}{*}{2--3} & 1.02 & -1.54\\ 
& & (0.56, 1.48) & (-1.67, -1.40)\\
& \multirow{2}{*}{3--4} & -0.57 & -0.95\\ 
& & (-0.63, -0.49) & (-1.02, -0.87)\\
& \multirow{2}{*}{4--5} & -1.14 & 0.00\\ 
& & (-1.71, -0.53) & (-0.00, 0.00)\\
& \multirow{2}{*}{5--6} & -0.03 & -0.04\\ 
& & (-0.06, -0.01) & (-0.06, -0.02)\\\hline
\multirow{6}{*}{$\partial_t^2\bpartial_\bs$} & \multirow{2}{*}{1--2} & -2.46 & 1.65\\ 
& & (-2.79, -2.14) & (-0.70, 4.11)\\
& \multirow{2}{*}{2--3} & 3.16 & -7.54\\ 
& & (1.05, 5.21) & (-8.67, -6.61)\\
& \multirow{2}{*}{3--4} & 1.39 &  -0.92\\ 
& & (1.05, 1.73) &  (-1.46, -0.31)\\
& \multirow{2}{*}{4--5} & -3.56 &  0.01\\ 
& & (-6.20, -1.11) &  (0.00, 0.01)\\
& \multirow{2}{*}{5--6} & -0.01 &  -0.04\\ 
& & ( -0.09, 0.05) &  (-0.06, -0.02)\\\hline
\multirow{6}{*}{$\partial_t^2\bpartial_\bs^2$} & \multirow{2}{*}{1--2} & 1.52 & -23.49\\ 
& & (-0.73, 4.14) & (-39.70, -6.29)\\
& \multirow{2}{*}{2--3} & -19.23 &  4.23\\ 
& & (-32.45, -5.83) & (-3.34, 12.63)\\
& \multirow{2}{*}{3--4} & 1.97 &  6.25\\ 
& & (-0.16, 4.29) & (1.67, 10.94)\\
& \multirow{2}{*}{4--5} & 3.44 &  -0.02\\ 
& & (-9.33, 18.73) & (-0.09, 0.04)\\
& \multirow{2}{*}{5--6} & 0.11 &  0.13\\ 
& & (-0.37, 0.60) & (-0.02, 0.31)\\\hline
\hline
\end{tabular}
}
\end{table}

\section{Figures}\label{sec:figs}

\subsection{Spatiotemporal Curvature}

\subsubsection{Pattern 2}

\begin{figure}[H]
	\centering
	\includegraphics[scale=0.78]{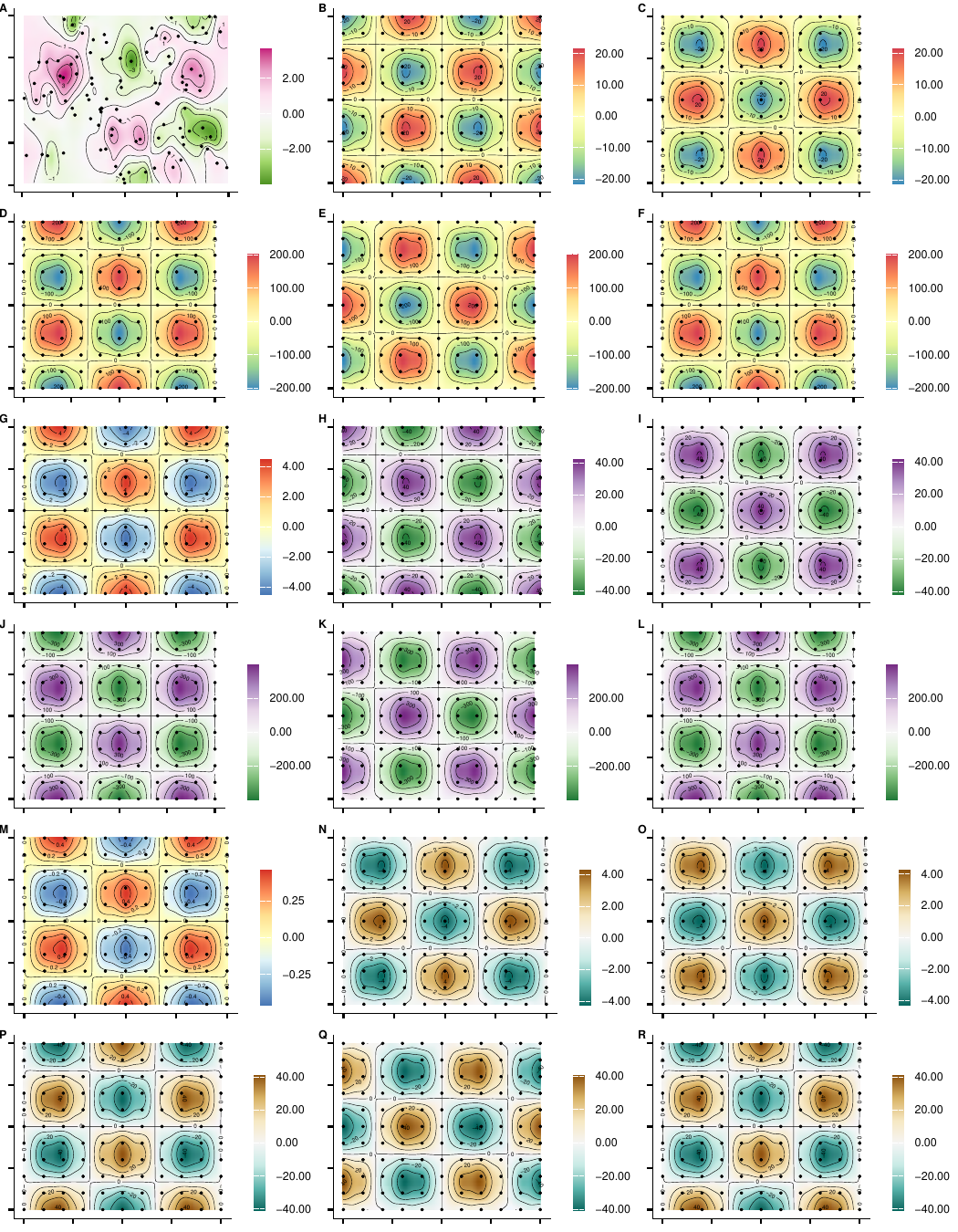}
	\caption{True spatiotemporal process, derivatives and curvature for $t=3$. (A) response from Pattern 2, (B) spatial gradient along $x$, (C) along $y$, (D) spatial curvature along $x$, (E) along $x$ and $y$ and (F) along $y$, (G--L) temporal gradients for processes in (A--F), (M--R) temporal curvature for processes in (A--F).}
	\label{fig:spt-grad-curv-2-3}
\end{figure}

\begin{figure}[H]
	\centering
	\includegraphics[scale=0.92]{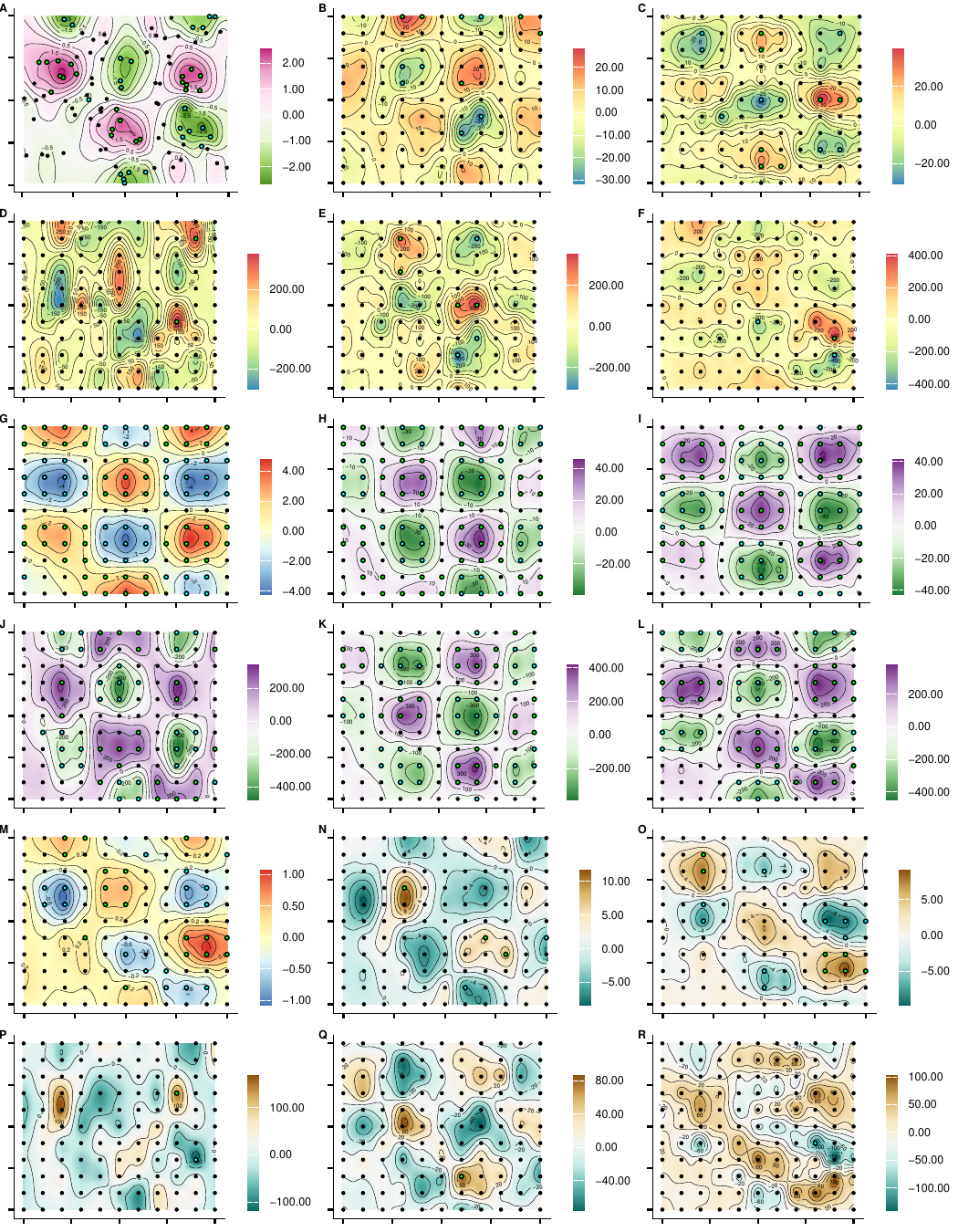}
	\caption{Estimated spatiotemporal process, derivatives and curvature for $t=3$. (A) response from Pattern 2, (B) spatial gradient along $x$, (C) along $y$, (D) spatial curvature along $x$, (E) along $x$ and $y$ and (F) along $y$, (G--L) temporal gradients for processes in (A--F), (M--R) temporal curvature for processes in (A--F).}
	\label{fig:est-spt-grad-curv-2-3}
\end{figure}

\subsection{Differential Geometric Operators}\label{figsec:diff-geo}

\begin{figure}[H]
	\centering
	\includegraphics[scale=0.92]{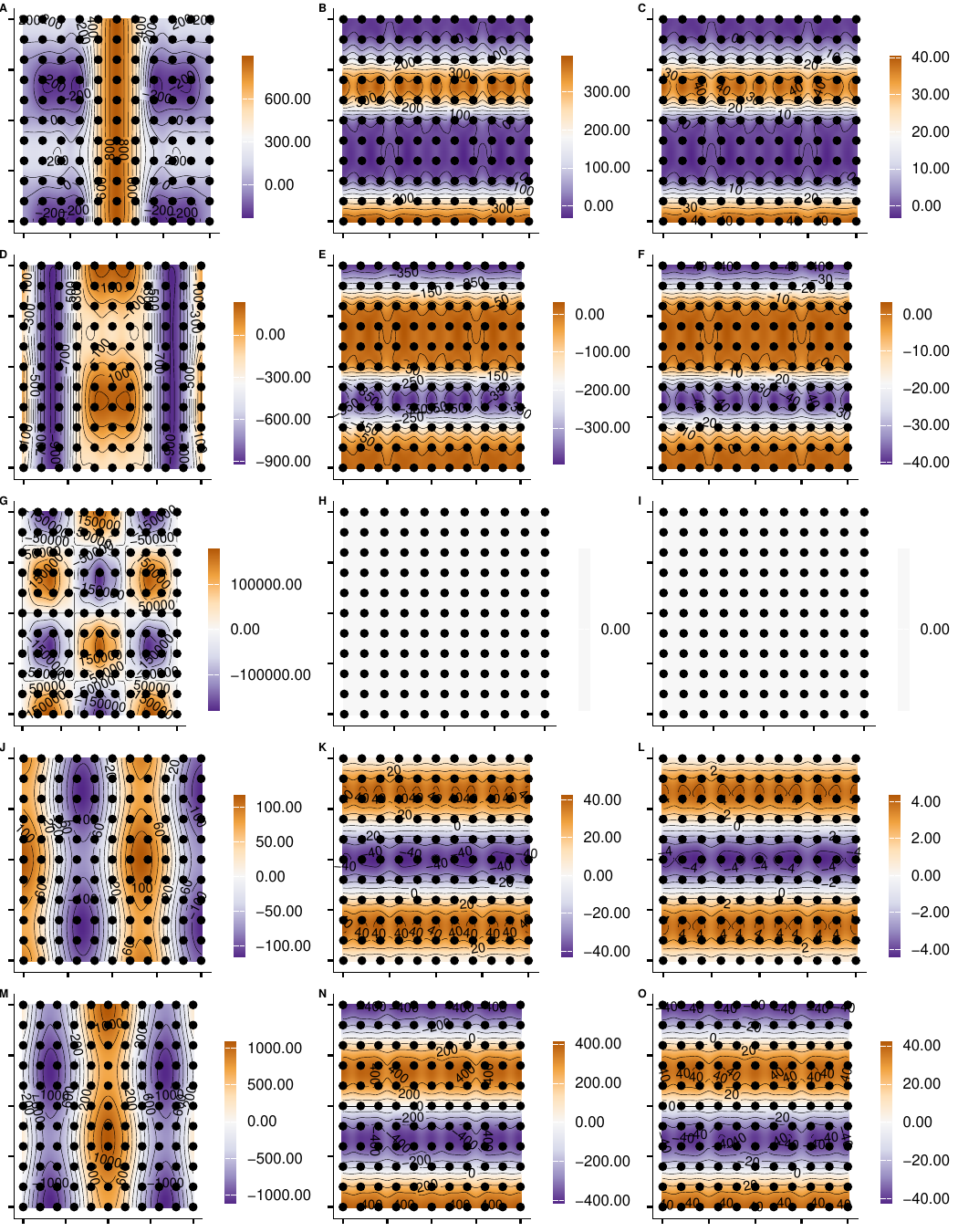}
	\caption{True spatiotemporal differential geometric processes for $t=3$. (A--C) Eigen-value 1, its temporal derivative and curvature (D--F) Eigen-value 2 its temporal derivative and curvature (G--I) Gaussian curvature, its temporal derivative and curvature (J--L) Divergence (see pp. 10--11) (M--O) Laplacian (see pp. 10--11).}\label{fig:diff-ge-true-1-3}
\end{figure}
\begin{figure}[H]
	\centering
	\includegraphics[scale=0.92]{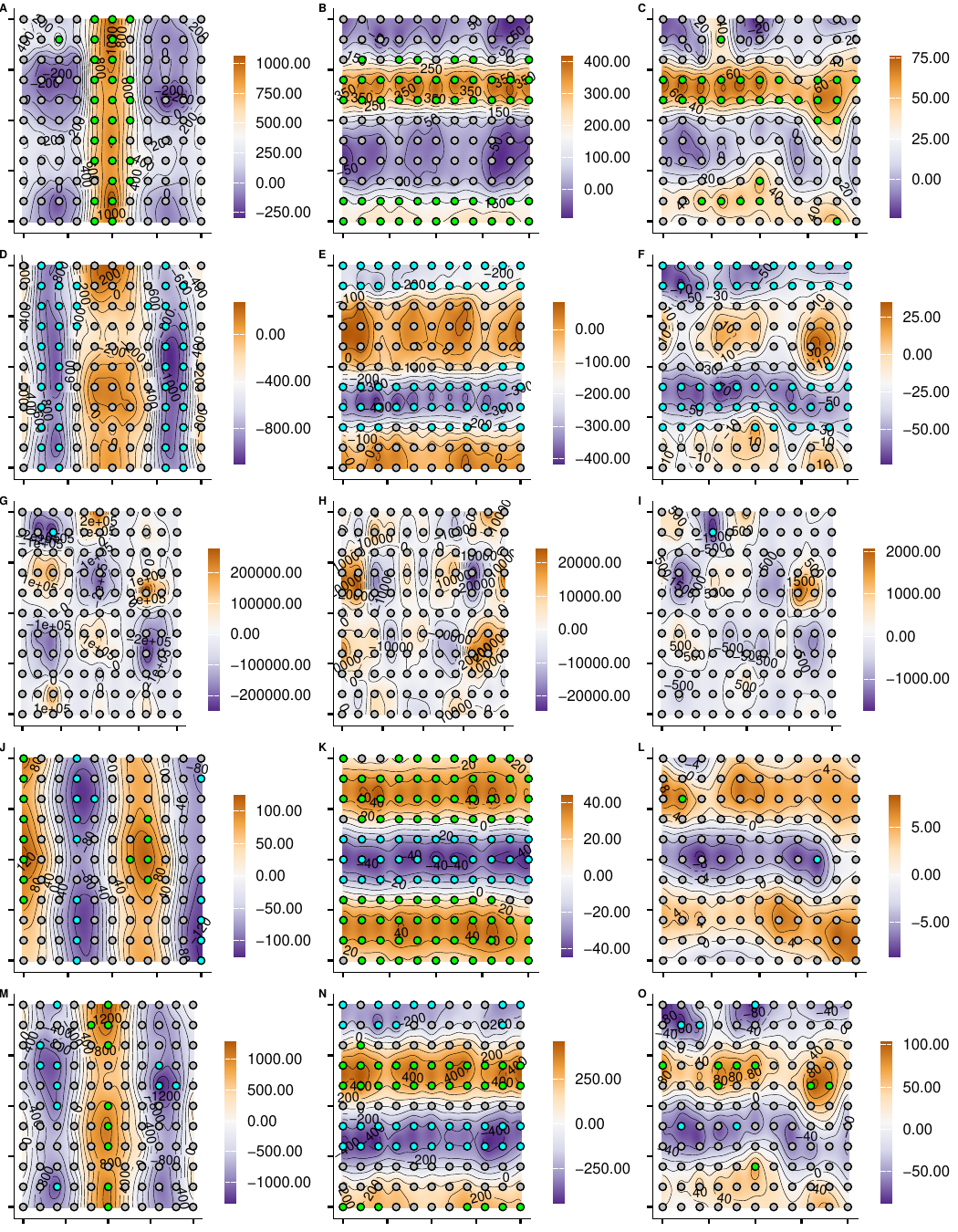}
	\caption{Estimated spatiotemporal differential geometric processes for $t=3$. (A--C) Eigen-value 1, its temporal derivative and curvature (D--F) Eigen-value 2 its temporal derivative and curvature (G--I) Gaussian curvature, its temporal derivative and curvature (J--L) Divergence (see pp. 10--11) (M--O) Laplacian (see pp. 10--11). Significant grid locations are in color---\emph{green} indicates significant positive and \emph{cyan} indicates significant negative values.}
	\label{fig:diff-ge-est-1-3}
\end{figure}

\subsection{Spatiotemporal Surface Wombling: Significance}
\begin{figure}[H]
\centering
\begin{subfigure}{\linewidth}
\includegraphics[scale = 0.19]{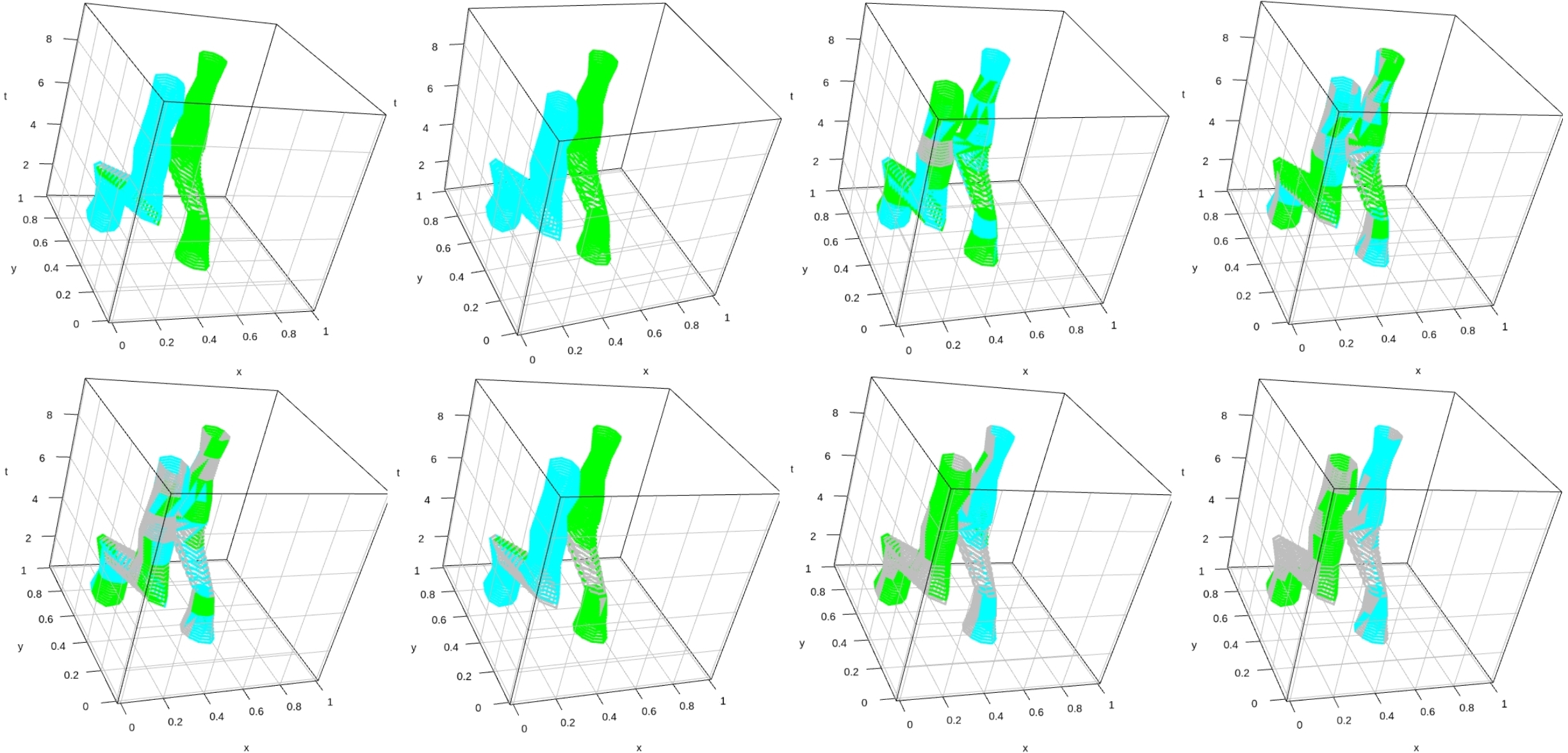}
\caption{Wombling Surfaces $A$ and $B$}\label{fig:sig-wombl-A-B}
\end{subfigure}

\vspace*{0.5in}

\begin{subfigure}{\linewidth}
\includegraphics[scale = 0.23]{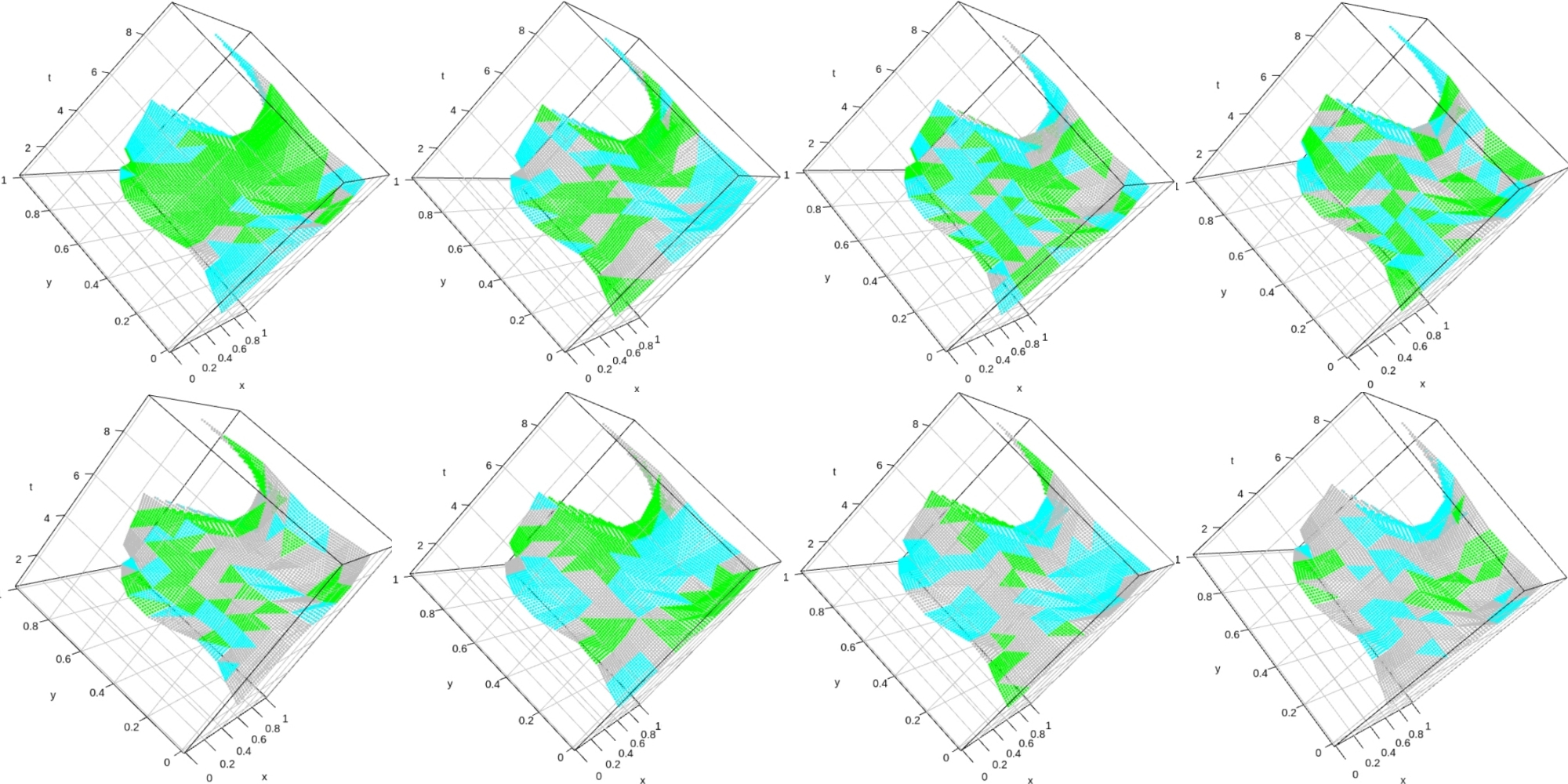}
\caption{Wombling Surface $C$}\label{fig:sig-wombl-C}
\end{subfigure}
\caption{Significance of spatiotemporal wombling measures for wombling surfaces (a) $A, B$ and (b) $C$ (see Section~5.3, fig.~4). (Row 1: left to right) spatial gradient, spatial curvature, temporal gradient, spatial-temporal gradient, (Row 2: left to right) temporal gradient in spatial curvature, temporal curvature, temporal curvature in spatial gradient, spatial-temporal curvature. Regions marked \texttt{green} (\texttt{cyan}) indicate positive (negative) significance while \texttt{grey} indicates no significance.}
\end{figure}

\subsection{Applications: Neuroimaging}
\begin{figure}[H]
\centering
\includegraphics[width=0.7\linewidth]{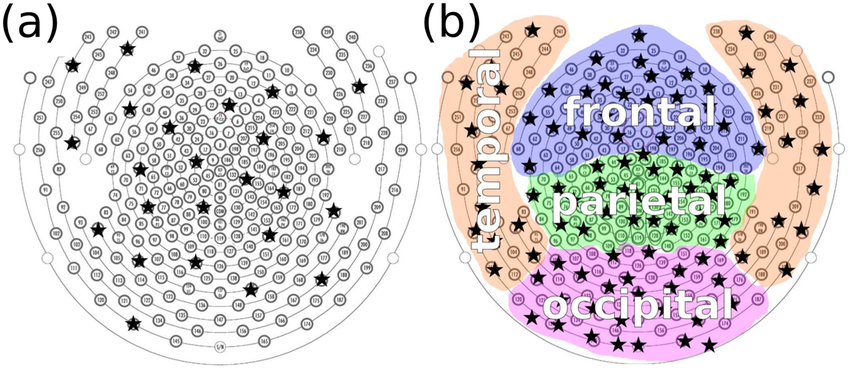}
\caption{Scalp regions overlaid over a 256 electrode EEG cap for recording measurements. Figure (a) shows the electrode markings/locations and (b) shows the regions.}
\label{fig:eeg-cap}
\end{figure}

\bibliographystyle{apalike}

\bibliography{sptwombling}

\end{document}